\documentclass[12pt,oneside,a4paper]{book}
\usepackage{setspace}
\doublespace

\usepackage[left=38mm,top=25mm,right=25mm,bottom=25mm]{geometry}
\def\DoDefinitions{true}

\def\DoTheorems{true}

\def\DoProofs{true}

\usepackage{float}
\usepackage{rotating}
\usepackage[latin1]{inputenc}
\usepackage{ifthen}
\usepackage{amsmath}
\usepackage{amsthm}
\usepackage{amssymb}
\usepackage{verbatim}
\usepackage{mathrsfs }
\usepackage{graphics}
\usepackage{graphicx}
\usepackage{hyperref}
\usepackage{calc}
\usepackage{fixltx2e}
\usepackage{bigints}
\usepackage{bbm}
\usepackage{dsfont}
\usepackage{cancel}
\usepackage[toc,page]{appendix}
\usepackage{color}
\usepackage{lineno}

\makeatletter
\DeclareRobustCommand*{\bfseries}{%
  \not@math@alphabet\bfseries\mathbf
  \fontseries\bfdefault\selectfont
  \boldmath
}
\makeatother

\newtheorem{jgbits}{HELLO}[section]

\ifthenelse{\equal{\DoTheorems}{true}}
{
  \theoremstyle{plain}
  \newtheorem{theorem}[jgbits]{Theorem}
  \newtheorem{lemma}[jgbits]{Lemma}
  \newtheorem{corrol}[jgbits]{Corollary}
}{
\newenvironment{theorem}
 {\begin{lrbox}\begin{minipage}{\textwidth}\addtocounter{jgbits}{1}}
 {\end{minipage}\end{lrbox}}
\newenvironment{lemma}
 {\begin{lrbox}\begin{minipage}{\textwidth}\addtocounter{lem}{1}}
 {\end{minipage}\end{lrbox}}
\newenvironment{corrol}
 {\begin{lrbox}\begin{minipage}{\textwidth}\addtocounter{cor}{1}}
 {\end{minipage}\end{lrbox}}
}

\ifthenelse{\equal{\DoProofs}{true}}
{\renewenvironment{proof}
{\noindent{\it Proof of \arabic{chapter}.\arabic{section}.\arabic{jgbits}}.}{\hfill$\Box$}}
{\renewenvironment{proof}
  {\begin{lrbox}{\junkbox}\begin{minipage}{\textwidth}\addtocounter{jgbits}{1}}
  {\end{minipage}\end{lrbox}}
}
\ifthenelse{\equal{\DoProofs}{true}}
{\newenvironment{proofs}
{\noindent{\it Proof of \Alph{chapter}.\arabic{section}.\arabic{jgbits}}.}{\hfill$\Box$}}
{\newenvironment{proofs}
  {\begin{lrbox}{\junkbox}\begin{minipage}{\textwidth}\addtocounter{jgbits}{1}}
  {\end{minipage}\end{lrbox}}
}

\ifthenelse{\equal{\DoDefinitions}{true}}
{
  \theoremstyle{definition}
  \newtheorem{definition}[jgbits]{Definition}
}{
\newenvironment{definition}
  {\begin{lrbox}\begin{minipage}{\textwidth}\addtocounter{def}{1}}
  {\end{minipage}\end{lrbox}}
}
\newenvironment{mapcode}
  {\begin{color}{red} \begin{singlespacing}\begin{flushleft}\begin{ttfamily}}
   {\end{ttfamily}\end{flushleft}\end{singlespacing}\end{color}}

\def\ccon{c}
\def\ask{\textasteriskcentered}
\def\ind{\textasciicircum}

\def\real{{\mathbb R}}

\def\norm#1{\left|#1\right|}
\def\innerprod(#1,#2){{\left<#1\,,\,#2\right>}}

\def\qe{\frac{q}{4 \pi \epsilon_0}}

\def\charge{q}

\def\l{\mathcal{L}}

\def\qquadand{\qquadtext{and}}
\def\quadtext#1{{\quad\text{#1}\quad}}
\def\qquadtext#1{{\qquad\text{#1}\qquad}}
\def\quadand{\quadtext{and}}

\def\funcf{f}
\def\funcg{g}

\def\point{x}
\def\atp{|_\point}

\def\ad{\dot{\alpha}}
\def\add{\ddot{\alpha}}
\def\atd{\dot{\alpha_{\theta}}}
\def\atdd{\ddot{\alpha_{\theta}}}
\def\apd{\dot{\alpha_{\phi}}}
\def\apdd{\ddot{\alpha_{\phi}}}
\def\atheta{\alpha_{\theta}}
\def\aphi{\alpha_{\phi}}

\def\P{\textup{P}}
\def\pdot{\dot{\textup{P}}}

\def\alw{\mathrm{A}}
\def\flw{\mathrm{F}}
\def\deg{\textup{deg}}

\def\stressk{\mathcal{S}_{\kkill}}
\def\stress{\mathcal{S}}
\def\stresst{\mathcal{T}}

\def\stresslw{\mathrm{S}}

\def\dega{p}
\def\degb{q}

\def\covec{\zeta}

\def\oneform{\nu}
\def\wonform{\omega}
\def\twoform{\Psi}

\def\oneformd{\widetilde{\nu}}
\def\wonformd{\widetilde{\omega}}

\def\aform{\alpha}
\def\bform{\beta}

\def\f{\mathcal{F}}
\def\fr{\flw_{\textup{R}}}
\def\fc{\flw_{\textup{C}}}

\def\jgen{\mathcal{J}}
\def\j{\mathrm{J}}

\def\jvec{j}

\def\pot{\mathcal{A}}
\def\e{\mathcal{E}}
\def\bvec{\mathcal{B}}
\def\b{\mathbf{B}}
\def\etilde{\widetilde{\e}}
\def\btilde{\widetilde{\bvec}}
\def\elw{\mathrm{E}}
\def\blw{\mathrm{B}}
\def\gund{\underline{\g}}

\def\rdcon{R_D^\circ}
\def\rcon{R^\circ}

\def\fret{\flw_{\textup{ret}}}
\def\fadv{\flw_{\textup{adv}}}

\def\fext{\flw_{\textup{ext}}}

\def \dxa{d y}
\def\co{y}

\def\xt{y^0}
\def\xx{y^1}
\def\xy{y^2}
\def\xz{y^3}

\def\dxt{d\xt}
\def\dxx{d\xx}
\def\dxy{d\xy}
\def\dxz{d\xz}

\def\xvec{\mathbf{x}}
\def\cvec{\mathbf{\c}}

\def\pitan{\pi}
\def\picot{\pi_c}

\def\mapMN{\phi}
\def\mapNO{\psi}
\def\maprmn{\phi}

\def\real{\mathbb{R} }
\def\m{\mathcal{M}}
\def\man{\mathbf{M}}
\def\manN{\mathbf{N}}
\def\manO{\mathbf{O}}
\def\tan{\textup{T}}

\def\dimM{m}
\def\dimN{n}

\def\mnotc{{(\m\backslash\c)}}

\def\fsmooth{\mathcal{F}}

\def\fourball{\mathbb{B}}

\def\tenst{\mathbf{T}}
\def\tenss{\mathbf{S}}
\def\tensr{\mathbf{R}}
\def\tensu{\mathbf{U}}
\def\n{n}
\def\w{W}
\def\x{X}
\def\xdual{\widetilde{\x}}
\def\y{Y}
\def\ydual{\widetilde{\y}}
\def\z{Z}

\def\vvec{V}
\def\vvecd{\widetilde{\vvec}}
\def\wvec{W}
\def\wvecd{\widetilde{\wvec}}

\def\fself{f_{\textup{self}}}
\def\florentz{f_{\textup{L}}}

\def\utilde{\widetilde{U}}
\def\u{U}

\def\xu{\underline{X}}
\def\ctau{\c(\tau_r)}

\def\v{V}
\def\vdual{\widetilde{V}}

\def\a{A}
\def\adual{\widetilde{A}}
\def\adot{\dot{A}}
\def\addot{\ddot{A}}
\def\adddot{\dddot{A}}
\def\adotdual{\widetilde{\dot{A}}}

\def\au{\underline{A}}

\def\c{C}

\def\cddot{\ddot{\c}}
\def\cdddot{\dddot{\c}}

\def\cdot{\dot{\c}}

\def\g{g}
\def\gdual{\g^{-1}}

\def\fracra{\frac{R}{\alpha}}

\def\dR{\frac{\partial}{\partial R}}

\def\dtau{\frac{\partial}{\partial \tau}}
\def\dth{\frac{\partial}{\partial \theta}}
\def\dph{\frac{\partial}{\partial \phi}}

\def\dt{\frac{\partial}{\partial \xt}}
\def\dx{\frac{\partial}{\partial \xx}}
\def\dy{\frac{\partial}{\partial \xy}}
\def\dz{\frac{\partial}{\partial \xz}}

\def\SigmaT{\Sigma_{\textup{T}}}
\def\Pdot{\pdot}

\def\ord{\mathcal{O}}
\def\constn{\kappa}
\def\tDirac{{\textup{D}}}

\def\ep{\epsilon_0}

\def\LW{Li$\text{\'{e}}$nard-Wiechert }

\def\VE{{\boldsymbol E}}
\def\VB{{\boldsymbol B}}
\def\VX{{\boldsymbol X}}
\def\Vn{{\boldsymbol n}}
\def\Vbeta{{\boldsymbol \beta}}
\def\Vx{{\boldsymbol x}}

\def\tauhat{\hat\tau}

\def\THat{\hat T}

\def\VEE{{\boldsymbol{E_0}}}

\def\VET{{\boldsymbol{E}}_{\textup{Tot}}}
\def\Va{{\boldsymbol a}}

\def\rhoLab{{\rho_{\textup{Lab}}}}
\def\tstraight{{t_\text{s}}}

\def\rhat{\hat{\rnew}}
\def\that{\hat{\theta}}
\def\phat{\hat{\phi}}

\def\VE{{\boldsymbol E}}

\def\VB{{\boldsymbol B}}

\def\VX{{\boldsymbol X}}
\def\Vn{{\boldsymbol n}}
\def\Vbeta{{\boldsymbol \beta}}
\def\Vx{{\boldsymbol x}}

\def\tauhat{\hat\tau}

\def\THat{\hat T}

\def\VEE{{\boldsymbol{E_0}}}
\def\VBB{{\boldsymbol{B_0}}}

\def\Va{{\boldsymbol a}}

\def\Lab{{{\textup{Lab}}}}
\def\coh{{\textup{coh}}}
\def\inc{{\textup{inc}}}
\def\cts{{\textup{cts}}}

\def\norm#1{\left\|#1\right\|}
\def\innerprod(#1,#2){{\left<#1\,,\,#2\right>}}
\def\Exx#1{{\left<#1\right>}}
\def\ExxOP#1{{\left<#1\right>_\textup{1P}}}

\def\rmag{\boldsymbol r}
\def\Xu{\VX}
\def\xu{\Vx}
\def\nund{\Vn}
\def\mund{\underline{\m}}
\def\betau{\Vbeta}
\def\au{\Va}
\def\elwd{\widetilde{\elw}}
\def\blwd{\widetilde{\blw}}

\def\vx{\underline{x}}
\def\vxd{\underline{\dot{x}}}
\def\vxdd{\underline{\ddot{x}}}
\def\vxddd{\underline{\dddot{x}}}

\def\rdcon{{R}_{\textup{D}0}}
\def\rcon{{R}_0}

\def\rrad{R}
\def\rnew{\mathsf{r}}

\def\totimes{\textstyle{\bigotimes}}
\def\kkill{\textup{K}}
\makeatletter

\newcommand{\Rmnum}[1]{\expandafter\@slowromancap\romannumeral #1@}
\makeatother

\usepackage{fancyhdr}
\fancyhfoffset[L,R]{\marginparsep+\marginparwidth}
\begin{document}
\begin{titlepage}

\begin{center}
\begin{minipage}{0.95\textwidth}
\vspace{1cm}
\begin{center}{\bf \huge

Problems in point charge electrodynamics

} \vspace{3cm}

{\bf \Large Michael Raymond Ferris}
\end{center}
\end{minipage}\vspace{1cm}
{\large Submitted for the degree of Doctor of Philosophy}\\
\vspace{8cm}

\begin{figure}[h]
\begin{center}
\includegraphics[width=0.3\textwidth]{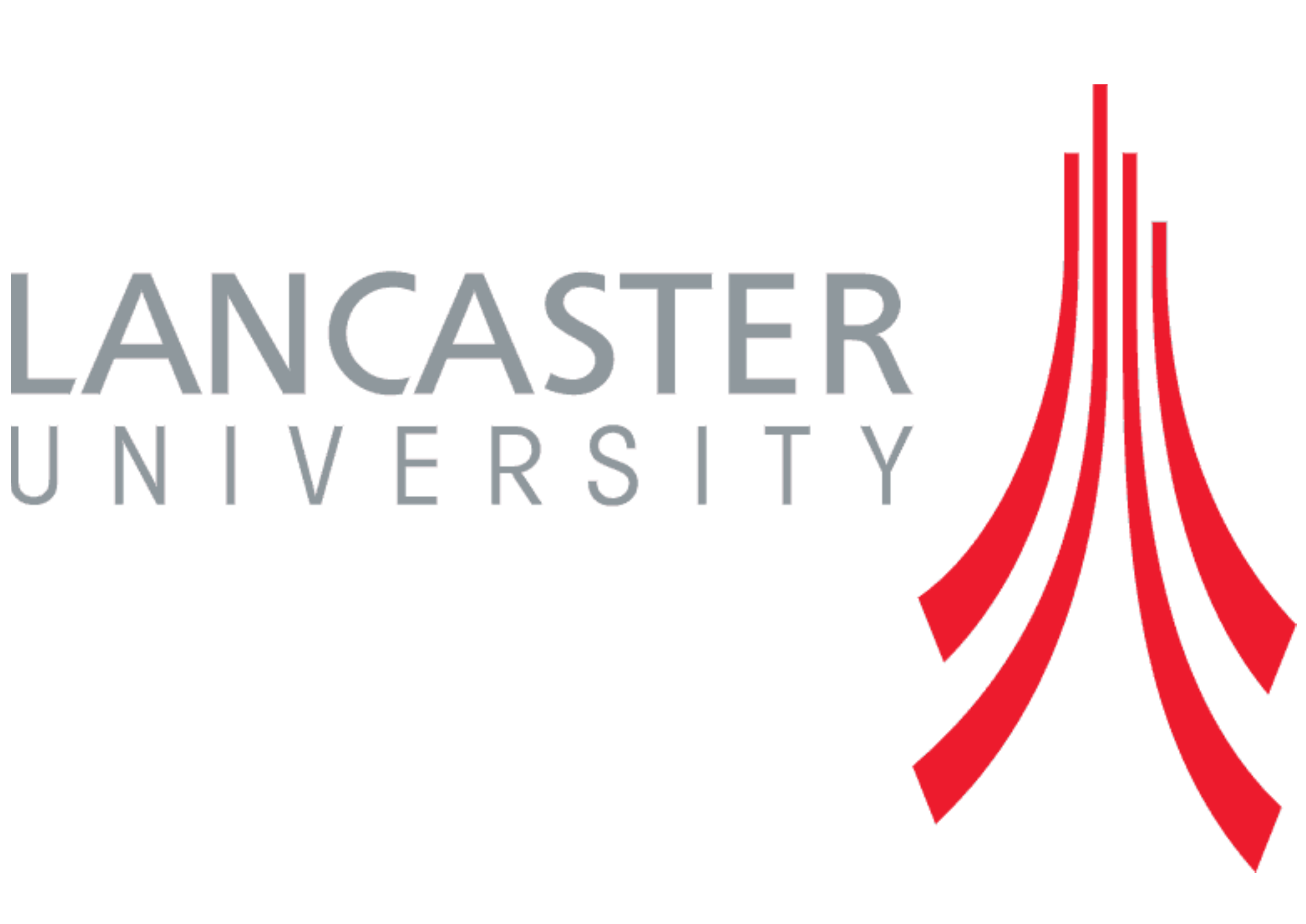}
\end{center}
\end{figure}
{\large Department of Physics
\\ Lancaster University,  August 2012}
\end{center}
\end{titlepage}

\begin{frontmatter}


 \newpage



\chapter*{Abstract}

This thesis consists of two parts.   In part I we consider  a discrepancy in the derivation of the electromagnetic self force for a point charge. The self force is given by the Abraham-von Laue vector, which consists of the radiation reaction term proportional to the $4$-acceleration, and the Schott term proportional to the $4$-jerk. In the point charge framework the self force can be defined as an integral of the \LW stress 3-forms over a suitably defined worldtube. In order to define such a worldtube it is necessary to identify a map which associates  a unique point along the worldline of the source with every field point off the worldline. One  choice of map is the Dirac time, which gives rise to a spacelike displacement vector field and a Dirac tube with spacelike caps. Another choice is the retarded time, which gives rise to a null displacement vector field and a Bhabha tube with null caps. In previous calculations which use the Dirac time the integration yields the complete self force, however in previous calculations which use the retarded time the integration produces only the radiation reaction term and the Schott term is absent. We show in this thesis that the Schott term may be obtained using a null displacement vector providing certain conditions are realized.\\[0.2cm]

 Part II comprises an  investigation into a problem in accelerator physics. In a high energy accelerator the cross-section of the beampipe is not continuous and there exist geometric discontinuities
 such as collimators and cavities.  When a  relativistic bunch of particles passes such a discontinuity  the field generated by a leading charge can interact with the wall and consequently affect the motion of trailing charges. The fields acting on the trailing charges are known as (geometric) wakefields. We model a bunch of particles as a one dimensional continuum of point charges and by calculating the accumulated \LW fields we address the possibility of reducing wakefields at a collimator interface by altering the path of the
  beam prior to collimation. This approach is facilitated by the highly relativistic regime in which lepton accelerators operate, where the Coulomb field given from the \LW potential is highly collimated
  in the direction of motion. It will be seen that the potential reduction depends upon the ratio of the bunch length to the width of the collimator aperture as well as the relativistic factor and  path of the beam. Given that the aperture of the collimator is generally on the order of millimetres we will see that for very short bunches, on the order of hundredths of a picosecond, a significant reduction is achieved.

\newpage
\chapter*{Author's declaration}
I declare that the original ideas contained in this thesis are the result of my own work conducted in collaboration
with my supervisor Dr Jonathan Gratus.  An article   based on the ideas in Part I has been published in Journal of Mathematical Physics (JMP) \cite{FerrisGratus11}.  A letter describing the key results of Part II has been submitted to European Physical Letters (EPL)\cite{GratusFerris11}.
\newpage
\chapter*{Acknowledgements}
I would like to express my very great appreciation to my supervisor Jonathan Gratus for his guidance
throughout my time at Lancaster and for his expertise and enthusiasm in our
 many discussions. I would like to thank other members of the Mathematical
Physics Group for their hospitality and willingness to discuss any ideas of interest. I would also like to thank the wider Department of Physics at Lancaster University for welcoming me and creating a friendly atmosphere
in which to live and work.

 I would like to thank the Cockcroft Institute for the many interesting lectures and discussions on topics in accelerator physics, and I would especially like to thank STFC for providing the funding to make this PhD thesis possible.
 I would like to offer my special thanks to my examiners David Burton and Andy Wolski for their extremely useful comments.

  Finally I would like to thank Uma Athale for her support and patience during the writing of this thesis, the many incredible people who have offered me their friendship during my time as a student, and my family, for their encouragement and continuous support in whatever I pursue.  

\tableofcontents{\markboth{\slshape{Contents}}{}}
\newpage
\phantomsection

\listoffigures
\addcontentsline{toc}{chapter}{List of Figures and Tables}
\begingroup
\let\clearpage\relax
\listoftables
\endgroup

\end{frontmatter}


\renewcommand{\headrulewidth}{0.0005cm}
\renewcommand{\footrulewidth}{0.cm}
\fancyhfoffset[L,R]{\marginparsep+\marginparwidth}
\fancyhead[LE,RO]{\slshape \rightmark}
\fancyhead[RE, LO]{\slshape \leftmark}
\fancyfoot[C]{\thepage}
\renewcommand{\sectionmark}[1]{\markboth{#1}{}}
\renewcommand{\subsectionmark}[1]{\markright{#1}{}}

\begin{mainmatter}

\chapter*{Guide to Notation}
\addcontentsline{toc}{chapter}{Guide to Notation}
All fields will be regarded as sections of tensor bundles over
appropriate domains of Minkowski space $\m$. Sections of the tangent bundle over $\m$
will be denoted $\Gamma \tan \m$ while sections of the bundle of exterior
$\dega$-forms will be denoted $\Gamma \Lambda^\dega \m$. Given a single worldline $\c$ in free space sections over the whole of spacetime
excluding the worldline will be written $\Gamma \tan \mnotc$ and
$\Gamma\Lambda^\dega \mnotc $.  We use the SI unit convention. Appendix \ref{app_dimensions} provides  a brief summary of the dimensions of various mathematical objects.

\chapter{Introduction}
\label{maxlorentz}
In this chapter we introduce the defining characteristics of Minkowski space, namely the metric and the affine structure, and the fundamental equations of Maxwell-Lorentz electrodynamics. We use the term \emph{Maxwell-Lorentz electrodynamics} to denote the microscopic vacuum Maxwell equations, first derived by Lorentz from the macroscopic Maxwell equations (see \cite{deGroot72, Rohrlich65}) and often called the \emph{Maxwell-Lorentz equations}, together with the Lorentz force equation. We introduce the  general form of the electromagnetic stress $3$-forms and show that they give rise to a set of conservation laws. A brief introduction to the necessary mathematics can be found in Appendix \ref{diffgeom}.

\section{Minkowski space}
\begin{definition}
\label{def_metric}
Minkowski space is the pseudo-Euclidean space defined by the pair $(\m, \g)$, where $\m$ is the four dimensional real vector space $\real^4$ and $\g$ is the Minkowski metric.
With respect to a global Lorentzian coordinate basis $(\co^0, \co^1, \co^2, \co^3)$  on $\m$ the Minkowski metric is defined by
\begin{align}
g= - \dxt\otimes \dxt+\dxx \otimes \dxx +\dxy \otimes \dxy +\dxz \otimes \dxz.
\label{gmink_def}
\end{align}
\end{definition}
\begin{lemma}
\label{lem_za}
Given a new set of coordinates $(z^0, z^1, z^2, z^3)$ on $\m$ we write $\g$ in terms of the new basis using the transformations (\ref{cov_trans}),
\begin{align}
\g=\g_{ab}d\co^a\otimes d\co^b&=\g_{ab} \frac{\partial \co^a}{\partial z^c}\frac{\partial \co^b}{\partial z^d} d z^c \wedge d z^d=\g^{(z)}_{cd}  d z^c \wedge d z^d
\end{align}
where
\begin{align}
\g^{(z)}_{cd}=\g_{ab} \frac{\partial \co^a}{\partial z^c}\frac{\partial \co^b}{\partial z^d}.
\label{gmat}
\end{align}
\end{lemma}
\begin{lemma}
\label{lem_starone}
The coordinate basis $(\xt, \xx, \xy, \xz)$ on $\m$ naturally gives rise to the basis $(\dxt, \dxx, \dxy, \dxz)$ of $1-$forms on $\tan^\ast \m$. This forms a $\g$-orthonormal basis and from definition \ref{def_starone} it follows
\begin{align}
\star 1 = \dxt\wedge \dxx \wedge \dxy \wedge \dxz
\end{align}
is the volume form on $\m$. In terms of a different coordinate basis $(z^0, z^1, z^2, z^3)$  it is given by
\begin{align}
\star 1= \sqrt{|\textup{det}(\g^z)|}d z^0\wedge d z^1 \wedge dz^2 \wedge d z^3
\label{starone_basis}
\end{align}
where $\g^z$ is the matrix of metric components $\g^{(z)}_{ab}$ defined by (\ref{gmat}).
\end{lemma}

\begin{definition}
\label{def_affine}
Let $V$ be a vector field over set of points $M$. If the  map
\begin{align}
M \times M \longrightarrow V, \qquad (x, y)\longrightarrow x-y,
\end{align}
exists and satisfies
\begin{align}
1.&\quad \textup{For all}\quad  x \in M,\quad \textup{for all}\quad v \in V \notag\\
&\quad \textup{there exists} \quad y\in M\quad  \textup{such that} \quad y-x= v,\notag\\
2.& \quad \textup{For all}\quad x, y, z \in M, \quad (x-y) + (z-x) = z-y,
\end{align}
 then $M$ is an affine space.
 \end{definition}

For any integer $n$ the space $\real^n$ is affine. It follows that Minkowski space is an affine space.
In some calculations it will be necessary to endow Minkowski space with an origin, thus transforming it into a vector space, however the results of such calculations will not depend on the vector space structure but only the affine structure.

\section{Maxwell-Lorentz Equations}

 The equations which describe the interaction between matter and the electromagnetic field were first formulated by Maxwell in 1865\cite{Maxwell65}. Maxwell's equations form a continuum theory of electrodynamics due to their origins in macroscopic experiment. In this thesis we are interested in the interaction of point charges and their fields, therefore we need equations which are valid on the microscopic scale.

\begin{definition}
\label{max_lor}
The Maxwell-Lorentz equations, or the microscopic vacuum Maxwell equations are given by
\begin{align}
d\f&=0\label{dF}\\
\epsilon_0 d\star \f&=\jgen\label{dstarF}
\end{align}
where $\f\in\Gamma\Lambda^2\m$ is the electromagnetic $2$-form,  $\jgen\in\Gamma\Lambda^3\m$ is the current $3$-form and $\epsilon_0$ is the permittivity of free space.
If we introduce the $1$-form potential
\begin{align}
\pot \in\Gamma\Lambda^1\m \qquad \textup{such that} \qquad \f=d\pot,
\label{A_maxwell}
\end{align}
then (\ref{dF}) and (\ref{dstarF}) reduce to the single equation
\begin{align}
\qquad\epsilon_0 d\star d\pot&=\jgen,
\label{A_maxwelltwo}
\end{align}
where (\ref{dF}) is satisfied automatically because the double action of the exterior derivative is zero.
\end{definition}

\begin{definition}
\label{max_lor_dist}
In terms of distributional forms (see \ref{dist_form_def})
\begin{align}
\pot^D \in\Gamma_D\Lambda^1\m, \qquad \f^D=d\pot^D,
\label{dist_maxwell}
\end{align}
Maxwell's second equation is given by
\begin{align}
\qquad\epsilon_0 d\star d\pot^D [\varphi]&=\jgen^D[\varphi]=\int_{\m} \varphi \wedge \jgen .
\label{A_maxwell_dist}
\end{align}
where $\varphi\in\Gamma_0 \Lambda^1\m$ is any test $1$-form (see \ref{def_test_form}).
\end{definition}
\subsection*{$3+1$ decomposition}
\begin{definition}
\label{def_split_one}
Given any velocity vector field $\u \in \Gamma \textup{T} \m$ satisfying\\
 $\g(\u, \u)=-1$, the electromagnetic $2$-form $\f$ may be written
\begin{align}
\f=\etilde \wedge \utilde + \ccon\b,
\label{f_def}
\end{align}
where $\etilde \in \Gamma \Lambda^1 \m$ and $\b\in \Gamma \Lambda^2\m$ are the electric $1$-form and magnetic $2$-form associated with $\u$ and $\f$, and satisfy
\begin{align}
i_{\u} \etilde=i_{\u} \b=0.
\end{align}
Here $\quad\widetilde{\quad}\quad$ is the metric dual operator defined by (\ref{dual_def}) and $\ccon$ is the speed of light in a vacuum.
\end{definition}
\begin{lemma}
According to observers whose worldlines coincide with integral curves of $\u$, the electric field $\e\in \Gamma \tan \m$ is given by
\begin{align}
\e=&\widetilde{i_{\u}\f},\notag\\
\label{eb_def}
\end{align}
\end{lemma}
\begin{proof}
Follows trivially from (\ref{f_def})
\end{proof}
\begin{definition}
\label{def_gthree}
We may use the vector field $\u$ to write the Minkowski metric $g$ in terms of a metric $\gund$ on the instantaneous 3-spaces
\begin{align}
\g=-\utilde \otimes \utilde + \gund.
\end{align}
Let $\#$ be the Hodge map associated with the instantaneous 3-space such that for $\alpha\in \Gamma\Lambda^p \m$
  \begin{align}
  \#: \Gamma\Lambda^p \m \rightarrow \Gamma \Lambda^{3-p}\m, \qquad \alpha \mapsto \# \alpha = (-1)^{p+1} i_\u \star \alpha
  \label{starthree}
  \end{align}
The Minkowski Hodge dual is then given by
\begin{align}
\star \alpha=(-1)^p \utilde \wedge \#\alpha.
\label{starfour}
\end{align}
\end{definition}
\begin{lemma}
The Hodge dual of $\f$ is given by
\begin{align}
\star \f= \#\etilde-\ccon\#\b\wedge\utilde
\label{star_f_three}
\end{align}
\end{lemma}
\begin{proof}\\
\begin{align}
\star \f =\star( \etilde \wedge \utilde) + \ccon\star\b,
\end{align}
It follows from (\ref{starthree}) that if $\alpha\in \Lambda^1\m$ then $\#\alpha=i_{\u}\star\alpha=\star(\alpha\wedge\utilde)$. Thus $\star( \etilde \wedge \utilde)= \#\etilde$. Similarly it follows from (\ref{starfour}) that if $\beta\in \Lambda^2\m$ then $\star\beta=\utilde\wedge \#\beta$, thus $\star \b=\utilde\wedge \#\b$.
\end{proof}
\begin{lemma}
\label{mag_def}
Let $\widetilde{\bvec}=-\#\b$ where $\bvec\in\Gamma \tan \m$ is the magnetic field, then according to observers whose worldlines coincide with integral curves of $\u$ it is given by
\begin{align}
\bvec=\frac{1}{\ccon}\widetilde{i_{\u}\star\f}, \qquad .
\label{bvecdef}
\end{align}
\end{lemma}
\begin{proof}
Consider (\ref{star_f_three}). Since $i_{\u} \#\aform=i_{\u} \# \bform=0$ it follows that $i_{\u} \star \f=-\ccon\#\b$.
\end{proof}
\begin{lemma}
In terms of $\e$ and $\bvec$ the 2-forms $\f$ and $\star \f$ are given by
\begin{align}
\f=&\etilde \wedge \utilde -c\#\widetilde{\bvec}=\etilde \wedge \utilde -\ccon\star(\widetilde{\bvec}\wedge \utilde)\label{f_plus},\\
\star \f=&\#\etilde+\ccon\widetilde{\bvec}\wedge\utilde= \star(\etilde\wedge \utilde)+\ccon\widetilde{\bvec}\wedge\utilde\label{starf_plus}.
\end{align}
 \end{lemma}
\begin{proof}
Since $\widetilde{\bvec}=-\#\b$ and $\#\#\b=\b$ it follows that $\#\widetilde{\bvec}=-\b$. Substituting this into (\ref{f_def}) yields (\ref{f_plus}). Similarly substituting the first relation into
(\ref{star_f_three}) yields  (\ref{starf_plus}).
\end{proof}
\subsection*{The Lorentz Force}
\begin{definition}
\label{lorentz_def}
Let $C:I\subset\real\to\m$ be the proper time parameterized inextendible worldline of a point particle with
observed rest mass $m$ and charge $q$. For $\tau\in I$
\begin{align}
\cdot=\c_{\ast}(d/d \tau),\qquad \cddot=\nabla_{\cdot}\cdot,\qquad \textup{and} \qquad \cdddot=\nabla_{\cdot} \nabla_{\cdot} \cdot
\end{align}
are the velocity, acceleration and jerk of the particle respectively. Here the pushforward map ${}_{\ast}$ is defined by (\ref{pushstart})-(\ref{pushend}) and  $\nabla$ is the Levi-Civita connection (see \ref{connection}).
In this introductory chapter and in Part II we assign the dimension of time to proper time $\tau$  such that 
\begin{align}
g(\cdot,\cdot)=-\ccon^2,
\label{gCdCd}
\end{align}
However the reader should note that Part I we will find it convenient to assign the dimension of length to proper time  so that
\begin{align}
g(\cdot,\cdot)=-1,
\label{gCdCdtwo}
\end{align}
For further details about dimensions see appendix \ref{app_dimensions}.
\end{definition}
\begin{lemma}
\label{orthog}
\begin{align}
&g(\cdot, \cddot)=0,\label{g_Cd_Cdd}\\
\qquad \textup{and} \qquad
&g(\cdot, \cdddot)=-g(\cddot, \cddot).\label{g_Cd_Cddd}
\end{align}
\end{lemma}
\begin{proof}
Equation (\ref{g_Cd_Cdd}) follows by differentiating (\ref{gCdCd}) with respect to $\tau$. Similarly, equation (\ref{g_Cd_Cddd}) follows by differentiating (\ref{g_Cd_Cdd}).
\end{proof}
\begin{definition}
\label{lorentz_deftwo}
The force on a point particle with worldline $\c(\tau)$ due to an external field $\fext\in\Gamma\Lambda^2\m$ is given by the Lorentz force $\florentz$, where
\begin{align}
\florentz\in\Gamma \tan \m,\qquad \florentz=\frac{\charge}{\ccon} \widetilde{i_{\cdot}\fext}.
\end{align}
\end{definition}
In 1916 Lorentz writes\cite{Lorentz16}
\begin{quote}
Like our former equations [Maxwell's equations], it is got by generalizing the results of electromagnetic experiments
\end{quote}

\section{Conservation Laws}
\label{chap_stress}

\begin{definition}
\label{kill_def}
A vector field $\vvec$ is a Killing field if it satisfies
\begin{align}
\l_{\vvec}  \g=0.
\label{Killing_def}
\end{align}
In terms of coordinate basis $\{y^i\}$ the metric may be written  $\g=g_{a b}(y^i) dy^a \otimes dy^b$, thus for vector field $\vvec=\frac{\partial}{\partial y^a}$ the left hand side of (\ref{Killing_def}) yields
\begin{align}
\l_{\partial_{\kkill}}  \g= \frac{\partial g_{a b}}{\partial y^{\kkill}} dy^a \otimes dy^b,
\end{align}
where $\partial_\kkill=\tfrac{\partial}{\partial y^\kkill}$. In Minkowski space $g_{0 0}=-1$ and $g_{a b} =\delta^a_b$ for $a=1, 2, 3$. Thus for the four translational vectors $\dt, \dx, \dy, \dz$ (\ref{Killing_def}) is trivially satisfied. In fact there are 10 killing vector fields on Minkowski space.

Let $\vvec$ be a Killing vector, then another property of Killing fields we shall use is
\begin{align}
\l_{\vvec} \star =\star \l_{\vvec},
\label{kill_star}
\end{align}

\end{definition}
\begin{definition}
\label{stressk_def}
The electromagnetic stress $3$-forms $\stressk \in \Gamma \Lambda^3 \m$ are given by
\begin{align}
\stressk=&\frac{\ep}{2 \ccon}\big(i_{\partial_{\kkill} }\f\wedge\star\f-i_{\partial_{\kkill} }\star\f\wedge\f\big)
\label{stress_def}
\end{align}
where $\partial_\kkill=\tfrac{\partial}{\partial y^\kkill}$ are the four translational Killing vectors. These 3-forms can be obtained from the Lagrangian density for the electromagnetic field using Noether's theorem, see \cite{ Obukhov} for a detailed exposition.
\label{stresst_lem}
The stress forms are
related to the symmetric stress-energy-momentum tensor $\displaystyle{\stresst \in \Gamma \totimes^{[\mathds{V}, \mathds{V}]} \m}$ by
\begin{align}
\stresst^{a \kkill}=  i_{\frac{\partial}{\partial y^a}} \star \stressk, \qquad \stressk=\star\Big(\big(\stresst(d y^{\kkill}, -)\big)\widetilde{\,\,\,}\Big)
\label{stresst_def}
\end{align}
where $\displaystyle{\stresst=\stresst^{a b} \frac{\partial}{\partial y^a}\otimes  \frac{\partial}{\partial y^b}}$.
\end{definition}
\begin{lemma}
The stress forms satisfy
\begin{align}
d \stressk=-\frac{1}{\ccon}i_{\partial_{\kkill} } \f\wedge \jgen,
\label{dtauk_def}
\end{align}
and thus for any source free region $N\subset\m$
\begin{align}
d \stressk=0.
\label{dtauk_0}
\end{align}
\end{lemma}
\begin{proof}
\begin{align}
d\stressk=&\frac{\ep}{2\ccon}d\big(i_{\partial_{\kkill} }\f\wedge\star\f-i_{\partial_{\kkill} }\star\f\wedge\f\big)\notag\\
=&\frac{\ep}{2\ccon}\big(d i_{\partial_{\kkill} }\f\wedge\star\f-i_{\partial_{\kkill} } \f\wedge d \star\f-d i_{\partial_{\kkill} }\star\f\wedge\f+i_{\partial_{\kkill} }\star\f\wedge d\f\big)\notag\\
=&\frac{\ep}{2\ccon}\big(d i_{\partial_{\kkill} }\f\wedge\star\f-d i_{\partial_{\kkill} }\star\f\wedge\f\big)-\frac{\ep}{2\ccon}i_{\partial_{\kkill} } \f\wedge d \star\f.\label{dtau_k}
\end{align}
From (\ref{dF}) and (\ref{cartan}) it follows that
\begin{align}
\l_{\partial_{\kkill} }\f= di_{\partial_{\kkill} }\f.\label{lief}
\end{align}
Using (\ref{lief}), (\ref{hodgeab}) and (\ref{kill_star}) respectively yields
\begin{align}
d i_{\partial_{\kkill} }\f\wedge\star\f=  \f \wedge \star d i_{\partial_{\kkill} }\f = \f \wedge \star \l_{\partial_{\kkill}}\f =\f \wedge  \l_{\partial_{\kkill}} \star\f = \f \wedge d i_{\partial_{\kkill} }\star \f+\f \wedge  i_{\partial_{\kkill}} d \star \f
\label{dikf}
\end{align}
Substituting (\ref{dikf}) into  (\ref{dtau_k}) yields
\begin{align}
d\stressk= &\frac{\ep}{2\ccon}\big(\f \wedge d i_{\partial_{\kkill} }\star \f+\f \wedge  i_{\partial_{\kkill}} d \star \f-d i_{\partial_{\kkill} }\star\f\wedge\f\big)-\frac{\ep}{2\ccon}i_{\partial_{\kkill} } \f\wedge d \star\f\notag\\
=&\frac{\ep}{2\ccon}\big(\f \wedge  i_{\partial_{\kkill}} d \star \f-i_{\partial_{\kkill} } \f\wedge d \star\f\big)\label{eqstress}
\end{align}
Since $\f$ is a 2-form and $d\star\f$ is a 3-form it follows that
\begin{align}
i_{\partial_{\kkill} }(\f\wedge d\star \f)=  i_{\partial_{\kkill} } \f\wedge d \star\f+\f \wedge  i_{\partial_{\kkill}} d \star \f=0.
\end{align}
Thus substituting  $\f \wedge  i_{\partial_{\kkill}} d \star \f=-i_{\partial_{\kkill} } \f\wedge d \star\f$ and (\ref{dstarF}) into (\ref{eqstress}) yields result.
\end{proof}
\begin{lemma}
For any source free region $N\subset\m$
\begin{align}
\int_{\partial \mathcal{N}} \stress_{\kkill}  = \int_{\mathcal{N}}d \stress_{\kkill} =0.
\label{Stokes}
\end{align}
\end{lemma}
\begin{proof}
Follows trivially from (\ref{dtauk_0}) and Stokes' theorem (\ref{stokes_def}).
\end{proof}
\begin{lemma}
If $\u$ is a timelike Killing vector then applying the 3+1 decomposition yields
\begin{align}
\stress_{\u}=\ep\etilde\wedge\btilde\wedge \utilde +\frac{\ep}{2\ccon}(\etilde\wedge\#\etilde+\ccon^2\btilde\wedge \#\btilde),
\end{align}
where $\ep\etilde\wedge\btilde$ is the Poynting $2$-form, and $\tfrac{\ep}{2\ccon}(\etilde\wedge\#\etilde+\ccon^2\btilde\wedge \#\btilde)$ the energy density 3-form.
\end{lemma}
\begin{proof}\\
Using definition \ref{stressk_def}
\begin{align}
\stress_{\u}=&\frac{\ep}{2\ccon}\big(i_{\partial_{\u}}\f\wedge\star\f-i_{\partial_{\u}}\star\f\wedge\f\big)\notag
\end{align}
Substituting (\ref{starf_plus}) and (\ref{f_plus}) and using the relations (\ref{eb_def}) and (\ref{bvecdef})  yields
\begin{align}
\stress_{\u}=&\frac{\ep}{2\ccon}\big(\etilde\wedge(\#\etilde+\ccon\btilde\wedge\utilde)-\ccon\btilde\wedge(\etilde\wedge\utilde-\ccon\# \btilde)\big),\notag\\
=&\frac{\ep}{2\ccon}\big(\etilde\wedge\#\etilde +\ccon^2\btilde\wedge\#\btilde\Big) +\ep\etilde\wedge\btilde\wedge\utilde.\notag
\end{align}

\end{proof}

\section{The source $\jgen$ for a point charge}
We now  consider the particular form of the current $3$-form $\jgen \in  \Lambda^3\m$ for a point charge. We use notation $\j=\jgen_{\textup{point charge}}$ in order to emphasize that $\j$ is a particular choice for $\jgen$.   The source is located only on the worldline of the particle therefore we expect the source distribution $\j^D\in \Gamma_D \Lambda^3 \m$ to have the form of a Dirac delta distribution.
\begin{definition}
\label{def_jvec}
Given the four $0$-form distributions $\jvec^a \in \Gamma_D \Lambda^0 \m$, where for $x\in \m$
\begin{align}
\jvec^a(x) = \charge \int_\tau  \cdot^a(\tau) \delta^{(4)}(x - \c(\tau))d \tau,
\label{ja_def}
\end{align}
we define the distributional current vector field by
\begin{align}
 \jvec= \jvec^a(x) \frac{\partial}{\partial y^a}, \qquad
\end{align}
The distributions $\jvec^a(x)$ are non-zero only when $x=\c(\tau)$.
The $3$-form $\j\in \Lambda^3 \m$ is given by
\begin{align}
 \j=\star\widetilde{\jvec}.
 \label{J_def}
\end{align}
\end{definition}
\begin{lemma}
\label{jd_def}
The current $3$-form distribution $\j^D\in\Gamma_D \Lambda^3 \m$ is given by
\begin{align}
\j^{D}[\varphi] &=\charge\int_I \c^\ast \varphi
\end{align}
for any test $1$-form $\varphi\in \Gamma_0 \Lambda^1 \m$.
\end{lemma}
\begin{proof}\\
From (\ref{J_def}) the distribution $\j^D$ is given by
\begin{align*}
\j^{D} [\varphi] &=\int_\m \star \widetilde{\jvec}\wedge \varphi ,\\
&=\int_\m i_\jvec \star 1\wedge \varphi ,\\
&= \int_\m  \jvec^a i_{\frac{\partial}{\partial x^a}} \varphi \star 1,\\
&=\int_\m  \jvec^a  \varphi_a \star 1.
\end{align*}
Substitution of (\ref{ja_def}) yields
\begin{align*}
\j^{D} [\varphi] &= \charge\int_\m \int_\tau \cdot^a (\tau) \delta^{(4)} (x - \c(\tau))d\tau \varphi_a \star 1\\
&= \charge \int_\tau \cdot^a (\tau)\varphi_a (\c(\tau))d\tau \\
&=\charge\int \varphi_a (\c(\tau)) \frac{d \c^a}{d \tau} d\tau\\
&= \charge\int \varphi_a (\c(\tau))d \c^a\\
&=\charge \int \varphi_a (\c(\tau))\c^\ast(d \co^a)\\
&=\charge\int_I \c^\ast \varphi
\end{align*}
where $ \c^\ast(\co^a) = \co^a \circ \c = \c^a$.
\end{proof}


\section{Worldline geometry}
\label{worldline_geom}

 Given the proper time parameterized inextendible worldline
\begin{align}
\c &: I \subset \real \rightarrow \m, \quad \tau \mapsto \c(\tau),
\label{C_def}
\end{align}
we required a way to locally map each point $\point\in \mnotc$ to a  unique point $\c(\tau'(\point))$ along the worldline.
Consider the region $N=\widetilde{N}\backslash \c$ where $\widetilde{N} \subset \m$ is a local neighborhood of the worldline.
The affine structure of $\m$ permits the construction of a unique displacement vector $\z|_{\point}$ defined as the difference between the two points (see figure \ref{vecz}),
 \begin{align}
 \z|_{\point}=\point-\c(\tau'(\point)).
 \end{align}
Note that the definition of $\z$ only requires the affine structure of $\m$. It does not require $\m$ to be converted into a vector space by assigning an origin.
 \newpage
 \begin{figure}
\setlength{\unitlength}{1cm}
\centerline{
\begin{picture}(10, 10)
\put(0, 0){\includegraphics[ width=10\unitlength]{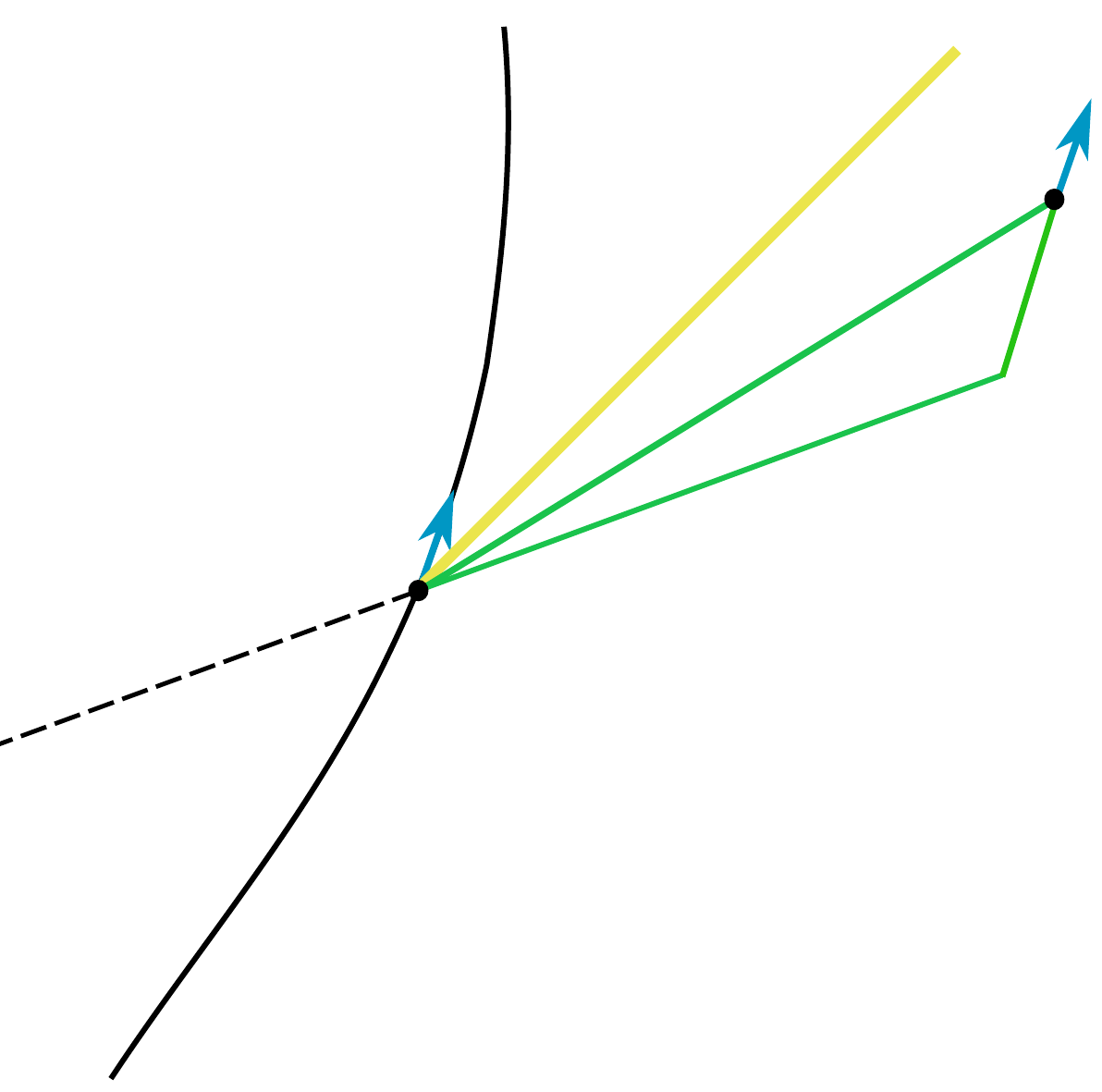}}
\put(0.4, 0.7){$\c(\tau)$}
\put(7.6, 8.7){\begin{turn}{45}null\end{turn}}
\put(9.4, 8.3){$x$}
\put(3, 5){$\cdot|_{\tau'}$}
\put(4, 4){$\tau'(x)$}
\put(10, 8.4){$\v'|_{x}$}
\put(7, 7){$\z|_{x}$}
\put(9.5, 7.3){$\z_{||}$}
\put(7, 5.4){$\z_{\perp}$}
\put(-1.7, 2.9){\begin{turn}{20}plane perpendicular \end{turn}}
\put(-1.5, 2.4){\begin{turn}{20}to $\cdot|_{\tau'}$\end{turn}}
\end{picture}
}
\caption{Displacement vector $\z|_x$}
\label{vecz}
\end{figure}
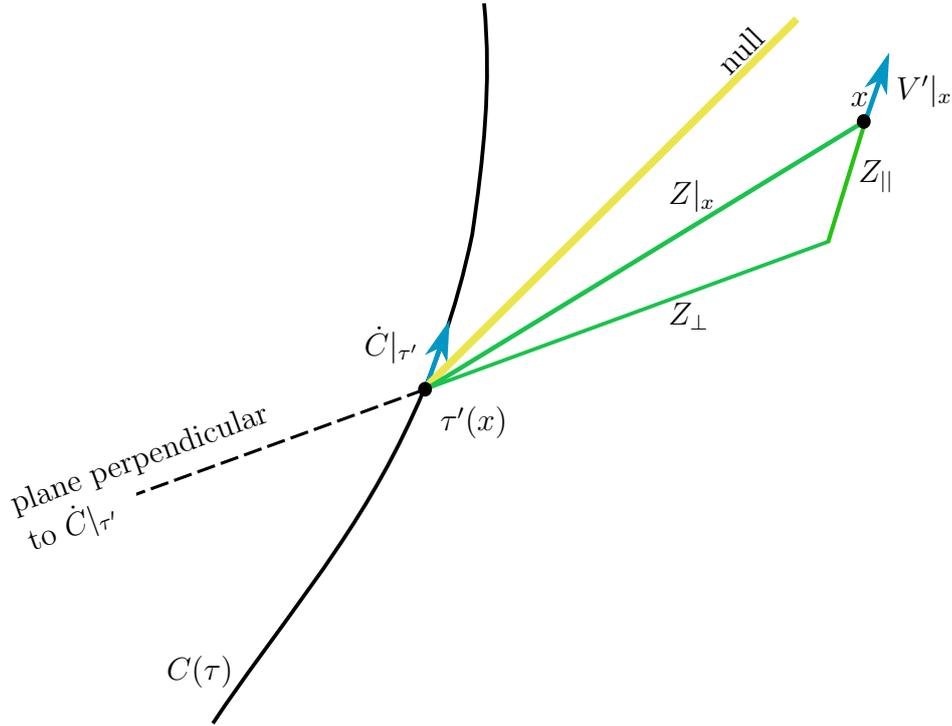

We may construct a local vector field $\z \in \Gamma \tan N$ such that
 \begin{align}
 \z=\z^a \frac{\partial}{\partial y^a},\qquad \text{where}\qquad \z^a= \point^a-\c^a(\tau'(\point)).
 \label{zvec_def}
 \end{align}
 for all $\point\in N$. Here $y^a (x)=x^a$.

 Since $\z$ is defined for every $\point \in  N$ the only requirement needed to define $\z$ completely is to fix $\tau'(\point)$.
 We are free to choose $\tau'(\point)$ in any way we like however particular choices are beneficial for certain problems. In one choice the vector $\z$ lies in the plane perpendicular to $\cdot(\tau'(x))$ (see figure \ref{vecY}). In this case we use the notation $\tau'=\tau_D$, where $\tau_D$ is the \emph{Dirac time}. The Dirac time associates each point $x \in N$ with the time $\tau_D(x)$ given by the solution to
 \begin{align}
 \g\big(x-\c(\tau_D(x)), \cdot(\tau_D(x))\big)=0.
 \end{align}
 In this case $\z|_x=\z_{\perp}=x-\c(\tau_D(x))$. We use the special notation
 \begin{align}
 \y=x-\c(\tau_D(x)).
 \end{align}

    \begin{figure}
\setlength{\unitlength}{1cm}
\centerline{
\begin{picture}(10, 10)
\put(0, 0){\includegraphics[ width=10\unitlength]{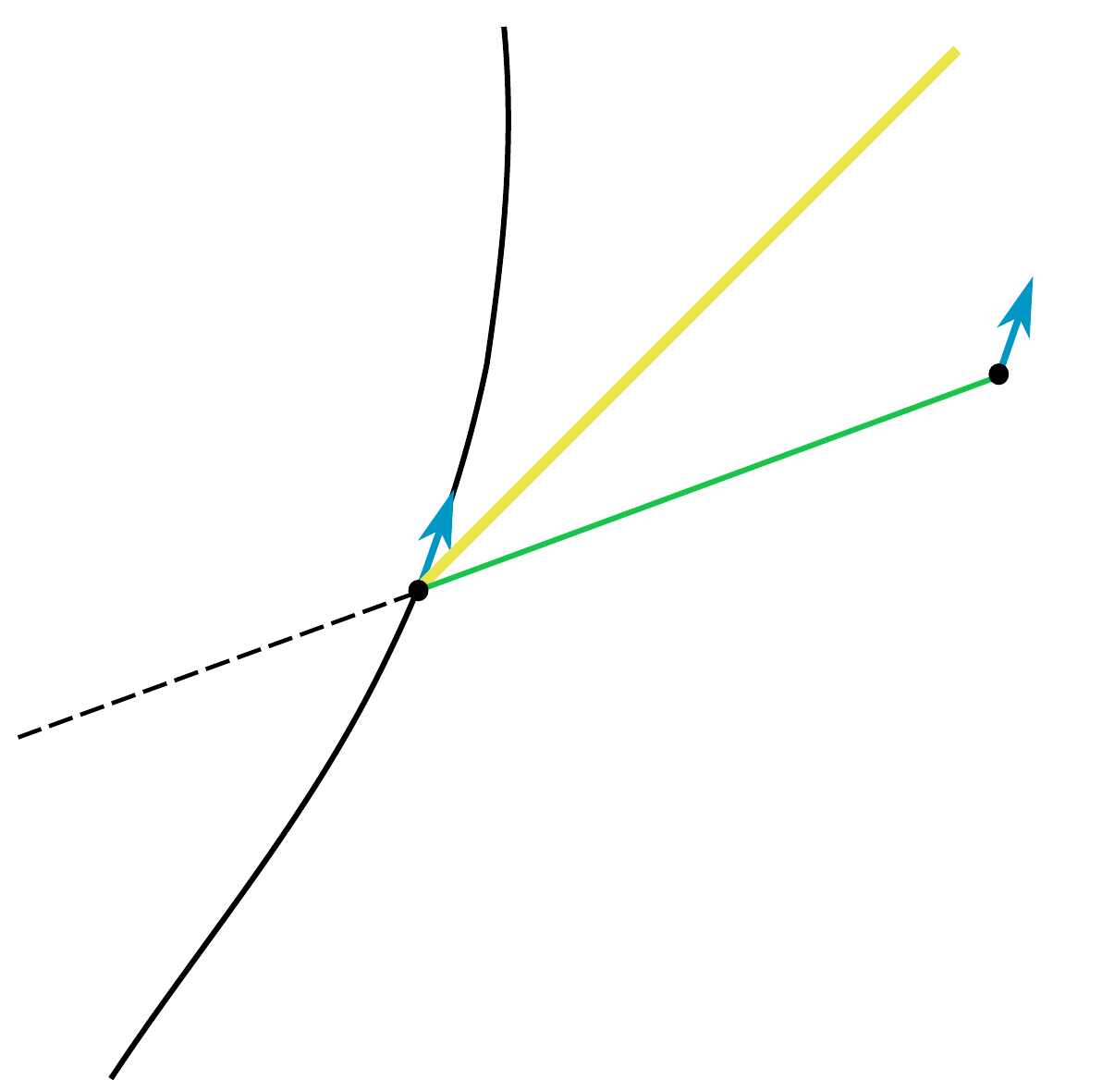}}
\put(0.4, 0.7){$\c(\tau)$}
\put(7.6, 8.7){\begin{turn}{45}null\end{turn}}
\put(9, 6){$x$}
\put(3, 5){$\cdot|_{\tau_D}$}
\put(4, 4){$\tau_D(x)$}
\put(9.5, 6.6){$\v_D|_{x}$}
\put(6.5, 5){$\y|_x$}
\put(-1.7, 2.9){\begin{turn}{20}plane perpendicular \end{turn}}
\put(-1.5, 2.4){\begin{turn}{20}to $\cdot|_{\tau_D}$\end{turn}}
\put(3.5, 2){Let the norm $||.||$ be defined by 
$||\z||=\sqrt{\g(\z, \z)^{2}}$.}
\put(3.5, 1){Then $||\y_{||}||=0$, and $||\y||=||\y_{\perp}||=\g(\y, \y)$. }
\end{picture}
}
\caption{Displacement vector $\y|_x$. }
\label{vecY}
\end{figure}
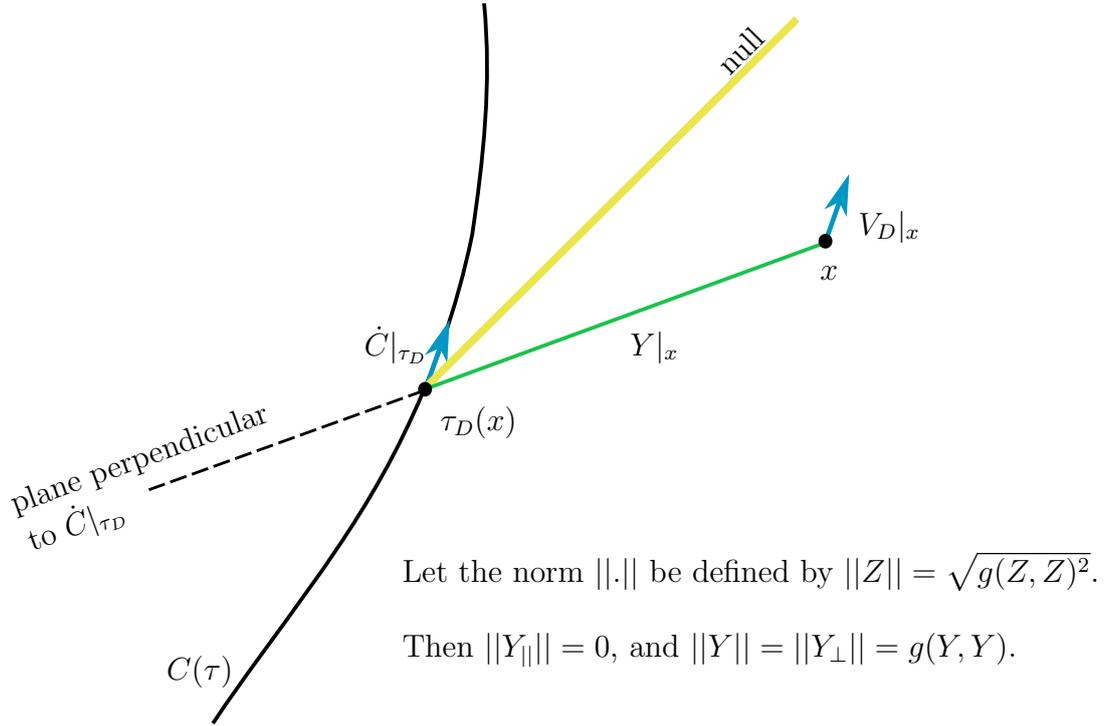
The map $\tau_D:\m \rightarrow \c$ is not unique for every $x\in \m$, for example in figure \ref{nontaud} we see that a single point can be mapped to multiple points along the worldline. However for a sufficiently small neighborhood $N \subset \mnotc$ uniqueness can be ensured. In appendix \ref{dirac_apend} we explore this geometry further.

  \begin{figure}
\setlength{\unitlength}{1cm}
\centerline{
\begin{picture}(10, 10)
\put(0, 0){\includegraphics[ width=10\unitlength]{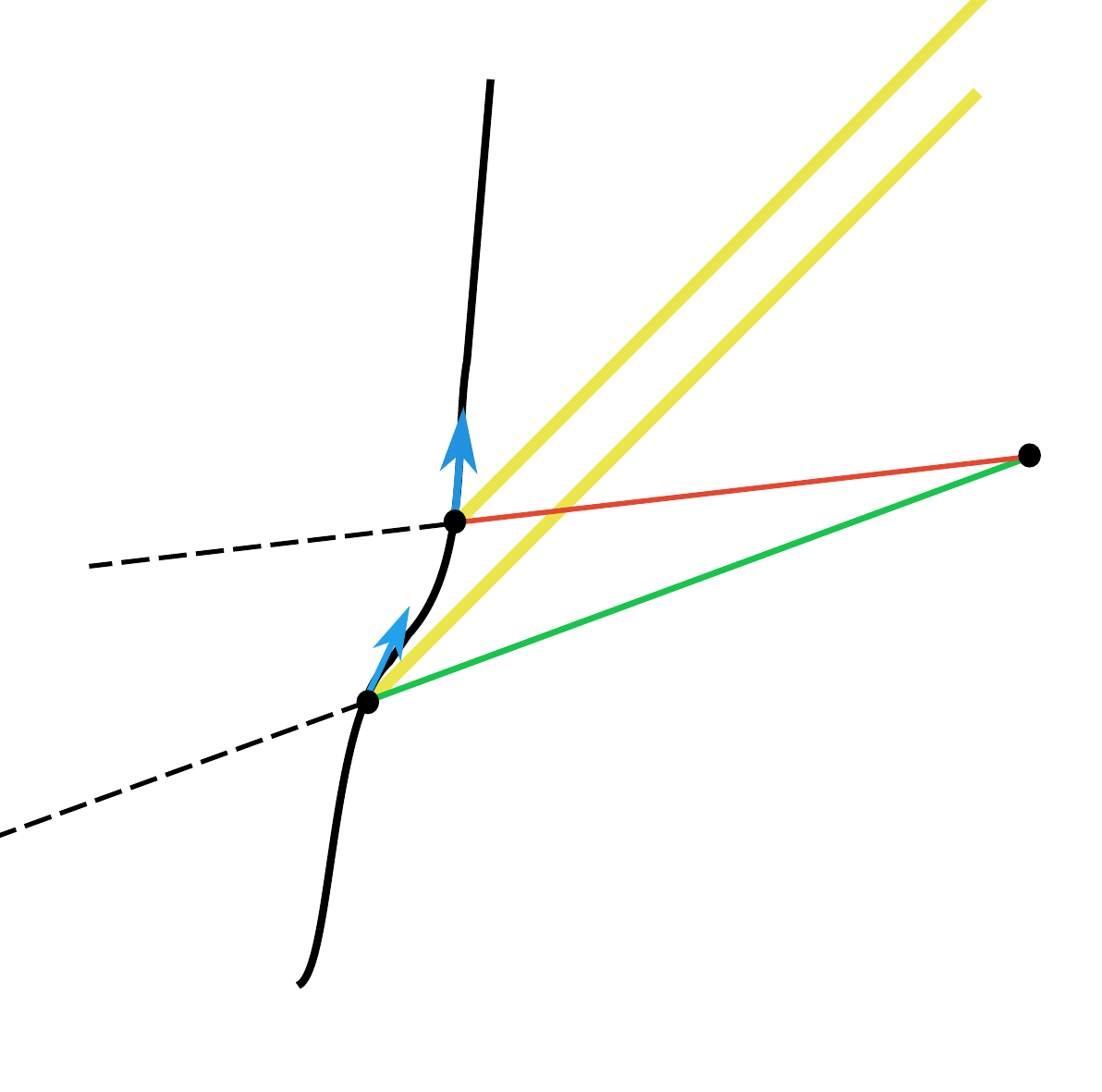}}
\put(2, 0.7){$\c(\tau)$}
\put(8, 9.3){\begin{turn}{45}null\end{turn}}
\put(8, 8.5){\begin{turn}{45}null\end{turn}}
\put(9, 6){$x$}
\put(2.4, 3.8){$\cdot|_{\tau_D}$}
\put(3.4, 3.3){$\tau_D(x)$}
\put(3, 5.7){$\cdot|_{\tau_D'}$}
\put(4.2, 4.8){$\tau_D'(x)$}
\put(6.5, 4){$\y|_x$}
\put(6.5, 5.7){$\y'|_x$}
\put(-1.7, 2){\begin{turn}{20}plane perpendicular \end{turn}}
\put(-1.5, 1.5){\begin{turn}{20}to $\cdot|_{\tau_D}$\end{turn}}
\put(-1.7, 4.7){\begin{turn}{6}plane perpendicular \end{turn}}
\put(-1.5, 4.2){\begin{turn}{6}to $\cdot|_{\tau_D'}$\end{turn}}
\end{picture}
}
\caption{Globally the map $\tau_D$ is non-unique.}
\label{nontaud}
\end{figure}
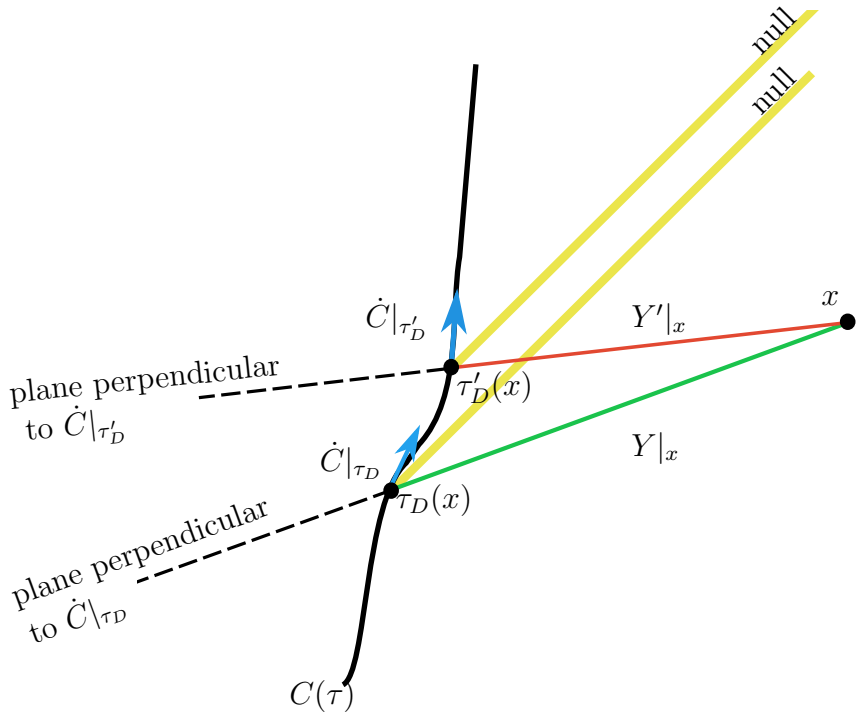

In another choice the vector $\z$ lies on the null cone (see figure \ref{vecX}). In this case we use the notation $\tau'=\tau_r$, where $\tau_r$ is the \emph{retarded time}. The retarded time associates a point $x \in \mnotc$ with the time $\tau_r(x)$ given by the solution to
 \begin{align}
 \g\big(x-\c(\tau_r(x)), x-\c(\tau_r(x))\big)=0, \qquad x^0>\c^0(\tau_r(x))
 \label{nullx_def}
 \end{align}
 In this case $\z|_x=x-\c(\tau_r(x))$. We use the special notation
 \begin{align}
 \x=x-\c(\tau_r(x)).
 \label{xnull_def}
 \end{align}

 \begin{figure}
\setlength{\unitlength}{1cm}
\centerline{
\begin{picture}(10, 10)
\put(0, 0){\includegraphics[ width=10\unitlength]{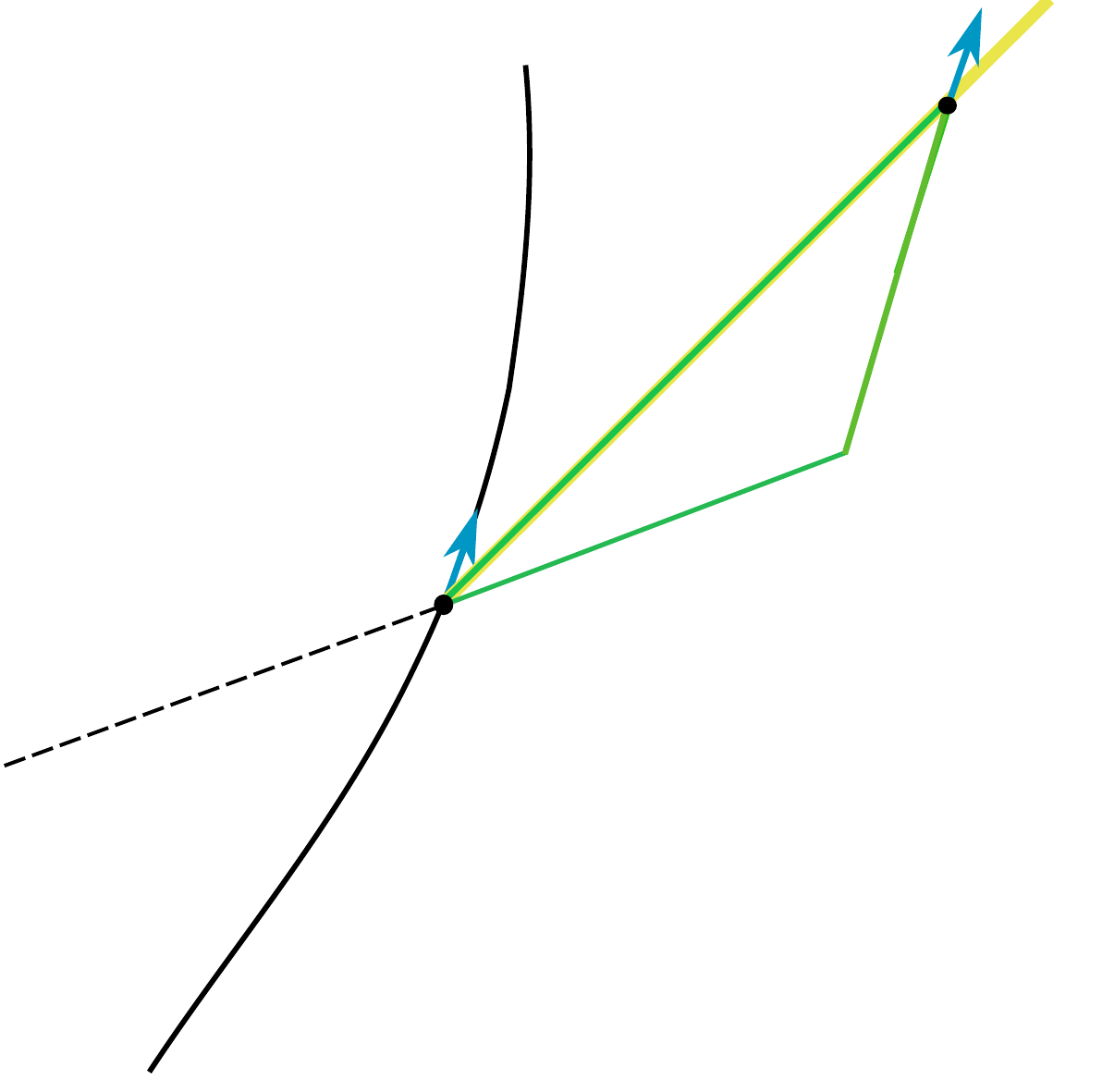}}
\put(0.4, 0.7){$\c(\tau)$}
\put(8.9, 9.6){\begin{turn}{45}null\end{turn}}
\put(9, 8.8){$x$}
\put(3, 5){$\cdot|_{\tau_r}$}
\put(4, 4){$\tau_r(x)$}
\put(7.8, 9.4){$\v|_{x}$}
\put(6, 7.2){$\x|_{x}$}
\put(8, 6.5){$\x_{||}$}
\put(6.5, 5){$\x_{\perp}$}
\put(-1.7, 2.9){\begin{turn}{20}plane perpendicular \end{turn}}
\put(-1.5, 2.4){\begin{turn}{20}to $\cdot|_{\tau_r}$\end{turn}}
\put(3.5, 2){Let the norm $||.||$ be defined by 
$||\z||=\sqrt{\g(\z, \z)^{2}}$.}
\put(3.5, 1){Then $||\x||=0$ and $||\x_{||}||=||\x_{\perp}||=\g(\x, \cdot|_{\tau_r(x)})^2$. }
\end{picture}
}
\caption{Displacement vector $\x|_x$.  }
\label{vecX}
\end{figure}
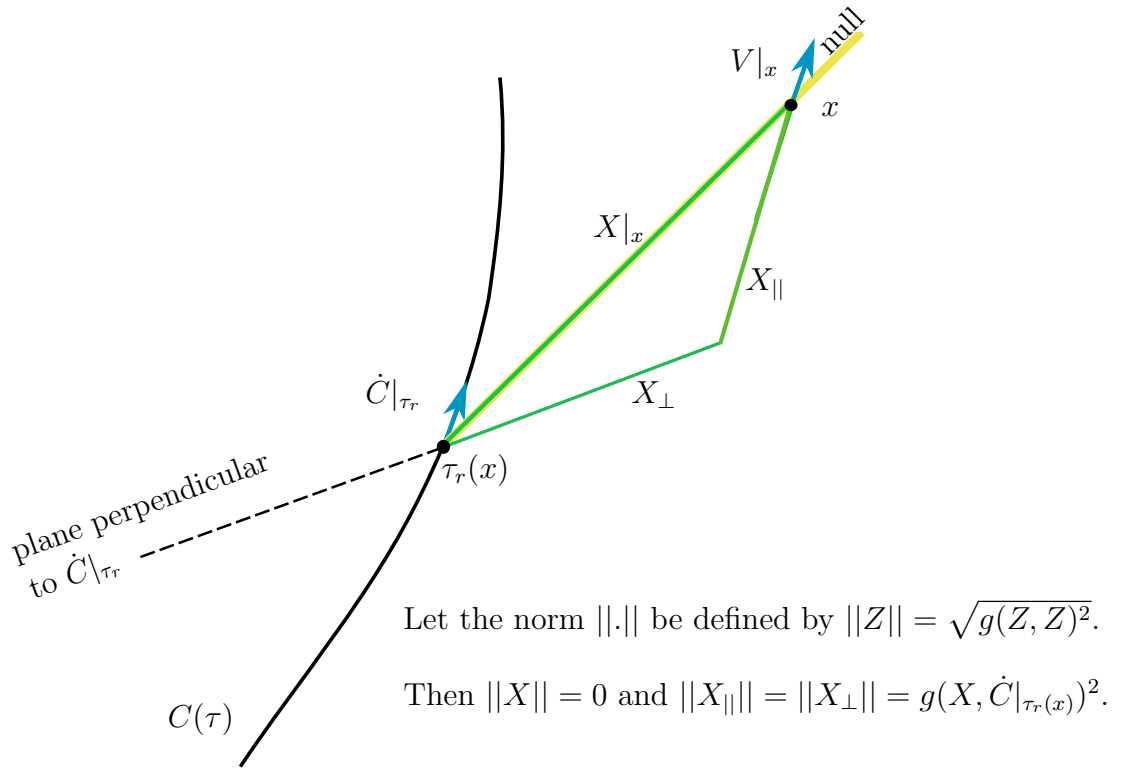
There is another possible choice in which $\z$ is a vector in the advanced null cone at $x$ (see figure \ref{vecX}). In this case we use the notation $\tau'=\tau_a$, where $\tau_a$ is the \emph{advanced time}. The advanced time associates a point $x \in \mnotc$ with the time $\tau_a(x)$ given by the solution to
 \begin{align}
 \g\big(x-\c(\tau_a(x)), x-\c(\tau_a(x))\big)=0, \qquad x^0<\c^0(\tau_r(x))
 \end{align}
 In this case $\z|_x=x-\c(\tau_a(x))$. We use the special notation
 \begin{align}
 \w=x-\c(\tau_a(x)).
 \end{align}

 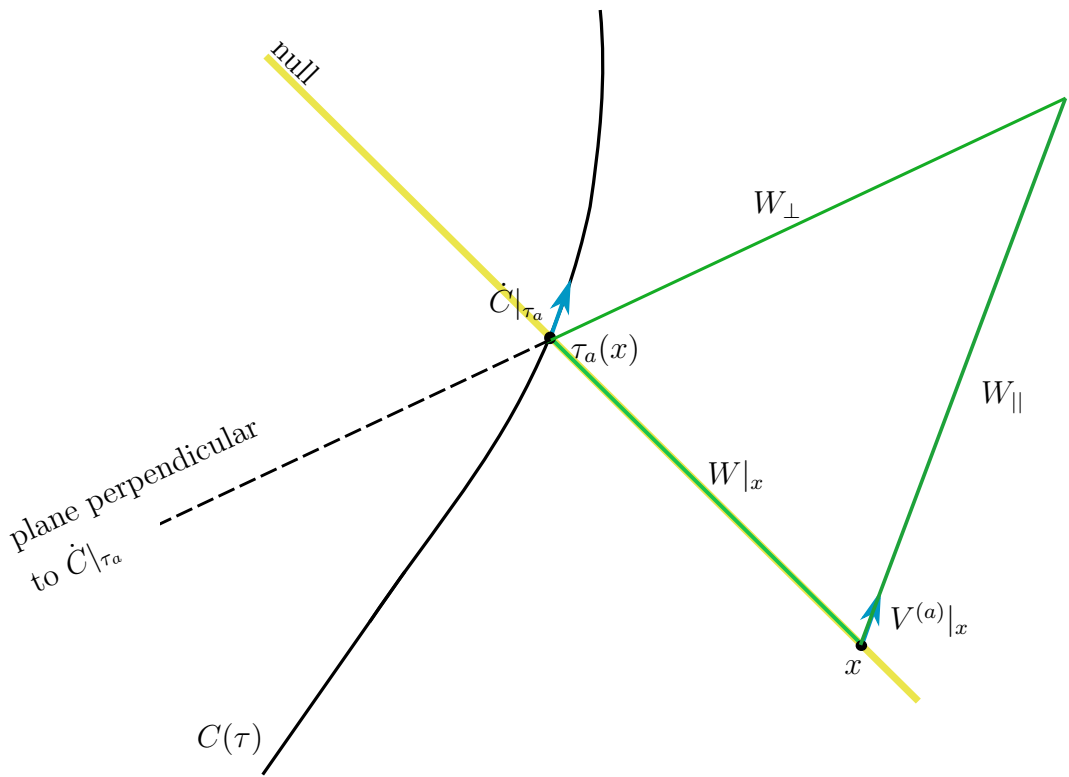
\begin{figure}
\setlength{\unitlength}{1.2cm}
\centerline{
\begin{picture}(10, 10)
\put(0, 0){\includegraphics[ width=10\unitlength]{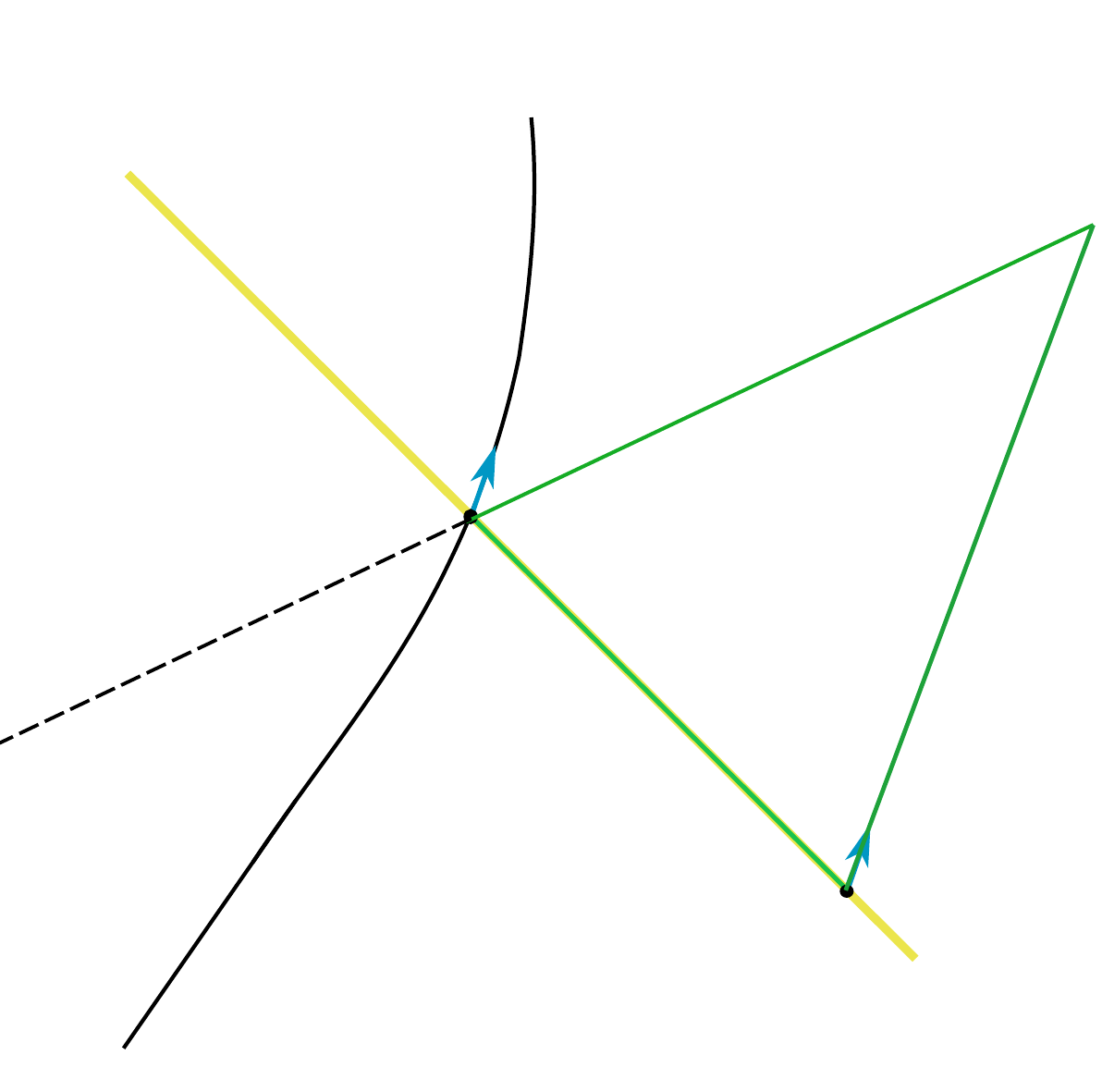}}
\put(0.4, 0.7){$\c(\tau)$}
\put(1.2, 8.4){\begin{turn}{-45}null\end{turn}}
\put(7.5, 1.5){$x$}
\put(3.6, 5.5){$\cdot|_{\tau_a}$}
\put(4.5, 5){$\tau_a(x)$}
\put(8, 2){$\v^{(a)}|_{x}$}
\put(6, 3.6){$\w|_{x}$}
\put(9, 4.5){$\w_{||}$}
\put(6.5, 6.6){$\w_{\perp}$}
\put(-1.7, 2.9){\begin{turn}{25}plane perpendicular \end{turn}}
\put(-1.5, 2.4){\begin{turn}{25}to $\cdot|_{\tau_a}$\end{turn}}
\end{picture}
}
\caption{Displacement vector $\w|_x$}
\label{vecW}
\end{figure}

The maps $\tau_r(x)$ and $\tau_a(x)$ are not necessarily defined for all $\x\in \m$.
Figure \ref{hyperbolae} shows the path of a curve undergoing constant acceleration. The backwards light cone from an arbitrary point in quadrant $\mathbf{A}$ or $\mathbf{B}$ intersects the worldline once in quadrant $\mathbf{B}$, hence the retarded map $\tau_r(x)$ is well defined in $\mathbf{A}$ and $\mathbf{B}$. However the backwards light cone from any point in quadrants $\mathbf{C}$ or $\mathbf{D}$ will never intersect the worldline, therefore the map $\tau_r(x)$ is not defined for $x\in \mathbf{C}$ or $x \in \mathbf{D}$. We can ensure the existence and uniqueness of the maps $\tau_r$ and $\tau_a$ by working exclusively in a sufficiently small (and appropriately chosen) neighbourhood $N\subset \mnotc$.

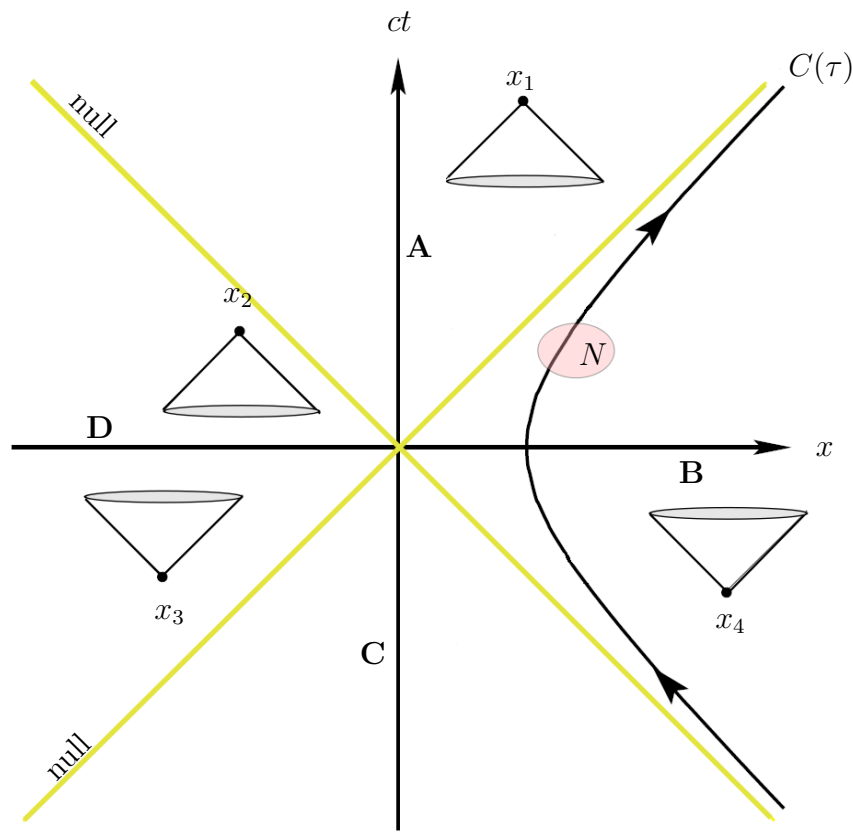
\begin{figure}
\setlength{\unitlength}{1.2cm}
\centerline{
\begin{picture}(10, 10)
\put(0, 0){\includegraphics[ width=10\unitlength]{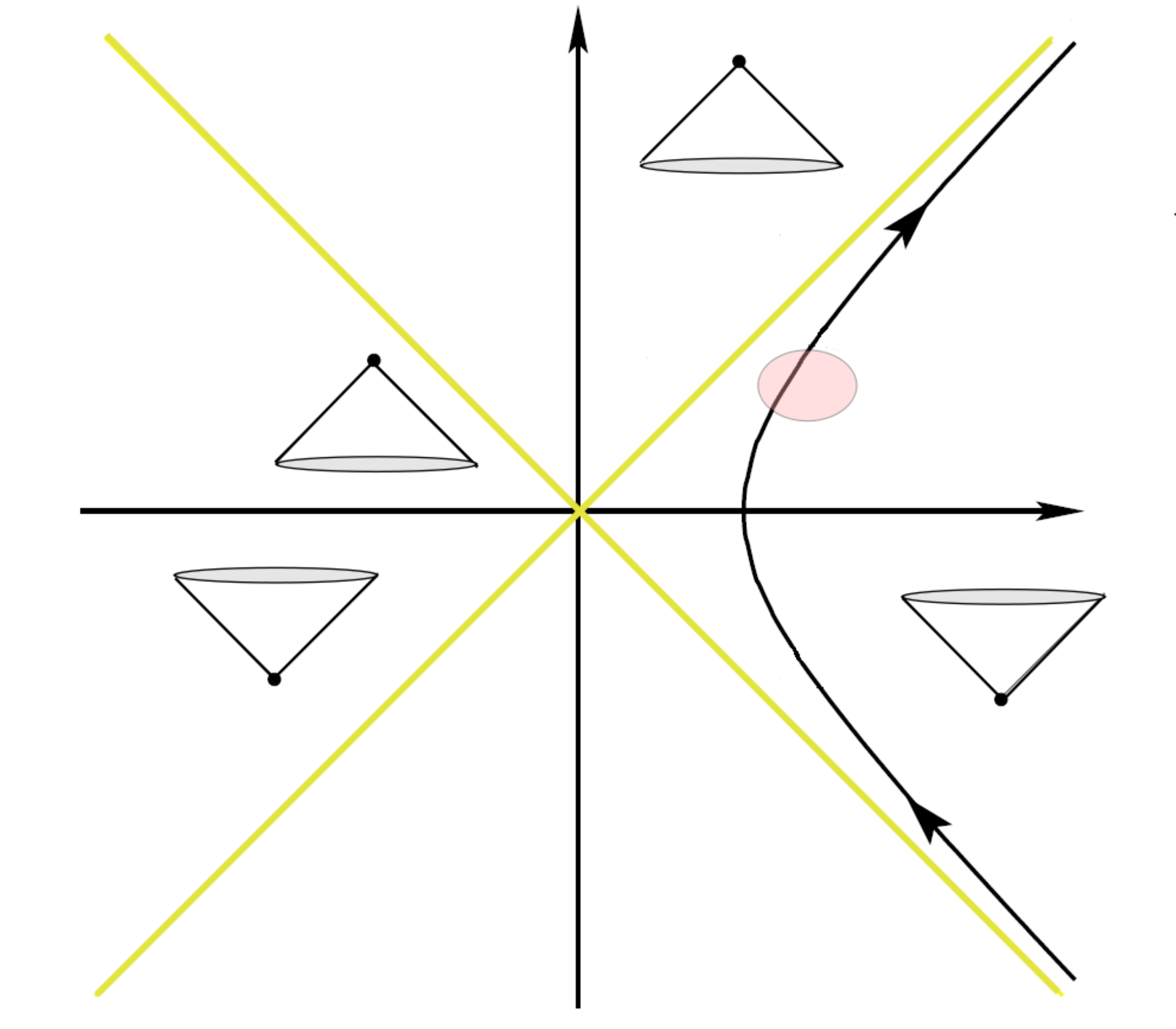}}
\put(4.8, 9){$ct$}
\put(9.5, 4.3){$x$}
\put(5, 6.5){$\mathbf{A}$}
\put(8, 4){$\mathbf{B}$}
\put(4.5, 2){$\mathbf{C}$}
\put(1.5, 4.5){$\mathbf{D}$}
\put(6.1, 8.4){$x_1$}
\put(3, 6){$x_2$}
\put(2.25, 2.5){$x_3$}
\put(8.4, 2.4){$x_4$}
\put(9.2, 8.5){$\c(\tau)$}
\put(1, 0.7){\begin{turn}{45}null\end{turn}}
\put(1.3, 8.2){\begin{turn}{-45}null\end{turn}}
\put(6.9, 5.3){$N$}
\end{picture}
}
\caption[Globally the maps $\tau_r$ and $\tau_a$ are not always defined]{Globally the retarded and advanced times are not necessarily defined for all $x\in\mnotc$.}
\label{hyperbolae}
\end{figure}

\subsection*{Null geometry}
In this section we explore further the consequences of choosing $\tau'=\tau_r$. This map is particularly suited to electromagnetic phenomena which propagate on the light cone. We call the resulting geometry \emph{null geometry}. We begin by consolidating equations  (\ref{C_def}), (\ref{zvec_def}),  (\ref{nullx_def}) and (\ref{xnull_def}).

\begin{definition}
\label{ret_map}
Given the one-parameter curve $\c(\tau)$ which traces the path of a point charge in spacetime, then for every field point $\point\in \mnotc$ there is at most one point $\tau_r(x)$ at which the worldline crosses the retarded light-cone with apex at $\point$.
\begin{align}
\c &: \real \rightarrow \m, \quad \tau \mapsto \c(\tau)\\
\tau_r &: \m \rightarrow \real,\quad \point \mapsto \tau_r(\point)
\label{tau_def}
\end{align}
\end{definition}
\begin{definition}
\label{Xvec_def}
The null vector $\x \in \Gamma \tan \mnotc$ is given by the difference between the field point
$\point$ and the worldline point $\c(\tau_r(\point))$
\begin{align}
&\x\atp= \point-\c \big(\tau_r(\point)\big),
\label{def_X}
\end{align}
where
\begin{align}
\g(\x, \x) = \g\big(\point - \c(\tau_r(\point)), \point - \c(\tau_r(\point))\big) = 0.
\label{null_cond}
\end{align}
\end{definition}
\begin{definition}
\label{def_vaadot}
The vector fields $\v, \a, \adot \in \Gamma \tan \mnotc $ are defined as
\begin{align}
\v\atp&=\cdot^j(\tau_r(\point)) \frac{\partial}{\partial \co^j},\quad\a\atp=\cddot^j(\tau_r(\point))\frac{\partial}{\partial \co^j}\quad \text{and} \quad\adot\atp=\cdddot^j(\tau_r(\point)) \frac{\partial}{\partial \co^j},
\label{def_V_A}
\end{align}
hence from lemma \ref{orthog} it follows
\begin{align}
\g(\v, \v)=-\ccon^2,\quad
\g(\a, \v)=0,\quad \text{and}\quad
\g(\adot, \v)=-\g(\a, \a).
\label{orthog_V_A}
\end{align}
\end{definition}
\begin{definition}
\label{def_nullvec}
We define the normalized null vector field by
\begin{align}
\n=\frac{\x}{R},\qquad \textup{where} \qquad R=-\g(\x, \v)
\end{align}
The normalized vector satisfies
\begin{align}
 \g(\n, \n)=0 \qquad \textup{and} \qquad \g(\n, \v)=-1.
\end{align}
\label{n_def}
\end{definition}
\begin{lemma}
The exterior derivative of the retarded proper time $\tau_r$ is given by
\begin{align}
d\tau_r = \frac{\xdual}{g(\x, \v)}.
\label{dtau_def}
\end{align}
\end{lemma}
\begin{proof}\\
Definition \ref{Xvec_def} requires only the affine structure of $\m$. For the following
proof we demand the stronger requirement that the points $\point\in \m$ and $\c(\tau_r(\point))\in \c(\tau)\subset \m$ are attributed with a vector structure on $\m$, such that
\begin{align}
\xvec\in \Gamma \tan \m,\qquad \xvec\atp=x^a \frac{\partial}{\partial \co^a} \qquad \text{and} \qquad \cvec \in \Gamma \tan \m,\qquad \cvec\atp=\c^a(\tau_r(\point))\frac{\partial}{\partial \co^a},\label{xc_vecs}
\end{align}
however the result (\ref{dtau_def}) requires only the affine structure.

We begin with the light cone condition
\begin{align}
0&=\g(\x, \x),\nonumber\\
&= \g(\xvec- \cvec, \xvec - \cvec),\nonumber\\
&=\g(\xvec- \cvec, \xvec) - \g(\xvec- \cvec, \cvec),\nonumber\\
&=\g(\xvec, \xvec) - 2\g(\cvec, \xvec) + \g(\cvec, \cvec).\label{vec_struct}
\end{align}
Therefore
\begin{align}
\qquad 0&= d\Big[\g(\xvec, \xvec) - 2\g(\cvec, \xvec) + \g(\cvec, \cvec)\Big],\nonumber\\
&= d\g(\xvec, \xvec) -2d\g(\cvec, \xvec) + d\g(\cvec, \cvec).
\label{gXX_def}
\end{align}
Now the first term in (\ref{gXX_def}) yields
\begin{align}
d\g(\xvec, \xvec)&= d(\g_{a b}\point^a \point^b),\nonumber\\
&= \g_{a b} (d \point^a) \point^b + \g_{a b} \point^a (d \point^b),\nonumber
\end{align}
Note that $x^a=y^a (x)$ thus $d x^a = d y^a$ and therefore
\begin{align}
d\g(\xvec, \xvec)&= \point_a dy^a + \point_a dy^a,\nonumber\\
&= 2\point_a d\co^a,\nonumber\\
&= 2 \widetilde{\xvec}.
\label{dgXX_def}
\end{align}
Similarly the second term  yields
\begin{align}
d\g(\cvec, \xvec) &= d(\g_{a b}\point^a \c^b(\tau_r)),\nonumber\\
&=\g_{a b} (d\point^a)\c^b(\tau_r) + \g_{a b} \point^a d(\c^b(\tau_r)),\nonumber\\
&=\c_a (\tau_r)d\co^a + \point_a d(\c^a(\tau_r)),\nonumber\\
\end{align}
where $\displaystyle{d(\c^a(\tau_r))=\frac{d}{d \tau_r}(\c^a(\tau_r))d\tau_r=\v^a d\tau_r}$, therefore
\begin{align}
d\g(\cvec, \xvec) &= \widetilde{\cvec} + \point_a  \v^a d \tau_r,\nonumber\\
& = \widetilde{\cvec} + \g(\xvec, \v) d\tau_r.
\label{dgCX_def}
\end{align}
The third term gives
\begin{align}
d\g(\cvec, \cvec) &= d(\g_{a b} \c^a(\tau_r) \c^b(\tau_r)),\nonumber\\
&=(d \c^a(\tau_r))\g_{a b} \c^b(\tau_r) + (d \c^a(\tau_r)) \g_{a b} \c^b(\tau_r),\nonumber\\
&=2 \c_a(\tau_r) \v^a d\tau_r,\nonumber\\
&=2 \g( \cvec, \v) d\tau_r.
\label{dgCC_def}
\end{align}
Substituting (\ref{dgXX_def}), (\ref{dgCX_def}) and (\ref{dgCC_def}) into (\ref{gXX_def}) yields
\begin{align*}
0 & = 2 \widetilde{\xvec} - 2\Big(\widetilde{\cvec} + g(\xvec, \v) d\tau_r\Big) + 2 g( \cvec, \v) d\tau_r,
\end{align*}
therefore
\begin{align*}
2 (\widetilde{\xvec} -  \widetilde{\cvec}) &= 2\Big(g(\xvec, \v) - \g( \cvec, \v)\Big)d\tau_r,\nonumber\\
\widetilde{\xvec}- \widetilde{\cvec}&= \g(\xvec-\cvec, \v)d\tau_r,\nonumber
\end{align*}
and upon rearrangement yields
\begin{align}
d\tau_r &= \frac{\widetilde{\xvec} - \widetilde{\cvec}}{\g(\x, \v)} = \frac{\xdual}{\g(\x, \v)}.
\label{gXX_yield}
\end{align}
\end{proof}

\begin{lemma}
\begin{align}
d\vdual = d\tau_r \wedge \adual
\label{dV_def}
\end{align}
\end{lemma}

\begin{proof}
\begin{align*}
d\vdual &= \frac{d\v_a}{d\tau_r}d\tau_r \wedge dy^a\\
&=\a_a d\tau_r \wedge dy^a\\
&=d\tau_r \wedge\notag \adual
\end{align*}
\end{proof}
\begin{corrol}
\begin{align}
d\vdual = \frac{\xdual\wedge \adual}{g(X, V)}
\label{dVd_def}
\end{align}
\end{corrol}
\begin{proof}
Follows directly from  (\ref{gXX_yield}) and (\ref{dV_def}).
\end{proof}

\begin{lemma}
\begin{align}
d( \star \vdual)=\frac{g(X, \a)}{g(X, V)} \star 1
\label{dstarVd}
\end{align}
\end{lemma}
\begin{proof}
\begin{align*}
d(\star \vdual) &= d \v^a \wedge i_{\frac{\partial}{\partial y^a}} \star 1\\
&= \a^a d\tau_r \wedge i_{\frac{\partial}{\partial y^a}} \star 1\\
&=d\tau_r \wedge \star \tilde{\a}\\
&= \frac{\xdual \wedge \star \tilde{\a}}{g(X, \v)}\\
&= \frac{g(X, \a)}{g(X, V)} \star 1
\end{align*}
\end{proof}
\begin{lemma}
\begin{align}
d\star(\xdual \wedge \vdual)=\frac{\xdual \wedge \star(\xdual \wedge \adual)}{g(X, \v)}- 3\star\vdual
\label{dstarXV}
\end{align}
\end{lemma}
\begin{proof}
\begin{align*}
d\star (\xdual \wedge \vdual)&=d\star(X_a dy^a \wedge \v_b dy^b)\\
&=d(\v_b X_a \star(dy^a \wedge dy^b))\\
&=d(\v_b X_a)\wedge\star(d y^a \wedge dy^b) + \v_b X_a d \star (dy^a \wedge dy^b)\\
&=d\v_b \wedge X_a \star (dy^a \wedge dy^b)+\v_b dX_a \wedge \star (dy^a \wedge dy^b)\\
&=d\v_b\wedge \star( \xdual\wedge dy^b) + dX_a\wedge \star (dy^a \wedge \vdual)\\
&=\frac{d\v_b}{d\tau_r} d\tau_r\wedge\star( \xdual\wedge dy^b)-d(\g^{a b}X_b) \wedge  i_{\partial_{y^a}} \star \vdual\\
&=A_b d\tau_r \wedge\star( \xdual\wedge dy^b) - dX^a \wedge  i_{\partial_{y^a}} \star \vdual\\
&=d\tau_r \wedge \star(\xdual \wedge \adual)- dy^a \wedge i_{\partial_{y^a}} \star \vdual +  d\c^a \wedge i_{\partial_{y^a}} \star \vdual\\
&=d\tau_r \wedge \star(\xdual \wedge \adual)- dy^a \wedge i_{\partial_{y^a}} \star \vdual +   d\tau_r \wedge  \star( \vdual \wedge \vdual)\\
&= \frac{\xdual \wedge \star(\xdual \wedge \adual)}{g(X, \v)}- dy^a \wedge i_{\partial_{y^a}} \star \vdual 
\end{align*}
Need also to show that
\begin{align}
 dy^a \wedge i_{\partial_{y^a}} \star \vdual = 3 \star \vdual \label{eq_show}
\end{align}
Let $\vdual=\v_a dy^a$, then
\begin{align*}
\star \vdual = \v_a \g^{a b} i_{\partial_{y^b}}\star 1= \v^b i_{\partial_{y^b}}\star 1
\end{align*}
Substituting $\star 1 = d y^0\wedge d y^1 \wedge d y^2 \wedge d y^3$ and contracting yields
\begin{align}
\star \vdual = \v_0 d y^0 \wedge d y^2 \wedge dy^3 -\v_1 dy^0 \wedge dy^2 \wedge dy^3 + \v_2 dy^0 \wedge dy^1 \wedge d y^3 -\v_3 dy^0 \wedge dy^1 \wedge dy^2 \label{eq_old}
\end{align}
Thus 
\begin{align}
dy^a \wedge i_{\partial_{y^a}} \star \vdual=& -\v_1 d y^0 \wedge d y^2 \wedge dy^3 + \v_2 d y^0 \wedge d y^1 \wedge dy^3-\v_3 d y^0 \wedge d y^1 \wedge dy^2\notag\\
&-\v_0 d y^1 \wedge d y^2 \wedge dy^3 -\v_2d y^1 \wedge d y^0 \wedge dy^3+ \v_3 d y^1 \wedge d y^0 \wedge dy^2\notag\\
&+\v_0 d y^2 \wedge d y^1\wedge dy^3 + \v_1d y^2 \wedge d y^0 \wedge dy^3 -\v_3 d y^2 \wedge d y^0 \wedge dy^1\notag\\
&-\v_0d y^3 \wedge d y^1 \wedge dy^2-\v_1 d y^3 \wedge d y^0 \wedge dy^2+ \v_2 d y^3 \wedge d y^0 \wedge dy^1 \label{eq_new}
\end{align}
Collecting terms in (\ref{eq_new}) and comparing with (\ref{eq_old}) yields (\ref{eq_show}).
\end{proof}
\begin{lemma}
\begin{align}
d\star(\xdual \wedge \adual)=\frac{\xdual \wedge \star(\xdual \wedge \adotdual)}{g(X, \v)}+ \frac{\xdual \wedge \star(\adual \wedge \vdual)}{g(X, \v)}-3\star\adual
\label{dstarXA}
\end{align}
\end{lemma}

\begin{proof}
\begin{align*}
d\star (\xdual \wedge \adual) &=d( \x_a \a_b \star (dy^a \wedge dy^b))\\
&=d(\x_a \a_b)\wedge\star(dy^a \wedge dy^b) + \x_a \a_b d\star (dy^a \wedge dy^b)\\
&=(d\x_a)\wedge \a_b \star(dy^a \wedge dy^b) +\x_a d\a_b\wedge  \star (dy^a \wedge dy^b)\\
&=d\x_a \wedge \star(dy^a \wedge \adual) + d\a_b \wedge \star (\xdual \wedge d y^b)\\
&=-d(\g^{a b} \x_a)\wedge i_{\partial_{y^b}}\star \adual +\adot_b d\tau_r\wedge \star(\xdual \wedge d y^b)\\
&=-d\x^a \wedge i_{\partial_{y^a}}\star \adual+ \frac{\xdual \wedge \star(\xdual \wedge \adotdual)}{g(X, \v)}\\
&= -d y^a\wedge i_{\partial_{y^a}}\star \adual+d\c^a \wedge i_{\partial_{y^a}}\star \adual+ \frac{\xdual \wedge \star(\xdual \wedge \adotdual)}{g(X, \v)}\\
&=-3\star \adual +\v^a d\tau_r\wedge i_{\partial_{y^a}}\star \adual + \frac{\xdual \wedge \star(\xdual \wedge \adotdual)}{g(X, \v)}\\
&=-3\star \adual + \frac{\xdual \wedge \star(\adual \wedge \vdual)}{g(X, \v)} + \frac{\xdual \wedge \star(\xdual \wedge \adotdual)}{g(X, \v)}
\end{align*}
\end{proof}

\begin{lemma}
\begin{align}
dg(\a, \a)=2\g(\a, \adot)d\tau_r
\label{dgAA_def}
\end{align}
\begin{align}
d\g(\x, \v)=\vdual + \Big(\frac{\g(\x, \a) +\ccon^2}{\g(\x, \v)}\Big)\xdual
\label{dgXV_def}
\end{align}
\begin{align}
d\g(\x, \a)=\adual +\Big(\frac{\g(\x, \adot)}{\g(\x, \v)}\Big) \xdual
\label{dgXA_def}
\end{align}
\end{lemma}

\begin{proof}\\
Proof of (\ref{dgAA_def})
\begin{align*}
d\g(\a, \a)=& d(g_{ab}\cddot^a(\tau_r)\cddot^b(\tau_r))\\
&=2g_{ab}d\cddot^a\cddot^b\\
&=2g_{ab}\cdddot^a\cddot_b d\tau_r\\
&=2\g(\a, \adot)d \tau_r
\end{align*}
Proof of (\ref{dgXV_def})
\begin{align*}
dg(X, \v)&= dg(x-\ctau, \v)\\
& = d\Big[g(x, \v) - g(\ctau, \v)\Big]\\
&=dg(x, \v) - dg (\ctau, \v)\\
&=d(g_{a b} x^a \v^b) - d(g_{a b}\c^a (\tau_r) \v^b)\\
&=g_{a b}\v^b dy^a + g_{a b} x^a d\v^b - g_{a b}\v^b d \c^a (\tau_r) - g_{a b}\c^a (\tau_r) d\v^b\\
&=\vdual + g(x, \a)d\tau_r - g(\v,\v)d\tau_r - g(\ctau, \a)d\tau_r\\
&=\vdual + \Big[g(x-\ctau, \a) - g(\v, \v)\Big]d\tau_r\\
&=\vdual + \Big[g(X, \a) - g(\v, \v)\Big]d\tau_r
\end{align*}
Substituting (\ref{gCdCd}) yields result.\\
Proof of (\ref{dgXA_def})
\begin{align*}
dg(X, \a)&= dg(x-\ctau, \a)\\
& = d\Big[g(x, \a) - g(\ctau, \a)\Big]\\
&=dg(x, \a) - dg (\ctau, \a)\\
&=d(g_{a b} x^a \a^b) - d(g_{a b}\c^a (\tau_r) \a^b)\\
&=g_{a b}\a^b dy^a + g_{a b} x^a d\a^b - g_{a b}\a^b d \c^a (\tau_r) - g_{a b}\c^a (\tau_r) d\a^b\\
&=\adual + g(x, \adot)d\tau_r - g(\a,\v)d\tau_r - g(\ctau, \adot)d\tau_r\\
&=\adual + \Big[g(x-\ctau, \adot) - g(\a, \v)\Big]d\tau_r\\
&=\adual + \Big[g(X, \adot) - g(\a, \v)\Big]d\tau_r
\end{align*}
Substituting (\ref{g_Cd_Cdd}) yields result.
\end{proof}

\section{Newman-Unti coordinates $(\tau, R, \theta, \phi)$}
\label{chap_coords}
We introduce a system of coordinates adapted to the null worldline geometry. The coordinates were first introduced in a general form for arbitrary manifolds by Temple in 1938 \cite{Temple38}, where they are referred to as \emph{optical coordinates}.  In 1963 Newman and Unti \cite{Newman} claim to introduce a new coordinate system ``intrinsically attached to an arbitrary timelike worldline", however the coordinate system they investigate is none other than the specialization of Temple's coordinates to Minkowski space. Since in this thesis we work explicitly with Minkowski space we have chosen to refer to the coordinates as Newman-Unti (N-U) coordinates in the spirit of Galt'sov and Spirin \cite{Galtsov02}, however the general class of coordinates should be attributed to Temple.  Similar coordinates were used by Trautman and Robinson \cite{Trautman62} in their work on gravitational waves, and  in the 1980's Ellis \cite{Ellis80} and others use similar coordinates in problems in relativistic cosmology where they are called \emph{Observational coordinates}.  Other variations on the name include \emph{retarded coordinates, null geodesic coordinates and lightcone coordinates}.

We recall from (\ref{def_X}) that
\begin{align}
\x=\point-\c(\tau)= -\frac{R}{\alpha}\Big(\dt + \sin(\theta)\cos(\phi)\dx+ \sin(\theta) \sin(\phi)\dy + \cos\theta \dz\Big).
\end{align}

 \begin{definition}
 \label{def_NU_coords}
 Given the global Lorentzian frame $(\xt, \xx, \xy, \xz)$ on $\m$, the Newman-Unti coordinates $(\tau, R, \theta, \phi)$ are defined
  by the coordinate transformation,
\begin{align}
& \xt=\c^0 (\tau) -\fracra\notag,\\
&\xx=\c^1(\tau) - \fracra\sin(\theta)\cos(\phi)\notag,\\
&\xy=\c^2 (\tau) -\fracra \sin(\theta)\sin(\phi)\notag,\\
\text{and}\qquad &\xz=\c^3(\tau) -\fracra \cos(\theta),
\label{nu_def}
\end{align}
where $\alpha \in\Gamma \Lambda^0 \m$ is defined by
\begin{align}
\alpha(\tau, \theta, \phi) =& -\frac{\g(\x, \dot{\c}(\tau))}{\g(\x, \partial_{\xt})}\nonumber\\
=&-\cdot^0(\tau) +\cdot^1 (\tau)\sin(\theta) \cos(\phi) + \cdot^2(\tau) \sin(\theta) \sin(\phi) +\cdot^3(\tau) \cos(\theta).
\label{def_alpha}
\end{align}
From (\ref{nu_def}) and (\ref{def_alpha}) it follows
\begin{align}
R=-\g(\x, \dot{\c}(\tau))
\quad\quad \textup{and}\quad\quad
\tau=\tau_r(\point(\tau, R, \theta, \phi)).
\label{tau_taur_R}
\end{align}
The spherical coordinates $\theta$ and $\phi$ are given naturally from the
global Lorentzian frame $(\xt, \xx, \xy, \xz)$.
\end{definition}
\begin{lemma}
In Newman-Unti coordinates the vector fields $\x$ and $\v$ are given by
\begin{align}
&\x=R\frac{\partial}{\partial R}
\label{X_NU}
\end{align}
and
\begin{align}
\v=\frac{\partial}{\partial \tau}+\x \frac{\ad}{\alpha}
\label{V_NU}
\end{align}
\end{lemma}
\begin{proof}\\
Proof of (\ref{X_NU})
Differentiating the coordinate transformation (\ref{nu_def}) with respect to $R$ yields
\begin{align}
\dR &= \frac{\partial \xt}{\partial R}\dt +\frac{\partial \xx}{\partial R}\dx +\frac{\partial \xy}{\partial R}\dy+\frac{\partial \xz}{\partial R}\dz\notag\\
&=-\frac{1}{\alpha}\Big(\dt + \sin(\theta)\cos(\phi)\dx+ \sin(\theta) \sin(\phi)\dy + \cos\theta \dz\Big).\notag
\end{align}
Therefore
\begin{align*}
\x&= \point - \c(\tau)\\
&=-\frac{R}{\alpha}\Big(\dt + \sin(\theta)\cos(\phi)\dx+ \sin(\theta) \sin(\phi)\dy + \cos\theta \dz\Big)\\
&=R\frac{\partial}{\partial R}\\
\end{align*}
Proof of (\ref{V_NU})
\begin{align*}
\dtau=& \frac{\partial \xt}{\partial \tau} \dt + \frac{\partial \xx}{\partial \tau}\dx+\frac{\partial \xy}{\partial \tau}\dy +\frac{\partial \xz}{\partial \tau y}\dz\\
=&\Big(\cdot^0(\tau)-R\dtau(\frac{1}{\alpha})\Big)\dt +\Big(\cdot^1(\tau)-R \sin(\theta)\cos(\phi)\dtau(\frac{1}{\alpha})\Big)\dx\\ &+\Big(\cdot^2(\tau)-R\sin(\theta)\sin(\phi)\dtau(\frac{1}{\alpha})\Big)\dy +\Big(\cdot^3(\tau)-R\cos(\theta)\dtau(\frac{1}{\alpha})\Big)\dz\\
=&\cdot^a(\tau)\frac{\partial}{\partial \co^a}+\frac{\partial \x}{\partial \tau}\\
=&\v- \x \frac{\ad}{\alpha}\\
\therefore \v=&\frac{\partial}{\partial \tau}+\x \frac{\ad}{\alpha}
\end{align*}
\end{proof}

\begin{lemma}
\label{nu_metric}
In Newman-Unti coordinate the Minkowski metric $\g\in \bigotimes^{[ \mathds{F}, \mathds{F}]} \man$ is given by
\begin{align}
g=&(-\ccon^2+2R\frac{\ad}{\alpha})d\tau \otimes d\tau
- (d\tau \otimes dR+ dR \otimes d\tau)\notag\\
&+ \frac{R^2}{\alpha^2}d\theta \otimes d\theta
+ \frac{R^2}{\alpha^2}\sin(\theta)^2  d\phi \otimes d\phi,
\label{g_nu_def}
\end{align}
and inverse metric $\gdual \in \bigotimes^{[ \mathds{V}, \mathds{V}]} \man$ is given by
\begin{align}
\gdual =&-\frac{-\ccon^2 \alpha+2 R \ad}{\alpha} \dR \otimes \dR +\frac{\alpha^2}{R^2} \dth \otimes \dth\notag\\
&+\frac{\alpha^2}{R^2 \sin(\theta)^2} \dph \otimes \dph-(\dtau \otimes \dR +\dR \otimes \dtau ).
\label{in_g_nu_def}
\end{align}
Let $z^0=\tau,\quad z^1=R,\quad z^2=\theta,\quad z^3=\phi$, then the matrices $G=G_{ab}=\g(\partial_{z^a},\partial_{z^b})$ and
$G^{-1}=G^{-1}_{ab}=\gdual(d z^a, d z^b)$ are given by
\begin{align}
G=&
\left( \begin{array}{cccc}
\displaystyle \frac{-\ccon^2\alpha+2R\ad}{\alpha} &-1                                    & 0 &0 \\
\displaystyle -1         & 0                    & 0                                    & 0 \\
0                        & 0                    &\displaystyle \frac{R^2}{\alpha^2}    & 0\\
0                        & 0                    & 0                                    &\displaystyle \frac{R^2\sin(\theta)^2}{\alpha^2} \end{array} \right),\label{gmatrix}
\end{align}
and
\begin{align}
G^{-1}=& \left( \begin{array}{cccc}
0                              &     -1           & 0                            & 0 \\
-1                              &   \displaystyle -\frac{-\ccon^2\alpha+2R\ad}{\alpha}                         &   0                         &  0  \\
0                             & 0                                                &\displaystyle \frac{\alpha^2}{R^2} &0\\
0                             & 0                                                & 0 &\displaystyle \frac{\alpha^2}{R^2\sin(\theta)^2} \end{array} \right)\label{gdualmatrix}\end{align}

\end{lemma}

\begin{proof}
Differentiation of the coordinate transformation (\ref{nu_def}) gives
\begin{footnotesize}
\begin{align}
\dxt&=\Big(\cdot^0(\tau)+R\frac{\ad}{\alpha^2}\Big)d\tau -\frac{1}{\alpha}d R+R\frac{\atheta}{\alpha^2}d\theta+R\frac{\aphi}{\alpha^2}d\phi\notag,\\
\dxx&=\Big(\cdot^1(\tau)+R\sin(\theta)\cos(\phi)\frac{\ad}{\alpha^2}\Big)d\tau  - \frac{\sin(\theta)\cos(\phi)}{\alpha} dR\notag \\
 &+\frac{R}{\alpha^2}\big(\atheta \sin(\theta)\cos(\phi) -\alpha\cos(\theta)\cos(\phi)\big)d \theta + \frac{R}{\alpha^2}\big( \aphi \sin(\theta)\cos(\phi)-\alpha \cos(\theta)\cos(\phi)\big)d\phi\notag,\\
\dxy&=\Big(\cdot^2(\tau)+R\sin(\theta)\sin(\phi)\frac{\ad}{\alpha^2}\Big)d\tau - \frac{\sin(\theta)\sin(\phi)}{\alpha} dR\notag\\
 &+\frac{R}{\alpha^2}\big(\atheta \sin(\theta)\sin(\phi) -\alpha\cos(\theta)\sin(\phi)\big)d \theta + \frac{R}{\alpha^2}\big( \aphi \sin(\theta)\sin(\phi)-\alpha \sin(\theta)\cos(\phi)\big)d\phi\notag,\\
\dxz&=\Big(\cdot^3(\tau)+R\cos(\theta)\frac{\ad}{\alpha^2}\Big) d\tau - \frac{\cos(\theta)}{\alpha} dR +\frac{R}{\alpha^2}\big( \atheta\cos(\theta). +\alpha\sin(\theta)\big) d\theta+\frac{R}{\alpha^2}\cos(\theta)\aphi d\phi
\label{dx_def}
\end{align}
\end{footnotesize}
where
\begin{align}
\ad=\frac{\partial\alpha}{\partial \tau}, \qquad\atheta=\frac{\partial\alpha}{\partial \theta}, \qquad\aphi=\frac{\partial\alpha}{\partial \phi}
\end{align}
Substitution of (\ref{dx_def}) into (\ref{gmink_def}) yields (\ref{g_nu_def}). The dual metric (\ref{in_g_nu_def}) follows from (\ref{gdualmatrix}).
\end{proof}

\begin{lemma}
The 1-forms $\xdual, \vdual \in \Gamma \Lambda^1 \mnotc$ are given by
\begin{align}
&\xdual= -R  d\tau\\
&\vdual = \frac{R\ad-\ccon^2\alpha}{\alpha} d\tau - dR
\label{nu_dual_vecs}
\end{align}
\end{lemma}
\begin{proof}
\begin{align}
\xdual &= R \g (\dR, -)\notag\\
&=-R d\tau\notag\\
\vdual &=  \g (\dtau+\x \frac{\ad}{\alpha}, -)\notag\\
&= \frac{R\ad-\ccon^2\alpha}{\alpha} d\tau - dR\notag
\end{align}
\end{proof}
\begin{lemma}

\begin{align}
\adual =& (R\frac{\ad^2}{\alpha^2}) d\tau -\frac{\dot{\alpha}}{\alpha} dR + R\frac{\ad\atheta-\alpha\atd}{\alpha^2}d\theta +R\frac{\ad\aphi-\alpha\apd}{\alpha^2} d\phi\\
\adotdual =& (\g(\adot, \v)+R\frac{\ad\add}{\alpha^2}) d\tau -\frac{\add}{\alpha} dR + R\frac{\add\atheta-\alpha\atdd}{\alpha^2}d\theta +R\frac{\add\aphi-\alpha\apdd}{\alpha^2} d\phi
\label{adual_def}
\end{align}

\end{lemma}
\begin{proof}
\begin{align}
\adual &= \frac{d \v_a}{d\tau}d \co^a\notag\\
&= -\cddot^0(\tau) \dxt + \cddot^1(\tau) \dxx + \cddot^2(\tau)\dxy +\cddot^3(\tau)\dxz\notag\\
\adotdual &= \frac{d \a_a}{d\tau}d \co^a\notag\\
&= -\cdddot^0(\tau) \dxt + \cdddot^1(\tau) \dxx + \cdddot^2(\tau)\dxy +\cdddot^3(\tau)\dxz\notag
\end{align}
result follows on substitution of (\ref{dx_def}).
\end{proof}
\begin{lemma}
\label{dtaudr}
\begin{align}
\widetilde{d\tau}=-\frac{\partial}{\partial R} \qquad \textup{and} \qquad \widetilde{d R}=-\frac{\partial}{\partial \tau}+\big(\ccon^2-2R\frac{\ad}{\alpha}\big)\frac{\partial}{\partial R}
\end{align}
\end{lemma}
\begin{proof}
Follows from (\ref{dual_def}) and (\ref{in_g_nu_def}).
\end{proof}
\begin{lemma}
\begin{align}
\a=& \Big(R\frac{\ad}{\alpha}-\ccon^2\Big)\frac{\ad}{\alpha}\dR +\frac{\ad}{\alpha}\dtau+\frac{\ad\atheta-\alpha \atd}{R}\dth +\frac{\ad \aphi-\alpha \apd}{R\sin^2(\theta)}\dph
\label{a_def}
\end{align}
\end{lemma}
\begin{proof}
follows from (\ref{gdual_def}),  (\ref{in_g_nu_def}) and  (\ref{adual_def}).
\end{proof}
\begin{lemma}
\label{gxa_coords}
\begin{align}
\g(\x, \a)=-R\frac{\ad}{\alpha}
\end{align}
\end{lemma}
\begin{proof}
Follows by substitution of (\ref{X_NU}) and (\ref{a_def}) into (\ref{g_nu_def}).
\end{proof}

\begin{lemma}
\begin{align}
\star 1=\frac{R^2\sin(\theta)}{\alpha^2}d\tau\wedge dR \wedge d\theta \wedge d\phi
\label{starone}
\end{align}
\end{lemma}
\begin{proof}
Follows from (\ref{starone_basis}) and (\ref{nu_def}).
\end{proof}

In appendix \ref{app_coords} we present a different coordinate system which we have called adapted N-U coordinates. They will be used in Part II of this thesis.


\section{The Li\'{e}nard-Wiechert field}
\label{chap_lw}

The Li\'{e}nard-Wiechert potential is the solution to the Maxwell-Lorentz equations when the source $\jgen \in \Gamma \Lambda^3 \m$ is given by $\j$, the current 3-form for a point charge moving arbitrarily in free space (\ref{J_def}). In this section the \LW potential and associated fields are given in term of the null geometry formalism developed in the proceeding section. We use the notation $\alw$ for the \LW 1-form potential as a special case for $\pot\in \Gamma \Lambda^1 \m$. It is the  solution to (\ref{A_maxwell}) given the source $\j$.

\begin{definition}
\label{lwboth_def}
The \textbf{\LW Potential} of the point charge at point $\point\in \mnotc$ is  given by
\begin{align}
\alw\atp \in \Gamma \Lambda^1 \mnotc, \quad \quad \alw\atp = \qe\frac{\vdual}{\g(\v, \x)}.
\label{LW_def}
\end{align}
We associate with $\alw$ the $1$-form distribution $\alw^D$ defined by its action on test $3$-form $\varphi\in\Gamma_0\Lambda^3\m$ by
\begin{align}
\alw^D \in \Gamma_D \Lambda^1 \m, \quad \quad \alw^D[\varphi] =\int_\m \varphi\wedge \alw.
\label{LWD_def}
\end{align}
\end{definition}
\begin{lemma}
\label{Flw_def}
The electromagnetic $2$-form attained by substituting $\pot=\alw$ in (\ref{A_maxwell}) will be called the \LW $2$-form and is given by
\begin{align}
\flw \in \Gamma \Lambda^2 \mnotc, \qquad  \flw\atp=\qe\frac{\g(\x, \v)\xdual\wedge \adual-\g(\x, \a)\xdual \wedge \vdual- \ccon^2\xdual\wedge \vdual}{\g(\x, \v)^3}.
\label{F_def}
\end{align}
\end{lemma}
\begin{proof}
\begin{align}
\flw=d\alw &= \qe d\Bigg(\frac{\vdual}{\g(\v, \x)}\Bigg)\notag\\
&= \qe d\Bigg(\frac{1}{\g(\x,\v)}\Bigg)\vdual + \qe\frac{1}{\g(\x,\v)}d(\vdual)\notag\\
&=\qe\Bigg[-\frac{1}{\g(\x, \v)^2}d\g(\x, \v) +\frac{1}{\g(\x,\v)}d(\vdual)\Bigg]\label{f_mid}
\end{align}
Substituting (\ref{dgXV_def}) and (\ref{dVd_def}) into (\ref{f_mid}) yields
\begin{align}
\flw=\qe \frac{\g(\x, \v)\xdual\wedge \adual-\g(\x, \a)\xdual \wedge \vdual -\ccon^2\xdual\wedge \vdual}{\g(\x, \v)^3}
\end{align}
\end{proof}
\begin{definition}
\label{Flw_dist_def}
We associate with $\flw$ the regular $2$-form distribution $\flw^D$ defined by its action on test $2$-form $\varphi\in\Gamma_0\Lambda^2\m$ by
\begin{align}
\flw^D \in \Gamma_D \Lambda^2 \m \qquad  \flw^D[\varphi]= \int_\m \varphi\wedge \flw.
\label{FD_def}
\end{align}
 Readers familiar with the 3-vector notation for the Electric and Magnetic \LW fields can look at lemma \ref{ecer_lem} to see how these relate to $\flw$.
\end{definition}
\begin{definition}
\label{def_frfc}
We split the \LW $2$-form $\flw$ into two terms where
 \begin{align}
 \fr = \qe\frac{\g(\x, \v)\xdual\wedge \adual-\g(\x, \a)\xdual \wedge \vdual}{\g(\x, \v)^3}
 \label{FR_def}
 \end{align}
 will be referred to as the radiation term, and
 \begin{align}
 \fc=\qe\frac{-\ccon^2\xdual\wedge \vdual}{\g(\x, \v)^3}
 \label{FC_def}
 \end{align}
 will be referred to as the Coulomb term.
\end{definition}

\begin{lemma}
The Li\'{e}nard-Wiechert potential (\ref{LW_def}) satisfies the Lorentz gauge condition
\begin{align}
d\star \alw=0.
\label{dstarLW_def}
\end{align}
\end{lemma}
\begin{proof}\\
Let
\begin{align}
\kappa=\qe,
\label{def_kappa}
\end{align}
then
\begin{align*}
 \frac{1}{\kappa}d \star \alw & =d\Big(\frac{\star \vdual}{\g(\x, \v)}\Big)\notag\\
&= -\frac{1}{g(\v, \x)^2}dg(\v, \x)\wedge\star \vdual +\frac{1}{g(\v, \x)} d(\star \vdual)\notag
\end{align*}
Substituting (\ref{dgXV_def}) and (\ref{dstarVd}) yields
\begin{align*}
  \frac{1}{\kappa}d \star \alw &= -\frac{1}{g(\v, \x)^2}\vdual\wedge \star \vdual  - \Big(\frac{g(X, \a) + \ccon^2}{g(X, \v)^3}\Big)\xdual\wedge\star \vdual +\frac{1}{g(\v, \x)} (\frac{g(X, \a)}{g(X, V)} \star 1)
\end{align*}
and using lemma \ref{one_one} gives
\begin{align}
 \frac{1}{\kappa} d \star \alw &= \frac{\ccon^2}{g(\v, \x)^2}\star 1   - \Big(\frac{\g(\x, \a) + \ccon^2}{\g(\x, \v)^2}\Big)\star 1 + \frac{g(\x, \a)}{g(\x, \v)^2} \star 1=0
\end{align}
\end{proof}
\begin{lemma}
\label{dstarF_def}
Off the worldline the \LW potential satisfies
\begin{align}
d\star d \alw=d \star \flw =0.\notag
\end{align}
Thus given arbitrary region $N \subset \mnotc$ it follows that
\begin{align}
\int_N \varphi \wedge d\star \flw = 0
\end{align}
for any test $1$-form $\varphi\in \Gamma_0 \Lambda^1 \m$.
\end{lemma}
\begin{proof}\\
From (\ref{F_def})
\begin{align*}
 \frac{1}{\kappa} \star \flw &= \frac{\g(\x, \v)\star(\xdual\wedge \adual)-(\g(\x, \a)+\ccon^2)\star(\xdual \wedge \vdual)}{\g(\x, \v)^3}\notag
\end{align*}
Thus
\begin{align}
 \frac{1}{\kappa} d \star \flw =& d\Big(\frac{\star(\xdual\wedge \adual)}{\g(\x, \v)^2}\Big)-d\Big(\frac{(\g(\x, \a)+\ccon^2)\star(\xdual \wedge \vdual)}{\g(\x, \v)^3}\Big)\notag\\
 =&\frac{d\star (\xdual \wedge \adual)}{\g(\x, \v)^2} +d\Big(\frac{1}{\g(\x, \v)^2}\Big)\wedge \star (\xdual \wedge \adual)- \frac{(\g(\x, \a)+\ccon^2)}{\g(\x, \v)^3}d\star(\xdual \wedge \vdual)\notag\\
 &- d\Big(\frac{\g(\x, \a)+\ccon^2}{\g(\x, \v)^3}\Big)\star(\xdual \wedge \vdual)
 \label{df_int}
\end{align}
Using the chain rule for differentiation yields
\begin{align}
d\Big(\frac{1}{\g{\x, \v}^2}\Big)=&-\frac{2d\g(\x, \v)}{\g(\x, \v)^2}\notag\\
\textup{and} \quad d\Big(\frac{\g(\x, \a)+\ccon^2}{\g(\x, \v)^3}\Big)=& \frac{d\g(\x, \a)}{\g(\x, \v)^3}-\frac{3(\g(\x, \a)+\ccon^2)}{\g(\x, \v)^4}d\g(\x, \v)\label{chain}
\end{align}
Substituting (\ref{dstarXV}), (\ref{dstarXA}), (\ref{dgXV_def}), (\ref{dgXA_def}) and (\ref{chain}) into (\ref{df_int}) and using lemma \ref{one_two} yields result.
\end{proof}
\begin{lemma}
In Newman-Unti coordinates the \LW potential $\alw\in\Gamma\Lambda^1\mnotc$ and the electromagnetic $2$-form $\flw\in \Gamma\Lambda^2\mnotc$ are given by
\begin{align}
&\alw=-\qe\Big(\Big(\frac{\ccon^2}{R}-\frac{\ad}{\alpha}\Big)d\tau + \frac{1}{R}dR \Big),\label{LW_coords}\\
&\fr=\qe\Big(\frac{\alpha\atd-\dot{\alpha}\atheta}{\alpha^2}d\tau \wedge d\theta + \frac{\alpha\apd-\dot{\alpha}\aphi}{\alpha^2}d\tau \wedge d\phi\Big),\label{FR_coords}\\
\text{and}\qquad &\fc=\qe\frac{\ccon^2}{R^2}d\tau \wedge dR.\label{FC_coords}
\end{align}
It follows from theorems \ref{dist_oneform} and \ref{dist_twoform} that the distributional $1$-form   $\alw_D\in \Gamma_D\Lambda^1\m$ and distributional $2$-form   $\flw_D\in \Gamma_D\Lambda^2\m$ are well defined.
\end{lemma}
\begin{proof}\\
Equations (\ref{LW_coords}) and (\ref{FC_coords}) follow by substitution of (\ref{tau_taur_R}) and (\ref{nu_dual_vecs}) into (\ref{LW_def}) and (\ref{FC_def}).
Equation (\ref{FR_coords}) follows by substitution of  (\ref{tau_taur_R}), (\ref{nu_dual_vecs}), (\ref{adual_def}) and (\ref{gxa_coords}) into (\ref{FR_def}).
\end{proof}
\begin{lemma}
\begin{align}
&\star \fr=\qe\Big(\frac{\aphi \dot{\alpha}-\alpha \apd}{\alpha^2 \sin(\theta)}d R \wedge d\theta - \frac{\sin(\theta)(\atheta \dot{\alpha} - \alpha \atd)}{\alpha^2}dR \wedge d\phi\Big),\notag\\
\text{and}\qquad &\star \fc=\qe\frac{\ccon^2\sin(\theta)}{\alpha^2}d\theta \wedge d\phi.\label{starFc_def}
\end{align}
\end{lemma}
\begin{proof}
Follows from definition \ref{hodge_def}, lemma \ref{dtaudr} and the equations (\ref{starone}),  (\ref{FR_coords}) and (\ref{FC_coords}).
\end{proof}

\begin{lemma}
The distributional \LW field $\flw^D\in\Gamma_D \Lambda^2 \m$ satisfies
\begin{align}
\epsilon_0 d\star \flw^D[\varphi]=\j^D[\varphi],
\end{align}
for any test $1$-form  $\varphi\in \Gamma_0 \Lambda^1 \m$.
\end{lemma}
\begin{proof}
First we consider the form of $\varphi$ close to the worldline.
A general test $1$-form $\varphi\in \Gamma_0 \Lambda^1 \m$ is given in Minkowski coordinates by
\begin{align}
\varphi= \varphi_i(y^0, y^1, y^2, y^3)dy^i
\label{phi_defz}
\end{align}
Now making transformation (\ref{def_NU_coords}) to Newman-Unti coordinates, such that
\begin{align}
\varphi= \hat{\varphi}_\tau(\tau, R, \theta,\phi)  d\tau + \hat{\varphi}_R(\tau, R, \theta,\phi)  dR + \hat{\varphi}_\theta(\tau, R, \theta,\phi)  d\theta + \hat{\varphi}_\phi(\tau, R, \theta,\phi)  d\phi
\label{phi_def1}
\end{align}
yields
\begin{align}
\hat{\varphi}_\tau=&Y_\tau^0(\tau)+Y_\tau^1(\tau, \theta, \phi) R,\notag\\
 \hat{\varphi}_R=& Y_R^0(\tau, \theta, \phi),\notag\\
 \hat{\varphi}_\theta=&  Y_\theta^1(\tau, \theta, \phi) R,\notag\\
  \hat{\varphi}_\phi=&Y_\phi^1(\tau, \theta, \phi)  R,
\end{align}
where $Y_\tau^0$ is a bounded function of $\tau$ and the rest of the $Y_i^l$'s are bounded functions of $\tau, \theta$ and $\phi$. Thus to zero order in $R$
\begin{align}
\varphi=Y_\tau^0d\tau+Y_R^0 dR +\ord(R),\label{test_small}
\end{align}
and
\begin{align}
d \varphi=-\frac{\partial Y_\tau^0}{\partial \theta}d\tau\wedge d\theta-\frac{\partial Y_\tau^0}{\partial \phi}d\tau\wedge d\phi+\frac{\partial Y_R^0}{\partial \tau}d\tau\wedge dR -\frac{\partial Y_R^0}{\partial \theta}dR \wedge d\theta -\frac{\partial Y_R^0}{\partial \phi}dR\wedge d\phi +\ord(R)\label{dvarphi}.
\end{align}
Now definition \ref{diff_forms} yields
\begin{align}
d\star \flw^D[\varphi]&= \star \flw^D [d\varphi]\notag\\
&=\int_\m d\phi \wedge \star \flw
\end{align}
We split the integral over $\m$ into a region away from the worldline and a region containing the worldline.
Let the four dimensional region $\fourball \subset \m$ be defined in N-U coordinates by
\begin{align}
\fourball=\{\quad\tau, R, \theta, \phi \quad| \quad\tau \in I , \quad 0\leq R \leq k,\quad 0\leq \theta \leq\pi,\quad 0 \leq \phi \leq 2\pi\quad\}
\label{fourball_def}
\end{align}
where $I$ is the domain of $\c$. The boundary $\partial \mnotc =\partial \fourball$ is given by
 \begin{align}
 \partial \fourball = \{\quad\tau, R, \theta, \phi\quad|\quad \tau \in I,\quad R=k,\quad 0\leq \theta \leq\pi,\quad 0 \leq \phi \leq 2\pi\quad\}.
 \end{align}
We calculate $d\star \f^D[\phi]$ with the assumption that $k \rightarrow 0$ so that the approximation (\ref{test_small}) remains valid
\begin{align}
d\star \flw^D[\varphi]&= \star \flw^D [d\varphi]\notag\\
&=\int_\m d\phi \wedge \star \flw\notag\\
&=\int_{\m \backslash  \fourball} d\varphi \wedge \star \flw + \int_{\fourball} d\varphi \wedge \star \flw\notag\\
&=\int_{\m \backslash  \fourball} d(\varphi \wedge \star \flw) + \int_{\m \backslash  \fourball} \varphi \wedge d \star \flw + \int_{\fourball} d\varphi \wedge \star \flw \label{three_terms}
\end{align}
The second term in (\ref{three_terms}) vanishes due to lemma (\ref{dstarF_def}). Consider the third term.
\begin{align}
\int_{\fourball} d\varphi \wedge \star \flw &= \int_{\fourball} d\varphi \wedge \star \fc +\int_{\fourball} d\varphi \wedge \star \fr
\label{third_term}
\end{align}
Using (\ref{starFc_def}) and (\ref{dvarphi}) yields
\begin{align}
\int_{\fourball} d\varphi \wedge \star \flw =& \qe\Big(\int_\m \frac{\partial Y_R^0}{\partial \tau}\frac{\ccon^2\sin(\theta)}{\alpha^2}d\tau \wedge dR \wedge d\theta \wedge d\phi\notag\\
&+ \int_\m \frac{\partial Y_\tau^0}{\partial \phi}\frac{\aphi \dot{\alpha}-\alpha \apd}{\alpha^2 \sin(\theta)}d\tau \wedge dR \wedge d\theta \wedge d\phi\notag\\
&+\int_\m \frac{\partial Y_\tau^0}{\partial \theta}\frac{\sin(\theta)( \alpha \atd-\atheta \dot{\alpha})}{\alpha^2}d\tau \wedge dR \wedge d\theta \wedge d\phi
\end{align}
 All three terms vanish under integration with respect to $R$ when $k \rightarrow 0$, therefore the third term in (\ref{three_terms}) vanishes. Finally we consider the first term.  We note that $R$ is constant on the boundary and therefore  $dR=0$. By Stokes' Theorem
\begin{align}
\int_{\m \backslash  \fourball} d(\varphi \wedge \star \flw)&= \int_{\partial \fourball} \varphi \wedge \star \flw\notag\\
&=\int_{\partial \fourball} \varphi \wedge \star \fc+\int_{\partial \fourball} \varphi \wedge \star \fr\label{two_terms}
\end{align}
The second term vanishes because $dR=0$. We are left with the first term,
\begin{align}
\int_{\partial \fourball} \varphi \wedge \star \fc=& \qe \int_{\partial \fourball} Y_\tau^0 (\tau )\frac{\ccon^2 \sin(\theta)}{\alpha^2}d\tau \wedge d\theta \wedge d\phi\\
=&\qe\int_{\tau=-\infty}^{\infty}Y_\tau^0 (\tau )\Big( \int_{\theta=0}^{\pi}\int_{\phi=0}^{2\pi}\frac{\ccon^2 \sin(\theta)}{\alpha^2}d\theta  d\phi\Big)d\tau
\end{align}
Let $\displaystyle{I=\int_{\theta=0}^{\pi}\int_{\phi=0}^{2\pi}\frac{\ccon^2 \sin(\theta)}{\alpha^2}d\theta  d\phi}$, then substituting (\ref{def_alpha}) we obtain
\begin{align}
 I=\int_{\theta=0}^{\pi}\int_{\phi=0}^{2\pi}\frac{\ccon^2 \sin(\theta)}{(-\cdot^0 +\cdot^1 \sin(\theta) \cos(\phi) +\cdot^2 \sin(\theta) \sin(\phi) +\cdot^3 \cos(\theta))^2}d\theta  d\phi
\end{align}
Let $\displaystyle z=e^{i\phi}$ such that
\begin{align}
 \sin(\phi)=\frac{1}{2i}(z-z^{-1}), \qquad \cos(\phi)=\frac{1}{2}(z+z^{-1}), \qquad d\phi=\frac{dz}{iz}
\end{align}
Substitution yields
\begin{align}
I=\ccon^2\int_{\mu(0, 1)} \frac{4 i \sin(\theta) z}{((-i \cdot^2 + \cdot^1)\sin(\theta) z^2 + (2\cdot^3 \cos(\theta) - 2\cdot^0)z + (i \cdot^2 +\cdot^1)\sin(\theta))^2}dz\wedge d\theta
\end{align}
where $\mu(0, 1)$ represents circle of radius 1 centred at the origin.
The quadratic:
\begin{align}
(-i \cdot^2 + \cdot^1)\sin(\theta) z^2 + (2\cdot^3 \cos(\theta) - 2\cdot^0)z + (i \cdot^2 +\cdot^1)\sin(\theta)=0
\end{align}
has roots at
\begin{align}
&\frac{1}{\sin(\theta)(-i\cdot^2+\cdot^1)}\Bigg(-\cdot^3 \cos(\theta) + \cdot^0\notag\\
  &\pm \sqrt{\cdot^0 - 2 \cdot^{0^2} \cdot^3 \cos(\theta) + \cdot^{3^2}\cos^2(\theta) - \cdot^{1^2} - \cdot^{2^2} +\cos^2(\theta)\cdot^{1^2} + \cos(\theta)^2 \cdot^{2^2}}\Bigg)
\end{align}
Denoting the roots by $\alpha(+), \beta(-)$ yields,
\begin{align}
I=\int_{\mu(0, 1)} \frac{4 i z\sin(\theta)}{(z-\alpha)^2(z-\beta)^2}dz d\theta
\end{align}
$|\alpha|>1$ therefore it lies outside the contour.
The residue of $I$ at $z=\beta$ is given by
\begin{align}
res= \frac{4 i \sin(\theta)(\alpha + \beta)}{(-\beta+\alpha)^3},
\end{align}
Therefore by the residue theorem  (see for example \cite{Howie}),
\begin{align}
I=2\pi \ccon^2 \int_0^\pi \frac{\sin(\theta) (\cdot^3\cos(\theta) -\cdot^0)}{((\cdot^3\cos(\theta)-\cdot^0)^2 +(\cdot^{2^2} +\cdot^{1^2})\sin^2(\theta))^{\frac{3}{2}}}d\theta
\end{align}
and integration using a computer gives
\begin{align}
I=\frac{4 \pi \ccon^2}{-\cdot^{0^2} + \cdot^{1^2} + \cdot^{2^2} + \cdot^{3^2}} =4\pi\frac{\ccon^2}{\ccon^2}
\end{align}
hence
\begin{align}
\int_{\partial \fourball} \varphi \wedge \star \fc=\frac{q}{\ep} \int_{\tau=-\infty}^{\infty}Y_\tau^0 (\tau)d\tau=\frac{q}{\ep} \int_{\tau=-\infty}^{\infty}\hat{\varphi}_\tau (\tau)d\tau= \frac{q}{\ep} \int_{I} \c^\ast \phi,
\end{align}
and comparison with lemma \ref{jd_def} yields
\begin{align}
\epsilon_0 d \ast \flw^D [\phi]= \charge \int_{I} \c^\ast \phi=\j^D[\phi]
\end{align}
\end{proof}

\newpage

\addcontentsline{toc}{chapter}{Part I - The self force and the Schott term discrepancy}

\vspace{3cm}
\begin{center}
\begin{minipage}{0.95\textwidth}
\vspace{6cm}
\begin{center}{\bf \huge

PART I\\
The self force and the Schott term discrepancy

} \vspace{3cm}

\end{center}
\end{minipage}\vspace{10cm}
\end{center}

\chapter{Introduction}
\label{intro}

In chapter \ref{maxlorentz} we assign the dimension of time to proper time $\tau$  so that (\ref{gCdCd}) is satisfied. We will return to this convention in Part II where we derive the electric and magnetic fields in the standard $3$-vector notation. In Part I it is convenient to assign the dimension of length to $\tau$  so that (\ref{gCdCdtwo}) is satisfied. For details see appendix \ref{app_dimensions}.

\section{The self force, mass renormalization and the equation of motion}

It is a consequence of Maxwell-Lorentz electrodynamics that any source of an electromagnetic field will be subject to interaction with that field. The resulting force on the source is known as the \emph{self force}. For a charged particle undergoing inertial motion the self force is zero, however for an accelerating charge the force is non-zero and tends to act as a damping term \cite{Landau80}. It is well known that an accelerating charge loses energy due to the emission of radiation,  where the instantaneous loss of momentum due to radiation, $\pdot_{\textup{RAD}}\in\Gamma\tan \m$, is given by the Larmor-Abraham formula \cite{Rohrlich97}
 \begin{align}
 \pdot_{\textup{RAD}}=\frac{q^2}{6 \pi \epsilon_0 } g(\cddot, \cddot)\cdot.
 \label{larmor_abraham}
 \end{align}
 The negative of this force is the \emph{radiation reaction} force, which must be a contribution to the self force. This has led many authors to use the term \emph{radiation reaction} synonymously with self force, however in the fully relativistic case there is an extra term in the self force in addition to the negative of (\ref{larmor_abraham}). This additional term is known as the \emph{Schott term}, and has lead to some controversy. The fully relativistic self force is given by the Abraham-von Laue vector \cite{erber61}
 \begin{align}
\fself=\frac{q^2}{6 \pi \epsilon_0 } (\cdddot-g(\cddot, \cddot)\cdot),
\label{abraham_laue}
 \end{align}
 where the Schott term is third order with respect to the worldline.
 
 The zeroth component of the radiation reaction force is Larmor's equation for the rate of radiation. The spatial part is proportional to the negative of the Newtonian velocity and may be interpreted as the radiation reaction force of the particle.
 The physical nature of the Schott term has been a topic for debate.
 Its presence leads to two interesting results: i) the self force can vanish even when the radiation rate is non-zero, for example in the case of uniform circular motion, and ii) the self force can be non-zero even when there is momentarily no radiation being emitted. Thus the identification of the whole of (\ref{abraham_laue}) as a radiation reaction force would be misleading. The Schott term is a total derivative, so it does not
 correspond to an irreversible loss of momentum by the
 particle, but plays an important role in the momentum
 balance between the radiation and the particle \cite{Rohrlich65}.
 
  With the self force given by (\ref{abraham_laue}) the resulting equation of motion for a charged particle undergoing arbitrary motion is given by the Abraham-Lorentz-Dirac (ALD) equation
 \begin{align}
m \nabla_{\cdot}\cdot=\florentz +\fself +f_{\textup{ext}},
\label{lorentz_dirac}
 \end{align}
 where $m$ is the observed rest mass of the particle, $\florentz\in\Gamma\tan\m$ is the Lorentz force due to the external field, $f_{\textup{ext}}\in\Gamma\tan\m$ is the force due to non-electromagnetic effects \footnote{In general these are not known but could include effects due to gravity or collision with neutral particles. It is common to assume $f_{\textup{ext}}=0$.}, and $\fself\in\Gamma\tan\m$ is given by (\ref{abraham_laue}). All three forces on the right of (\ref{lorentz_dirac}) are vector fields with support on finite closed regions of the worldline\footnote{When looking at the solution of this equation it is often useful to consider the external force (EM or non-EM) to be a pulse \cite{Dirac38, Eliezer, Bonnor}, however this is a mathematical idealization.}   thus Rohrlich's dynamic asymptotic condition \cite{Rohrlich65},
 \begin{align}
 \lim_{|\tau|\rightarrow \infty} \nabla_{\cdot}\cdot=0,
 \end{align}
is satisfied.
  The third order nature of the Schott term has instilled doubts about the validity of the ALD equation, since it leads to particular classes of solution which are foreign to classical physics. These solutions include \emph{preacceleration}, where a particle may begin to accelerate before a force has been applied, and \emph{runaway solutions}, where a particle may continue to accelerate exponentially even for a static force (see \cite{Rohrlich65, Parrott86, Poisson99}).  Further doubts about the validity of the ALD equation are raised by the fact that there remains to this day no derivation of (\ref{lorentz_dirac}) which is completely free from ambiguity. The most widely known difficulty is that of mass renormalization.

 The origin of mass renormalization can be found in the early attempts to calculate the self force based on extended models for the electron. At the dawn of the twentieth century the limitations imposed by quantum physics were unknown and it was widely believed the dynamics of an electron could be established by supposing a classical model for the particle. The model was based on the idea of a macroscopic charged object reduced to the microscopic scale. There is an inherent problem with this approach because macroscopic charged objects are stable only because of the intermolecular forces binding them together. As an elementary particle the electron is necessarily devoid of these forces, thus within such a model the particle would have a tendency to blow itself apart due to the mutual repulsion of its volume elements. The solution of this difficulty, proposed by Poincar$\acute{\textup{e}}$, was to postulate the existence of an additional cohesive force which would exactly cancel the repulsion. This cohesive force would enable the electron to remain stable, however it would by definition have no effect on the motion of the particle and its physical nature would remain unknown.

 If we accept the Poincar$\acute{\textup{e}}$'s hypothesis and assume an extended model for the electron, then the self force may be calculated using the Lorentz force law. It is possible to calculate the Lorentz force acting on a particular volume element due to the rest of the charge distribution. The self force is then given by net force on the particle due to the respective Lorentz forces on each of the volume elements. In order to calculate this force it is necessary to postulate an additional condition on the model, that of rigidity. The most common notion of rigidity is that of Born rigidity, where the particle is rigid in its rest frame.  In the early 1900s Lorentz \cite{Lorentz} and Schott\cite{Schott}, amongst others, were able to calculate the resulting force for a number of different charge distributions. Non-relativistically, for a Born rigid , spherically symmetric charge distribution instantaneously at rest, the calculation yields\cite{Rohrlich65}
\begin{align}
\underline{f_{\textup{self}}}\approx -\frac{2}{3c^2}q\kappa U \vxdd +  \frac{2}{3c^3}q\kappa \vxddd -\frac{2}{3c^2}q\kappa\sum_{n=2}^\infty \frac{(-1)^n}{n!}\frac{d^{n} \vxdd}{c^n dt^{n}}\ord( r^{n-1}) ,
\end{align}
where $\underline{f_{\textup{self}}}$ is a 3-vector, $\vxd$ is the 3-acceleration of the charge and the dot denotes differentiation with respect to time. The constant $\kappa$ is defined by (\ref{def_kappa}) and $r$ denotes the radius of the distribution. The constant $U$ is given by,
\begin{align}
U= \int \int \frac{n(\vx)n(\vx')}{r}d^3x d^3x'
\end{align}
where $n(\vx)/q$ is the normalized charge distribution.
  In the limit $r\rightarrow 0$, i.e. the point charge limit,  the terms in the summation vanish.  The resulting equation of motion for $r\rightarrow 0$ is given by
\begin{align}
m_0\vxdd \approx -\frac{2}{3c^2}q\kappa U \vxdd   + \frac{2}{3c^3}q\kappa \vxddd  +\underline{\florentz} + \underline{f_{\textup{ext}}},
\label{eq_early}
\end{align}
where $m_0$ is the \emph{bare mass} and $\underline{\florentz}$ and $\underline{f_{\textup{ext}}}$ are 3-vectors. We notice the first term on the right hand side is proportional to the acceleration of the electron. This led to the identification of the coefficient $m_e=\frac{2}{3c^2}q\kappa U $ as an electromagnetic contribution to the observed rest mass of the particle. This enables the term to be shifted to the left hand side of (\ref{eq_early}), resulting in the equation of motion
\begin{align}
m\vxdd =  \frac{2}{3c^3}q\kappa \vxddd  +\underline{\florentz}+\underline{f_{\textup{ext}}},
\label{eq_early2}
\end{align}
where $m=m_0+m_e$ is the observed rest mass given by the sum of the electromagnetic mass and bare mass. This is known as the Lorentz-Abraham equation, and is the non-relativistic limit of (\ref{lorentz_dirac}). The process of shifting the term $m_e\vxd$ to the left hand side is known as \emph{mass renormalization}. In the point charge approach mass renormalization is still required, however the electromagnetic mass of the point particle is found to be infinite. This means the bare mass must be assumed to be negatively infinite in order to leave a finite observed mass and a meaningful equation of motion. This process of adding two infinite quantities to give a finite mass is undesirable and brings into question the validity of the resulting equation of motion. \footnote{There have been attempts to eradicate mass normalization, for example see \cite{Norton}, where different mathematical techniques are used to cancel the singular terms, however there remains no physical justification. The inability of classical physics to consistently treat field divergences has lead to further needs for renormalization in quantum field theory.}.

 With the advance of physics since the early twentieth century it is now clear that any notion of rigidity is incompatible with special relativity. It is also known that electrons and other charged elementary particles exhibit wave particle duality and other quantum behavior. This has lead to almost complete abandonment of the macroscopic model in favor of other models which do not cling to the idea of miniature classical distributions of charge. The simplest such model is that of a point charge.  However if we adopt the point charge model from the outset it is not obvious how to define the self force because the \LW field is singular at the position of the particle.  In 1938 Dirac proposed a method by which the self force arises as an integral of the stress-energy-momentum tensor associated with the \LW field.  In 1973 Rohrlich writes \cite{Mehra73}
 \begin{quote}
 Whatever one may think today of Dirac's reasons in developing a classical theory of a point electron, it is by many contemporary views (and I completely concur), the correct thing to do: if one does not wish to exceed the applicability limits of classical (i.e., non-quantum) physics one cannot explore the electron down to distances so short that its structure (whatever it might be) would become apparent. Thus for the classical physicist the electron is a point charge within his limits of observation.
 \end{quote}

\section{The point charge approach and the Schott term discrepancy}

Within the point model framework the components of the instantaneous change in  electromagnetic 4-momentum $\pdot_{\textup{EM}}\in \Gamma \tan \m$ arise as integrals of the \LW stress 3-forms over a suitable three dimensional domain of spacetime. This instantaneous change in 4-momentum is identified as the negative of the self force but with an additional singular term which can be discarded by mass renormalization.
 In Dirac's calculation the domain is the side
$\Sigma^\tDirac_T$ of a narrow tube, of spatial radius
$\rdcon$, enclosing a section of the worldline $\c$. See
FIG. \ref{fig_Tubes}.
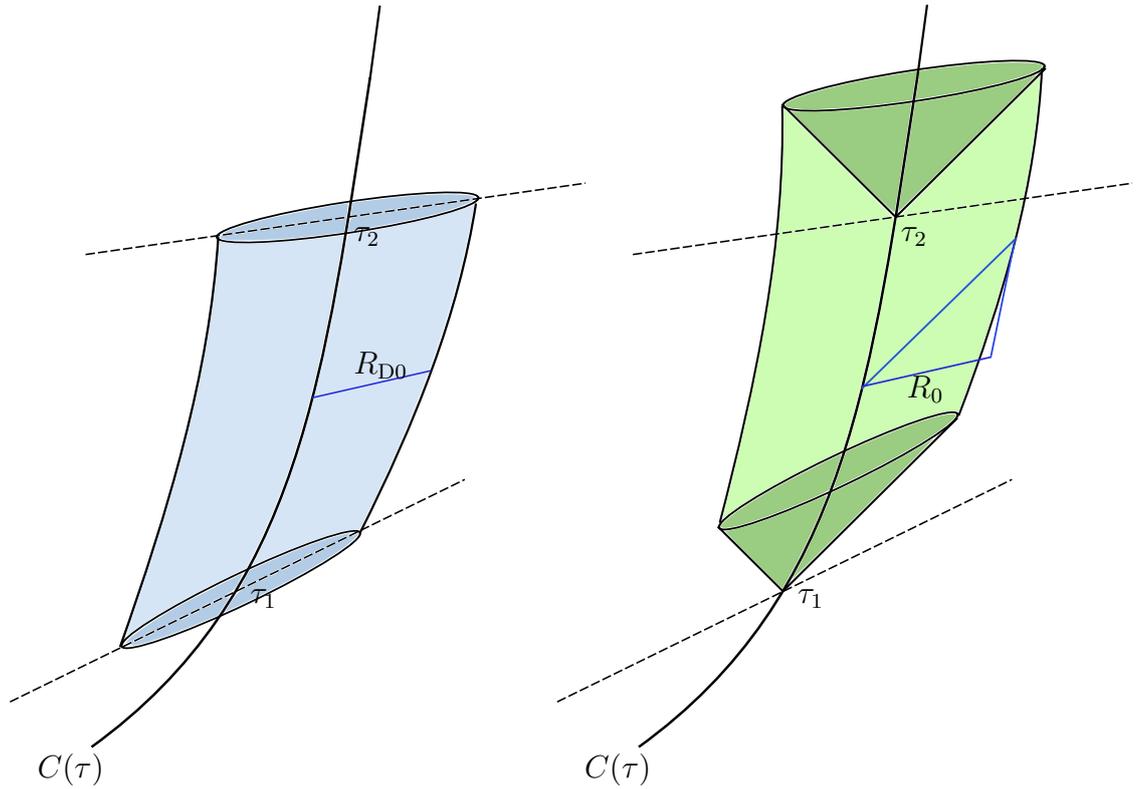
\begin{figure}
\begin{center}
\setlength{\unitlength}{0.8cm}
\begin{picture}(20, 13.5)
\put(0, 0){\includegraphics[ width=10\unitlength]{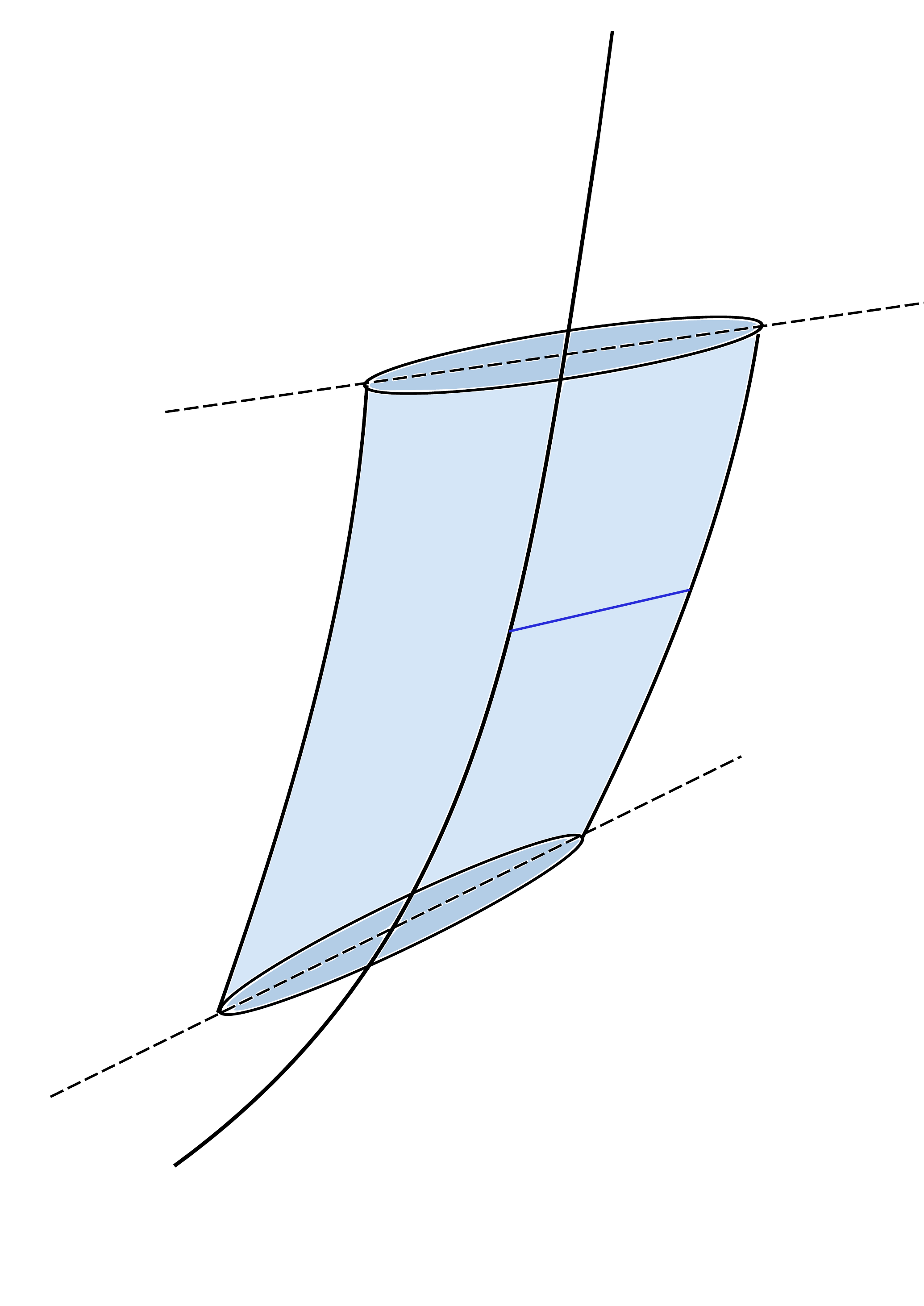}}
\put(9, 0){\includegraphics[ width=10\unitlength]{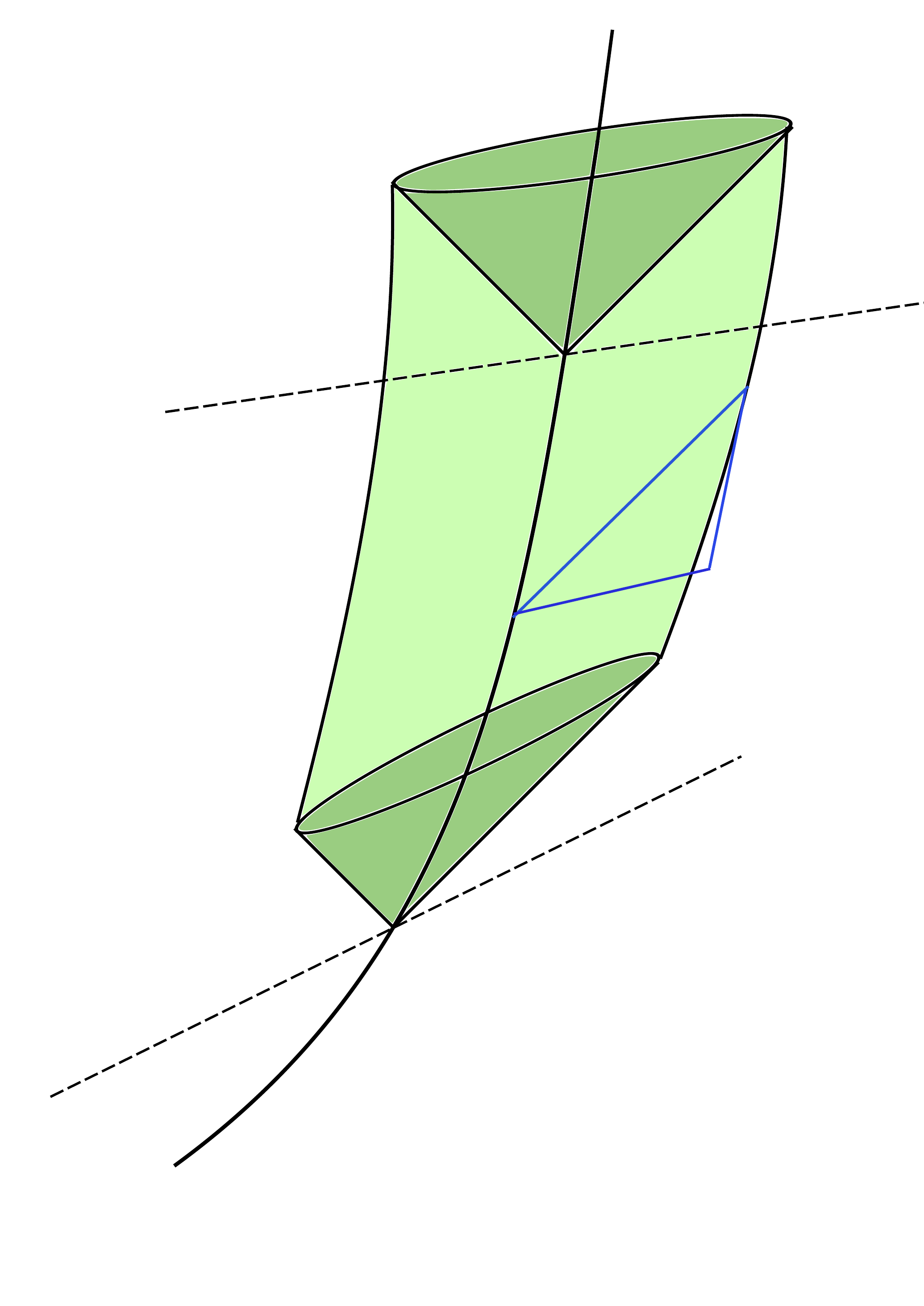}}
\put(6.2, 9.9){$\tau_2$}
\put(4.5, 3.9){$\tau_1$}
\put(15.2, 9.9){$\tau_2$}
\put(13.5, 3.9){$\tau_1$}
\put(10, 1){$\c(\tau)$}
\put(1, 1){$\c(\tau)$}
\put(6.2, 7.7){$\rdcon$}
\put(15.3, 7.3){$\rcon$}
\end{picture}
\end{center}
\caption[The Dirac and Bhabha tubes]{The Dirac tube (left) and Bhabha tube (right)}
\label{fig_Tubes}
\end{figure}
 The displacement vector $\y$ defining
 the Dirac tube is spacelike, therefore the \LW potential
is not naturally given in terms of the Dirac time $\tau_D(\point)$. However in appendix \ref{dirac_apend} we show that for small $\displaystyle{R_D=\g(\y, \y)}$,  and hence small $\tau_D-\tau_r$ due to (\ref{dirac_taur}),  it can be expressed as the series
\begin{align}
\frac{1}{\kappa}\alw|_{\point}=&-\frac{\v_D}{R_D} +\big(\a_D +\frac{1}{2}\g(\n_D, \a_D)\v_D\big)\notag\\
&+\Big(\v_D\big(\frac{1}{8}\g(\a_D, \a_D)-\frac{1}{8}\g(\n_D, \a_D)^2-\frac{1}{3}\g(\n_D, \adot_D)\big) -\frac{1}{2}\adot_D-\frac{1}{2}\g(\n_D, \a_D)\a_D\Big)R_D\notag\\
& +\ord(R_D^2),\tag{\ref{lw_dirac}}
\end{align}
where the vector fields $\v_D, \a_D, \adot_D \in \Gamma \tan \mnotc $ are defined as
\begin{align}
\v_D\atp&=\cdot^j(\tau_D(\point)) \frac{\partial}{\partial \co^j},\quad\a_D\atp=\cddot^j(\tau_D(\point))\frac{\partial}{\partial \co^j}\quad \text{and} \quad\adot_D\atp=\cdddot^j(\tau_D(\point)) \frac{\partial}{\partial \co^j},
\tag{\ref{defdirac_V_A}}
\end{align}

When using a Dirac tube  the integration of the stress 3-forms
gives for the instantaneous EM 4-momentum \cite{Dirac38,Teitelboim70,Galtsov02,Norton}
\begin{align}
-\Pdot^{\textup{D}}_{\textup{EM}}=q
\kappa\Big(\frac{2}{3}\big( \cdddot-g(\cddot, \cddot)\cdot\big)
-\lim_{\rdcon \rightarrow 0}\frac{1 }{2 \rdcon}\cddot\Big),
\label{f_self}
\end{align}

This is the Abraham-von Laue vector with the additional singular term which depends on
the shrinking of the Dirac tube onto the worldline.

An alternative approach, first used by Bhabha\cite{Bhabha39} in 1939 ,
is to integrate the \LW  stress forms over the side
$\Sigma_{\textup{T}}$ of the Bhabha tube with spatial radius $\rcon$. The
principal advantage of this approach is that the displacement vector
$X$ is lightlike and as a result the \LW potential can written explicitly,
\begin{align}
\frac{1}{\kappa}\alw\atp = \frac{\vdual}{\g(\v, \x)}.\tag{\ref{LW_def}}
\end{align}
It follows that the corresponding stress 3-forms can also be written explicitly.
However previous articles which use a
Bhabha tube to evaluate the instantaneous EM 4-momentum give the following expression
\cite{Norton, Bhabha39,Poisson99,Tucker06,Parrott86}
\begin{align}
-\Pdot_{\textup{EM}}&=-q\kappa\Big(\frac{2}{3}g(\cddot,
\cddot)\cdot+\lim_{\rcon \rightarrow 0}\frac{1 }{2 \rcon}\cddot\Big)
=
-\Pdot^{\textup{D}}_{\textup{EM}}-\frac{2}{3}q\kappa\cdddot.
\label{f_self_wrong}
\end{align}
This is the radiation reaction force with the additional singular term which depends on
the shrinking of the Bhabha tube onto the worldline. The Schott term $2q\kappa\frac{\cdddot}{3}$ is missing from the approaches employing the Bhabha tube.
In 2006 Gal'tsov and Spirin \cite{Galtsov02} draw attention to this discrepancy. They claim the Schott term should
arise directly from the electromagnetic stress-energy-momentum tensor and provide a derivation using Dirac geometry in order to show this. However they propose the missing term in \eqref{f_self_wrong} is a consequence of the null geometry used to define the Bhabha tube. We show in this thesis that the term may be obtained using null geometry providing certain conditions are realized.

\section{Regaining the Schott term}


\subsection*{Addition to non-EM momentum}

  The standard approach which has been used in articles \cite{Norton,Bhabha39,Poisson99,Tucker06}, is to simply add the term to the non-EM momentum of the particle. This method will give the correct form for the ALD equation, however it is not physically justified since the self force is by nature an electromagnetic effect.

 We suppose a balance of momentum
\begin{align}
\Pdot_{\text{PART}}  +\Pdot_{\text{EM}}=f_{\text{ext}}
\label{Ppart_PEM}
\end{align}
where total momentum has been separated into electromagnetic contribution $P_{\text{EM}}$ and non-electromagnetic
contribution $P_{\text{PART}}$, and $\dot{P}=\nabla_{\cdot}P$. All the external forces acting on the particle are denoted by $f_{\text{ext}}$.
A suitable choice for the non-electromagnetic momentum
$P_{\text{PART}}$ has to be made.
Most external forces $f_{\text{ext}}$, including the Lorentz force,
 are orthogonal to $\cdot$:
\begin{align}
g(f_{\text{ext}},\cdot)=0.
\label{f_ext_orthog}
\end{align}
For such an external force, if (\ref{f_self}) is obtained then a natural choice for $P_{\text{PART}}$ is
\begin{align}
P_{\text{PART}}=m_0  \cdot.
\label{P_not_inc_Schott}
\end{align}
This is the correct term for the 4-momentum of a particle if its spin has been neglected.
Combining (\ref{f_self}), (\ref{Ppart_PEM}) and
(\ref{P_not_inc_Schott}) gives
\begin{equation}
\begin{aligned}
m_0  \cddot &= f_{\text{ext}}-\Pdot^{\textup{D}}_{\textup{EM}}\\
 &= f_{\text{ext}}-q\kappa\Big(
\tfrac{2}{3}g(\cddot, \cddot)\cdot
+\cdddot
- \lim_{\rdcon \rightarrow 0} \frac{1}{2\rdcon}\cddot\Big).
\end{aligned}
\label{balance}
\end{equation}
Thus assuming the observed rest mass $m$ to be given by
\begin{align}
m=m_0+\lim_{\rdcon \rightarrow 0} \frac{q \kappa}{2\rdcon}\cddot
\label{inf_terms}
\end{align}
 we satisfy the orthogonality condition (\ref{g_Cd_Cdd}).  By contrast, if (\ref{f_self_wrong}) is obtained one cannot set
\begin{align}
m_0\cddot = f_{\text{ext}} -\Pdot_{\textup{EM}}\qquad \textup{and} \qquad m=m_0+\lim_{\rcon \rightarrow 0} \frac{q\kappa}{2\rcon}\cddot
\end{align}
and satisfy (\ref{g_Cd_Cdd}). This has lead some authors \cite{Norton,Bhabha39,Poisson99,Tucker06} to add an ad hoc term
to the non-electromagnetic contribution to the force.
\begin{align}
\dot{P}^{\textup{B}}_{\text{PART}}=m_0  \cddot + \tfrac23q\kappa \cdddot.
\label{P_inc_Schott}
\end{align}
This ad hoc term will ensure the orthogonality condition is satisfied and hence compensate for the missing Schott term.

\subsection*{Regaining the term by careful analysis of limits}

We will show that the calculation of the self force using null geometry requires
three limits to be taken (see figure \ref{limits}), the shrinking of the Bhabha tube $\SigmaT$ onto the
worldline $\c$ i.e. $\rcon\to 0$,  and the bringing together of the lightlike caps $\Sigma_1$ and $\Sigma_2$
onto the lightlike cone with vertex $\c(\tau_0)$ i.e. $\tau_1\to\tau_0$  $\tau_2\to\tau_0$, where $\tau_0$ is the proper time at which we wish to evaluate the self force (see FIG.\ref{fig_Tubes}).
We therefore have the freedom to choose the order of these limits.  We choose to let the three limits take place simultaneously,
subject to the constraint that
\begin{align}
\lambda=\raisebox{0.4cm}{$\displaystyle{\lim_{\substack{\rcon
        \rightarrow 0\\\tau_1 \rightarrow \tau_0\\\tau_2 \rightarrow
        \tau_0}}}$}
\Big(  \frac{\tau_1+\tau_2-2\tau_0}{4\rcon}\Big)
\label{def_lambda}
\end{align}
where $\lambda\in\real$ is finite. This gives the self force as
\begin{align}
f_{\textup{self}}
&=
-q\kappa\Big(
\tfrac{2}{3}g(\cddot, \cddot)\cdot
+\lambda\cdddot
+ \lim_{\rcon \rightarrow 0} \frac{1}{2\rcon}\cddot
\Big)
\label{Intr_Pdot_res_expan}
\end{align}
which is in agreement with $f^{\textup{D}}_{\textup{self}}$ if
$\lambda=-\tfrac23$, hence the Schott term arises by direct integration of the stress forms using null geometry.

\begin{figure}
\setlength{\unitlength}{1cm}
\centerline{
\begin{picture}(14, 19)
\put(0, 0){\includegraphics[ width=14\unitlength]{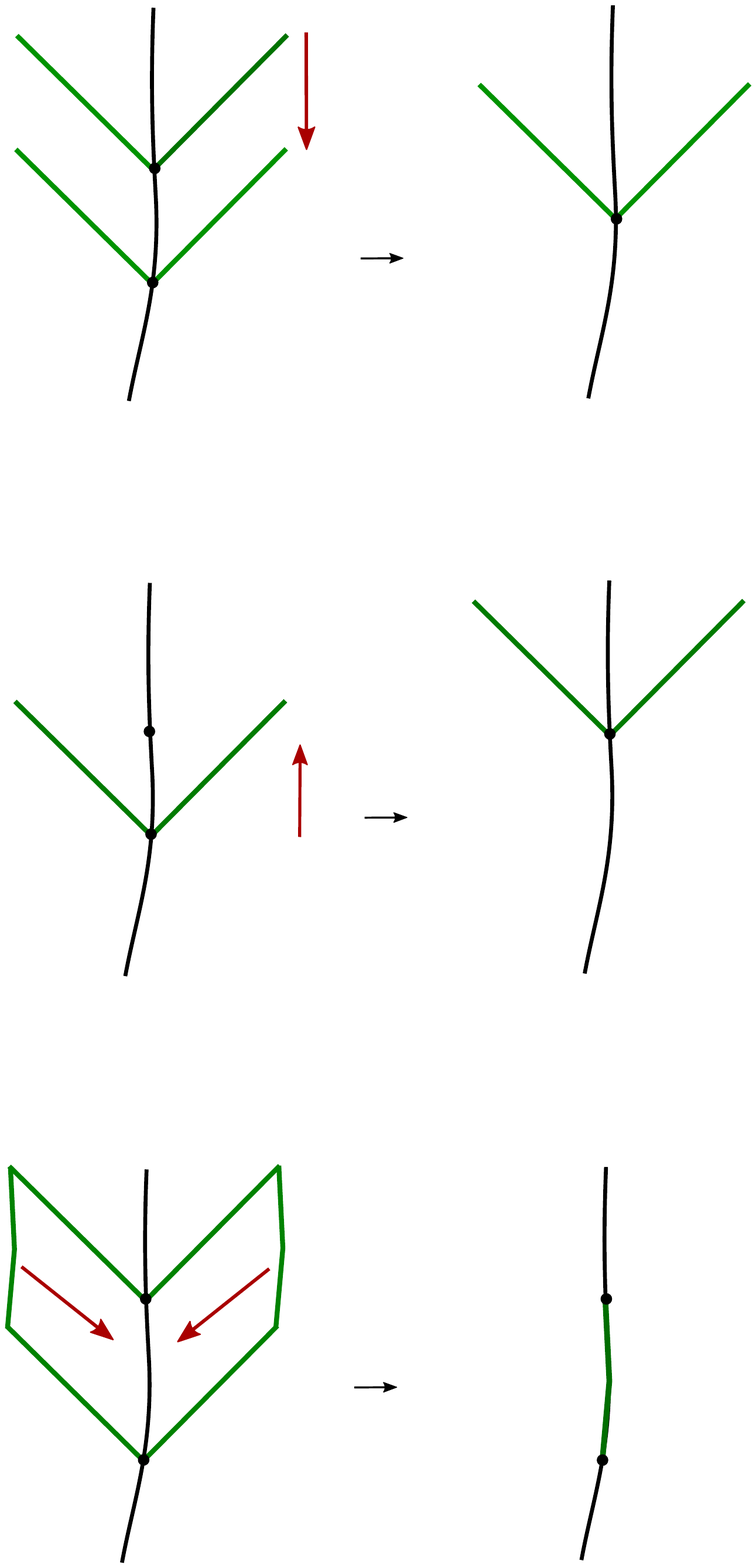}}
\put(4, 17){$\tau_2$}
\put(3.9, 15.7){$\tau_1$}
\put(10.3, 16.4){$\tau_1=\tau_2$}
\put(4, 10.5){$\tau_1=\tau_2$}
\put(3.8, 9.3){$\tau_0$}
\put(9.8, 10.3){$\tau_0=\tau_1=\tau_2$}
\put(3.4, 4.2){$\rcon$}
\put(9.8,3){$\rcon=0$}
\end{picture}
}
\caption[Limits required in definition of self force]{Three independent limits are required. The converging of the caps $\tau_1 \rightarrow \tau_2$, the movement of the apex of the squashed tube to the point where the self force will be evaluated $\tau_2 \rightarrow \tau_0$, and the shrinking of the radius $\rcon \rightarrow 0$.}
\label{limits}
\end{figure}

\chapter{Defining the self force for a point charge}

In this chapter we formally define the Dirac and Bhabha tubes. We also present the definition of the self force based on conservation of four momentum within a Bhabha tube. For this definition we take the limit as the tube approaches an arbitrary point on the worldline.

\section{The Dirac and Bhabha tubes}

\begin{definition}
\label{worldtube_def}
Consider the region $N=\widetilde{N}\backslash \c$ where $\widetilde{N} \subset \m$ is a local neighborhood of the worldline. Suppose the two continuous maps
\begin{align}
&\tau': N \rightarrow \real, \quad x \mapsto \tau'(x)\\
&R' : N \rightarrow \real^+, \quad x \mapsto R' (x),
\end{align}
 are well defined for all $x \in N$. Here $\real^+$ denotes the positive real numbers and we shall call $R'(x)$ the displacement of $x$ from $\c(\tau'(x))$. Furthermore for $\lambda \in \real^+$ let
\begin{align}
&\tau'\Big(\lambda\Big(x-\c\big(\tau'(x)\big)\Big)+\c\big(\tau'(x)\big)\Big)=\tau'(x),\\\notag
\textup{and}\quad & R'\Big(\lambda\Big(x-\c\big(\tau'(x)\big)\Big)+\c\big(\tau'(x)\big)\Big)=\lambda R'(x).
\end{align}
These relations ensure that for an arbitrary point $\tau_0$ on the worldline with $\c (\tau_0) \in \widetilde{N}$
\begin{align}
\lim_{x \rightarrow \c (\tau_0)} \tau'(x) = \tau_0 \qquad \textup{and} \qquad \lim_{x \rightarrow \c (\tau_0)}  R'(x)=0
\end{align}
\end{definition}
Since $N$ is open there exist values $\tau_{\textup{min}}, \tau_{\textup{max}}, R_{\textup{max}} \in N$ such that the 4-region
\begin{align}
\mathds{S}= \Big\{\quad x\quad \Big|\quad \tau_{\textup{min}}<\tau'(x)<\tau_{\textup{max}},\quad 0<R'(x)<\rcon'\quad\},
\end{align}
where $\rcon'<R_{\textup{max}}$, is well defined .
\begin{definition}
\label{max_vals}
The 3-boundary of this region $\mathds{T}=\partial\mathds{S}$ is a known as a \begin{bf}worldtube\end{bf} and is defined by $\mathds{T}=\Sigma_1' \cup \Sigma_2' \cup \SigmaT'$ where for $i=1, 2$
\begin{align}
\Sigma_i' &= \Big\{\quad x\quad\quad \Big|\tau'(x)=\tau_i,\quad 0< R'(x)\leq \rcon'\quad\},
\label{SigmaT_Zcoords}
\end{align}
where $\rcon'$ and  $\tau_i$ are constants and $\tau_i=\tau'(x)$ for all $x\in \Sigma_i'$.
\begin{align}
\SigmaT' &= \Big\{\quad x\quad\Big|\quad\tau_2\leq\tau'\leq\tau_1,\quad R'=\rcon'\quad\}.
\label{SigmaT_Zcoords}
\end{align}
We call $\Sigma_i' $ the caps of the worldtube $\mathds{T}$ and they are surfaces of constant $\tau'$ whose boundaries are topological 2-sheres. We call $\SigmaT'$ the side of $\mathds{T}$ and it is a timelike surface of constant $R'=\rcon'$ topologically equivalent to a cylinder.  We call $\tau'$ the \emph{worldline map} associated with $\mathds{T}'$  and $R'$ the \emph{displacement map}.
\end{definition}

\newpage
\begin{lemma}
\label{dtube_def}
When using Dirac geometry $\tau'=\tau_D$ and $R'=R_D=\sqrt{g(Y, Y)}$. The surfaces $\Sigma^\tDirac_{i}$ are subregions of the planes of simultaneity according to an observer comoving at $\c(\tau_i)$.
The \begin{bf}Dirac tube\end{bf} $\mathds{T}_{D}$ is defined by
$\displaystyle{\Sigma^\tDirac_1 \cup
  \Sigma^\tDirac_2 \cup
  \Sigma^\tDirac_{\textup{T}}}$ where for $i=1, 2$
\begin{align*}
\Sigma^\tDirac_i&= \Big\{\c(\tau_i)+Y  \quad\Big|\quad g(Y,
\cdot)=0,\quad g(Y, Y)<(\rdcon)^2\Big\}\,,
\\
\Sigma^\tDirac_{\textup{T}} &= \Big\{\c(\tau)+Y \quad\Big|\quad g(Y, \cdot)=0\quad g(Y, Y)=(\rdcon)^2, \quad\tau_1\leq\tau\leq\tau_2\Big\}.
\end{align*}
The parameter  $\rdcon>0$ is a measure of the cross-sectional radius of the Dirac tube, see figure \ref{D_Tube}.
\end{lemma}

\begin{figure}
\begin{center}
\setlength{\unitlength}{1.2cm}
\begin{picture}(10, 11)
\put(0, -2){\includegraphics[ width=10\unitlength]{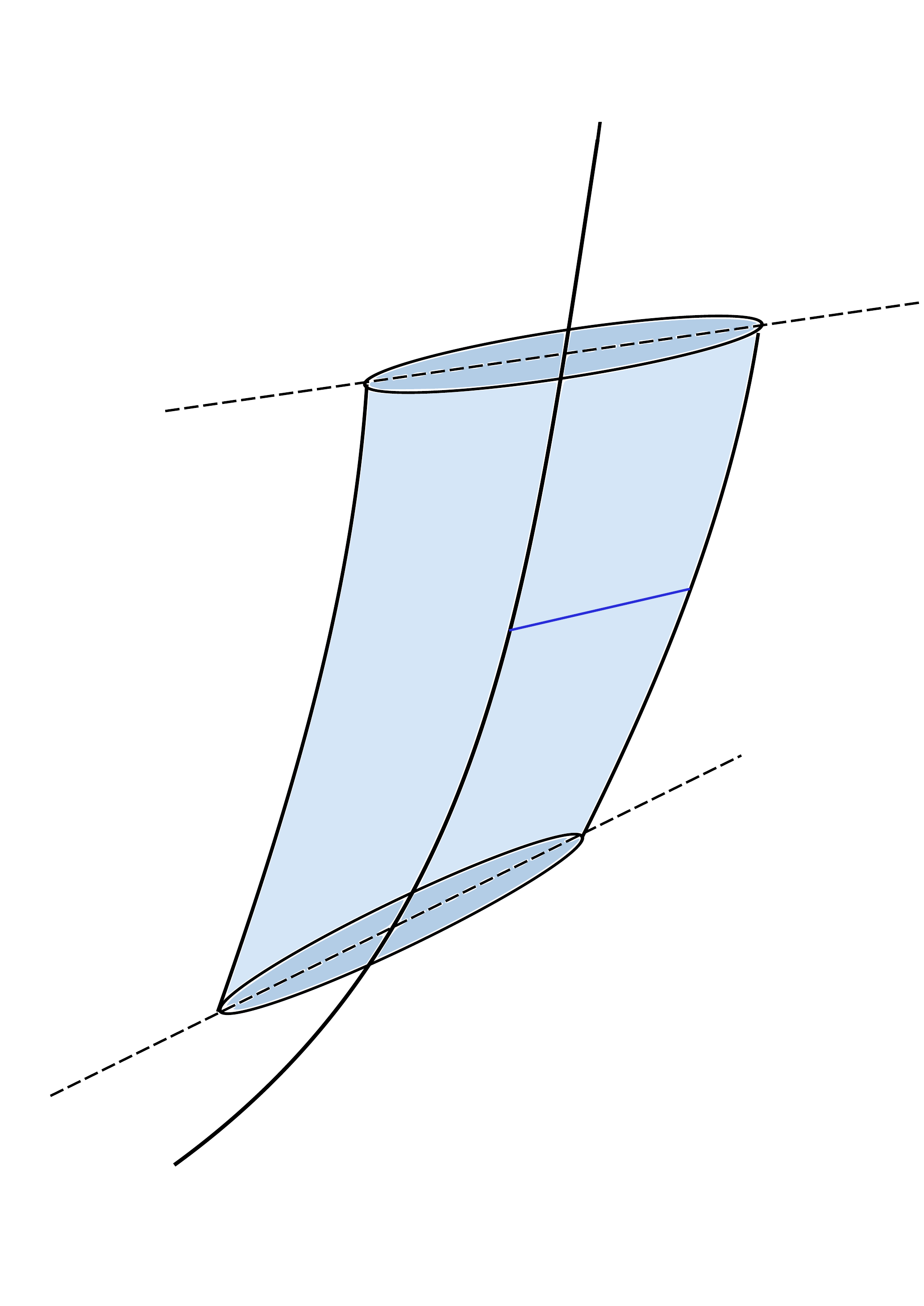}}
\put(4, 6){$\SigmaT^D$}
\put(5, 8){$\Sigma_2^D$}
\put(3, 1.2){$\Sigma_1^D$}
\put(6.2, 8.3){$\tau_0+\delta$}
\put(4.4, 2){$\tau_0$}
\put(6.2, 5.7){$\rdcon$}
\end{picture}
\caption{The Dirac Tube}
\end{center}
\label{D_Tube}
\end{figure}
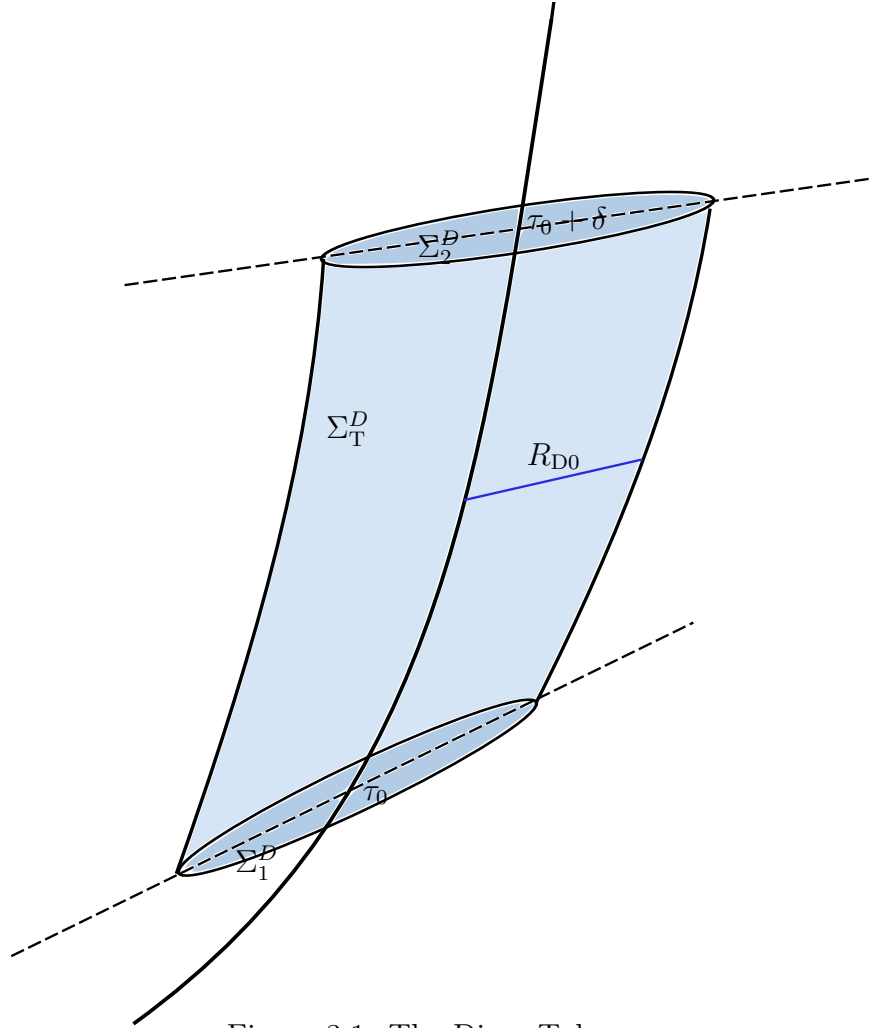
\newpage
\begin{lemma}
\label{btube_def}
When using null geometry $\tau'=\tau_r$ and $R'=R=-\g(\x, \v)$. The surfaces $\Sigma_{i}$ are subregions  of the forward null cones at $\c(\tau_i)$.
The \begin{bf}Bhabha tube\end{bf} $\mathds{T}_B$ is given by $\displaystyle{\Sigma_1 \cup \Sigma_2 \cup \SigmaT}$ where for $i=1, 2$
\begin{align}
\Sigma_i &= \Big\{\c(\tau_i)+X\quad\Big|\quad g(X, X)=0,\quad -g(X,
\cdot)<\rcon\Big\}\,,\notag
\\
\SigmaT &= \Big\{\c(\tau)+X\quad\Big|\quad g(X,
X)=0,\quad-g(X, \cdot)=\rcon, \quad\tau_1\leq\tau\leq\tau_2\Big\}.
\label{Bhabha_Tube}
\end{align}
The parameter  $\rcon>0$ is a measure of the cross-sectional radius of the Bhabha tube, see figure \ref{B_tube}.
\end{lemma}

\begin{figure}
\begin{center}
\setlength{\unitlength}{1.2cm}
\begin{picture}(10, 12)
\put(0, -2){\includegraphics[ width=10\unitlength]{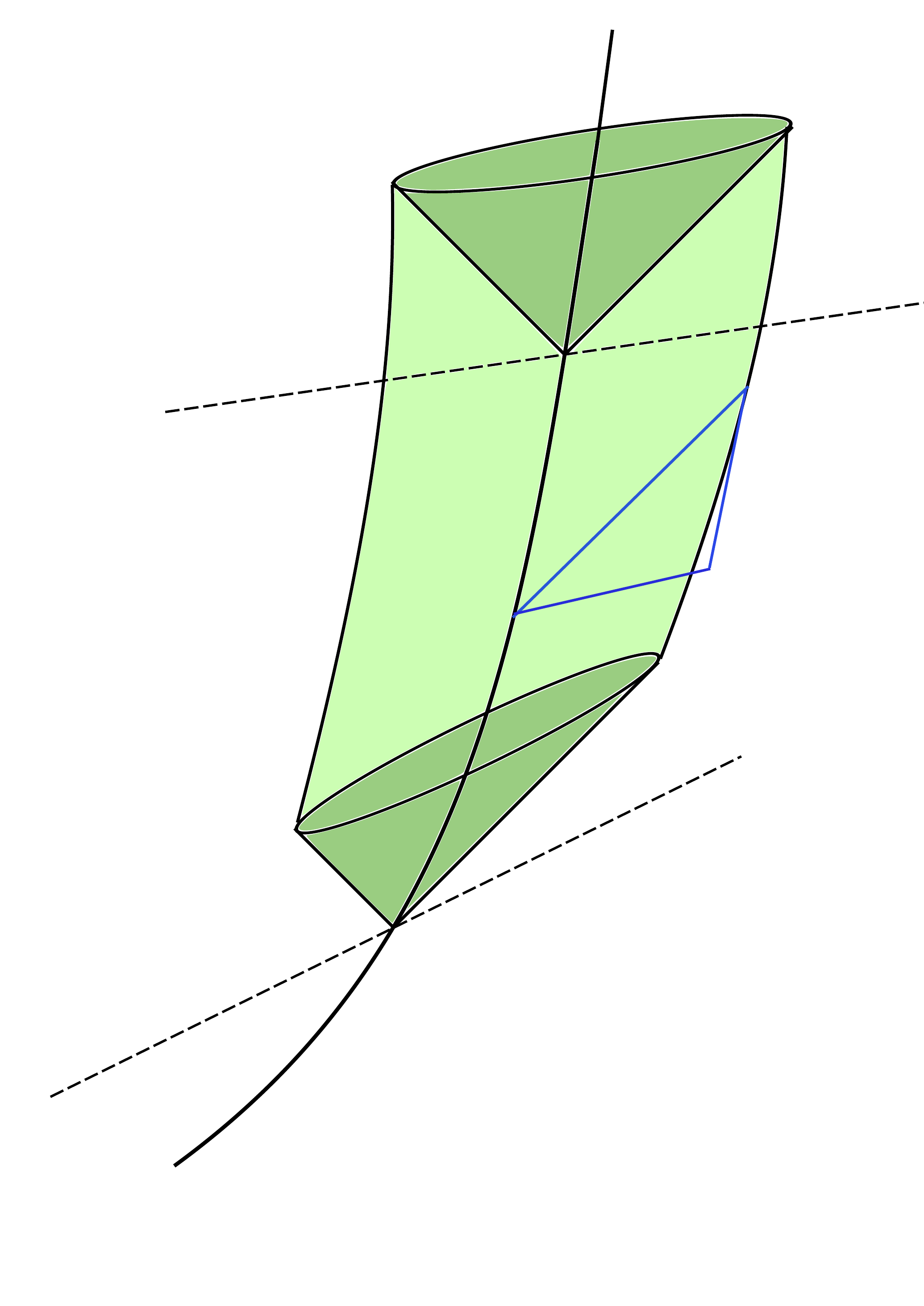}}
\put(4, 6){$\SigmaT$}
\put(5.4, 9.2){$\Sigma_2$}
\put(3.8, 2.6){$\Sigma_1$}
\put(6.2, 8.3){$\tau_0+\delta$}
\put(4.4, 2){$\tau_0$}
\put(6.2, 5.3){$\rcon$}
\end{picture}
\end{center}
\caption{The Bhabha Tube}
\label{B_tube}
\end{figure}

\section{Conservation of 4-momentum}

\begin{lemma}
\label{cons_lem}
Consider figure \ref{N_tube}. The two Bhabha tubes $\mathds{T}^{\textup{in}}$ and $\mathds{T}^{\textup{out}}$ given by
 \begin{align}
 \mathds{T}^{\textup{in}}=& \Sigma^{\textup{in}}_1 \cup \Sigma^{\textup{in}}_2 \cup \SigmaT^{\textup{in}},\notag\\
 \textup{and} \qquad \mathds{T}^{\textup{out}}=& \Sigma^{\textup{out}}_1 \cup \Sigma^{\textup{out}}_2 \cup \SigmaT^{\textup{out}},
 \end{align}
  have different radii $R^{\textup{in}}$ and $R^{\textup{out}}$. The surfaces $\Sigma^{\textup{diff}}_1$ and $\Sigma^{\textup{diff}}_2$ are the differences between the caps of the two tubes,
   \begin{align}
  \Sigma^{\textup{diff}}_1=& \Sigma^{\textup{out}}_1\backslash \Sigma^{\textup{in}}_1\notag\\
  \textup{and}\qquad  \Sigma^{\textup{diff}}_2=& \Sigma^{\textup{out}}_2\backslash \Sigma^{\textup{in}}_2.
  \end{align}
  Let the 4-region $\mathds{S}$ enclosed by the two tubes be finite and source free, with boundary
  \begin{align}
  \partial \mathds{S}= \Sigma^{\textup{diff}}_1-\Sigma^{\textup{diff}}_2+\SigmaT^{\textup{out}}-\SigmaT^{\textup{in}}.
  \end{align}
   then the following relation is true
  \begin{align}
\int_{\Sigma^{\textup{diff}}_1}\stresslw_{\kkill} -\int_{\Sigma^{\textup{diff}}_2}\stresslw_{\kkill} =\int_{\SigmaT^{\textup{in}}}\stresslw_{\kkill} -\int_{\SigmaT^{\textup{out}}}\stresslw_{\kkill}
\label{mom_flux}
\end{align}
\end{lemma}
\begin{proof}
Using Stokes Theorem (\ref{stokes_def}) and (\ref{dtauk_0}) it follows that
\begin{align}
0=&\int_{N}d \stresslw_{\kkill} =\int_{\partial N} \stresslw_{\kkill} \notag\\
=&\int_{\Sigma^{\textup{diff}}_1}\stresslw_{\kkill} -\int_{\Sigma^{\textup{diff}}_2}\stresslw_{\kkill} +\int_{\SigmaT^{\textup{out}}}\stresslw_{\kkill} -\int_{\SigmaT^{\textup{in}}}\stresslw_{\kkill} ,
\label{stokes_N}
\end{align}
\end{proof}

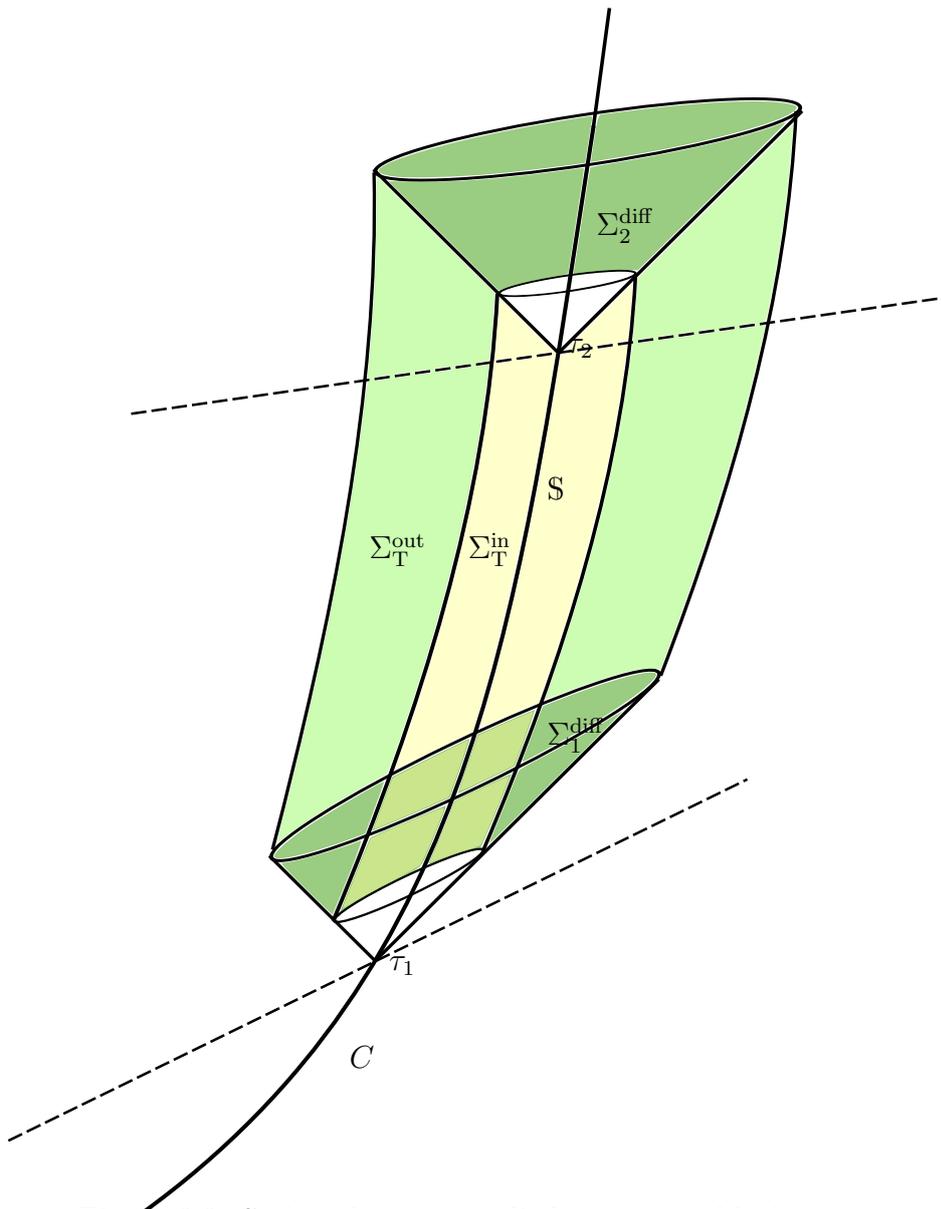
\begin{figure}
\setlength{\unitlength}{1.3cm}
\centerline{
\begin{picture}(10, 12)
\put(0, -2){\includegraphics[ width=10\unitlength]{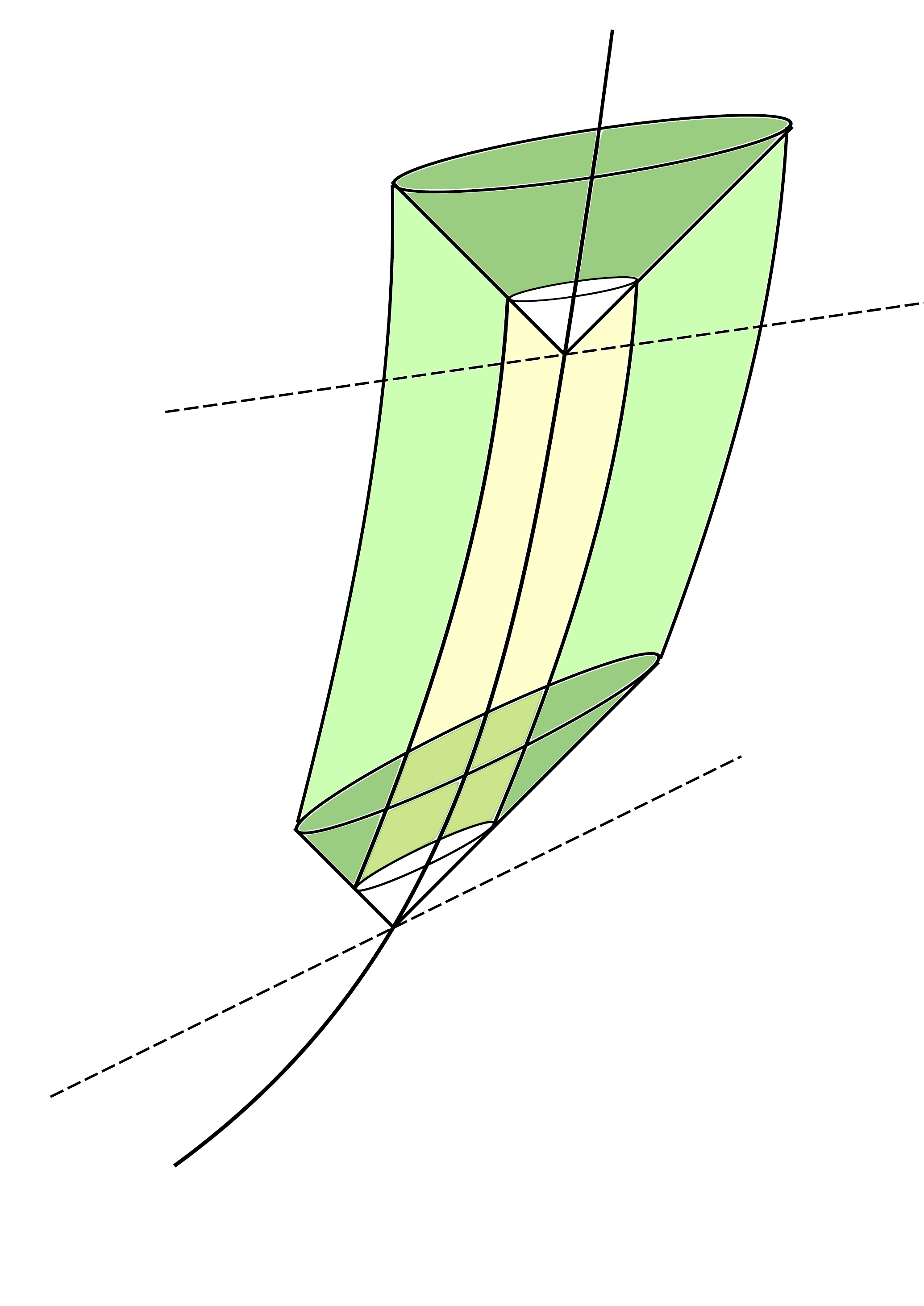}}
\put(6.2, 8.3){$\tau_2$}
\put(4.4, 2){$\tau_1$}
\put(5.2, 6.2){$\SigmaT^{\textup{in}}$}
\put(6.5, 9.5){$\Sigma^{\textup{diff}}_2$}
\put(6, 4.3){$\Sigma^{\textup{diff}}_1$}
\put(4.2, 6.2){$\SigmaT^{\textup{out}}$}
\put(6, 6.8){$\mathds{S}$}
\put(4,1){$\c$}
\end{picture}
}
\caption{Stokes theorem applied to two worldtubes}
\label{N_tube}
\end{figure}

\section{The instantaneous force at an arbitrary point on the worldline}

\begin{definition}
\label{mom_flux_def}
The k-component of 4-momentum flux $\mathrm{P}_{\kkill} ^{(\Sigma)}\in\Gamma\Lambda^0 \man $ through an arbitrary source-free 3-surface $\Sigma\subset \m$ is defined by
\begin{align}
\mathrm{P}_{\kkill} ^{(\Sigma)} &=\int_{\Sigma} \stresslw_{\kkill} ,
\label{def_Pk}
\end{align}
\end{definition}


We cannot use definition \ref{mom_flux_def} to give the flux of momentum through the caps $\Sigma_i$ of the Bhabha tube because they are not source free; there is a singularity at the point intersected by the worldline. Instead in the following we employ Stokes theorem in order to heuristically justify defining the difference in momentum between the two caps as an integral over the side $\SigmaT$.

\begin{definition}
\label{pk_def}
Consider the Bhabha tube $\mathds{T}= \Sigma_1 \cup\Sigma_2\cup\SigmaT$ with $R=\rcon$ where $\rcon$ is constant.
Let $\tau_1=\tau_r(\Sigma_1)$ and $\tau_2=\tau_r(\Sigma_2)$ with $\tau_2>\tau_1$. The instantaneous change in 4-momentum at arbitrary proper time $\tau_0$ is defined by
\begin{align}
\Pdot_{\kkill}  (\tau_0)=&\raisebox{0.4cm}{$\displaystyle{\lim_{\substack{\rcon \rightarrow 0\\\tau_1 \rightarrow \tau_0\\\tau_2 \rightarrow \tau_0}}}$}\bigg(\frac{1}{\tau_2-\tau_1}\int_{\SigmaT}\stresslw_{\kkill} \bigg).
\label{def_Pdot}
\end{align}
\end{definition}

This definition
is justified heuristically as follows. Inspired by definition \ref{mom_flux_def} we wish to
write
\begin{align}
\Delta \mathrm{P}_{\kkill} =\P_{\kkill} ^{(\Sigma_2)}-\P_{\kkill} ^{(\Sigma_1)}= \int_{\SigmaT}\stresslw_{\kkill}
\label{def_Pk_diff}
\end{align}
However the integrals  $\P_{\kkill} ^{(\Sigma_1)}$ and $\P_{\kkill} ^{(\Sigma_2)}$ are
both infinite \footnote{There are methods by which these infinities can be avoided, for example Rowe \cite{Rowe75} uses distribution theory in order to eradicate singularity,
and Norton \cite{Norton} uses a method where the zero limit is not used.}. Ignoring this, we assert
\begin{equation}
\begin{aligned}
\Pdot_{\kkill}  (\tau_0) = \raisebox{0.4cm}{$\displaystyle{\lim_{\substack{\rcon \rightarrow 0\\\tau_1 \rightarrow \tau_0\\\tau_2 \rightarrow \tau_0}}}$}\bigg(\frac{1}{\tau_2-\tau_1}\Big(\P_{\kkill} ^{(\Sigma_2)}-\P_{\kkill} ^{(\Sigma_1)}\Big)\bigg)
\end{aligned}
\label{def_Pk_dot}
\end{equation}
Inserting (\ref{def_Pk_diff}) into (\ref{def_Pk_dot}) yields (\ref{def_Pdot}).
\begin{definition}
\label{pdem_def}
The vector $\Pdot_{\text{EM}}(\tau_0) \in  T_{\c(\tau_0)}\m$ is defined by
\begin{align}
\Pdot_{\text{EM}}(\tau_0)=\Pdot_{\kkill} (\tau_0) g^{\kkill l}\frac{\partial}{\partial y^l}
\label{def_Pdot_dual}
\end{align}
where $g^{\kkill l}=g^{-1}(dy^\kkill, dy^l)$ and $g^{-1}$ is the inverse metric on $\m$.
Since $\tau_0$ is arbitrary there is an induced vector field $\Pdot_{\text{EM}}$ on the curve $\c$.
\end{definition}

\begin{lemma}
\label{stress_tensor_def}
Let \LW stress 3-forms $\stresslw_{\kkill} \in \Gamma \Lambda^{3}\m$ be given by
\begin{align}
\stresslw_{\kkill}= \stresslw_{\kkill}^{\textup{R}}+\stresslw_{\kkill}^{\textup{C}}+ \stresslw_{\kkill}^{\textup{R C}}
\end{align}
where
\begin{align}
\stresslw_{\kkill}^{\textup{R}}=& \frac{\ep}{2 }\big(i_{\partial_{\kkill} }\fr\wedge\star\fr-i_{\partial_{\kkill} }\star\fr\wedge\fr\big)\notag\\
\stresslw_{\kkill}^{\textup{C}}=&\frac{\ep}{2 }\big(i_{\partial_{\kkill} }\fc\wedge\star\fc-i_{\partial_{\kkill} }\star\fc\wedge\fc\big)\notag\\
\stresslw_{\kkill}^{\textup{RC}}=&\frac{\ep}{2 }\big(i_{\partial_{\kkill} }\fc\wedge\star\fr-i_{\partial_{\kkill} }\star\fc\wedge\fr+i_{\partial_{\kkill} }\fr\wedge\star\fc-i_{\partial_{\kkill} }\star\fr\wedge\fc\big)\label{srr}
\end{align}
Then
\begin{align}
\stresslw_{\kkill}^{\textup{R}}=&q^2\frac{\g(\a, \a)-\g(\n, \a)^2}{16\pi^2 \epsilon_0 R^2}\n_{\kkill} \star\widetilde{\n},\notag\\
\stresslw_{\kkill}^{\textup{C}}=&\frac{q^2}{16\pi^2 \epsilon_0R^4}\big(\n_{\kkill} \star\vdual + \v_{\kkill}  \star\widetilde{\n} - \n_{\kkill} \star\widetilde{\n} + \g_{ \kkill a}\star dx^a\big),\notag\\
\stresslw_{\kkill}^{\textup{RC}}=&-\frac{q^2}{8\pi^2\epsilon_0 R^3}\big(\a_{\kkill} \star\widetilde{\n} + \n_{\kkill} \star\adual +\g(\a, \n)(\v_{\kkill} \star\widetilde{\n} +\n_{\kkill} \star\vdual-2\n_{\kkill} \star\widetilde{\n})\big).
\label{null_stress}
 \end{align}
 Note that the factor of $\ccon$ in (\ref{stress_def})  is absent from (\ref{srr}). This is due to our decision to assign the dimension of length to proper time.
  \end{lemma}
\begin{proof}\\
We will show only for for $\stresslw_{\kkill}^{\textup{R}}$, the other results follow similarly.
From (\ref{FR_def}) we have
\begin{align}
\frac{1}{\kappa}i_{\partial_{\kkill}}\fr=& \frac{1}{\g(\x, \v)^2}i_{\partial_{\kkill}}(\xdual \wedge \adual)- \frac{\g(\x, \a)}{\g(\x, \v)^3}i_{\partial_{\kkill}}(\xdual \wedge \vdual) \notag\\
=&\frac{1}{\g(\x, \v)^2}(\x_{\kkill} \adual - \a_{\kkill} \xdual)- \frac{\g(\x, \a)}{\g(\x, \v)^3} (\x_{\kkill} \vdual - \v_{\kkill} \xdual)\notag
\end{align}
where $\kappa$ is defined by (\ref{def_kappa}). Thus
\begin{align}
\frac{1}{\kappa^2}i_{\partial_{\kkill}}\fr\wedge \star \fr=& \frac{1}{\g(\x, \v)^4} \Big( \x_{\kkill} \adual \wedge \star (\xdual \wedge \adual) -\a_{\kkill} \xdual \wedge \star (\xdual \wedge \adual)\Big)\notag\\
&- \frac{\g(\x, \a)}{\g(\x, \v)^5}\Big( \x_{\kkill} \adual \wedge \star (\xdual \wedge \vdual) -\a_{\kkill} \xdual \wedge \star (\xdual \wedge \vdual)\Big)\notag\\
&-\frac{\g(\x, \a)}{\g(\x, \v)^5}\Big( \x_{\kkill} \vdual \wedge \star (\xdual \wedge \adual) -\v_{\kkill} \xdual \wedge \star (\xdual \wedge \adual)\Big)\notag\\
&+\frac{\g(\x, \a)^2}{\g(\x, \v)^6}\Big( \x_{\kkill} \vdual \wedge \star (\xdual \wedge \vdual) -\v_{\kkill} \xdual \wedge \star (\xdual \wedge \vdual)\Big)
\end{align}
Now using lemma \ref{one_two} yields
\begin{align}
\frac{1}{\kappa^2}i_{\partial_{\kkill}}\fr\wedge \star \fr=&\frac{1}{\g(\x, \v)^4} \Big( \x_{\kkill}\g(\a, \a)\star \xdual - \x_{\kkill}\g(\a, \x)\star \adual-\a_{\kkill}\g(\x, \a)\star \xdual\Big)\notag\\
&- \frac{\g(\x, \a)}{\g(\x, \v)^5}\Big( -\x_{\kkill}\g(\a, \x)\star \vdual -\a_{\kkill}\g(\v, \x)\star \xdual \Big)\notag\\
&- \frac{\g(\x, \a)}{\g(\x, \v)^5}\Big( -\x_{\kkill}\g(\v, \x)\star \adual -\v_{\kkill}\g(\a, \x)\star \xdual \Big)\notag\\
&+\frac{\g(\x, \a)^2}{\g(\x, \v)^6}\Big(-\x_{\kkill}\star \xdual -\x_{\kkill}\g(\v, \x)\star \vdual -\v_{\kkill}\g(\v, \x)\star \xdual \Big)\notag
\end{align}
Expanding and cancelling like terms yields
\begin{align}
 \frac{1}{\kappa^2}i_{\partial_{\kkill}}\fr\wedge \star \fr=&\Bigg(\g(\a, \a)-\frac{\g(\x, \a)^2}{\g(\x, \v)^2}\Bigg)\frac{\x_{\kkill}\star \xdual}{\g(\x, \v)^4}
 \end{align}
 Now we need to calculate the second term in $\stresslw_{\kkill}^{\textup{R}}$ in (\ref{srr}). From (\ref{FR_def}) we have
 \begin{align}
\frac{1}{\kappa^2} i_{\partial_{\kkill} }\star\fr\wedge\fr=&  \frac{1}{\g(\x, \v)^4} \Big( i_{\partial_{\kkill}} \star(\xdual\wedge \adual)\wedge \xdual \wedge \adual\Big)\notag\\
&- \frac{\g(\x, \a)}{\g(\x, \v)^5}\Big( i_{\partial_{\kkill}} \star(\xdual\wedge \adual)\wedge \xdual \wedge \vdual\Big)\notag\\
&-\frac{\g(\x, \a)}{\g(\x, \v)^5}\Big( i_{\partial_{\kkill}} \star(\xdual\wedge \vdual)\wedge \xdual \wedge \adual\Big)\notag\\
&+\frac{\g(\x, \a)^2}{\g(\x, \v)^6}\Big( i_{\partial_{\kkill}} \star(\xdual\wedge \vdual)\wedge \xdual \wedge \vdual\Big)
\label{frfrstar}
\end{align}
Now using lemma \ref{fiveoneforms} yields
\begin{small}
\begin{align}
i_{\partial_{\kkill}} \star(\xdual\wedge \adual)\wedge \xdual \wedge \adual=& \big(\a_{\kkill} \g(\x, \a)-\x_{\kkill} \g(\a, \a)\big)\star \xdual+ \x_{\kkill} \g(\x, \a)\star \adual - \g(\x, \a)^2 \g_{\kkill a}\star d y^a\notag\\
i_{\partial_{\kkill}} \star(\xdual\wedge \adual)\wedge \xdual \wedge \vdual=& \g(\a, \x)\v_{\kkill}\star \xdual+ \x_{\kkill} \g(\x, \v)\star \adual - \g(\x, \a)\g(\x, \v) \g_{\kkill a}\star d y^a\notag\\
i_{\partial_{\kkill}} \star(\xdual\wedge \vdual)\wedge \xdual \wedge \adual=&\g(\a, \x)\x_{\kkill}\star \vdual- \x_{\kkill} \g(\x, \a)\star \xdual - \g(\x, \a)\g(\x, \v)\g_{\kkill a} \star d y^a\notag\\
i_{\partial_{\kkill}} \star(\xdual\wedge \vdual)\wedge \xdual \wedge \vdual=&\big(\v_{\kkill} \g(\x, \v)+\x_{\kkill} \big)\star \xdual+ \x_{\kkill} \g(\x, \v)\star \vdual - \g(\x, \v)^2 \g_{\kkill a}\star d y^a
\label{relslem}
\end{align}
\end{small}
 Substituting (\ref{relslem}) into (\ref{frfrstar}) and expanding yields
\begin{align}
\frac{1}{\kappa^2}i_{\partial_{\kkill} }\star\fr\wedge\fr=& \Bigg(\frac{\g(\x, \a)^2}{\g(\x, \v)^2}-\g(\a, \a)\Bigg)\frac{\x_{\kkill}\star \xdual}{\g(\x, \v)^4}
\end{align}
Thus
\begin{align}
\stresslw_{\kkill}^{\textup{R}}=& \kappa^2 \frac{\ep}{2 }\big(i_{\partial_{\kkill} }\fr\wedge\star\fr-i_{\partial_{\kkill} }\star\fr\wedge\fr\big)\notag\\
=& \frac{q^2}{16 \pi^2 \epsilon_0}\Bigg(\g(\a, \a)-\frac{\g(\x, \a)^2}{\g(\x, \v)^2}\Bigg)\frac{\x_{\kkill}\star \xdual}{\g(\x, \v)^4}
\end{align}

\end{proof}
\chapter{The resulting expression}

In this chapter we give the result of carrying out the integration in definition (\ref{pk_def}). We use a computer to carry out the calculation and make use of the Newman-Unti coordinate system. The input code for the MAPLE mathematical software can be found in appendix \ref{maple_ptwo}. Here we state the result.

\section{Arbitrary co-moving frame}
Setting $\tau=\tau_0+\delta$ we expand $\stress_{\kkill} $ in powers of $\delta$. We adapt the global Lorentz frame such that
\begin{equation}
\begin{aligned}
&y^j(\c(\tau_0))=0\quad\text{for} \quad j=0, 1, 2, 3
\\
\text{and} \quad \quad&\cdot(\tau_0)= \frac{\partial}{\partial y^0},\quad\cddot(\tau_0)= a\frac{\partial}{\partial y^3},\quad \cdddot(\tau_0)=b^j \frac{\partial}{\partial y^j}
\end{aligned}
\label{def_x_i}
\end{equation}
where $a,b^j \in\real$ are constants given by
\begin{align}
a=\sqrt{g(\cddot(\tau_0), \cddot(\tau_0))}, \quad\quad
b^j=dy^j(\cdddot(\tau_0))
\label{def_ab}
\end{align}
and from (\ref{g_Cd_Cdd})
\begin{align*}
b^{0}=a^2
\end{align*}
Thus expanding $\cdot$ and $\cddot$ we have
\begin{equation}
\begin{aligned}
\cdot(\delta+\tau_0)
&=
\Big(1+ \frac{b^0}{2}\delta^2\Big)\frac{\partial}{\partial
  y^0}+\frac{b^1}{2}\delta^2\frac{\partial}{\partial y^1}+
\frac{b^2}{2}\delta^2\frac{\partial}{\partial y^2}+\Big(a \delta+
\frac{b^3}{2}\delta^2\Big)
\frac{\partial}{\partial y^3}+\ord(\delta^3),
\\
\cddot(\delta+\tau_0)
&=
b^0 \delta \frac{\partial}{\partial y^0}+b^1 \delta \frac{\partial}{\partial y^1}+b^2 \delta \frac{\partial}{\partial y^2}+\Big(a+b^3 \delta\Big) \frac{\partial}{\partial y^3}+\ord(\delta^2)
\end{aligned}
\label{cd_expansion}
\end{equation}
From (\ref{def_V_A}) and (\ref{tau_taur_R}) we have
\begin{align}
\v|_{(\delta+\tau_0,R,\theta,\phi)}
&=
\cdot(\delta+\tau_0)
\qquad\text{and}\qquad
\a|_{(\delta+\tau_0,R,\theta,\phi)}=\cddot(\delta+\tau_0)
\label{va_expand}
\end{align}
It is useful to express $\v$ and $\a$ in mixed coordinates, with the
basis vectors in terms of the global Lorentz coordinates, but the
coefficients expressed in terms of the Newman-Unti coordinates.

\subsection*{The result}
Here we outline the steps taken in the MAPLE code.
We begin with the expression (\ref{LW_coords}) for the \LW potential $\alw$ in Newman-Unti
coordinates. Taking the exterior derivative we obtain
the field $2-$form $\flw$ and its Hodge dual $\displaystyle{\star\flw}$.
We obtain expressions for the four translational Killing vectors
$\frac{\partial}{\partial y^k}$ in Newman-Unti coordinates and using
(\ref{stress_def}) we obtain expressions for the four electromagnetic
stress $3-$forms $\stresslw_{\kkill} $. Substituting the expansions (\ref{cd_expansion})
into these expressions we obtain the integrands, and finally using (\ref{SigmaT_coords})
 we integrate over $\SigmaT$. The result is
\begin{equation}
\begin{aligned}
\frac{1}{q\constn}\int_{\SigmaT} S_{0}&= -\frac{1}{4}b^0 \frac{\delta_2^2-\delta_1^2}{\rcon}-\frac{2}{3}a^2(\delta_2-\delta_1)-\frac{2}{3}a b^3(\delta_2^2-\delta_1^2)+\ord(\delta_1^3)+\ord(\delta_2^3),
\\
\frac{1}{q\constn}\int_{\SigmaT} S_{1}&= +\frac{1}{4}b^1 \frac{\delta_2^2-\delta_1^2}{\rcon}+\ord(\delta_1^3)+\ord(\delta_2^3),
\\
\frac{1}{q\constn}\int_{\SigmaT} S_{2}&= +\frac{1}{4}b^2 \frac{\delta_2^2-\delta_1^2}{\rcon}+\ord(\delta_1^3)+\ord(\delta_2^3),
\\
\frac{1}{q\constn}\int_{\SigmaT} S_{3}&=+\frac{1}{4}b^3
\frac{\delta_2^2-\delta_1^2}{\rcon}+\frac{1}{2}a\frac{\delta_2-\delta_1}{\rcon}+\frac{1}{3}a^3(\delta_2^2-\delta_1^2)
+\ord(\delta_1^3)+\ord(\delta_2^3)
\end{aligned}
\label{int_Sk}
\end{equation}
where
\begin{align*}
\delta_1=\tau_1-\tau_0, \quad\quad \delta_2=\tau_2-\tau_0
\end{align*}
and $\kappa$ is given by (\ref{def_kappa}).

Combining (\ref{int_Sk}) into a single expression and using
(\ref{def_Pdot}) and (\ref{def_Pdot_dual}) we obtain the following
expression for $\Pdot(\tau_0)\in  \textup{T}_{\c(\tau_0)}\m$
\begin{align}
\frac{1}{q\constn}\Pdot(\tau_0)=&\tfrac{2}{3}a^2 \tfrac{\partial}{\partial y^0} +\lim_{\rcon \rightarrow 0} \frac{1}{2\rcon}a \tfrac{\partial}{\partial y^3}+ \raisebox{0.4cm}{$\displaystyle{\lim_{\substack{\rcon  \rightarrow 0\\\tau_1 \rightarrow \tau_0\\\tau_2 \rightarrow \tau_0}}}$} \Big( \frac{\tau_1+\tau_2-2\tau_0}{4 \rcon}\Big)b^j \tfrac{\partial}{\partial y^j} +\ord\big(\delta_1^2\big)+\ord\big(\delta_2^2\big).
\label{res_raw}
\end{align}
Hence from the definition of $\lambda$ (\ref{def_lambda}) and (\ref{def_x_i}),
\begin{equation}
\begin{aligned}
\frac{1}{q\constn}\Pdot(\tau_0)&=+\tfrac{2}{3}g\big(\cddot(\tau_0),
\cddot(\tau_0)\big)\cdot(\tau_0) + \lim_{\rcon \rightarrow 0}
\frac{1}{2\rcon}\cddot(\tau_0)+\lambda\cdddot(\tau_0)
+\ord\big(\delta_1^2\big)+\ord\big(\delta_2^2\big).
\end{aligned}
\label{Pdot_res_expan}
\end{equation}
The first term in (\ref{Pdot_res_expan}) is the standard radiation
reaction term and the second term is the singular term to be renormalized.
The third term is proportional to $\cdddot(\tau_0)$ and therefore may
be recognised as the Schott term providing the coefficient is well
defined in the limit.

If $\lambda$ is chosen to be finite it follows immediately that all
higher order terms in the series vanish. This is because $\rcon^{-1}$ is
the most divergent power of $\rcon$ appearing in the
series. Mathematically we are free to choose $\lambda$ to diverge, in
which case higher order terms could be made finite. However this would require
extra renormalization in order to accommodate the $\lambda$ terms and
the resulting equation of motion would not resemble the
Lorentz-Abraham-Dirac equation.

Choosing $\lambda$ to be finite yields for
$\Pdot_{\text{EM}}(\tau)\in  \textup{T}_{\c(\tau)}\m$
\begin{align}
\frac{1}{q\constn}\dot{\mathrm{P}}_{\text{EM}} &= \tfrac23 g(\cddot, \cddot)\cdot +\lambda \cdddot +\lim_{\rcon \rightarrow 0}\frac{1}{2 \rcon}\cddot.
\end{align}
The value of $\lambda$ may be fixed by satisfying the orthogonality
condition (\ref{g_Cd_Cdd}),
\begin{align}
0=\frac{1}{q\constn}g(\Pdot_{\text{EM}}, \cdot)&=-\tfrac23g(\cddot, \cddot) +\lambda g(\cdddot, \cdot)=-(\tfrac23+\lambda )g(\cddot, \cddot).
\end{align}
Therefore $\lambda=-\tfrac{2}{3}$ and the final covariant expression
for $f_{\textup{self}}$ is given by
\begin{align}
f_{\textup{self}}=-\Pdot_{\textup{EM}} &= \tfrac23\kappa \big(\cdddot-g(\cddot, \cddot)\big)\cdot-\lim_{\rcon \rightarrow 0}\frac{\kappa}{2 \rcon}\cddot,
\end{align}
which is identical to (\ref{f_self}).
\section{Conclusion}

We have shown that the complete self force may be obtained directly from the \LW stress 3-forms when using the null geometry with the Bhabha tube as the domain of integration. This eliminates the need to introduce the extra  ad hoc term in \eqref{P_inc_Schott}. It also proves the reason for the missing term in previous calculations is the procedure followed in taking the limits, and not the nature of the coordinates used as proposed by Gal'tsov and Spirin \cite{Galtsov02}.

 We have seen that a requirement for the term to appear is that the ratio of limits $\lambda$, which describes the way in which the Bhabha tube is collapsed onto the worldline, is made finite. This is a natural choice because it demands $\delta_1$, $\delta_2$ and $\rcon$ to be the same order of magnitude. The specific value $\lambda=-\tfrac23$ is fixed by the orthogonality condition (\ref{g_Cd_Cdd}), however the physical justification for imposing this particular geometry on the Bhabha tube is currently unknown.

Consider figure \ref{limits}. It is easily seen that definition (\ref{pk_def}) is equivalent to
\begin{align}
\Pdot_{\kkill}  (\tau_0)=&\raisebox{0.4cm}{$\displaystyle{\lim_{\substack{\rcon \rightarrow 0\\\tau_1 \rightarrow \tau_2\\\tau_2 \rightarrow \tau_0}}}$}\bigg(\frac{1}{\tau_1-\tau_2}\int_{\SigmaT}\stresslw_{\kkill} \bigg).
\label{def_Pdot_lim}
\end{align}
It turns out that the limit $\tau_1\rightarrow \tau_2$ may be taken before the other two limits. This results in the following lemma.
\begin{lemma}
\label{lem_final}
\begin{align}
\Pdot_{\kkill} (\tau_0)=&\raisebox{0.25cm}{$\displaystyle{\lim_{\substack{\rcon \rightarrow 0\\ \tau_1 \rightarrow \tau_0}}}$}\int_{S^2(\tau_1)}i_{\frac{\partial}{\partial \tau}}\stresslw_{\kkill}
\label{sphere_def}
\end{align}
where $S^2(\tau_1)$ is the 2-sphere given in Newman-Unti coordinates by
\begin{align}
S^2 &= \Big\{(\tau , R, \theta, \phi)\Big| \tau=\tau_1,\quad R=\rcon,\quad 0\leq\theta\leq \pi, \quad 0\leq\phi\leq 2\pi\Big\}
\label{sphere_def_lem}
\end{align}
\end{lemma}
\begin{proof}
In Newman-Unti coordinates the side $\SigmaT$ of the Bhabha tube is given by
\begin{align}
\SigmaT &= \Big\{(\tau, R, \theta, \phi)\Big|\tau_1\leq\tau\leq\tau_2,\quad R=\rcon,\quad 0\leq\theta\leq\pi,\quad 0\leq \phi\leq 2\pi\Big\}.
\label{SigmaT_coords}
\end{align}
It follows from (\ref{def_Pdot_lim})that
\begin{align}
\Pdot_{\kkill}  (\tau_0)=&\raisebox{0.4cm}{$\displaystyle{\lim_{\substack{\rcon \rightarrow 0\\\tau_2 \rightarrow \tau_1\\\tau_1 \rightarrow \tau_0}}}$}\bigg(\frac{1}{\tau_1-\tau_2}\int_{\tau=\tau_1}^{\tau_2}\int_{S^2(\tau)} \stresslw_{\kkill} (\tau, \rcon, \theta, \phi)\bigg)\notag\\
=&\raisebox{0.25cm}{$\displaystyle{\lim_{\substack{\rcon \rightarrow 0\\\tau_1 \rightarrow \tau_0}}}$}\Bigg(\lim_{\tau_2 \rightarrow \tau_1}\Big(\frac{1}{\tau_1-\tau_2}\int_{\tau=\tau_1}^{\tau_2}\int_{S^2(\tau)} \stresslw_{\kkill} (\tau, \rcon, \theta, \phi)\Big)\Bigg)
\end{align}
Applying theorem (\ref{lie_int}) yields result.
\end{proof}

Lemma \ref{lem_final} shows the important step in the calculation is taking the limits $\rcon \rightarrow 0$ and $\tau_1=\tau_2 \rightarrow \tau_0$ simultaneously. If we use (\ref{sphere_def}) instead of (\ref{pk_def}) for our definition of the self force then we obtain the following  in place of (\ref{res_raw}),

\begin{align}
\frac{1}{q\constn}\Pdot(\tau_0)=&\tfrac{2}{3}a^2 \tfrac{\partial}{\partial y^0} +\lim_{\rcon \rightarrow 0} \frac{1}{2\rcon}a \tfrac{\partial}{\partial y^3}+ \raisebox{0.25cm}{$\displaystyle{\lim_{\substack{\rcon  \rightarrow 0\\\tau_1 \rightarrow \tau_0}}}$} \Big( \frac{\tau_1-\tau_0}{2 \rcon}\Big)b^j \tfrac{\partial}{\partial y^j} +\ord\big(\delta_1^2\big)+\ord\big(\delta_2^2\big).
\end{align}
and the orthogonality condition (\ref{g_Cd_Cdd}) yields the key result
\begin{align}
\raisebox{0.25cm}{$\displaystyle{\lim_{\substack{\rcon  \rightarrow 0\\\tau_1 \rightarrow \tau_0}}}$} \Big( \frac{\tau_1-\tau_0}{2 \rcon}\Big)= \raisebox{0.25cm}{$\displaystyle{\lim_{\substack{\rcon  \rightarrow 0\\\delta_1 \rightarrow 0}}}$}\Big( \frac{\delta_1}{2 \rcon}\Big)=-\frac{2}{3}.
\label{key_result}
\end{align}
In Dirac's calculation $\tau_1 = \tau_0$ and so the Schott term arises naturally without having to take this limit. However when using Dirac geometry the \LW potential and stress forms have to be expanded in a Taylor series about the retarded time $\tau_r$ (see appendix \ref{dirac_lw}).  In this process the relationship (\ref{dirac_taur}) is used which gives a relationship between $\rdcon$ and $\delta_r=\tau_D-\tau_r$ which is similar to (\ref{key_result}).

\newpage

\addcontentsline{toc}{chapter}{Part II - A new approach to the reduction of wakefields}

\vspace{3cm}
\begin{center}
\begin{minipage}{0.95\textwidth}
\vspace{6cm}
\begin{center}{\bf \huge

PART II\\
A new approach to the reduction of wakefields

} \vspace{3cm}

\end{center}
\end{minipage}\vspace{10cm}
\end{center}

\chapter{Introduction}

\section{Collimation and Wakefields in a particle accelerator }

 It is common for accelerators to have bunches of order $10^8$ particles or more. For example, ALICE, at Daresbury Laboratory, uses bunches with bunch charges of $20$pC to $80$pC, which represents approximately $1.25\times 10^8$ to $5\times 10^8$ electrons. As the bunch traverses the accelerator some of these particles will be perturbed from the ideal orbit or trajectory. This may be due to collective instabilities elsewhere in the beam or deflection due to residual gas that could not be removed from the vacuum chamber. In addition particles from the wall of the beam pipe can be accelerated in a process known as \emph{self injection}. All of these stray particles will form a low density region of charge around the beam which is called the \emph{beam halo}.

 The presence of a large beam halo is generally undesirable. In colliders the halo particles reduce the accuracy of measurements at the interaction region, and  in medical accelerators they can cause severe consequences as a result of highly energetic particles missing the desired target. In order to remove the halo from a beam specially designed apparatus called collimating systems are used.

 In general collimating systems incorporate regions where the cross-sectional area of the beam pipe is reduced. Collimators are specific sections of the beam pipe which undergo a narrowing in one or both of the transverse dimensions. There are many possible configurations depending on the design requirements of individual projects. In high energy accelerators the presence of collimators can also have an adverse effect on the beam due to so called \emph{wakefields}. Electromagnetic fields due to highly relativistic particle beams can interact with the walls of the collimator and induce \emph{image charges} on the wall. The fields resulting from these image charges are known as wakefields and can effect the motion of trailing charges, often inducing instabilities and emittance growth.  Generally fields caused by large scale geometric discontinuities, for example in cavities and collimators, are known as \emph{geometric wakefields} and fields caused by resistivity in the wall are known as \emph{resistive wakefields}.

\section{Present approaches to the reduction of wakefields}

There is much interest in methods for reducing the geometric wakefields produced by a charged bunch of particles passing through a
collimator. The customary approach is to reduce the taper angle of the collimator. Early work on the calculation of wakefields from smoothly tapered structures was pioneered by Yokoya \cite{Yokoya90}, Warnock \cite{Warnock93} and Stupakov \cite{Stupakov96, Stupakov01}. More recent investigations by Stupakov,  Bane and Zagorodnov \cite{Stupakov07, Bane07, Bane10} and Podobedov and Krinsky \cite{Podobedov06, Podobedov07} have also looked at the effect of altering the transverse cross section of the collimator. A detailed analysis of the numerical and analytic calculation of collimator wakefields, including an informative introduction to the topic, may be found in \cite{Smith11}. All of the present methods for minimizing wakefields rely on altering the geometry of the collimator. In this thesis we propose a new method where the trajectory of the beam is altered.

\section{The relativistic \LW field}

A relativistic particle undergoing nonlinear acceleration will generate a radiation field primarily in the instantaneous direction of motion (see figure \ref{fig_SR}). This is known as \begin{bf}synchrotron radiation\end{bf}.  For high $\gamma$-factors the bulk of the field lies inside an angle $\Delta\phi\sim 1/\gamma$ where $\Delta\phi$ is the angle from the direction of motion.
\begin{figure}
\begin{center}
\setlength{\unitlength}{0.4cm}
\begin{picture}(30, 15)
\put(0, 0){\includegraphics[ width=30\unitlength]{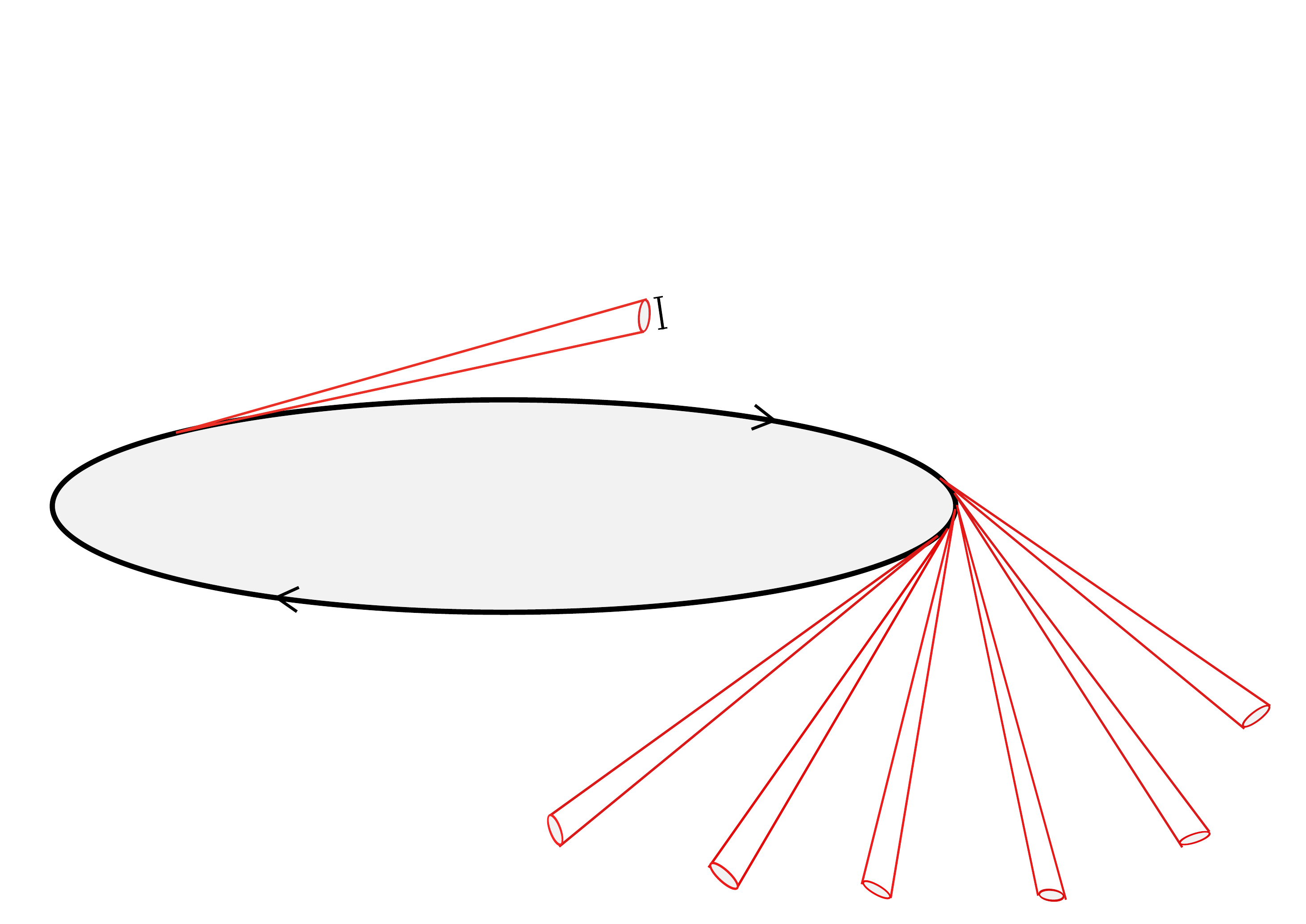}}
\put(15.3, 13.7){$\Delta\phi$}
\end{picture}
\end{center}
\caption{Synchrotron radiation}
\label{fig_SR}
\end{figure}
By contrast, a relativistic particle with constant velocity will generate no radiation field. This is easily seen from the form of $\fr$ in equation (\ref{FR_def}) since when the acceleration is zero this term vanishes. It is well known in accelerator physics that the Coulomb field $\fc$ generated by a relativistic particle moving with constant velocity is flattened towards the plane orthogonal to the direction of motion, and is often called a \begin{bf}pancake field\end{bf} (see figure \ref{fig_pcake}).
\begin{figure}
\begin{center}
\includegraphics[width=0.6\textwidth]{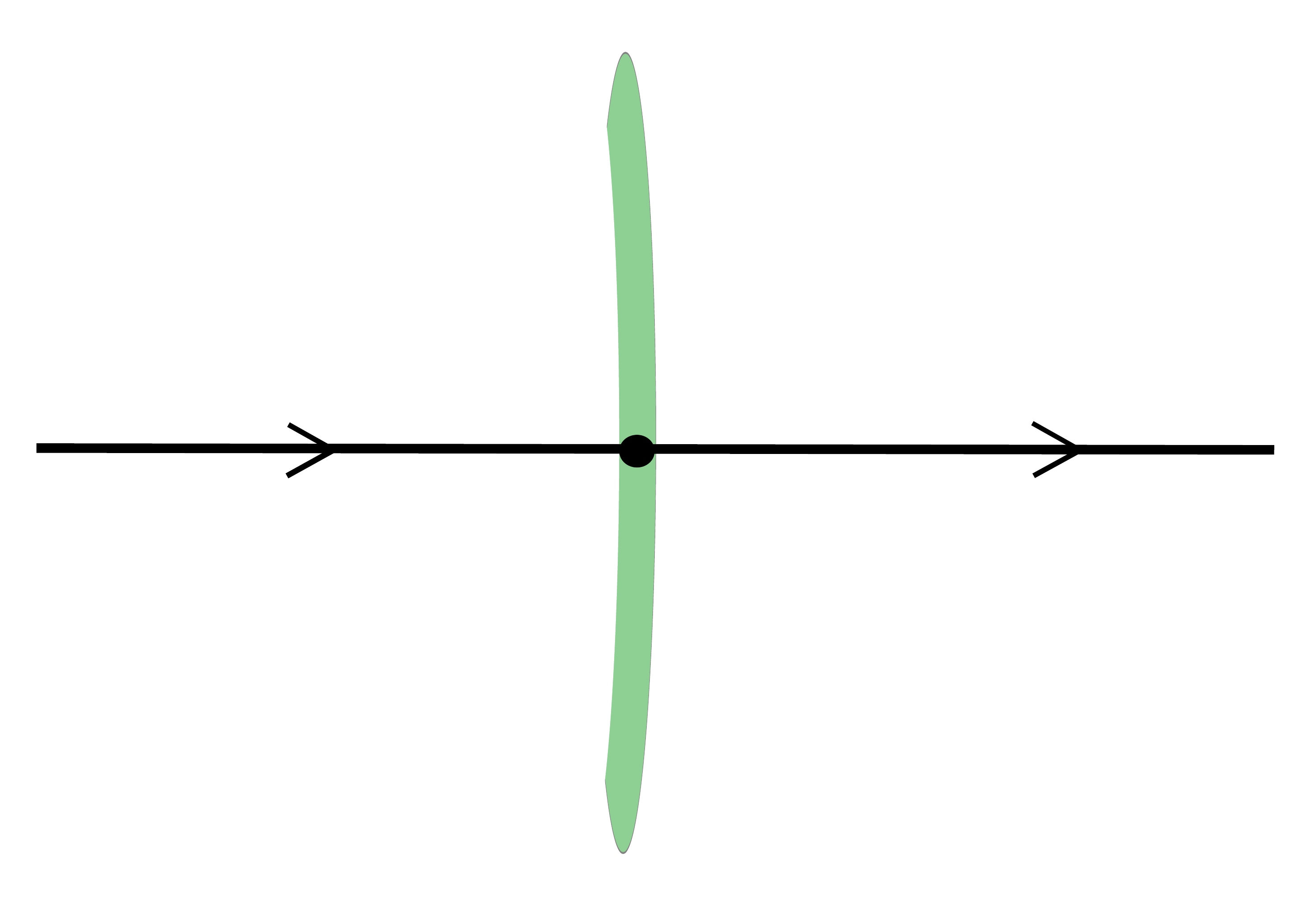}
\end{center}
\caption{The pancake field }
\label{fig_pcake}
\end{figure}

Consider figure \ref{fig_Fields}. The magnitude of the Coulombic $\fc$ and radiative $\fr$ \LW fields are plotted as height above the sphere for high $\gamma$ . In both cases the field is distributed in a narrow spike protruding from the sphere in a small angle from the direction of motion. In both cases the bulk of the field lies inside an angle $\Delta\phi\sim 1/\gamma$ from the direction of motion.

\begin{figure}
\setlength{\unitlength}{0.4cm}
\begin{center}
\begin{picture}(30, 10)
\put(0, 1.8){\includegraphics[width=14\unitlength,viewport=22 293 553 515]{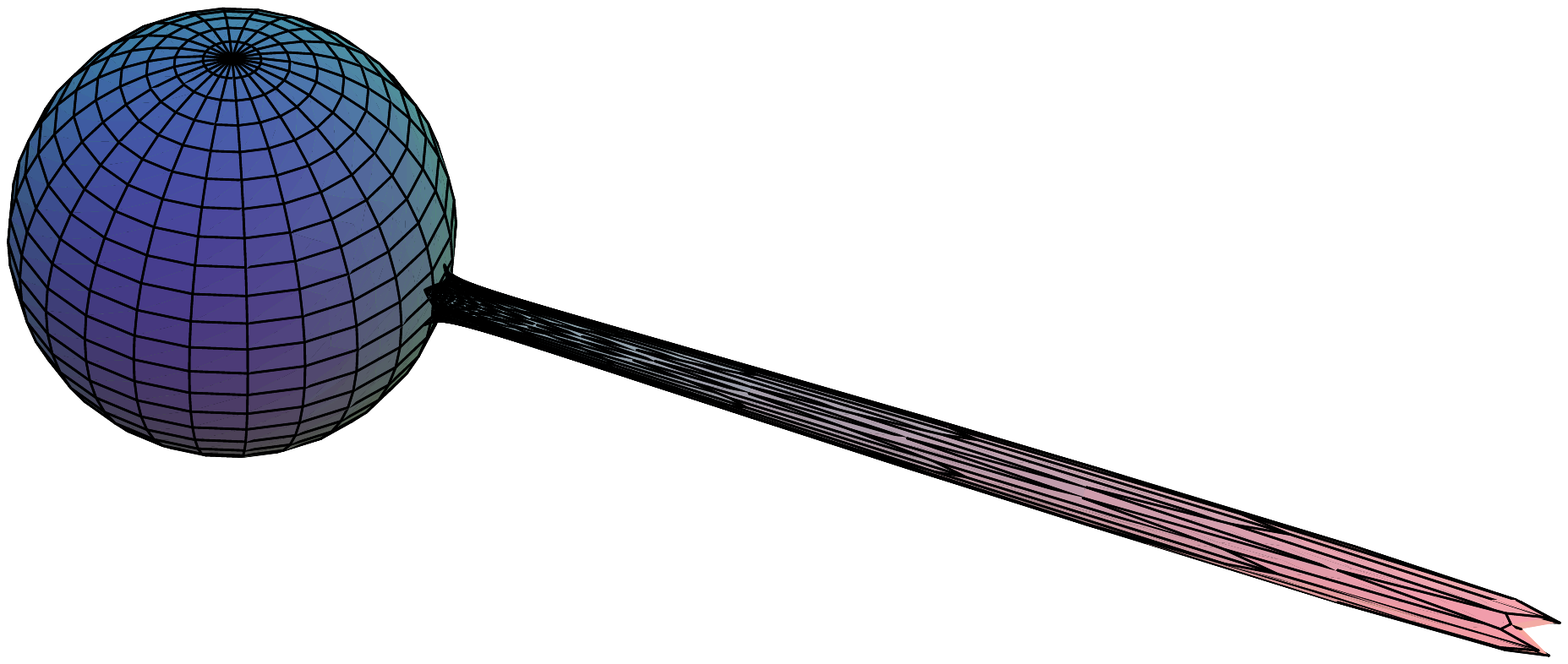}}
\put(15, 1){\includegraphics[width=18\unitlength,viewport=25 306 540 500]{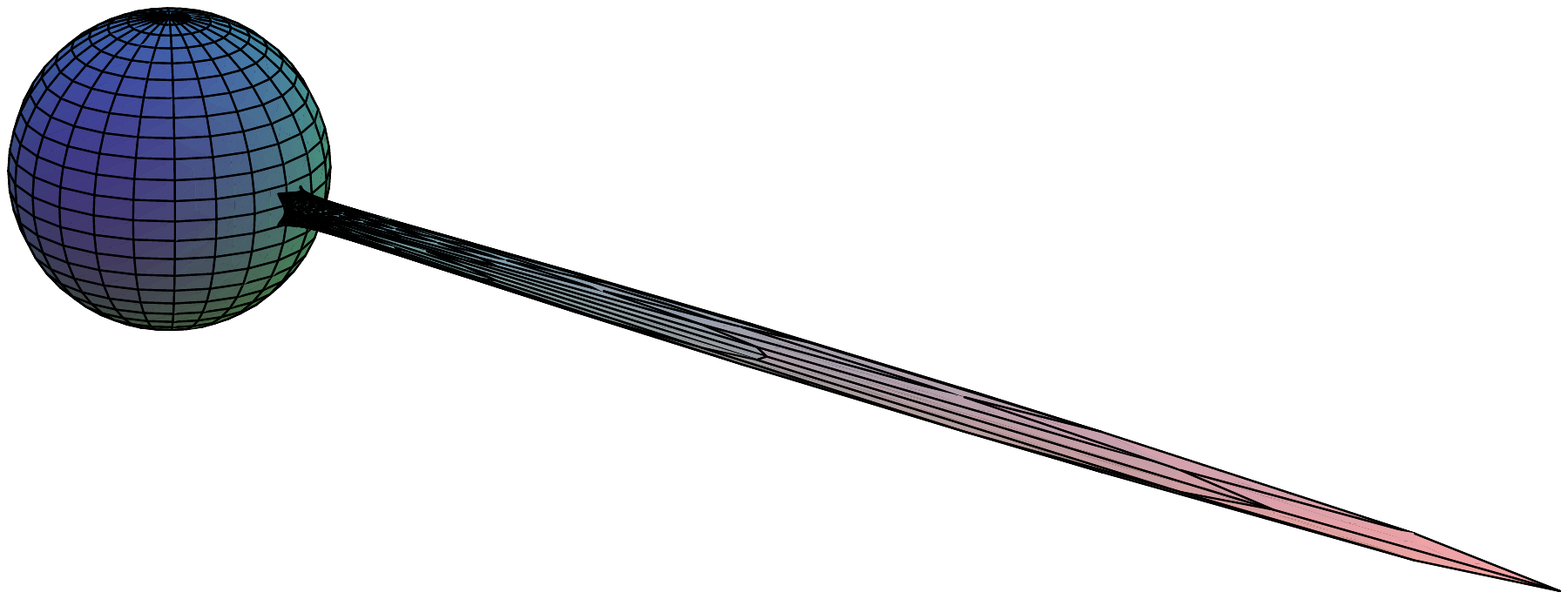}}
\put(3, 0){Coulomb field}
\put(18, 0){Radiation field}
\end{picture}
\end{center}
\caption[relativistic \LW fields as heights above the sphere]{The magnitude of the Coulomb and radiative fields for a high
  $\gamma$, given as height above the sphere. The bulk of the fields is
in the direction of motion.}
\label{fig_Fields}
\end{figure}

At first glance the plot of the Coulomb field seems to contradict figure \ref{fig_pcake} which shows the field flattened in the transverse plane. It is reasonable to ask how these two
radically different behaviours can be consistent.
\begin{figure}
\centerline{
\setlength{\unitlength}{0.09\textwidth}
\begin{picture}(10,3)
\put(0,0){\includegraphics[width=10\unitlength]{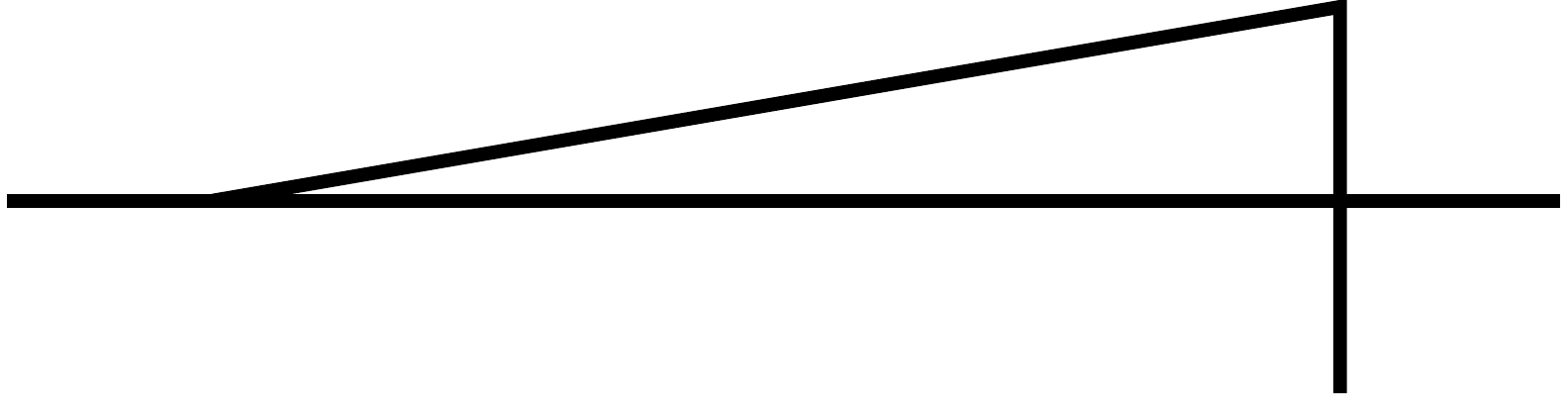}}
\put(8.6,1.5){$h$}
\put(5,0.7){$v\tstraight$}
\put(5,2.1){$c\tstraight$}
\put(8.6,2.5){R}
\put(8.6,.7){Q}
\put(1.0,.7){P}
\end{picture}
}
\caption{Showing the communication between a particle and its pancake}
\label{fig_Catchup}
\end{figure}
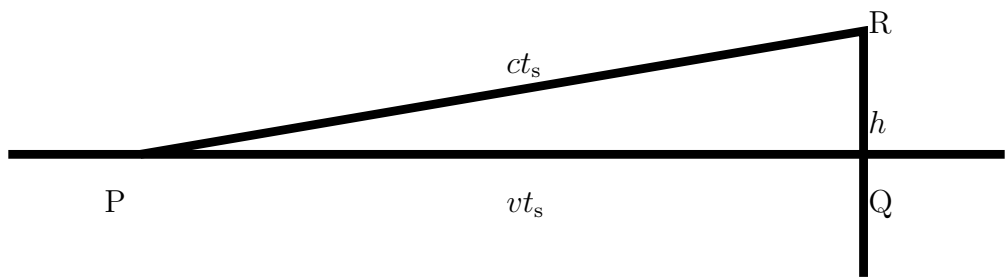
Consider a particle moving at velocity $v$ along the horizontal line $PQ$
in figure \ref{fig_Catchup}. Let $R$ be a point in the pancake a
distance $h$ from the particle, when the particle is at $Q$. The last
point at which the particle could communicate with the point $R$ is at
$P$, a length $v\tstraight$ from $Q$. Here $\tstraight$ is the time
it takes for light to travel from $P$ to $R$ and also the time for
the particle to travel from $P$ to $Q$. Then
$\norm{PR}=c\tstraight$ and $\norm{PQ}=v\tstraight$. Thus
$(c\tstraight)^2=h^2+(v\tstraight)^2$. Hence
$h^2=c^2\tstraight^2(1-v^2/c^2)=c^2\tstraight^2/\gamma^{2}$ so
$\tstraight=\gamma h/c$ and $\norm{PQ}=\gamma h v/c$. Thus a
particle needs to have travelled in a straight line for a length
$\norm{PQ}=\gamma h v/c$ in order for a pancake of radius $h$ to
develop. Looking at the fields which originate at $P$ and arrive at $R$, they
are at an angle approximately $\norm{RQ}/\norm{PR}=1/\gamma$. This is
consistent with figure \ref{fig_Fields}.

\section{Proposal}
We investigate the possibility of reducing wakefields in accelerators by placing structures which give rise to geometric wakefields, such as collimators and cavities, directly after a bending dipole. We model a beam of charged particles as a one dimensional continuum of point charges undergoing the same motion in space, but at a different time.
 In our analysis we envisage a collimator as the source of wakefields and we calculate the field strength at the entrance of the collimator due to the collective \LW field of the particles in the beam. We do not consider any boundary conditions imposed by the beam pipe or the collimator itself. We propose the new method of reduction of wakefields should be used parasitically on existing bends so that there is no additional beam disruption due to coherent synchrotron radiation wakefields (CSR wakes) or loss or energy due to synchrotron radiation (SR). In particular an accelerator which requires the following:
\begin{itemize}
\item
Short bunches (much shorter bunch length that the aperture of the
collimator).
\item
The bending of the bunches, via the use of dipoles.
\item
Collimation.
\end{itemize}
can achieve the collimation for free, i.e. with no additional loss of
energy or disruption to the bunches from geometric or CSR wakes, by
placing the collimator just after the bend. \\

We have seen that in order for a pancake of radius $h$ to develop the particle needs to have travelled in a straight line for a distance $\gamma h v/c$. For highly relativistic motion this is approximately $\gamma h$. Our proposed method of reduction of wakefields relies on this result.
 The idea is to bend the beam slightly before it enters the collimator (see figure \ref{setup}). Most of the Coulomb
field generated by the particle before the bend will continue in a straight line. By sufficiently
enlarging the beam pipe in this direction the wakefield due to this part of the
field can be neglected. If the distance, $Z$, of the
straight line segment from the terminus of the bend to the centre
of the collimator is sufficiently small, then the resulting pancake field will be
too small to reach the sides of the structure. Of course bending
the beam will generate additional radiation fields, however by judicious choice of
geometry of the beam these can be minimized.
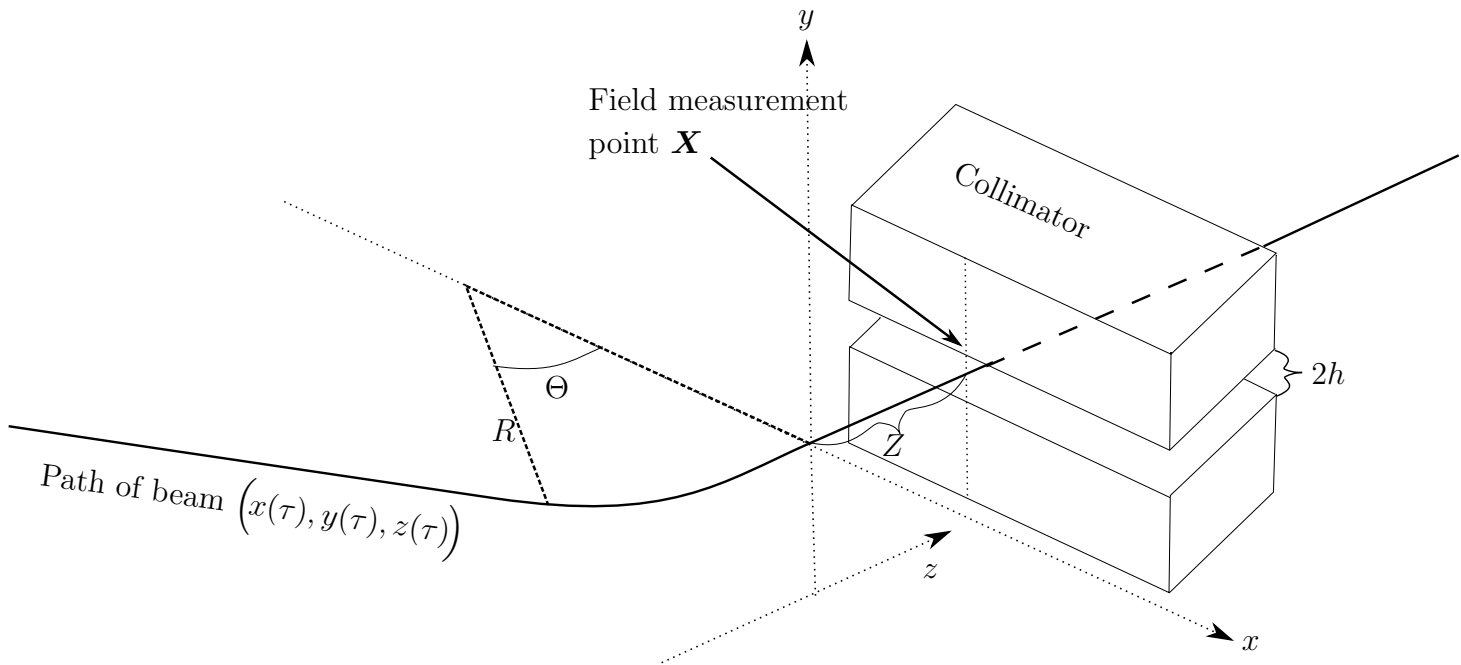
\begin{figure}
\centerline{
\setlength{\unitlength}{0.013\textwidth}
\begin{picture}(100,50)
\put(0,0){\includegraphics[width=100\unitlength]{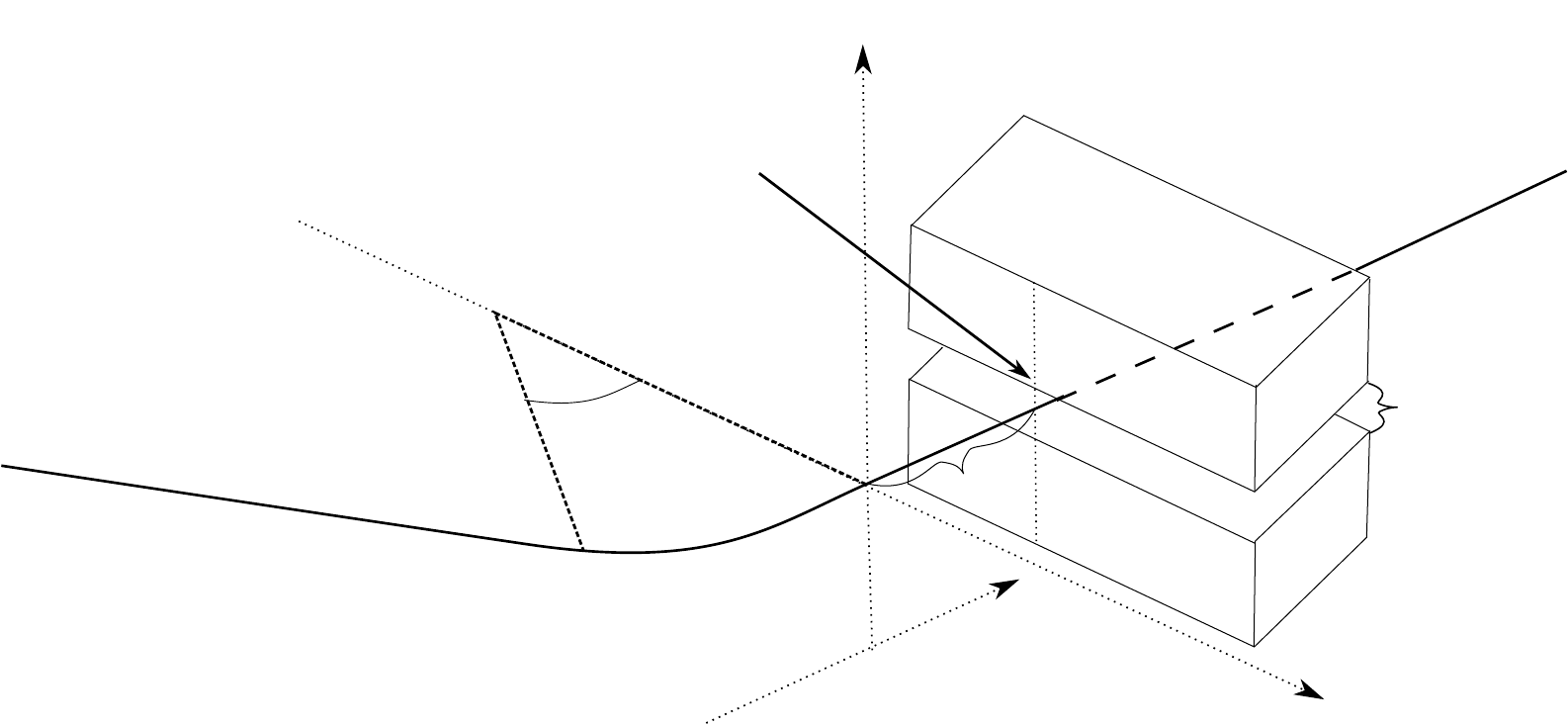}}
\put(2,15){\makebox(0,0)[tl]{\rotatebox{-8}{Path of beam $\Big(\!x(\tau),y(\tau),z(\tau)\!\Big)$}}}
\put(65,35){\makebox(0,0)[tl]{\rotatebox{-24}{Collimator}}}
\put(40,40){\makebox(0,0)[tl]{Field measurement}}
\put(40,37){\makebox(0,0)[tl]{point $\VX$}}
\put(35,17){\makebox(0,0)[tr]{$R$}}
\put(37,20){\makebox(0,0)[tl]{$\Theta$}}
\put(61,16){\makebox(0,0)[t]{$Z$}}
\put(92.1,21){\makebox(0,0)[tr]{$2h$}}
\put(63,7){\makebox(0,0)[tl]{$z$}}
\put(85,2){\makebox(0,0)[tl]{$x$}}
\put(55,44){\makebox(0,0)[b]{$y$}}
\end{picture}
}
\caption{Setup for beam trajectory and collimator}
\label{setup}
\end{figure}

Let $h$ denote half the aperture of the
collimator and let $L$ represent the spatial length of the bunch. The following two scenarios will be considered:
\begin{itemize}
\item Long smooth bunches
where $L>h$ and
any variation in the density of the bunch is over length
scales longer than $h$,
\item Bunches where variation in density is over short length scales less
than about $0.2h$. This includes the case of very short bunches where
$L\ll 0.2 h$.
\end{itemize}

These two scenarios are both applicable to present day machines, where the bunch length depends upon the
specific objectives and engineering considerations of individual projects.

In chapter \ref{chap_pointline} we show that the coherent electric (magnetic) fields due to a bunch modeled as a 1D continuum of point particles are given by the convolution of the electric (magnetic) field due to a single particle with the charge profile. In chapter \ref{bendingbeams} we carry out a numerical investigation using the mathematical software MAPLE. Assuming a Gaussian charge profile we minimize the electric field generated by the bunch by calculating the field due to different beam trajectories. Having optimized the trajectory we calculate the electric field at a specific point, representing a point on the collimator wall, for a selection of different bunch lengths which are attainable at present day facilities.
Calculating the secondary electromagnetic fields generated by the collimator is a
boundary value problem,  hence calculating the full wakefield kick due to the collimator and
a bent beam would require detailed knowledge of the geometry and material
properties of the beam pipe. This will not be undertaken in this thesis. However we will show that the field incident on the boundary may be reduced by a factor of 7, and since the wakefields are, to a large extent, proportional to the fields at the boundary, the field in the beam pipe will automatically be reduced by approximately the same factor.
We will find that for short bunches, or bunches with large
amounts of micro-bunching, it is possible to make a significant reduction
in wakefields. This is applicable to present day free electron lasers, which employ bunch compressors to produce
very short bunches, for example in LCLS $L/c \approx 0.008$ps. Assuming a collimator of half aperture $h=0.5$mm then in this case $L=0.0048h$. It turns out that electromagnetic fields due to
long smooth bunches may not be reduced significantly. In many present day colliders the bunches are designed to be long and
smooth, however in the future short bunch colliders may be desirable (see Table \ref{table1}).

\begin{table}
\begin{center}
\caption{Bunch lengths for some modern colliders and FELs}
\label{table1}
\begin{tabular}{|c| c |c|}\hline
Collider & Year of & Bunch length  [ps]\\
 &Commissioning & \\\hline
SLC, SLAC   &  $1989$            & $3$\\
ILC         &  $\geq 2015$       & $1$\\
CLIC        &  $\geq 2025$       & $0.15$\\\hline
Free Electron Laser  &   & Min. bunch length  [ps]\\\hline
FLASH, DESY &  $2005$            & $0.05$\\
LCLS, SLAC  &  $2009$            & $0.008$\\
XFEL, DESY  &  $2014$            & $0.08$\\\hline
\end{tabular}
\end{center}
\end{table}

\chapter{The field of a 1D continuum of  point charges}

\label{chap_pointline}

In this chapter we consider the field generated by a 1D continuum of point charges on an arbitrary trajectory. The key result is that the electric field for the continuum is given by the convolution of the electric field for a single point charge with the charge profile. This result will be used in the next chapter where we adopt the 1D continuum as our model for a bunch of particles in an accelerator.
\section{The Li$\acute{\textup{e}}$nard-Wiechert field in $3$-vector notation}

\begin{definition}
\label{def_X_three}
Given a choice of time coordinate $t$ such that $\frac{\partial}{\partial t}$ is Killing we can write $\m=\real \times \mund$, where $\mund$ is Euclidean three space. We denote the points $\point\in \mnotc$ and $\c(\tau)\in \m$  by
\begin{align}
\point=(\ccon T, \Xu), \qquad \text{and} \qquad \c(\tau)=(\ccon t,\xu)
\end{align}
where $T, t\in \real$ and $\Xu, \xu \in \mund$. The null displacement vector $\x$ is given by
\begin{align}
\x=(\ccon T-\ccon t, \Xu-\xu), \qquad \textup{where} \qquad t=\gamma\tau_r(\ccon T, \Xu)
\end{align}
where the difference $\Xu-\xu$ is a $3$-vector at point $\Xu\in \mund$. It follows from the definition of $\tau_r$ that $T>t$.
\end{definition}
\begin{definition}
\label{def_rn}
The spatial displacement between the field point $\Xu$ and the emission point $\xu$ will be denoted by
\begin{align}
\rmag=||\Xu-\xu||,\label{r_def}
\end{align}
where $||.||$ is the Euclidean norm. We define the unit $3$-vector $\nund \in \tan_{\Xu}\mund$ by
\begin{align}
\nund=\frac{\Xu-\xu}{||\Xu-\xu||}=\frac{\Xu-\xu}{\rmag}, \qquad \nund \centerdot \nund=1
\end{align}
where the dot denotes the standard scalar product.
\end{definition}
\begin{lemma}
\label{rt_lem}
It follows from the null condition that
\begin{align}
\rmag=\ccon T-\ccon t\label{rt}
\end{align}
\end{lemma}
\begin{proof}
\begin{align}
0=\g(\x, \x)=&\g\big((\ccon T-\ccon t, \Xu-\xu), (\ccon T-\ccon t, \Xu-\xu)\big)\notag\\
=&-(\ccon T-\ccon t)^2 +||\Xu-\xu||^2\notag
\end{align}
Thus
\begin{align}
||\ccon T-\ccon t||=||\Xu-\xu||=\rmag
\end{align}
The result (\ref{rt}) follows from noticing $T>t$.
\end{proof}

It follows trivially from definitions  \ref{def_X_three} and \ref{def_rn} and lemma \ref{rt_lem} that the null 4-vector is given by
\begin{align}
\x=\rmag(1, \nund)\label{x_three}
\end{align}

\begin{definition}
\label{beta_def}
The normalized Newtonian velocity $\betau$ and acceleration $\au$ are defined by
\begin{align}
\betau=&\frac{1}{\ccon } \frac{d\xu}{d t }=\frac{1}{\ccon\gamma} \frac{d\xu}{d\tau}, \qquad\text{and}\qquad \au= \frac{d\betau}{d t }=\frac{1}{\gamma} \frac{d\betau}{d\tau}
\end{align}
\begin{lemma}
\label{lemfourvecs}
Thus the $4$-vectors $\cdot, \cddot \in \tan_{\c(\tau)}\m$ are given by
\begin{align}
&\cdot=\ccon\gamma(1, \betau)\notag\\
\textup{and} \qquad &\cddot=\ccon\gamma^4(\au\centerdot\betau)(1, \betau)+\ccon\gamma^2(0, \au)
\label{v_three}
\end{align}
\end{lemma}
\begin{proof}
First note that
\begin{align*}
  \frac{d\gamma}{d t}= \frac{d}{d t}(1-\betau\centerdot\betau)^{-\frac{1}{2}} =-\frac{1}{2}(1-\betau\centerdot\betau)^{-\frac{3}{2}}\frac{d}{dt}(-\betau\centerdot\betau),
\end{align*}
thus since $\frac{d}{dt}(-\betau\centerdot\betau)=-2\au\centerdot\betau$ it follows that
\begin{align}
  \frac{d\gamma}{d t}=\gamma^3\au\centerdot\betau
  \label{gam_t}
\end{align}
From (\ref{def_X_three}) we have $\c=(\ccon t, \xu)$. Thus
\begin{align*}
  \cdot=\frac{d\c}{d \tau}=\gamma \frac{d \c}{d t}=\gamma (\ccon, \frac{d\xu}{d t})\notag
\end{align*}
Also
\begin{align*}
  \cddot=\frac{d\cdot}{d\tau}=\gamma \frac{d\cdot}{d t}=\gamma \ccon \frac{d \gamma}{d t} (1, \betau)+\gamma^2 \ccon \frac{d}{d t}(1, \betau).
\end{align*}
Substituting (\ref{gam_t}) and (\ref{beta_def}) yields the result \ref{lemfourvecs}.
\end{proof}
\end{definition}
\begin{lemma}The following relations are true
\begin{align}
&\g(\x, \cdot)=\rmag \ccon\gamma (\nund \centerdot \betau-1)\notag\\
\textup{and} \qquad &\g(\x, \cddot)= \rmag \ccon \gamma^4 (\betau\centerdot\nund-1)(\au\centerdot\betau)+\rmag \ccon \gamma^2 (\au\centerdot\nund)\label{gwv_three}
\end{align}
\end{lemma}
\begin{proof}
\begin{align*}
\g(\x, \cdot)=& \g\big( \rmag (1, \nund), \ccon \gamma (1, \betau)\big)= \rmag \ccon \gamma (\betau\centerdot\nund-1)
\end{align*}
Also
\begin{align*}
\g(\x, \cddot)=&\g\big(\rmag (1, \nund), \gamma \ccon \frac{d \gamma}{d t} (1, \betau)\big)+ \g\big(\rmag (1, \nund), \gamma^2 \ccon (0, \au) \big)\\
=&\rmag \ccon \gamma^4 \big((\au \centerdot\betau)\betau\centerdot\nund-(\au\centerdot\betau)\big)+\rmag \ccon \gamma^2 (\au\centerdot\nund).
\end{align*}
\end{proof}
\begin{lemma}
\label{ecer_lem}
In chapter \ref{maxlorentz} equations (\ref{eb_def}) and (\ref{bvecdef}) we define the electric and magnetic 1-forms $\etilde, \widetilde{\bvec}\in \Gamma \tan \m$ for a timelike observer curve $\u$. Given a coordinate chart $(y^0, y^1, y^2, y^3)$ let  $\u=\partial_{y^0}$ and let $\elw=\elw_{\textup{C}}+\elw_{\textup{R}}$ where
\begin{align}
\elwd_{\textup{C}}=i_{\partial_{y^0}}\fc\qquad \textup{and} \qquad \elwd_{\textup{R}}= i_{\partial_{y^0}}\fr.
\label{ebvec_def}
\end{align}
If $\elwd_{\textup{C}}=\elw_{\textup{C}a}dy^a$ and $\elwd_{\textup{R}}=\elw_{\textup{R}a}dy^a$ for $a=1, 2, 3$ then
\begin{align}
\elw_{\textup{C}a}= \frac{q}{4\pi \epsilon_0}\frac{(\nund-\betau)_a}{\rmag^2 \gamma^2 (1-\nund\centerdot\betau)^3}\qquad \textup{and}\qquad \elw_{\textup{R}a}= \frac{q}{4\pi \epsilon_0}\frac{(\nund\times(\nund-\betau)\times\au)_a}{\rmag \ccon \gamma (1-\nund\centerdot\betau)^3}.
\end{align}
\end{lemma}
\begin{proof}
Let $\fc=\flw_{\textup{C} a b} dz^a \wedge dz^b$, then from (\ref{FC_def}) and (\ref{gwv_three}) it follows that
\begin{align}
\frac{1}{\kappa}\flw_{\textup{C} a b}=\frac{-\ccon^2(\x_a \cdot_b-\x_b\cdot_a)}{(rc\gamma (\nund \centerdot \betau-1))^3}.
\end{align}
where $\kappa=\frac{q}{4\pi \epsilon_0}$. Thus
\begin{align}
\frac{1}{\kappa}\elw_{\textup{C} a}=\frac{1}{\kappa}\flw_{\textup{C} a 0}= \frac{-\ccon^2(\x_a \cdot_0-\x_0\cdot_a)}{(rc\gamma (\nund \centerdot \betau-1))^3}.
\end{align}
It follows from (\ref{x_three}) and (\ref{v_three}) that
\begin{align*}
\x_0=-\rmag, \qquad \x_a=\rmag\nund_a, \qquad \cdot_0=-\ccon \gamma \qquad \textup{and}\qquad \cdot_a=\ccon \gamma\betau_a,
\end{align*}
thus
\begin{align}
\frac{1}{\kappa}\elw_{\textup{C} a}=&\frac{\ccon ^3\gamma \rmag\nund_a-\rmag \ccon ^3\gamma\betau_a}{(\rmag \ccon \gamma (\nund \centerdot \betau-1))^3}=\frac{(\nund-\betau)_a}{\rmag^2\gamma^2 (1-\nund \centerdot \betau)^3}
\end{align}
Similarly let $\fr=\flw_{\textup{R} a b} dz^a \wedge dz^b$. It follows from (\ref{FR_def}) and (\ref{gwv_three}) that
\begin{align*}
\frac{1}{\kappa}\flw_{\textup{R} a b}=&\frac{\rmag \ccon\gamma (\nund \centerdot \betau-1)(\x_a \cddot_b-\x_b\cddot_a)}{(\rmag \ccon\gamma (\nund \centerdot \betau-1))^3}-\frac{\rmag \ccon\gamma \nund \centerdot \au(\x_a \cdot_b-\x_b\cdot_a)}{(\rmag \ccon\gamma (\nund \centerdot  \betau-1))^3}\\
=&\frac{(\x_a \cddot_b-\x_b\cddot_a)}{(\rmag \ccon \gamma (\nund \centerdot  \betau-1))^2}-\frac{\rmag \ccon \gamma \nund \centerdot  \au(\x_a \cdot_b-\x_b\cdot_a)}{(\rmag \ccon\gamma (\nund \centerdot  \betau-1))^3}
\end{align*}
Thus
\begin{align}
\frac{1}{\kappa}\elw_{\textup{R} a}=\frac{1}{\kappa}\flw_{\textup{R} a 0}= \frac{(\x_a \cddot_0-\x_0\cddot_a)}{(\rmag \ccon\gamma (\nund \centerdot  \betau-1))^2}-\frac{\rmag \ccon\gamma (\nund \centerdot  \au)(\x_a \cdot_0-\x_0\cdot_a)}{(\rmag \ccon\gamma (\nund \centerdot  \betau-1))^3}\centerdot
\label{ea_def}
\end{align}
It follows from  (\ref{v_three}) that
\begin{align}
\cddot_0=-\ccon\gamma^4\au\centerdot\betau  \textup{and} \qquad \cddot_a=\ccon\gamma^2\au_a +\ccon\gamma^4(\au\centerdot\betau)\betau_a,
\end{align}
thus the first term in (\ref{ea_def}) yields
\begin{align}
\textup{first term}=&\frac{-\rmag \ccon \gamma^4(\au\centerdot\betau)\nund_a+\rmag\big(\ccon \gamma^2 \au_a+\ccon \gamma^4(\au\centerdot\betau)\betau_a\big)}{(\rmag \ccon\gamma (\nund \centerdot \betau-1))^2}\notag\\
=&\frac{\au_a}{\rmag \ccon (\nund\centerdot\betau-1)^2}+\frac{\gamma^2(\au\centerdot\betau)(\betau_a-\nund_a)}{\rmag \ccon (\nund\centerdot\betau-1)^2}
\label{first_tm}
\end{align}
and the second term yields
\begin{align*}
 \textup{second term}=& -\frac{\big((-\rmag \ccon\gamma \nund_a+\rmag \ccon \gamma \betau_a)(\rmag \ccon \gamma^4(\au\centerdot \betau)(\betau\centerdot \nund-1)+\rmag \ccon \gamma^2(\au\centerdot \nund))\big)}{(\rmag \ccon\gamma (\nund \centerdot  \betau-1))^3}\\
 =&\frac{\gamma^2(\nund_a-\betau_a)(\au\centerdot \betau)}{\rmag \ccon (\nund\centerdot \betau-1)^2}+\frac{(\nund_a-\beta_a)(\au\centerdot\nund)}{\rmag \ccon (\nund\centerdot \betau-1)^3}
  \label{sec_tm}
\end{align*}
Thus adding the two term yields
\begin{align*}
 \frac{1}{\kappa} \elw_{\textup{R} a}=\frac{(\nund\centerdot\betau-1)\au_a}{\rmag \ccon (\nund\centerdot\betau-1)^3}+\frac{(\nund_a-\beta_a)(\au\centerdot\nund)}{\rmag \ccon (\nund\centerdot\betau-1)^3}
\end{align*}
The result follows from the rule for triple vector products.
\end{proof}
\begin{definition}
\label{mag_def}
Similarly let $\blwd=\blwd_{\textup{C}}+\blwd_{\textup{R}}$ where
 \begin{align}
\bvec_{\textup{C}}= \frac{1}{\ccon}\widetilde{i_{\partial_{z^0}}\star\fc}\qquad \textup{and} \qquad \bvec_{\textup{R}}= \frac{1}{\ccon}\widetilde{i_{\partial_{z^0}}\star\fr},
\end{align}
Then if $\blwd_{\textup{C}}=\blw_{\textup{C}a}dy^a$ and $\blwd_{\textup{R}}=\blw_{\textup{R}a}dy^a$ it can be show that
\begin{align}
\bvec_{\textup{C}a}=\frac{1}{\ccon}(\nund \times \VE _{\textup{C}})_a \qquad \textup{and} \qquad \bvec_{\textup{R}a}=\frac{1}{\ccon}(\nund \times \VE _{\textup{R}})_a
\end{align}
\end{definition}

\section{The model of a beam}
\label{beam_model}
We model our bunch of particles  as a one dimensional bunch where each particle undergoes the
same motion in space but at a different time. This bunch is moving at
a constant speed with relativistic factor $\gamma$. Let $\nu$ label the points in the bunch,
which will be called body points. The profile of the bunch is given by
$\rho(\nu)$.\footnote{ Note that $\nu$ has the dimension of time.}

\begin{definition}
\label{model_beam}
 Let $\Vx_\nu(\tau)$ represent the position of body point $\nu$ at proper time $\tau$, and for each body point $\nu$ let
\begin{align}
t_\nu(\tau)=(\tau+\nu)/\gamma.
\label{eqns_path_nu}
\end{align}
\end{definition}


\begin{definition}
The retarded time for the body
point $\nu$ corresponding to the fields measured at $\VX$ at
laboratory time $T$ is denoted by $\tauhat(\VX,T,\nu)$. Similarly the arrival time at $\VX$ of the field
generated by body point $\nu$ at proper time $\tau$ is denoted by $\THat( \nu,\tau,\VX)$.\\
For the $\nu=0$ particle we define
\begin{align}
\tauhat_{0}(\VX, T)=\tauhat(\VX,T,0)
\quadand
\THat_{0}(\tau,\VX)=\THat(0,\tau,\VX).
\label{eqns_zero}
\end{align}
\end{definition}
\begin{lemma}
\label{T_tau_lem}
It follows that
\begin{align}
\tauhat_0(\VX,\THat(\nu,\tau,\VX)-\nu/\gamma)=\tau.
\label{eqns_tau_T_Invers_sub}
\end{align}
and
\begin{align}
\tauhat(\VX,T,\nu)=\tauhat_{0}(\VX, T-\nu/\gamma).
\label{eqns_tau_nu}
\end{align}
\end{lemma}
\begin{proof}
The retarded time condition is given by
\begin{align}
cT-ct_\nu\big(\tauhat(\VX,T,\nu)\big)=
\norm{\VX-\Vx_\nu\big(\tauhat(\VX,T,\nu)\big)},
\label{eqns_ret_time}
\end{align}
and hence
\begin{align}
cT-c\tauhat(\VX,T,\nu)/\gamma-c\nu/\gamma =
\norm{\VX-\Vx_\nu\big(\tauhat(\VX,T,\nu)\big)}.
\label{eqns_ret_time_res}
\end{align}
Thus
\begin{align}
c\THat(\nu,\tau,\VX)
&=
ct_\nu(\tau)+\norm{\VX-\Vx_\nu(\tau)}\nonumber\\
&=c(\tau+\nu)/\gamma+\norm{\VX-\Vx_\nu(\tau)}.
\label{eqns_def_That}
\end{align}
From (\ref{eqns_ret_time_res}) and (\ref{eqns_def_That})
\begin{align}
cT&=c\big(\tauhat(\VX,T,\nu)+\nu\big)/\gamma+\norm{\VX-\Vx_\nu\big(\tauhat(\VX,T,\nu)\big)}\notag\\
&=c\THat(\nu,\tauhat(\VX,T,\nu),\VX).
\label{eqns_tau_T_Invers_a}
\end{align}
Since $\THat$ is increasing and the range of $\tauhat$ is
from $-\infty$ to $+\infty$ it follows that $\THat$ and $\tauhat$ are inverse to each other, yielding (\ref{eqns_tau_T_Invers_a})
and
\begin{align}
\tauhat(\VX,\THat(\nu,\tau,\VX),\nu)=\tau.
\label{eqns_tau_T_Invers_b}
\end{align}
Now $\THat(\nu,\tau,\VX)$ and $\tauhat(\VX,T,\nu)$ may be written in terms of
$\THat_{0}(\tau,\VX)$ and $\tauhat_{0}(\VX,T)$. From (\ref{eqns_def_That})
\begin{align}
\THat(\nu,\tau,\VX)=\THat_{0}(\tau,\VX)+\nu/\gamma.
\label{eqn_THat_nu}
\end{align}
From (\ref{eqns_tau_T_Invers_a}), (\ref{eqns_tau_T_Invers_b}) and
(\ref{eqns_zero}),
\begin{align}
\THat_0(\tauhat_0(\VX,T), \VX)=T
\label{eqns_tau_T_Invers_zero_a}
\end{align}
and
\begin{align}
\tauhat_0(\VX,\THat_0(\tau,\VX))=\tau.
\label{eqns_tau_T_Invers_zero_b}
\end{align}
Substituting (\ref{eqn_THat_nu}) into (\ref{eqns_tau_T_Invers_zero_b}) leads to
\begin{align}
\tauhat_0(\VX,\THat(\nu,\tau,\VX)-\nu/\gamma)=\tau.
\end{align}
Substituting $\tau=\tauhat(\VX,T,\nu)$ and using (\ref{eqns_tau_T_Invers_a}) yields
\begin{align}
\tauhat(\VX,T,\nu)=\tauhat_{0}(\VX, T-\nu/\gamma).
\end{align}
\end{proof}

\section*{Statistics for independent identical distributions}
\begin{definition}
We assume the $\nu$ for each particle
has the identical distributions $\rho(\nu)$, where
\begin{align}
\int \rho(\nu) d\nu =1
\label{Stat_normal}
\end{align}
That is the probability
that particle $k$ has displacement $\nu_k$ is $\rho(\nu_k)d\nu$.
\end{definition}
\begin{definition}
Given a function of $H(\nu_1,\ldots,\nu_N)$ of all the random variables
we define the expectation of $H$ as
\begin{align}
\Exx{H}=
\int d\nu_1 \rho(\nu_1) \cdots
\int d\nu_N \rho(\nu_N) H(\nu_1,\ldots,\nu_N)
\label{Stat_def}
\end{align}
\end{definition}
\begin{lemma}
\label{sum_lem}
For a function which is simply the sum of functions $\sum_{k=1}^N
h(\nu_k)$ we have

\begin{align}
\Exx{\sum_{k=1}^N h(\nu_k)}
=
N \ExxOP{h}
\label{Stat_Expect_Sum}
\end{align}
where
\begin{align}
\ExxOP{h}
=
\int  \rho(\nu) h(\nu)\,d\nu
\label{Stat_Expect_onepart_def}
\end{align}
is the one particle expectation.
\end{lemma}
\begin{proof}
\begin{align*}
\Exx{\sum_{k=1}^N h(\nu_k)}
&=
\int d\nu_1 \rho(\nu_1) \cdots
\int d\nu_N \rho(\nu_N) \sum_{k=1}^N h(\nu_k)
\\&=
\sum_{k=1}^N \int d\nu_1 \rho(\nu_1) \cdots
\int d\nu_N \rho(\nu_N) h(\nu_k)
\\&=
\sum_{k=1}^N
\int d\nu_k \rho(\nu_k) h(\nu_k)
=
N \int d\nu \, \rho(\nu) h(\nu)
\end{align*}
\end{proof}
\begin{lemma}
The expectation sum of product of the two functions
\begin{align*}
H(\nu_1,\ldots,\nu_N)
&=
\Big(\sum_{k=1}^N h(\nu_k)\Big)
\Big(\sum_{m=1}^N g(\nu_m)\Big)
\end{align*}
is given by
\begin{align}
\Exx{H}=
N\ExxOP{hg}+ (N^2-N)
\ExxOP{h}\ExxOP{g}
\label{Stat_Expect_prod}
\end{align}
This is important since it corresponds to components of the energy, momentum and stress of the electromagnetic field. In particular, the energy of the electromagnetic field is determined in section \ref{sec_energy}.
\end{lemma}
\begin{proof}
\begin{align*}
\Exx{H}
=&
\Exx{
\Big(\sum_{k=1}^N h(\nu_k)\Big)\Big(\sum_{m=1}^N g(\nu_m)\Big)
}
=
\sum_{k=1}^N \sum_{m=1}^N \Exx{
h(\nu_k)
g(\nu_m)
}
\\=&
\sum_{k=1}^N \sum_{m=1}^N
\int d\nu_1 \rho(\nu_1) \cdots
\int d\nu_N \rho(\nu_N)
h(\nu_k)
g(\nu_m)
\\=&
\sum_{k=1}^N \sum_{m=k}
\int d\nu_1 \rho(\nu_1) \cdots
\int d\nu_N \rho(\nu_N)
h(\nu_k)
g(\nu_m)\\
&+
\sum_{k=1}^N \sum_{m\ne k}
\int d\nu_1 \rho(\nu_1) \cdots
\int d\nu_N \rho(\nu_N)
h(\nu_k)
g(\nu_m)
\\=&
\sum_{k=1}^N
\int d\nu_k \rho(\nu_k)
h(\nu_k)
g(\nu_k)
+
\sum_{k=1}^N \sum_{m\ne k}
\int d\nu_k \rho(\nu_k)
\int d\nu_m \rho(\nu_m)
h(\nu_k)
g(\nu_m)
\\=&
N\int d\nu \rho(\nu)
h(\nu)
g(\nu)
+
\sum_{k=1}^N \sum_{m\ne k}
\Big(\int d\nu_k \rho(\nu_k) h(\nu_k)\Big)
\Big(\int d\nu_m \rho(\nu_m)g(\nu_m)\Big)
\\=&
N\ExxOP{hg}+
\sum_{k=1}^N \sum_{m\ne k}
\ExxOP{h}\ExxOP{g}
=
N\ExxOP{hg}+ (N^2-N)
\ExxOP{h}\ExxOP{g}
\end{align*}
Note the structure of the expectation of $H$, in particular the appearance of $N$ and $N^2-N$.
\end{proof}
\begin{lemma}
\label{shift}
The one particle expectation of a shifted function is given by
\begin{align}
\ExxOP{g(T-\gamma^{-1}\nu)}
&=
\int \rho_\Lab(T-T')g(T')d T'.
\label{Stat_1P_shift}
\end{align}
\end{lemma}
\begin{proof}
\begin{align*}
\ExxOP{g(T-\gamma^{-1}\nu)}
&=
\int \rho(\nu) g(T-\gamma^{-1}\nu)\,d\nu\\
&=
\int \gamma \rho\big(\gamma (T-T')\big)g(T')d T'\\
&=
\int \rho_\Lab(T-T')g(T')d T'
\end{align*}
where $T'=T-\gamma^{-1}\nu$, and
\begin{align}
\rho_\Lab(T)=\gamma\rho(\gamma T)
\label{Stat_rho_Lab}
\end{align}
is the charge density as measured in the laboratory frame.
\end{proof}

\section{Expectation of electric and magnetic fields}
\begin{definition}
\label{defEXT}
For a particle of charge $q$ undergoing arbitrary motion $\Vx(\tau)$, where
$\tau$ is the particle's proper time, the Li\'{e}nard-Wiechert fields
at point $\VX$ and time $T$ are given \cite{Jackson99} by
\begin{align}
\VE(\VX,T)=\elw\Big(\VX-\Vx(\tau_R),\Vbeta(\tau_R),\Va(\tau_R)\Big)
\label{EB_TX}
\end{align}
and
\begin{align}
\VB(\VX, T)=\blw\Big(\VX-\Vx(\tau_R),\Vbeta(\tau_R),\Va(\tau_R)\Big).
\label{BB_TX}
\end{align}
where $\elw$ and $\blw$ are defined in lemma \ref{ecer_lem} and definition \ref{mag_def} respectively.

For the body point $\nu$ the Li\'{e}nard-Wiechert electric and magnetic
fields at point $\VX$ and time $T$ are given by substituting
$\tau_R=\tauhat(\VX,T,\nu)$ into (\ref{EB_TX}),
\begin{align*}
&\VE(\VX,T,\nu)=
\elw\Big(\VX-\Vx\big(\tauhat(\VX,T,\nu)\big),
\Vbeta\big(\tauhat(\VX,T,\nu)\big),
\Va\big(\tauhat(\VX,T,\nu)\big)\Big)
\end{align*}
and likewise for $\VB(\VX,T,\nu)$.
\end{definition}
Let $\VEE(\VX,T)$ be the electric field at point $\VX$ and time $T$ due to the body point $\nu=0$ given by
\begin{align*}
\VEE(\VX,T)=
\elw\Big(\VX-\Vx\big(\tauhat_0(\VX,T)\big),
\Vbeta\big(\tauhat_0(\VX, T)\big),
\Va\big(\tauhat_0(\VX,T)\big)\Big)
\end{align*}
Using
(\ref{eqns_tau_nu}) it follows
\begin{align}
\VE(\VX,T,\nu)
&=
\VEE(\VX,T-\nu/\gamma).
\label{E_shift}
\end{align}
and
\begin{align}
\VB(\VX,T,\nu)
&=
\VBB(\VX,T-\nu/\gamma).
\label{B_shift}
\end{align}

\begin{lemma}
The total electric and magnetic fields at time $T$ at the point $\VX$
are given by
\begin{align}
\VE_{\textup{Tot}}(\VX, T,\nu_1,\ldots,\nu_N)
=
\sum_{k=1}^N \VE(\VX, T,\nu_k)
\label{Eng_E_n1N}
\end{align}
and
\begin{align}
\VB_{\textup{Tot}}(\VX, T,\nu_1,\ldots,\nu_N)
=
\sum_{k=1}^N \VB(\VX, T,\nu_k).
\label{Eng_B_n1N}
\end{align}
It follows that
\begin{align}
\Exx{\VE_{\textup{Tot}}(\VX, T,\nu_1,\ldots,\nu_N)}
=
N \int  \rho(\nu) \VEE(\VX,T-\nu/\gamma)\,d\nu
\label{Eng_E_n1N}
\end{align}
and
\begin{align}
\Exx{\VB_{\textup{Tot}}(\VX, T,\nu_1,\ldots,\nu_N)}
=
N \int  \rho(\nu) \VBB(\VX,T-\nu/\gamma)\,d\nu.
\label{Eng_B_n1N}
\end{align}
\end{lemma}
\begin{proof}\\
The result follows from lemma \ref{sum_lem} and equations (\ref{E_shift}) and (\ref{B_shift}).
\end{proof}\\
Let total electric field at the point $\VX$ at time $T$ be given by
\begin{align}
\VET(\VX,T)=\frac{1}{N}\Exx{\VE_{\textup{Tot}}(\VX, T,\nu_1,\ldots,\nu_N)}
\end{align}
We notice that (\ref{E_shift}) and (\ref{B_shift}) are functions
where the dependence on $\nu$ is simply shifted $g(T-\gamma^{-1}\nu)$. Thus by lemma \ref{shift} it follows
\begin{align*}
\VET(\VX,T)
&=
\int\rho(\nu)\VE(\VX,T,\nu)d\nu\\
&=
\int\rho(\nu)\VEE(\VX,T-\nu/\gamma)d\nu
\\&=
\int\gamma\rho\big(\gamma(T-T')\big)\VEE(\VX, T')d T'\\
&=
\int \rhoLab(T-T')\VEE(\VX,T')d T',
\end{align*}
where $T'=T-\nu/\gamma$, and $q\rhoLab(T)=q\gamma\rho(\gamma T)$ is the
charge density as measured in the laboratory frame. Thus the key result
is that the total electric field is given by the
convolution
\begin{align}
\VET(\VX, T)
&=
\int \rhoLab(T-T')\VEE(\VX, T')d T'
\label{E_Tot}.
\end{align}
The above can be repeated for the total magnetic field
$\VB_{\text{Tot}}(\VX, T )$.
Clearly $\VEE(\VX, T')$ will depend on the energy of the beam $\gamma$
and the path of the beam $\Vx(\tau)$.

\section{Expectation of field energy and coherence}
\label{sec_energy}
\begin{definition}
The energy of the electromagnetic field at time $T$ at the
point $\VX$ for the $N$ particles is defined as the expectation
\begin{align}
\phi(\VX,T)
=
\Exx{
\norm{\VE_{\textup{Tot}}(\VX, T,\nu_1,\ldots,\nu_N)}^2+
\norm{\VB_{\textup{Tot}}(\VX,T,\nu_1,\ldots,\nu_N)}^2}
\label{Eng_Eng_def}
\end{align}
\end{definition}
\begin{lemma}
\begin{align}
\phi(\VX,T)=N\phi_\inc(\VX,T)+(N^2-N)\phi_\coh(\VX,T)
\label{Eng_coh_inc}
\end{align}
where the incoherent field is given by
\begin{align}
\phi_\inc(\VX,T) =
\ExxOP{\norm{\VE(\VX,T,\nu)}^2+\norm{\VB(\VX,T,\nu)}^2}
\label{Eng_phi_inc}
\end{align}
and the coherent field is given by
\begin{align}
\phi_\coh(\VX,T) = \norm{\VE_\cts(\VX,T)}^2 + \norm{\VB_\cts(\VX,T)}^2
\label{Eng_phi_coh}
\end{align}
where the one particle continuous electromagnetic fields are given by
\begin{align}
\VE_\cts(\VX,T)=\ExxOP{\VE(\VX,T,\nu)}
\qquadand
\VB_\cts(\VX,T)=\ExxOP{\VB(\VX,T,\nu)}
\label{Eng_E_B_cts}
\end{align}
I.e. $\VE_\cts(\VX,T)$ and $\VB_\cts(\VX,T)$ correspond to the
electric and magnetic fields due to a continuous distributions of
charge with distribution given by $\rho(\nu)$.

\end{lemma}
\begin{proof}
Expanding (\ref{Eng_Eng_def}) we see that this is simply a sum of products
\begin{align*}
\phi(\VX,T)
=&
\sum_{i}^3
\Exx{
E_{i,\textup{Tot}}(\VX,T,\nu_1,\ldots,\nu_N)\
E_{i,\textup{Tot}}(\VX,T,\nu_1,\ldots,\nu_N)}\\
&+
\sum_{i}^3
\Exx{
B_{i,\textup{Tot}}(\VX,T,\nu_1,\ldots,\nu_N)\
B_{i,\textup{Tot}}(\VX,T,\nu_1,\ldots,\nu_N)}
\end{align*}
where $E_{i,\textup{Tot}}$ is the $i$'th component of $\VE_{\textup{Tot}}$.
From (\ref{Stat_Expect_prod}) we have
\begin{align*}
\phi(\VX,T)
&=
N\sum_{i}^3 \ExxOP{E_i(\VX,T,\nu)^2} + (N^2-N)\sum_{i}^3 \ExxOP{E_i(\VX,T,\nu)}^2
\\&\qquad+
N\sum_{i}^3 \ExxOP{B_i(\VX,T,\nu)^2} + (N^2-N)\sum_{i}^3 \ExxOP{B_i(\VX,T,\nu)}^2
\end{align*}
where $E_i(\VX,T,\nu)$ is the $i$'th component of $\VE(\VX,T,\nu)$.
\end{proof}

We've already seen from (\ref{E_shift}) and (\ref{B_shift}) that the  electric and
magnetic fields are simply shifted functions so we can use
(\ref{Stat_1P_shift}) to give the coherent and incoherent fields in
terms of convolutions
\begin{align}
\phi_\inc(\VX,T) =
\int \rho_\Lab(T-T') \Big(\norm{\VEE(\VX,T')}^2+\norm{\VBB(\VX,T')}^2\Big) d T'
\label{Eng_phi_inc_conv}
\end{align}
and
\begin{align}
\VE_\cts(\VX,T)=&\int \rho_\Lab(T-T')\VEE(\VX,T') d T'\notag\\
 \textup{and}\quad \VB_\cts(\VX,T)=&\int \rho_\Lab(T-T')\VBB(\VX,T') d T'
\label{Eng_E_B_cts_conv}
\end{align}

\chapter[Numerical results]{Numerical results}
\label{bendingbeams}

In this chapter we present the results of a numerical investigation carried out with the mathematical software MAPLE. The relevant code can be found in appendix \ref{mapletwo}. In the following we give a brief outline of the calculations involved and state the main results.

\section{The field at a fixed point $\VX$ for a single particle }
\sectionmark{The field of a single particle}

Consider figure \ref{setup}. Half the aperture of the collimator is given by distance $h$.  We have seen (figure \ref{fig_Catchup}) that for high $\gamma$,  a pancake of radius $h$ can develop only if the particle  has been travelling in a straight line over a displacement of at least $\gamma h v/c \approx \gamma h$. Therefore for our proposal to be effective the distance $Z$ in figure \ref{setup} should be less than $\gamma h$. Preliminary results show that the optimum value for $Z$ is $Z\lesssim 10h$, with the field varying little with lower values, thus in the following analysis we fix the field measurement point $\VX=(0, h, 10h)$ and consider the magnitude of the electric field at $\VX$ due a single particle approaching and passing through the collimator. In all calculations we use $q=-1.60217 \times 10^{-19}$C. 

 We consider the path constructed from a straight line followed by an arc
of a circle of radius $R$ followed by another straight line. Observe
that this path is unrealistic since it would require large magnets to
remove the \emph{magnetic leakage}. Magnetic leakage is defined as the passage of magnetic flux outside the path along which it can do useful work. In general in a bending dipole the magnetic leakage causes the path of the charge to be slightly smoothed out at the ends of the dipole, so that in a real bending magnet the path of the charge would not be precisely the arc of a circle. We assume that the smoothing of the path
corresponding to real dipoles would not significantly change the nature
of the result. 

 Let $\Theta$ denote the angle of arc. The
coordinate system is chosen so that the direction of the second
straight line is along the $z$ axis and the arc is in the $x-z$ plane,
finishing at the origin.  We refer to this trajectory as the \emph{pre-bent} trajectory in contradistinction with that of a particle approaching from $(x, y, z)=(0, 0 -\infty)$ on a straight line towards the origin, which we refer to as the \emph{straight} trajectory.

The pre-bent trajectory is given by $\Vx(\tau)=(x(\tau), y(\tau), z(\tau))$ where
\begin{align}
x(\tau)=&\begin{cases}
R(\cos\Theta-1)+(\Theta R+\gamma v\tau)\sin\Theta \quad \textup{for} \quad -\infty<\tau<-{ R\Theta}/{\gamma v}\\
R\Big(\cos({\gamma v\tau}/{R})-1\Big)\quad \textup{for} \quad -{R\Theta}/{\gamma v}<\tau<0\\
0\quad \textup{for} \quad 0<\tau<\infty,
\end{cases}\notag\\[0.3cm]
y(\tau)=&\begin{cases}
 0 \quad \textup{for} \quad -\infty<\tau<\infty,
\end{cases}\label{prebent_path}\\[0.3cm]
\textup{and}\quad z(\tau)=&\begin{cases}
 -R\sin\Theta+(\Theta R+\gamma v\tau)\cos\Theta\quad \textup{for} \quad -\infty<\tau<-{ R\Theta}/{\gamma v}\\
 R\sin({\gamma v\tau}/{ R})  \quad \textup{for} \quad    -{ R\Theta}/{\gamma v}<\tau<0\\
 \gamma v\tau   \quad \textup{for} \quad 0<\tau<\infty.
\end{cases}\notag
\end{align}
The straight trajectory is given by
\begin{align}
(x(\tau), y(\tau), z(\tau))=(0, 0, \gamma v \tau) \quad \textup{for} \quad  -\infty<\tau<\infty.
\end{align}

\subsubsection*{Calculating the field at $\VX$ due to a specific path}
We use a coordinate system $\{\tau, \rnew, \that, \phat\}$ adapted from the Newman-Unti coordinates $\{\tau, R, \theta, \phi\}$. The coordinate transformation is given by (\ref{coords_new}). Comparison with (\ref{nu_def}) yields $\rnew=\frac{R}{\alpha}$ where $R=-\g(\x, \v)$ is the Newman-Unti radial parameter.
 We require the electric and magnetic fields due to a particle on a given trajectory.  For fixed field point $(\VX, T)=(X_0, Y_0, Z_0, T_0)$ there exist parameters $\rhat$, $\that$, and $ \phat$ which satisfy
\begin{align}
&\ccon T_0=\c^0 (\tau) + \rhat\notag\\
&X_0=\c^1(\tau) + \rhat\sin(\that)\cos(\phat)\notag\\
&Y_0=\c^2 (\tau) + \rhat \sin(\that)\sin(\phat)\notag\\
&Z_0=\c^3(\tau) +\rhat \cos(\that).
\label{T0_def}
\end{align}
Rearranging yields the relations
\begin{align}
\rhat=& \sqrt{(X_0-\c^1)^2+(Y_0-\c^2)^2+(Z_0-\c^3)^2}\notag\\
\ccon T_0(\tau)=&\rhat + \c^0\notag\\
\cos(\that)=&\frac{Z_0-\c^3}{\rhat}\notag\\
\sin(\that)=&\frac {\sqrt{(X_0-\c^1)^2+(Y_0-\c^2)^2}}{\rhat}\notag\\
\cos(\phat)=&\frac{X_0-\c^1}{\sqrt{(X_0-\c^1)^2+(Y_0-\c^2)^2}}\notag\\
\sin(\phat)=&\frac{Y_0-\c^2}{\sqrt{(X_0-\c^1)^2+(Y_0-\c^2)^2}}
\label{path_ref}
\end{align}
These relations can be substituted into the expressions (\ref{e_newcoords}) for the radiative $\elw_{\textup{R}}(\tau, \rnew, \theta, \phi)$ and Coulombic $\elw_{\textup{C}}(\tau, \rnew, \theta, \phi)$ electric fields (or magnetic fields). This gives the electric field (magnetic field) as a function of the components $\c^0, \c^1, \c^2, \c^3$ and the coordinates $T_0, X_0, Y_0, Z_0$.

When considering the electric field due to a particle on a specific trajectory we need only substitute the correct components for $\c$. For example in order to calculate the electric field for the pre-bent path we consider the three sections of the path independently. For each of the three intervals in (\ref{prebent_path}) we input the trajectory by defining components
\begin{align}
\c^0=\gamma \tau\notag\\
\c^1(\tau)=x(\tau)\notag\\
\c^2(\tau)=y(\tau)\notag\\
\c^3(\tau=z(\tau)
\label{worldline_comps}
\end{align}
where the corresponding values for $x(\tau), y(\tau)$ and $z(\tau)$ are defined in (\ref{prebent_path}). See appendix \ref{mapletwo} lines {\footnotesize \color{red} \tt 131-147}.

In the Maple code we have written a procedure which will take
 a selection of variable input parameters and output any field as a function of $\tau$. See (\ref{get_fields}). The variable input parameters are the the components  $\c^0, \c^1, \c^2, \c^3$ and a list of numerical values for the parameters $X_0, Y_0, Z_0$ and $\rrad, \Theta$, $\gamma$ as well as an initial value for $\tau$.

\subsubsection*{Lab time}

The ranges of $\tau$ for the three different trajectories are obtained by substituting the numerical input values for $\rrad, \Theta$ and $\gamma$ into the intervals in (\ref{prebent_path}). Thus for a given set of input variables we are able to plot any desired field component for a particular section of the path against $\tau$ for the range of $\tau$ appropriate to that section. In order to plot the field component against $\tau$ for the whole path we simply display the three plots corresponding to the three sections of the trajectory on the same graph.

The lab time $T_0(\tau)$ is a different function of $\tau$ for each of the three sections. This follows from (\ref{T0_def}). For a particular section we may obtain $T_0(\tau)$  by substituting our variable input parameters  into (\ref{T0_def}) and thus we may plot any desired field against $T_0$ for that section of path by plotting the field and the time $T_0$ as parametric equations in $\tau$. To plot the field over the whole range of $T_0$ we simply display all three plots on the same graph as before.

\subsubsection*{Optimizing the values of $\rrad$ and $\Theta$}

 We have control over the variable parameters $\rrad$ and $\Theta$.  In order to establish the optimum set of parameters to minimize the field at $\VX$ we calculate the peak energy of the electric field $||\VEE(\VX, T)||$ for a range of $T$, performing a parameter sweep for a selection of values for $\rrad$ and $\Theta$.  We set $\gamma=1000$, which represents an energy level easily obtained in modern accelerators. The procedure used in MAPLE is given in \ref{minimize_app} and the results are displayed in figure \ref{minimize}.  We have displayed three different views of the same graph. The numerical values for the electric field are absent because we are interested only in the relative values for the different trajectories. The values for $\Theta$ and $\rrad$ used in these plots are a selection of the values tested, however they are sufficient to show the trend. 
 
  Recall that $\rrad$ determines the curvature of the bend and $\Theta$ determines the length of the bend. Looking at the first graph we see that in the far corner of the graph, where $\Theta$ and $\rrad$ are at a minimum, the field is at a maximum. As we approach the near region of the graph the magnitude decreases very rapidly with increasing $\Theta$. We interpret this as follows. For a short bend the radial distance from the point $\VX$ to the continuation of the straight section will be small, and thus the pancake which developed on the straight section will be strongly encountered at point $\VX$. For a longer bend this contribution will be greatly reduced due to both the increased radial distance of $\VX$ from the continuation of the straight section and the increased longitudinal distance of $\VX$ from the terminus of the straight section. The latter distance is important because once the straight section ends and the bend begins the pancake is no longer travelling with the particle and the field strength within the pancake is decreasing. In consequence we interpret the ridge in the graph where the steep section ends as the cut off where the point X no longer encounters a significant field due to the pancake.
 
  In reality we cannot adopt the smallest $\rrad$ and largest $\Theta$ because they are impractical in the design considerations of real machines. We chose to restrict the trajectory to the values $\Theta=0.13$rad and $R=0.5$m because the bend is sufficient to reduce the field at $X$ while also maintaining a minimal length and intensity in order to suppress radiation and CSR wakes. In addition a bend of this size would be practical from an engineering perspective.

\begin{table}
\begin{center}
\begin{singlespacing}
\begin{tabular}{|c|c| c|}\hline
& $\rrad$ & $\Theta $\\\hline
min & $500$ & $1/95$\\
&$1000$&$1/90$\\
&$1500$&$1/85$\\
&$2000$&$1/80$\\
&$2500$&$1/75$\\
&$3000$&$1/70$\\
&$3500$&$1/65$\\
&$4000$&$1/60$\\
&$4500$&$1/55$\\
&$5000$&$1/50$\\
&$5500$&$1/45$\\
&$6000$&$1/40$\\
&$6500$&$1/35$\\
&$7000$&$1/30$\\
&$7500$&$1/25$\\
&$8000$&$1/20$\\
&$8500$&$1/15$\\
&$9000$&$1/10$\\
&$9500$&$1/5$\\
max&$10000$&$1$\\\hline
\end{tabular}
\end{singlespacing}
\end{center}
\caption{Input values for $\rrad$ and $\Theta$.}
\label{minimize_table}
\end{table}

\begin{figure}
\setlength{\unitlength}{0.9cm}
\begin{center}
\begin{picture}(15, 23)
\put(2, 10){\includegraphics[width=12\unitlength]{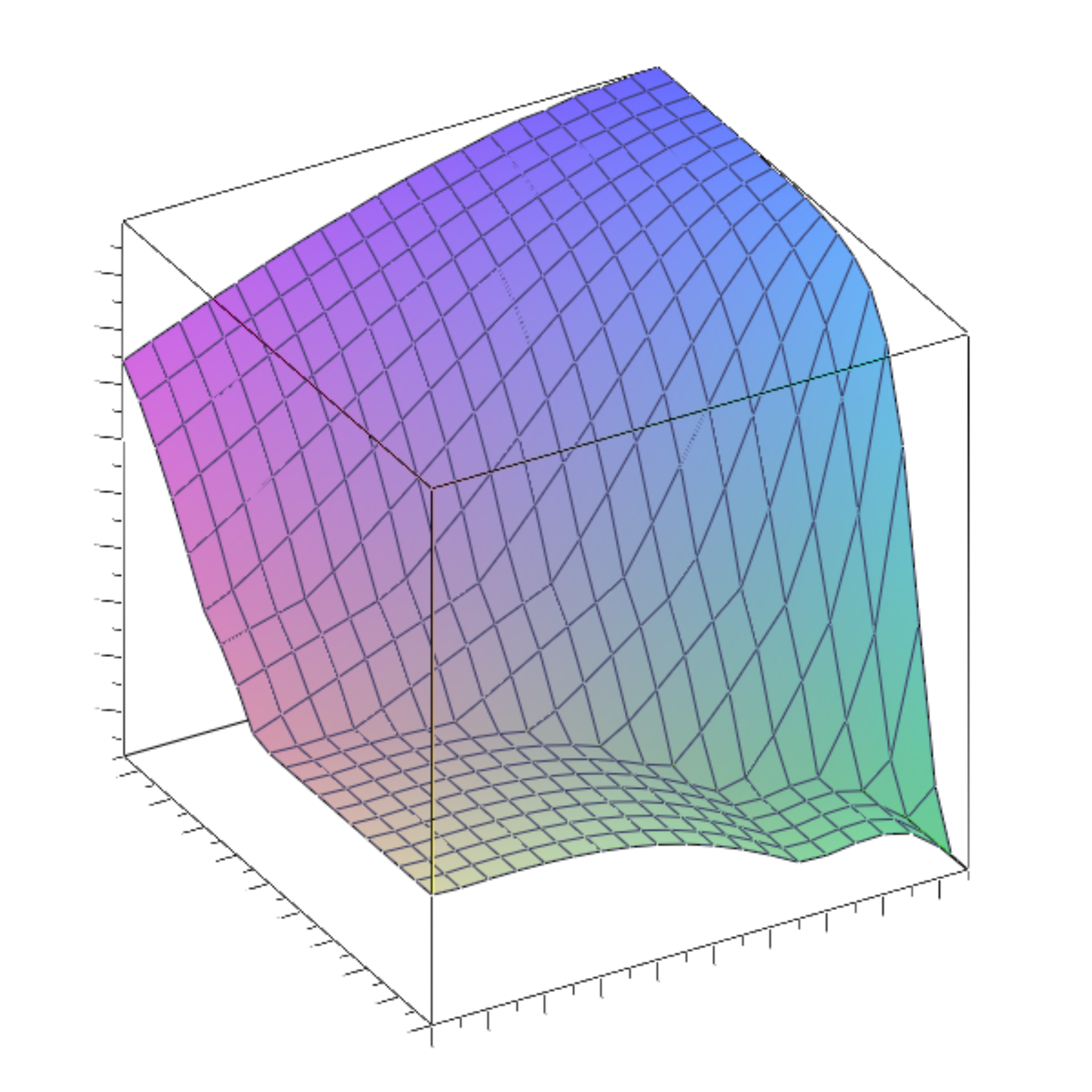}}
\put(-1.5,0){\includegraphics[width=10\unitlength]{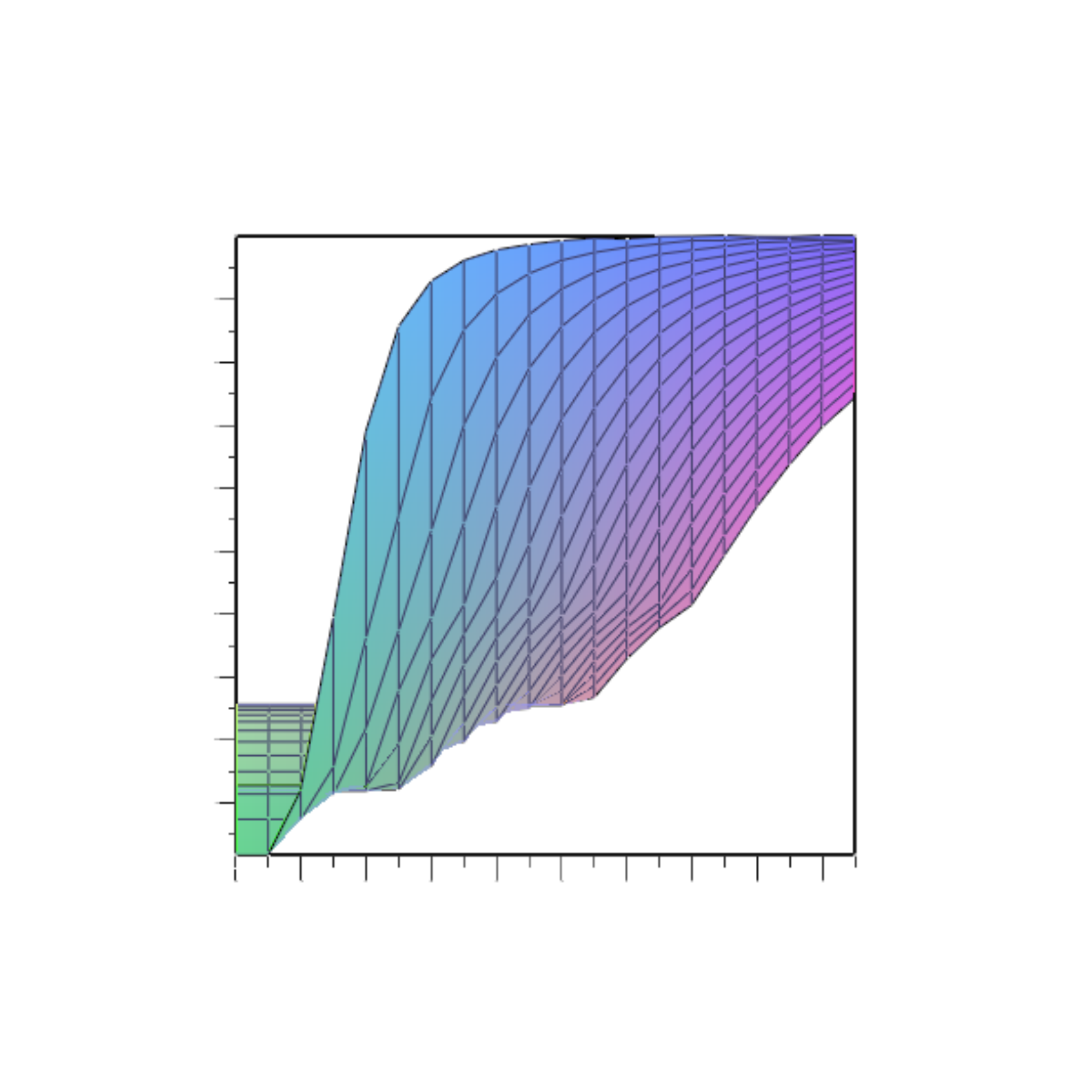}}
\put(7.5, 0){\includegraphics[width=10\unitlength]{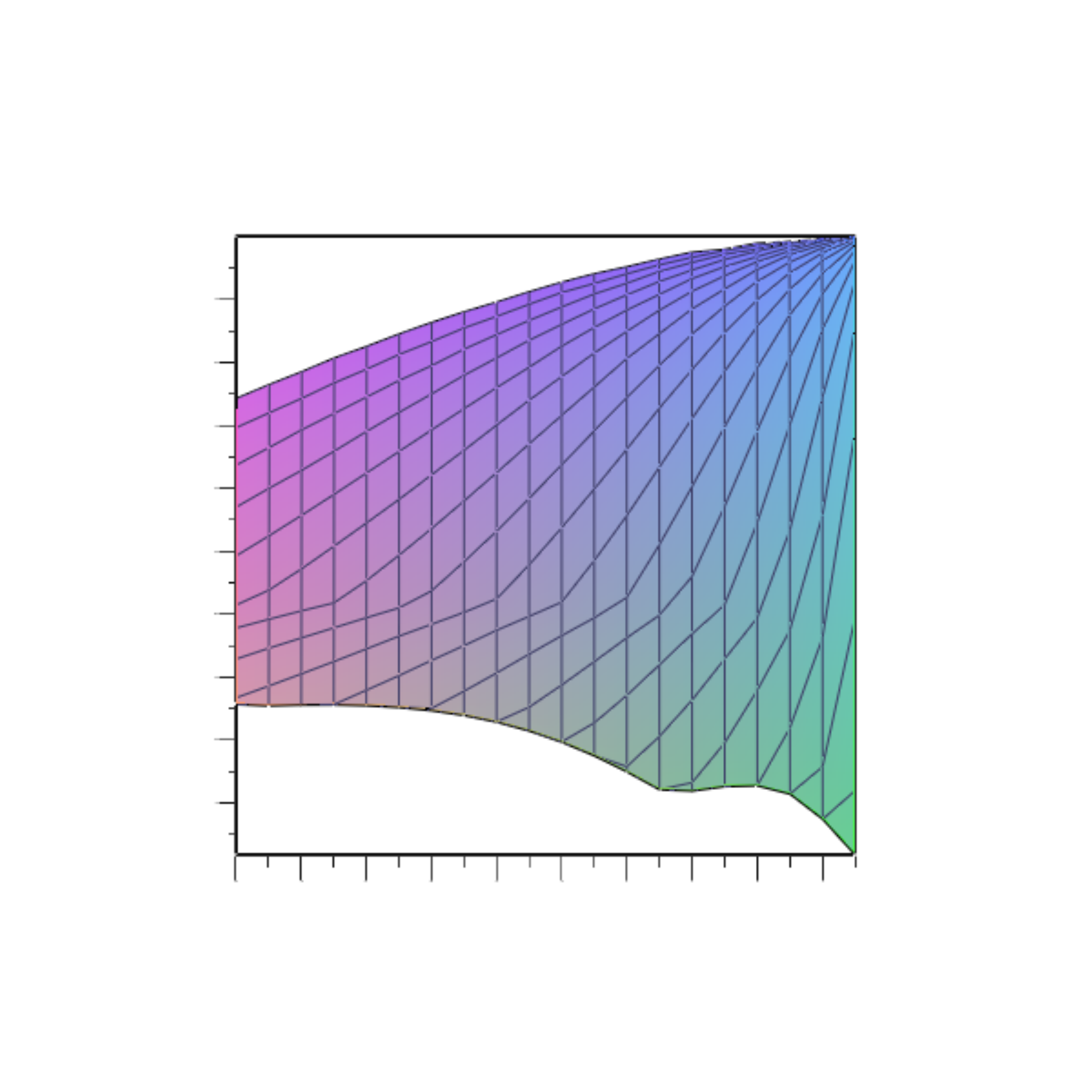}}

\put(4, 12){$\Theta$}
\put(10, 11){$\rrad$}

\put(3, 1){$\Theta$}
\put(12, 1){$\rrad$}

\put(0, 1){\small{max}}
\put(6, 1){\small{min}}
\put(9, 1){\small{max}}
\put(15, 1){\small{min}}
\put(2, 13){\small{min}}
\put(5.4, 10){\small{max}}
\put(7.3, 10){\small{max}}
\put(12.5, 11.5){\small{min}}

\put(-0.6, 4){\rotatebox{90}{{ $ ||\VEE(\VX, T)||$}}}
\put(8.6, 4){\rotatebox{90}{{ $ ||\VEE(\VX, T)||$}}}
\put(2.2,16){\rotatebox{90}{{ $ ||\VEE(\VX, T)||$}}}
\end{picture}
\end{center}
\caption{Field strength for different values of $R$ and $\Theta$. We see clearly that the minimum field energy occurs when the $\rrad$ is at its minimum value and $\Theta$ is at its maximum value.}
\label{minimize}
\end{figure}
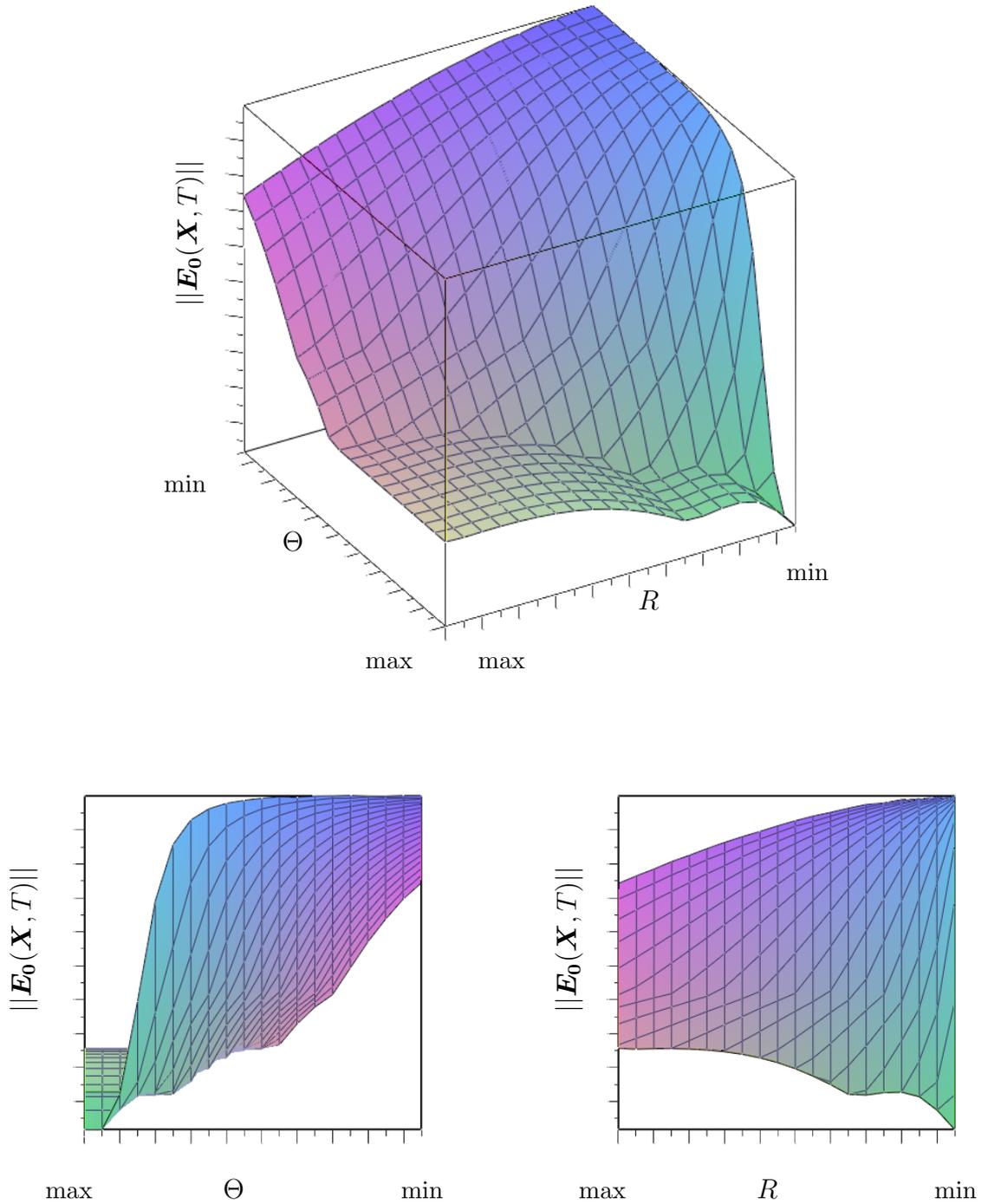

\subsection*{The field due to a particle on the optimised pre-bent trajectory compared with that of a particle on the straight trajectory}
Consider the two cases given in figure \ref{fig_fields} in which
$\gamma=1000$, $x=0$, $y=h$ and $z=10h$. In the straight
line case the peak field is $\approx75$Vm$^{-1}$ and the majority of the field arrives within an interval of
$0.015$ps. In fact it is easy to show that for a straight
line path the peak field increases with $\gamma$ and the width decreases
with $\gamma$ leading to the classic pancake. By contrast for the pre-bent case
 the peak field is significantly reduced to only $\approx7.7$Vm$^{-1}$, however
 the interval over which the field arrives is now $0.35$ps for the right hand peak,
and $0.1$ps for the left hand peak. The reason for these
two peaks is that the left hand peak is the coulomb field due to
the first straight line segment, whereas the second peak is due to the
radiation from the circular part of the beam path. The discontinuity is a result of
 the discontinuity in acceleration for this trajectory. Repeating the
calculation with higher $\gamma$-factors does not significantly change the
height or shape of the second peak.

Figure \ref{fig_fields_comp} shows the cartesian components of the electric field for the particle on the pre-bent trajectory. We see that the field is largely in the $x$ and $y$ directions, with the peak field in the $y$ direction. By contrast for a particle on the straight trajectory the $y$ component dominates with the $z$ component negligible and the $x$ component zero. This means that for a straight trajectory the field is primarily directed in the transverse direction as expected. The nonzero $x$ component in the pre-bent case is due to the incident angle of the initial straight section as well as the radiation caused by the bend.
\setlength{\unitlength}{0.9cm}
\begin{figure}
\begin{center}
\begin{picture}(10.5,10.7)
\put(0.2,0.5){
\includegraphics[height=10.7\unitlength]{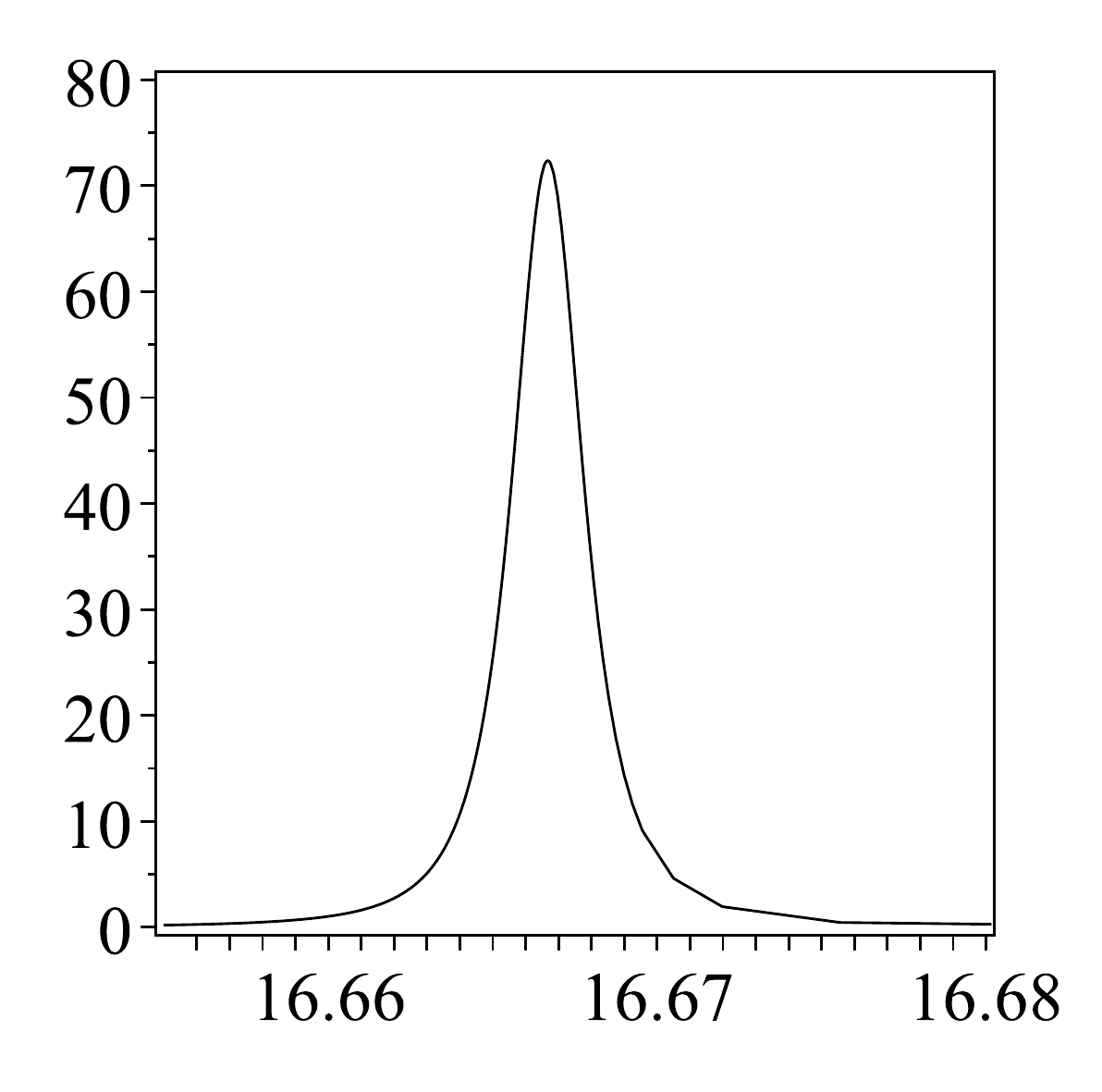}
}
\put(0,4){\rotatebox{90}{{ $||\VEE(\VX, \THat_0)||$ $[\textup{Vm}^{-1}]$}}}
\put(5,0.4){{$\THat_0$ [ps]} }
\put(4, -0.6){Straight path}
\end{picture}\\[1.5cm]
\begin{picture}(10.5,10.7)
\put(0.5,0.5){
\includegraphics[height=10.7\unitlength]{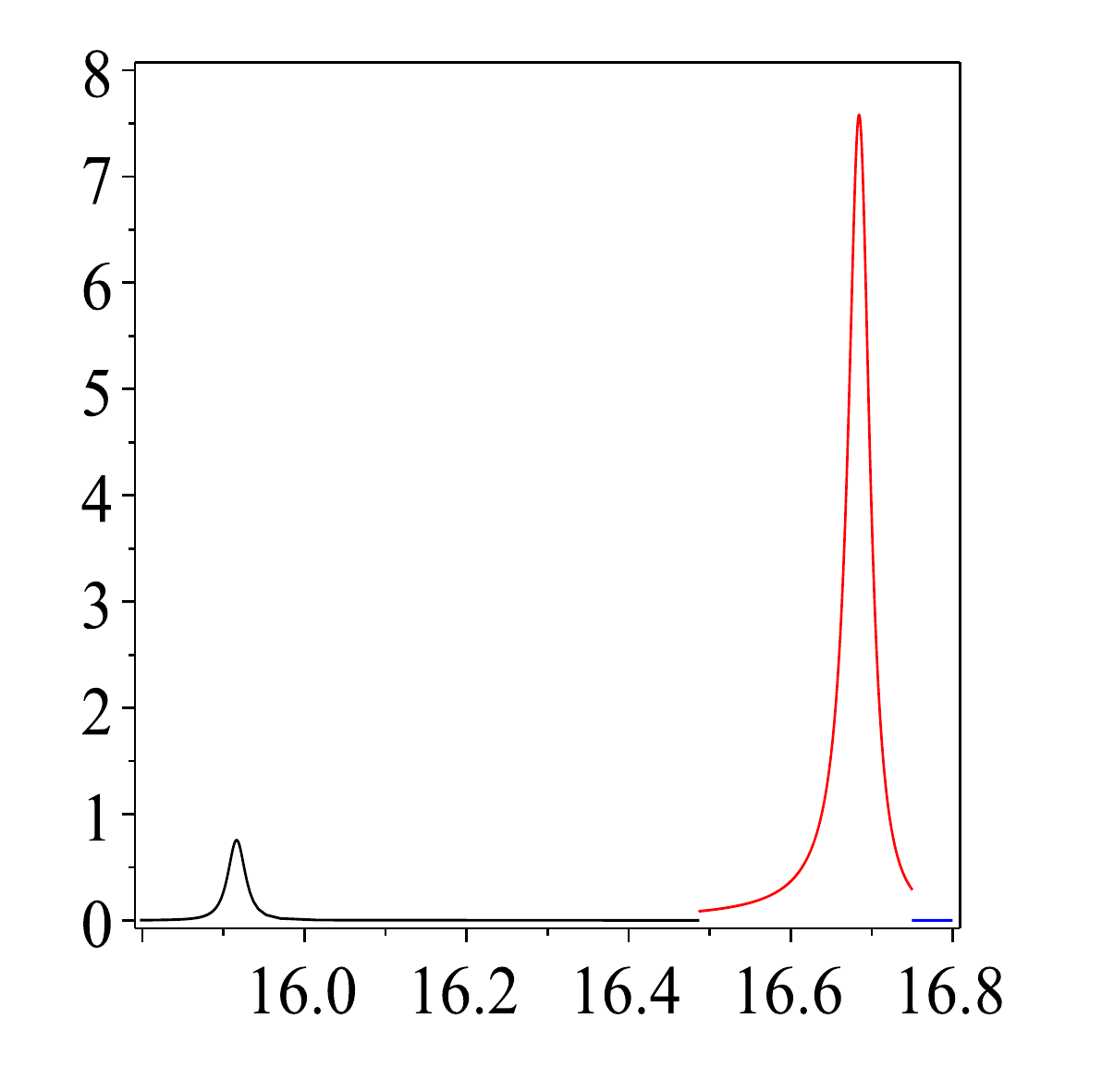}
}
\put(0,4){\rotatebox{90}{{$||\VEE(\VX, \THat_0)||$ $[\textup{Vm}^{-1}]$}}}
\put(5,0.4){{$\THat_0$ [ps]}}
\put(4, -0.6){Pre-bent path}
\end{picture}
\end{center}
\caption[Electric field $||\VEE(\VX, T)||$ for straight and pre-bent trajectories]{The electric field strength $||\VEE(\VX, T)||$ at $\VX=(0, h, 10 h)$, with $h=0.5$mm, due to a body point following a
 straight path along the $z$-axis and a body point following the \emph{pre-bent} path given in
 (\ref{prebent_path}) with $\Theta=0.13$rad and $R=0.5$m. }
 \label{fig_fields}
\end{figure}

\setlength{\unitlength}{0.7cm}
\begin{figure}
\begin{center}
\begin{picture}(10.5,11.2)
\put(-0.5,0.5){
\includegraphics[height=10.7\unitlength]{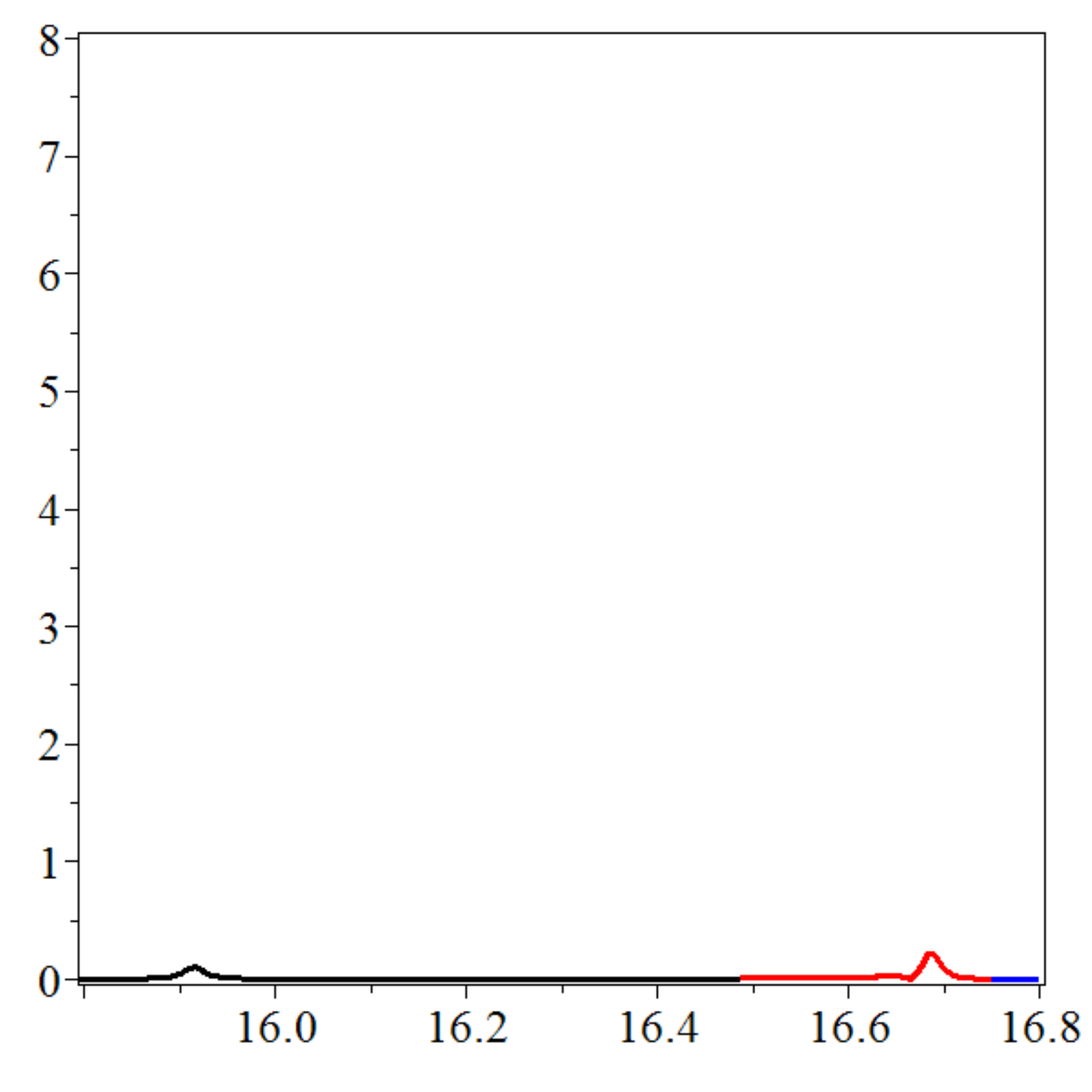}
}
\put(4, 11.2){$z$-component}
\put(-1,4){\rotatebox{90}{{ $(\VEE)_z(\VX, \THat_0)$ $[\textup{Vm}^{-1}]$}}}
\put(5,0){{$\THat_0$ [ps]} }
\end{picture}\\[1.2cm]
\begin{picture}(21,10.7)
\put(-0.5,0.5){
\includegraphics[height=10.7\unitlength]{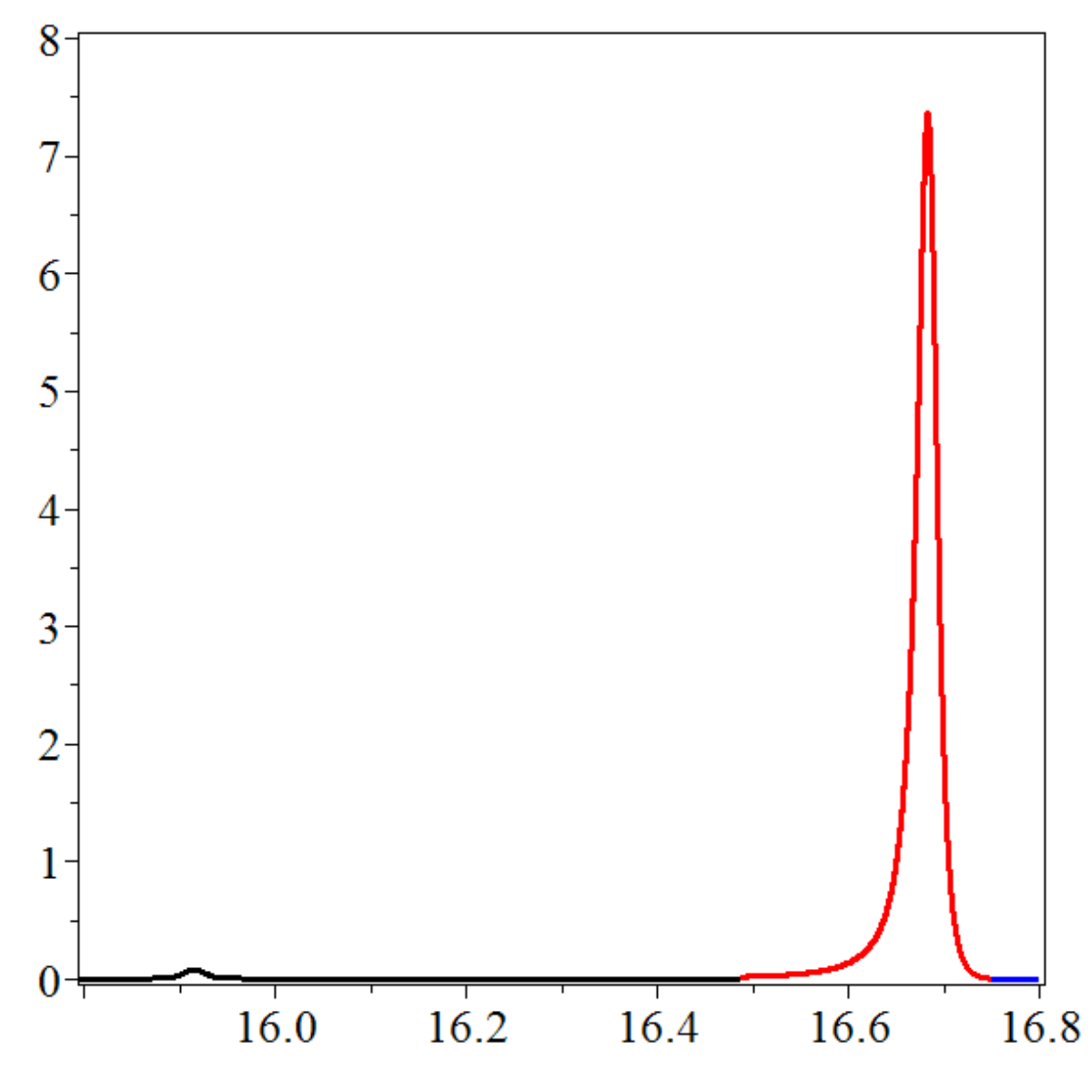}
}
\put(4, 11.2){$y$-component}
\put(-1,4){\rotatebox{90}{{$(\VEE)_y(\VX, \THat_0)$ $[\textup{Vm}^{-1}]$}}}
\put(5,0){{$\THat_0$ [ps]}}
\put(11,0.5){
\includegraphics[height=10.7\unitlength]{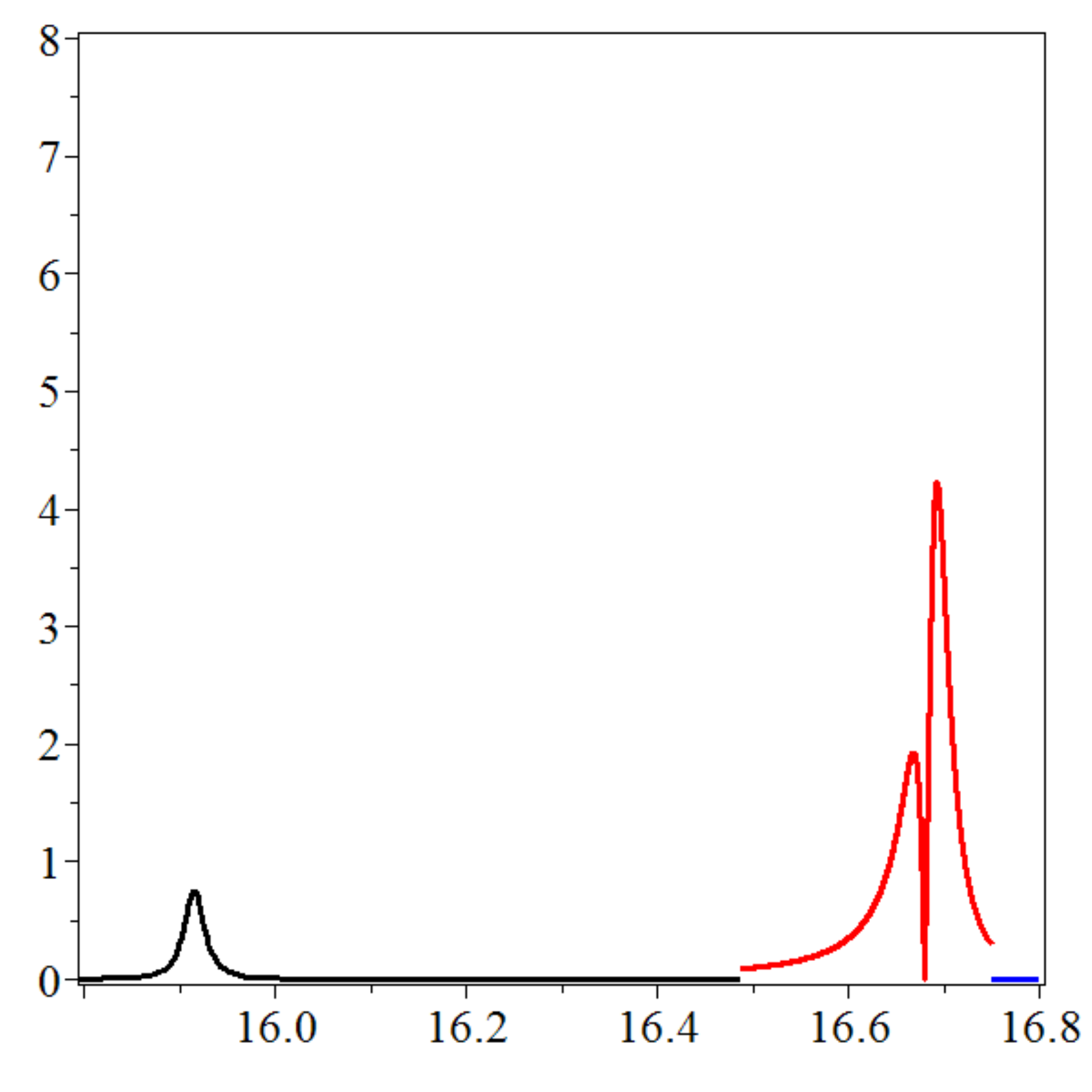}
}
\put(15.7, 11.2){$x$-component}
\put(10.5,4){\rotatebox{90}{{$(\VEE)_x(\VX, \THat_0)$ $[\textup{Vm}^{-1}]$}}}
\put(15.5,0){{$\THat_0$ [ps]}}
\end{picture}
\end{center}
\caption[Electric field components $(\VEE)_x, (\VEE)_y$ and $(\VEE)_z$ for  pre-bent trajectories]{Electric field components $(\VEE)_x, (\VEE)_y$ and $(\VEE)_z$ at point $\VX=(0, h, 10 h)$, with $h=0.5$mm, due to a body point following a
 \emph{pre-bent} path with $\Theta=0.13$rad and $R=0.5$m. }
 \label{fig_fields_comp}
\end{figure}

\section{The coherent field at $\VX$ due to a bunch}

Consider a bunch modelled as a one dimensional continuum of point
particles with a low density halo. The fields generated in the halo will
not be considered. The one dimensional continuum is a good model for beams with transverse
dimension significantly smaller than the bunch length. The
assumption is made that the majority of the bunch charge is contained in the
one-dimensional core and only the halo is removed by collimation. Within
this model, each particle in the core undergoes the same motion in space but at a
different time and is moving at a constant speed with
relativistic factor $\gamma$.

Using (\ref{E_Tot}) we calculate the field due to a Gaussian particle distribution $\rho_{\text{Lab}}$ for the two cases
in figure \ref{fig_fields}.  The code for the convolution can be found in \ref{convolution_app}. We input the time $T_0=t$ at which the convolution should be centered,
  the bunch length (FWHM of the Gaussian) as upper and lower values of $T_0$, and the number
   of points $N$ over which the samples should be taken. The procedure can be summed up in the
   following  steps which are followed in a do loop for $j=1..N$.
 \begin{itemize}
 \item Define $\rhoLab:= (t, a, b) -> 1/(a \sqrt(2\pi))\exp\big((-(t-b)^2)/(2 a^2)\big)$

 \item Solve $T_0(\tau_j)=t_j$ for $\tau_j$, where $t_j=t-a-(b-a)/N)(j+1/2)$ and $a$ and $b$ are the upper and lower bounds on the range of $T_0$ respectively.\\

 \item Substitute $\tau=\tau_j$ into the electric field $\elw(T_0(\tau))$ to give the field strength at time $T_0=t_j$\\

 \item Multiply $\elw(t_j)$ by $\rhoLab(t_j)$, where $\rhoLab$ is a specific Gaussian defined by inputting FWHM.

 \item Sum the result over $j$, $\textup{sum}=\displaystyle{\sum_{j=1}^{N}  \elw(t_j) \rhoLab(t_j)}$

 \item Normalize by dividing $\textup{sum}$ by   $\displaystyle{\sum_{j=1}^{N}  \rhoLab(t_j)}$

 \end{itemize}

 Table \ref{table3} displays the results for a selection of bunch lengths which are attainable in some present day machines.
\begin{table}
\begin{center}
\caption{Peak field strength for different sized bunches with h=0.5mm.}
\begin{tabular}{| @{\ \ }l @{\ \ }| @{\ \ } l@{\  \ }| @{\ \ }l@{\ \ }|@{\ \ }l@{\ \ }|}\hline
\multicolumn{2}{|c|@{\ \ }}{Bunch Length} &\multicolumn{2}{c|}{Peak $||\VET(\VX,T)||$ [Vm$^{-1}$]}\\\hline
 \multicolumn{1}{|c|@{\ \ }}{L[h]} & \multicolumn{1}{c|@{\ \ }}{(L/c)[ps]} &  \multicolumn{1}{c|@{\ \ }}{straight}&  \multicolumn{1}{c|}{pre-bent} \\\hline
 $1.8 \times 10^1   $   &  $3.00\times 10^0  $ &  $1.97 \times 10^{-1}$  &  $1.97 \times 10^{-1}$\\
 $6.00 \times 10^{-1}$   &  $1.00 \times 10^0 $ &  $5.91 \times 10^{-1}$  &  $5.89 \times 10^{-1}$\\
 $9.00 \times 10^{-2}$   &  $1.50 \times 10^{-1}$ &  $3.93 \times 10^0   $  &  $3.48 \times 10^0   $\\
 $4.80 \times 10^{-2}$   &  $8.00 \times 10^{-2} $ &  $7.33 \times 10^0   $  &  $5.27 \times 10^0   $\\
 $3.00\times 10^{-2} $   &  $5.00 \times 10^{-2} $ &  $1.16 \times 10^1   $  &  $6.36\times 10^0    $\\
 $4.80 \times 10^{-3}$   &  $8.00\times 10^{-3}  $ &  $5.12 \times 10^1   $  &  $7.53 \times 10^0   $\\\hline
\end{tabular}
\end{center}
\label{table3}
\end{table}

\subsection{Long bunches}

If the bunch is long and smooth, i.e. longer than the collimator
aperture, so that there is no significant change in $\rhoLab$ over the
width of $\VEE(\VX,T')$, then $\VEE(\VX, T')$ may be crudely regarded as a
$\delta$-function and $\VET(\VX,T)$ is given  by
\begin{align}
\VET(\VX,T)
&\approx
\rhoLab(T)
\int \VEE(\VX,T')d T'.
\label{E_Tot_smooth}
\end{align}
Integration of $\VEE(\VX,T')$ for the straight and pre-bent trajectories
 reveals that
\begin{align}
||\VET(\VX, T)||\approx\frac{q}{2\pi\epsilon_0 c}\frac{\rhoLab(T)}{||\VX||}
\label{E_Tot_smooth_res}
\end{align}
This value of $\VET$ is independent of $R$ and
$\Theta$ for all paths where $R$ is large compared to $L$. To see why
this is the case consider our one dimensional beam of particles as a
continuous flow of charge, similar to a line charge in a wire but
without the background ions. The fields due to this flow may be
calculated using the Biot-Savart law. Since $h\ll R$ the field is
dominated by the nearby current and hence no variation of $R$,
$\Theta$ or $Z$ will alter the fields.
We find that calculations using the Biot-Savart Law agree very closely with equation (\ref{E_Tot_smooth_res}), thus providing verification of our code.

\subsection{Short bunches}

If the beam has bunches of length $L\lesssim 0.05 h$ then it follows from
(\ref{E_Tot}) and figure \ref{fig_fields} that a considerable reduction in fields is possible. If $\rho_{\text{Lab}}$ has full width
at half maximum $L/c=0.008$ps with corresponding bunch length $L=0.0048 h$, then the peak value for the total electric field in the straight line case is given by $\approx 51.2 \textup{Vm}^{-1}$. By contrast, in the pre-bent case the peak value
for the total electric field is $\approx 7.5\textup{Vm}^{-1}$,
giving an approximate factor of 7 reduction in field. This is approaching
the maximal factor of 10 improvement one can achieve with $\gamma=1000$, which occurs when
the bunch length is small enough that the convolution gives the peak values for the fields in figure \ref{fig_fields}. With
higher energies and shorter bunch lengths the radiation peak remains unchanged, whereas the electric field for the
straight path grows linearly with $\gamma$. Thus even greater
improvements can be made.

\section{Conclusion}

In our analysis we have chosen a specific point $\VX=(0, 0.5$mm$, 5$mm$)$ and minimized the field at this point. We have shown that the magnitude of the electric field due to a single particle can be considerably reduced by altering the path of the beam, however the duration for which the field is non-zero is increased. We have used this result to show that the coherent field for a short bunch can be reduced significantly by bending the beam, with reductions of up to $85\%$ feasible for some present day FELs and future colliders. No reduction in the coherent field can be made for long smooth bunches. We assume the coherent field will dominate the incoherent field because of the $N$ Vs $N^2$ behaviour given in (\ref{Eng_coh_inc}), however the incoherent fields are always present.

 If the field point $\VX$ is instead displaced in the positive $x$ direction, then a
significant increase in field strength is observed. This increase results from
both the Coulomb field from the straight section of the path before
the arc and the radiation from the circular part of the path.

 Consider figure \ref{fig_contours}. The beam has been pre-bent from the left before passing through the origin of the graph, hence while on the bend the direction of motion is in the positive $x$ direction. The magnitude of the field is shown as a contour plot. We see the magnitude increasing as we pass from the origin along the x-axis and then decreasing again after a very dense region. This pattern is what we expect to find from the synchrotron radiation emitted from the bend.   The darkest parts of the graph are where the majority of the synchrotron radiation passed through the x-y plane. There are four discrete spots;  two very dark spots at approximately $x=1.4$mm and two slightly less intense spots to the left of these at $x=1.0$mm. The two dark spots are approximately ten times the magnitude of the other two. All the remaining field shown in the plot is negligible in comparison to these four peaks. It will be necessary to alter the shape of the collimator to avoid the high field regions interacting with the material in the
collimator. This need not affect the efficacy of the collimator to
remove the halo, for example see figure \ref{fig_Modified_coll}.

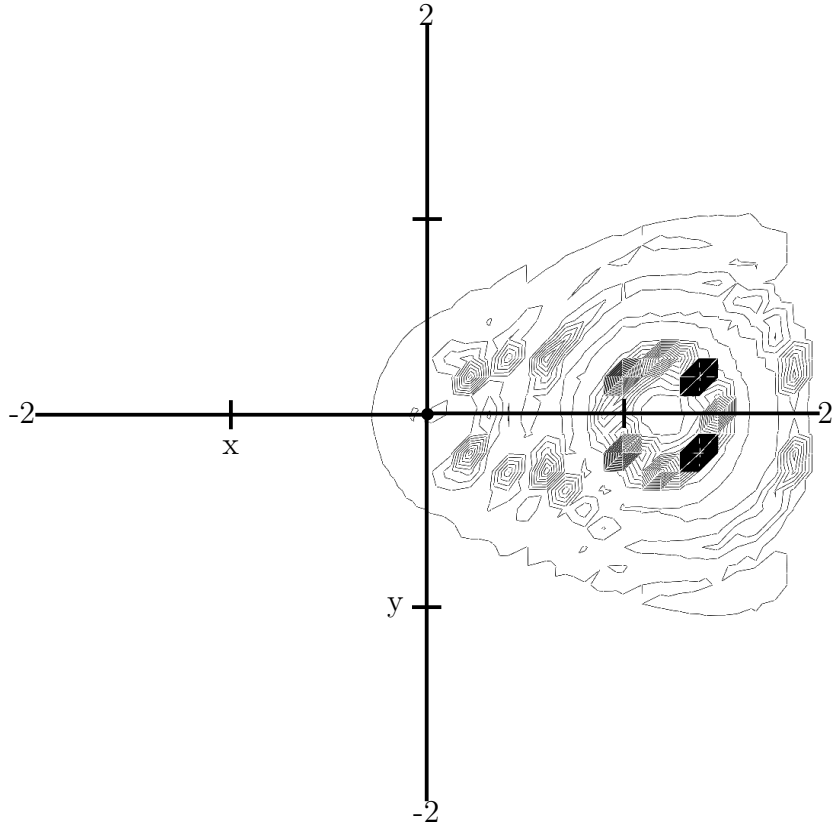
\begin{figure}
\begin{center}
\setlength{\unitlength}{0.035\textwidth}
\begin{picture}(20,20)
\put(0,0){\includegraphics[height=20\unitlength]{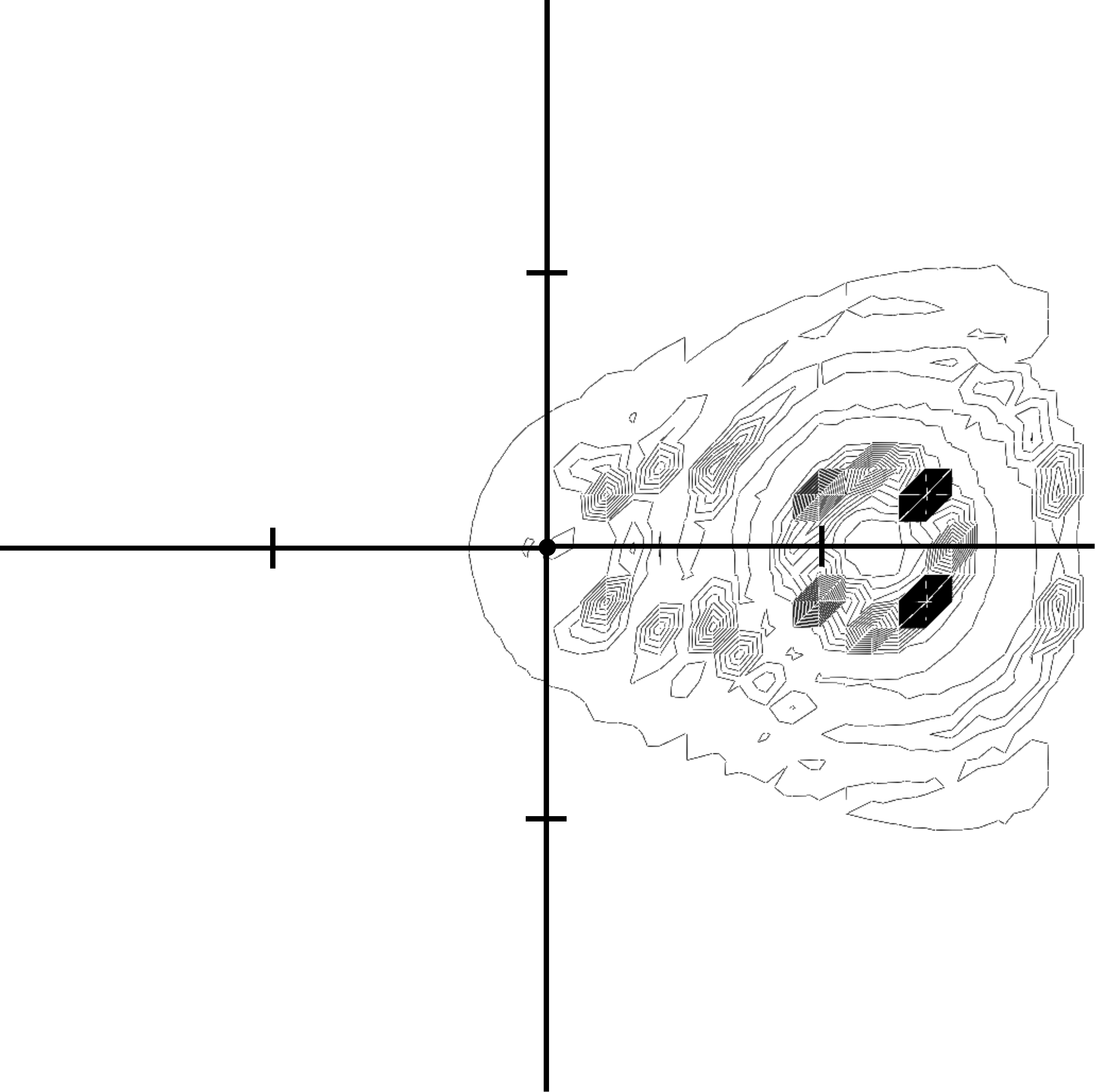}}
\put(0,10){\makebox(0,0)[r]{-2}}
\put(20,10){\makebox(0,0)[l]{2}}
\put(10,0){\makebox(0,0)[t]{-2}}
\put(10,20){\makebox(0,0)[b]{2}}
\put(9.,5){\makebox(0,0)[l]{y}}
\put(5,9.3){\makebox(0,0)[t]{x}}
\end{picture}
\end{center}
\caption[Electric field magnitude in $(x, y, Z)$ plane]{Contour plot for the maximum fields $||\VEE(\VX, T)||$ in the $(x,y)$ plane transverse to beam at $z=Z$. The plot
  represents a 4mm$\times$4mm region around the beam. The beam
  has been bent from the left before passing though the origin of the
  graph. Twenty timesteps were taken and the lack of
  smoothness is due to numerical errors. Most of the field is between 0$\textup{Vm}^{-1}$ and 100$\textup{Vm}^{-1}$ in magnitude however the two black spots represent regions where the field is approximately 1000 times greater.}
\label{fig_contours}
\end{figure}

\begin{figure}
\begin{center}
\setlength{\unitlength}{0.04\textwidth}
\begin{picture}(10,10)
\put(0,0){\rotatebox{90}{\includegraphics[width=10\unitlength]{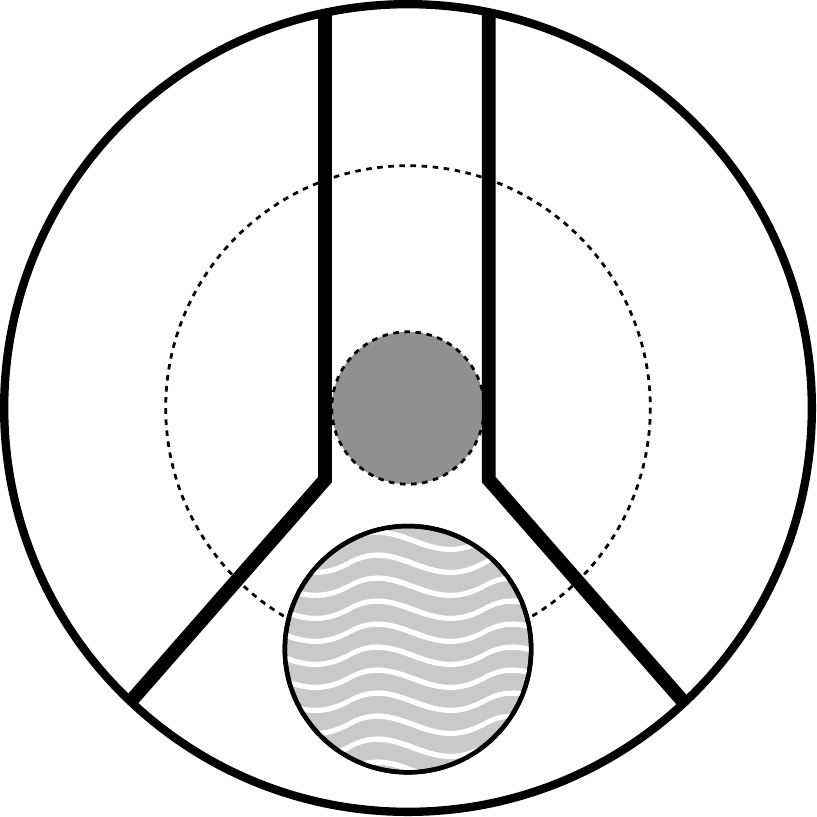}}}
\put(-0.5,7){\rotatebox{45}{beam Pipe}}
\put(2.2,3){\rotatebox{-45}{halo}}
\put(4,4.8){\rotatebox{0}{beam}}
\put(1,6.2){\rotatebox{0}{collimator}}
\put(6.7,5){{High}}
\put(6.7,4.3){{Fields}}
\end{picture}
\end{center}
\vspace{-2em}
\caption{Modified collimator in the plane transverse to the path of the beam.}
\label{fig_Modified_coll}
\end{figure}
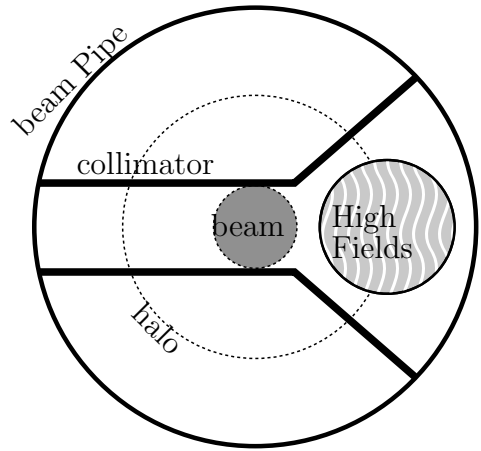

One criticism of our work is the fact that we have used the \LW field, which is strictly accurate only for a particle in free space. In the accelerator community the following formula is often employed to describe the field of a bunch traversing a circular beam pipe
\begin{align}
E_r=\frac{2\lambda}{r},
\end{align}
where $\lambda$ is the longitudinal charge distribution and $r$ is the radial distance from the axis. This formula assumes $\sigma_z \geq r/\gamma$, where $\sigma_z$ is the rms bunch length. For typical machines this is the normal regime, for example with $r=1cm$, $\gamma=1000$ then $\sigma_z \geq 10\mu m$. In our investigation we have shown that the most substantial reductions in wakefield are expected for very small bunch lengths, therefore even without the difficulty introduced by the bend this formula would be inappropriate. The correct form for the field in a curved beam pipe is a boundary value problem depending on the intricate geometry of the beampipe, see for example \cite{GotoTucker}.\\[3pt]

In this investigation we have adopted the \emph{rigid beam approximation} so that the charge profile
remains constant throughout the whole trajectory. This approximation is generally adopted in calculating
the field generated by a relativistic beam traveling in a straight line, however for a pre-bent trajectory
there will be CSR wakes and energy loss due to radiation which in general will disrupt the charge profile.
In order to determine whether there will be an advantage in bending the beam before collimation it will be necessary to calculate the net effect of bending plus reduced collimation wakes compared with collimation wakes on a straight beam trajectory. It is probable that the adverse effects of implementing an additional bend will be too severe for there to be any advantage in this new approach. However all accelerators, even linacs, already have to bend the beam using dipoles in certain places. Therefore it seems natural to
place a collimator directly after a bending magnet in order to prevent additional adverse effects. The optimum design of the beam path, beam
tube and collimator shape, for particular machines will require a
combination of analytic, numerical and experimental research.
Clearly long tapers will reduce the advantage gained by bending the beam
since it will give time for the pancake to form. However it may be
advantageous to use a short taper.

\begin{appendices}

 \chapter{Dimensional Analysis}
\label{app_dimensions}
 The SI base units are given by
 \begin{align}
 L&=\textup{metre},\quad m\notag\\
 T&=\textup{second}, \quad s\notag\\
 M&=\textup{kilogram}, \quad kg\notag\\
 A&=\textup{Ampere}, \quad A
 \end{align}
 It is more convenient to use the dimension of charge $Q$ instead of $A$, with derived unit the Coulomb $C=sA$. We use square brackets to denote the dimensions of the enclosed object. The constants $\epsilon_0, \mu_0$ where $\ccon^{-2}=\epsilon_0 \mu_0$ have dimensions
 \begin{align}
 [\epsilon_0]=&\frac{Q^2 T^2}{M L^3}\notag\\
 [\mu_0]=&\frac{ML}{Q^2}
 \end{align}
 The electric current 3-form has the dimension of charge $[\underset{{\footnotesize (3)}}{\jgen}]=Q$. Here the \begin{footnotesize}$(3)$\end{footnotesize} denotes the degree of the differential form. Dimensions of the electric and magnetic fields are given by
\begin{align}
[\underset{{\footnotesize (0)}}{\e^i}]=&\frac{M L}{Q T^2}, \qquad     [\e]=\frac{M}{Q T^2}  \qquad \textup{and} \qquad
[\underset{{\footnotesize (1)}}{\etilde}]=\frac{M L^2}{Q T^2}\notag\\
[\underset{{\footnotesize (0)}}{\bvec^i}]=&\frac{M }{Q T},\qquad     [\bvec]=\frac{M}{L Q T^2}  \qquad \textup{and} \qquad
[\underset{{\footnotesize (1)}}{\widetilde{\bvec}}]=\frac{M L }{Q T}
\end{align}
and dimensions of the electromagnetic 1-from potential $\pot$ and 2-forms $\f$ and $\star\f$ are given by
\begin{align}
[\underset{{\footnotesize (1)}}{\pot}]=[\underset{{\footnotesize (2)}}{\f}]=[\underset{{\footnotesize (2)}}{\star\f}]=\frac{M L^3}{Q T^2}
\end{align}
Also
\begin{align}
[\underset{{\footnotesize (3)}}{\stressk}]=& [\underset{{\footnotesize (0)}}{\textup{P}_{\textup{\kkill}}}]=\frac{ML}{T}\qquad\textup{and}\qquad [\underset{{\footnotesize (0)}}{\pdot_{\textup{\kkill}}}]=  [\textup{force}]= \frac{ML}{T^2}.
\end{align}
 Base quantities have dimensions
 \begin{align}
[x^a]=[dx^a]=L, \qquad\textup{and} \qquad\Big[\frac{\partial}{\partial x^a}\Big]=\frac{1}{L}
\end{align}
such that $[g]= L^2$ and $[\gdual]= \tfrac{1}{L^2}$.

\section{Dimensions  in chapter \ref{maxlorentz} and Part II}

We choose proper time $\tau$ to have the dimension of time $T$ such that
\begin{align}
[\c^a(\tau)]&=L, \quad\quad [\cdot^a(\tau)]=\frac{L}{T}, \quad\quad [\cddot^a(\tau)]=\frac{L}{T^2}
\end{align}
It follows
\begin{align}
[\x]=&[x^a-\c^a(\tau)]\big[\frac{\partial}{\partial x^a}\Big]=1 , \qquad [\widetilde{\x}]=[g(-,\x)]=L^2,\notag\\
[\v]=&[\cdot^a(\tau)]\Big[\frac{\partial}{\partial x^a}\Big]=\frac{1}{T}, \qquad [\vdual]=[g(-, \v)]=\frac{L^2}{T},\notag\\
[\a]=&[\cddot^a(\tau)]\Big[\frac{\partial}{\partial x^a}\Big]=\frac{1}{T^2}, \qquad[\adual]=[g(-, \a)]=\frac{L^2}{T^2}
\end{align}
and
\begin{align}
[g(\v, \v)]=&\frac{L^2}{T^2}, \quad\quad [g(\x, \v)]=\frac{L^2}{T}, \quad\quad [g(\x, \a)]=\frac{L^2}{T^2}
\end{align}

\section{Dimensions in  Part I}

We choose proper time $\tau$ to have the dimension of time $L$ such that
\begin{align}
[\c^a(\tau)]&=L, \quad\quad [\cdot^a(\tau)]=1, \quad\quad [\cddot^a(\tau)]=\frac{1}{L}
\end{align}
It follows
\begin{align}
[\x]=&[x^a-\c^a(\tau)]\big[\frac{\partial}{\partial x^a}\Big]=1 , \qquad [\widetilde{\x}]=[g(-,\x)]=L^2,\notag\\
[\v]=&[\cdot^a(\tau)]\Big[\frac{\partial}{\partial x^a}\Big]=\frac{1}{L}, \qquad [\vdual]=[g(-, \v)]=L,\notag\\
[\a]=&[\cddot^a(\tau)]\Big[\frac{\partial}{\partial x^a}\Big]=\frac{1}{L^2}, \qquad[\adual]=[g(-, \a)]=1
\end{align}
and
\begin{align}
[g(\v, \v)]=&1, \quad\quad [g(\x, \v)]=L, \quad\quad [g(\x, \a)]=1
\end{align}

\chapter{Differential Geometry}
\label{diffgeom}

In this appendix we present a brief introduction to the geometric constructs and notations encountered in the thesis. This introduction is by no means complete, and for a deeper understanding of the subject we refer the reader to the vast collection of introductory books on the subject, of which \cite{BennTucker, Frankel03} are good examples.

\section{Tensor Fields and differential forms}
\subsection*{Vector fields}
\begin{definition}
Let $\man$ be an arbitrary differential manifold of dimension $\dimM$, and $\point$ an arbitrary point in $\man$. The \begin{bf} space of smooth  real valued $c^{\infty}$ functions \end{bf} over $\man$ is denoted $\fsmooth(\man)$,
\begin{align}
\funcf \in\fsmooth(\man)\quad \text{implies} \quad \funcf :  \man &\rightarrow \real, \qquad \point \mapsto \funcf(\point)
\end{align}
\end{definition}
\begin{definition} A (contravariant)\begin{bf}  vector at a point \end{bf} $\vvec|_\point$ is a map from $\fsmooth(\man)$ to $\real$,
\begin{align}
\vvec|_{\point} : \fsmooth(\man)\rightarrow \real, \qquad \funcf \mapsto \vvec|_{\point}(\funcf),
\end{align}
which satisfies
 \begin{align}
 \vvec|_{\point} (\funcf + \funcg) &= \vvec|_{\point}(\funcf) + \vvec|_{\point}(\funcg)\notag,\\
 \vvec|_{\point}(\lambda \funcf)&= \lambda \vvec|_{\point}(\funcf)\notag,\\
\text{and}\qquad  \vvec|_{\point} (\funcf \funcg) &= \vvec|_{\point}(\funcf)\funcg + \funcf \vvec|_{\point}(\funcg),
\label{vp_def}
 \end{align}
 where $f, g \in \fsmooth(\man)$ and $\lambda$ is an arbitrary scalar.
 \end{definition}

\begin{definition}
For every $\point\in \man$ the \begin{bf}tangent space\end{bf} to $\man$ at point $\point$, written $\tan_\point \man$, is the $\dimM$ dimensional vector space whose elements are the vectors at $\point$, i.e. $\vvec|_{\point} \in \tan_\point \man$. The \begin{bf} tangent bundle\end{bf} is the $2\dimM$ dimensional manifold given by the set theoretic union of the tangent spaces $\tan_\point \man$ for all $\point\in \man$,
\begin{align}
\tan \man=\bigcup_{\point\in \man} \tan_\point \man.
\end{align}
We call $\man$ the base space.
\end{definition}
\begin{definition}
Given the projection map
\begin{align}
\pitan:\tan \man\rightarrow \man,\qquad \vvec|_{\point}\mapsto \point,
\end{align}
a \begin{bf} section  of the tangent bundle\end{bf} is a continuous map
\begin{align}
&\vvec: \man \rightarrow \tan \man,\qquad \point \mapsto \vvec|_{\point}\notag\\
\text{such that}\qquad
&\pi(\vvec|_{x})=\point \qquad \text{for all}\qquad \point\in \man.
\label{vecfield_def}
\end{align}
The map $\vvec$ identifies a vector at a point for each point in the base space\footnote{We have used the notation $\vvec|_{x}$ to denote a vector at a point as well as a vector field evaluated at a point}, therefore when acting on a smooth function it is the map
\begin{align}
\vvec: \fsmooth(\man) \rightarrow \real, \qquad \funcf\mapsto \vvec(\funcf)
\end{align}
 with the properties (\ref{vp_def}). A smooth section of the tangent bundle is a \begin{bf} vector field \end{bf}. The space of  vector fields over $\man$ is written $\Gamma \tan \man$.
 \end{definition}
 \begin{definition}
Given a local coordinate basis $\co^a$ on $\man$ there exists an induced local basis $\tfrac{\partial}{\partial \co^a}$ on $\tan \man$. In terms of this basis a vector field $\vvec\in\tan\man$ is given by
\begin{align}
\vvec=\vvec^a\frac{\partial}{\partial \co^a}
\end{align}
where the Einstein summation convention is used for $a=1..m$ and $\vvec^a$ are smooth functions on $\man$. In terms of a different local coordinate basis $z^a$
\begin{align}
\vvec=\vvec^a \frac{\partial}{\partial \co^a}=\vvec^a \frac{\partial z^b}{\partial \co^a}\frac{\partial}{\partial z^b}.
\end{align}
where $\vvec^a$ are the components of $\vvec$ in the $y^a$ coordinate basis and $\vvec^a \partial z^b / \partial y^a$ are the components of $\vvec$ in the $z^b$ coordinate basis.
\end{definition}
\subsection*{Differential $1$-forms}
\begin{definition}
The space dual to the tangent space $\tan_\point \man$ is called the cotangent space at $\point$ and is denoted by $\tan_\point^\ast \man$.
Elements $\covec|_{\point}\in \tan_\point^\ast \man $ are covariant vectors or covectors at $\point$ and   satisfy
\begin{align}
\covec|_{\point}:  \tan_\point \man \rightarrow \real, \qquad  \vvec\mapsto \covec|_{\point}(\vvec),
\end{align}
with the properties
 \begin{align}
\covec|_{\point}(\vvec|_{\point} + \wvec|_{\point}) &= \covec|_{\point}(\vvec|_{\point}) +\covec|_{\point}(\wvec|_{\point})\notag,\\
\covec|_{\point} (\funcf \vvec|_{\point} )&= \funcf \covec|_{\point} (\vvec|_{\point}) .
\label{flin_def}
 \end{align}
 \end{definition}
 \begin{definition}
The \begin{bf} cotangent bundle\end{bf} is the $2 \dimM$ dimensional manifold $\tan^\ast \man$ defined as the set theoretic union of the cotangent spaces $\tan_\point^\ast \man$ for all $\point\in \man$,
\begin{align}
\tan^\ast \man=\bigcup_{\point\in \man} \tan_\point^\ast \man.
\end{align}
\end{definition}
\begin{definition}
Given the projection map
\begin{align}
\picot:\tan^{\ast}\man\rightarrow \man,\qquad \covec|_{\point} \mapsto \point,
\end{align}
a \begin{bf} section  of the cotangent bundle\end{bf} is a continuous map
\begin{align}
&\covec: \man \rightarrow \tan^\ast \man,\qquad \point \mapsto \covec|_{\point}\notag\\
\text{such that}\qquad
&\picot(\covec|_{\point})=\point \qquad \text{for all}\qquad \point\in \man.
\end{align}
The map $\covec$ identifies a covector at a point for each point in the base space. When acting on a vector field it is the map \begin{align}
\covec: \Gamma \tan \man \rightarrow \fsmooth(\man), \qquad \vvec\mapsto \covec(\vvec)
\end{align}
 with the properties
  \begin{align}
\covec(\vvec + \wvec) &= \covec(\vvec) +\covec(\wvec)\notag,\\
\covec (\funcf \vvec )&= \funcf \covec (\vvec) .
\label{flin_def_field}
 \end{align}
 A smooth section of the tangent bundle is called a \begin{bf} covector field \end{bf}or \begin{bf} (differential) $1$-form \end{bf}. The space of $1$-forms is written $\Gamma \tan^\ast \man$.
\end{definition}
\begin{lemma}
 The duality of $\tan_\point \man$ and $\tan_\point^\ast \man$ demands
 \begin{align}
 \covec|_{\point} (\vvec|_{\point})=\vvec|_{\point}(\covec|_{\point})
 \label{eq_def}
 \end{align}
 where $\vvec|_{\point}(\covec|_{\point})$ satisfies the reversal of (\ref{flin_def}) with respect to vectors and covectors.
\end{lemma}
 \begin{definition}
Given a local coordinate basis $\co^a$ on $\man$ there exists an induced local basis $d\co^a$ on $\tan^\ast \man$.
In terms of this basis a differential 1-form $\covec\in\tan^\ast\man$ is given by
\begin{align}
\covec=\covec_a d\co^a
\end{align}
where $\covec_a$ are smooth functions on $\man$. In terms of a different local coordinate basis $z^a$
\begin{align}
\covec=\covec_a d\co^a=\covec_a \frac{\partial \co^a}{\partial z^b} d z^b.
\label{cov_trans}
\end{align}
where $\covec^a$ are the components of $\covec$ in the $\co^a$ coordinate basis and $\covec_a \frac{\partial \co^a}{\partial z^b} $ are the components of $\covec$ in the $z^a$ coordinate basis. 
\end{definition}

\subsection*{Tensor fields}

\begin{definition}
The degree of an arbitrary tensor will be represented as an ordered list
$s$ of $0$ or more entries. Each entry is either the symbol $\mathds{F}$ (for 1-form) or $\mathds{V}$
(for vector) e.g. $s = [\mathds{V}, \mathds{F}, \mathds{F},  \mathds{V}]$. The space of tensors of degree $s$ over $\man$ is
denoted $\bigotimes^s \man$ with sections $\Gamma \bigotimes^s \man$.  Let the tangent space and the cotangent space be denoted
\begin{align}
\tan \man= \totimes^{[\mathds{V}]} \man \qquad \textup{and} \qquad \tan^{\ast} \man= \totimes^{[\mathds{F}]} \man
\end{align}
then arbitrary degree tensors are constructed using the tensor product,
\begin{align}
\otimes : \totimes^s \man \times \totimes^t \man \rightarrow \totimes^{[s, t]} \man, \qquad (\tenst, \tenss) \mapsto \tenst \otimes \tenss
\end{align}
where $s$ and $t$ are ordered lists and $[s, t]$ is simply the concatenation of the two lists. For example given vector field $\vvec\in \Gamma \tan \man$ and $1$-forms $\aform, \bform \in \Gamma \tan^{\ast} \man$ we may define a degree $[\mathds{V}, \mathds{F}, \mathds{F}]$ tensor field by
\begin{align}
\vvec \otimes \aform \otimes \bform \in \Gamma \totimes^{[\mathds{V}, \mathds{F}, \mathds{F}]} \man.
\end{align}
The tensor product satisfies
\begin{align}
\tenst \otimes (\tenss \otimes \tensr) =& (\tenst \otimes \tenss) \otimes \tensr\notag\\
(\tenst_1 + \tenst_2) \otimes \tenss =& \tenst_1 \otimes \tenss + \tenst_2 \otimes \tenss
\end{align}
and
\begin{align}
\funcf (\tenst \otimes \tenss)=& (\funcf \tenst) \otimes \tenss= \tenst \otimes (\funcf\tenss)
\end{align}
for $\tenst, \tenst_1, \tenst_2 \in  \totimes^s \man$, $\tenss\in \totimes^t \man$, $\tensr \in \totimes^u \man$ and $\funcf \in \fsmooth(\man)$.
The dual space of $\totimes^{s}\man$ is denoted $\totimes^{\bar{s}}\man$, where $\bar{s}$ is the list obtained by interchanging the symbols $\mathds{F}$ and $\mathds{V}$ in $s$. The total contraction of elements in $\totimes^{\bar{s}}\man$ with elements in $\totimes^{s}\man$ is written
\begin{align}
\totimes^{\bar{s}}\man \times \totimes^{s}\man \rightarrow \fsmooth(\man), \qquad  (\tenst, \tensr) \mapsto \tenst:\tensr
\end{align}
where $\tenst\in \totimes^{\bar{s}}\man $ and $ \tensr \in \totimes^{s}\man $. It is defined inductively via
\begin{align}
\vvec:\aform=\aform:\vvec=\aform(\vvec) \qquad \textup{where} \qquad \aform \in  \totimes^{[\mathds{F}]} \man\qquad \textup{and}\qquad \vvec \in \totimes^{[\mathds{V}]},
\end{align}
and extended to arbitrary tensors by
\begin{align}
(\tenss \otimes \tenst):(\tensr \otimes \tensu)=(\tenss : \tensr)(\tenst:\tensu)
\end{align}
\end{definition}

\subsection*{The metric}
\begin{definition}
Of special importance is the symmetric, non-degenerate degree $[ \mathds{F}, \mathds{F}]$ tensor field $\g\in \bigotimes^{[ \mathds{F}, \mathds{F}]} \man$ with the properties
\begin{align}
&\g(\vvec, \wvec)=\g(\wvec, \vvec),\notag\\
\text{and}\quad &\g(\vvec, \wvec)=0\quad \text{for all}\quad \vvec \neq 0 \Rightarrow \wvec=0.
\end{align}
This tensor is called the metric. It provides an isomorphism between the covariant and contravariant vector fields.
The metric dual $\vvecd$ of vector field $\vvec$ is the differential $1$-form given by
\begin{align}
\vvecd=\g(\vvec, -),
\label{dual_def}
\end{align}
where $(-)$ denotes an empty argument. There exists a symmetric, non-degenerate degree $[ \mathds{V}, \mathds{V}]$ tensor field $\gdual \in \bigotimes^{[ \mathds{V}, \mathds{V}]} \man$ which satisfies
\begin{align}
&\gdual(\vvecd, \wvecd)=\g(\vvec, \wvec)\quad\text{and}\quad \vvec=\gdual(\vvecd, -).
\label{gdual_def}
\end{align}
Given the local coordinate basis $\co^a$ on $\man$ the metric $\g$ is given by
\begin{align}
\g=\g_{ab}\dxa^a \otimes \dxa^b
\end{align}
where the functions $\g_{ab}$ are determined by
\begin{align}
\g_{ab}=\g\Big(\frac{\partial}{\partial \co^a},\frac{\partial}{\partial \co^b} \Big)
\end{align}
\end{definition}
\subsection*{Differential $p$-forms}
\begin{definition}
An important subspace of $ \bigotimes \man$ is the space of totally antisymmetric degree $[\mathds{F}_1,...,\mathds{F}_p]$ tensors denoted $\Lambda^\dega \man$. The space of smooth sections of $\Lambda^\dega \man$ is denoted $\Gamma \Lambda^\dega \man$  and elements $\Psi\in \Gamma \Lambda^\dega \man$ are called (differential) $\dega$-forms. For example for $\oneform, \wonform \in \Gamma \bigotimes^{[\mathds{F}]} \man$ the totally antisymmetric part of the degree $[\mathds{F}, \mathds{F}]$ tensor field $\oneform \otimes \wonform$ is the difference
\begin{align}
\frac{1}{2}(\oneform\otimes\wonform-\wonform \otimes \oneform)=\twoform
\label{form_def}
\end{align}
since reversing the positions of $\oneform$ and $\wonform$ yields
\begin{align}
\frac{1}{2}(\wonform\otimes\oneform-\oneform \otimes \wonform)=-\twoform,
\label{form_def}
\end{align}
  threfore $\twoform \in\Gamma \Lambda^2 \man$ is a differential $2$-form. The space of differential $0$-forms $\Gamma \Lambda^0 \man$ is defined as the space of smooth functions over $\man$, i.e. $\fsmooth (\man) = \Gamma \Lambda^0 \man$ and the space of differential $1$-forms $\Gamma \tan^\ast \man=\Gamma \bigotimes^{[\mathds{F}]} \man$ is now also written as as $\Gamma \Lambda^1 \man$.

Higher degree forms are obtained using the \begin{bf} exterior or wedge product\end{bf}. The wedge product of $\aform\in \Gamma \Lambda^\dega \man$  and $ \bform \in \Gamma \Lambda^\degb \man$ is the map
\begin{align}
\wedge : \Gamma \Lambda^\dega \man \times \Gamma \Lambda^\degb \man \quad \rightarrow \quad \Gamma \Lambda^{\dega+\degb} \man, \qquad \aform, \bform \quad \mapsto \quad \aform\wedge\bform,
\end{align}
where the $(\dega+\degb)$-form $\aform\wedge \bform$ is the totally antisymmetric part of the tensor field $\aform\otimes \bform$.
The wedge product satisfies
\begin{align}
\aform \wedge (\bform \wedge \gamma) =& (\aform \wedge \bform) \wedge \gamma ,\notag\\
(\aform_1 + \aform_2) \wedge \bform =& \aform_1 \wedge \bform + \aform_2 \wedge \bform,
\end{align}
and
\begin{align}
\funcf (\aform \wedge \bform)=& (\funcf \aform) \wedge \bform= \aform \wedge (\funcf\bform),
\end{align}
and
\begin{align}
\aform \wedge \bform=(-1)^{\dega \degb} \bform\wedge \aform
\end{align}
for $\aform,\aform_1, \aform_2  \in \Lambda^{\dega}\man$, $\bform \in \Lambda^{\degb}\man$, $\gamma\in \Lambda^{r}\man$  and $\funcf \in \fsmooth(\man)$. It follows any arbitrary $\dega$-form can be reduced to a linear superposition of wedge products of $\dega$ differential $1$-forms. Given local coordinate basis $y^a$ on $\man$,
\begin{align}
\aform\in\Gamma\Lambda^\dega \man, \qquad \aform=\aform_{a_1 a_2 ..a_\dega}d\co^{a_1}\wedge d\co^{a_2}\wedge..\wedge d\co^{a_\dega}.
\label{form_comp}
\end{align}
\end{definition}

\section{Differential operators}
\begin{definition}For the following definitions it is useful to define the map
\begin{align}
&\eta: \Lambda^\dega \man \rightarrow \Lambda^\dega \man, \qquad \aform \mapsto \eta(\aform),\notag\\
\text{where} \quad &\eta(\aform) = (-1)^{\dega}\aform \quad \text{for all}\quad \aform\in \Lambda^\dega \man.
\end{align}
\end{definition}
\subsection*{Exterior derivative}
\begin{definition}
The \begin{bf}exterior derivative\end{bf} of a differential form is the map
\begin{align}
d: \Gamma\Lambda^\dega \man \rightarrow \Gamma \Lambda^{\dega+1}\man, \qquad \aform \mapsto d\aform,
\end{align}
where
\begin{align}
d\funcf (\vvec)&=\vvec(\funcf),\notag\\
d(\aform \wedge \bform)&=d\aform\wedge\bform+\eta (\aform) \wedge d\bform,\notag\\
d(d\aform)&=0
\end{align}
for all $\funcf \in\Gamma\Lambda^0 \man$, $\vvec \in \Gamma \tan \man$, $\aform \in \Gamma \Lambda^\dega \man$ and $\bform \in \Gamma \Lambda^\degb \man$
\end{definition}
\begin{lemma}
For $1$-form  $\oneform \in \Gamma\Lambda^1 \man$ and vector fields $\vvec, \wvec \in \Gamma \tan \man$ the following relation holds
\begin{align}
d\oneform (\vvec, \wvec) =& \vvec(\oneform( \wvec))-\wvec (\oneform( \vvec))-\oneform([\vvec, \wvec]),
\label{twoform}
\end{align}
where $[, ]$ denotes the \emph{Lie Bracket}.
\end{lemma}
\begin{proofs}
It is sufficient to prove for the $1-form$ $\oneform=f dg$, where $\funcf, \g \in \fsmooth(\man)$. In this case $d\oneform=df\wedge dg$ and thus
\begin{align}
d\oneform (\vvec, \wvec)=& df \wedge dg (\vvec, \wvec),\notag\\
=&\frac{1}{2}\big(df ( \vvec)dg(\wvec) - dg ( \vvec)df ( \wvec)\big),\notag\\
=&\frac{1}{2}\big(\vvec(f)\wvec(g)-\vvec(g)\wvec(f)\big).
\end{align}
Here $d\oneform (\vvec, \wvec)$ is the action of the $2-$form $d\oneform$ on the ordered pair of vector fields $(\vvec, \wvec)$.
Also
\begin{align}
\vvec(\oneform (\wvec))-\wvec(\oneform(\vvec))&-\oneform( [\vvec, \wvec])\notag\\
=& \vvec(fdg(\wvec))-\wvec(fdg(\vvec))-fdg([\vvec, \wvec])\notag\\
=&\vvec(f ( \wvec(g)))-\wvec(f \vvec(g))-f[\vvec, \wvec](g)\notag\\
=&\vvec(f)\wvec(g)+f\vvec(\wvec(g))-\wvec(f)\vvec(g) \notag\\
&-f\wvec(\vvec(g))-f[\vvec, \wvec](g)\notag\\
=&\vvec(f)\wvec(g)-\wvec(f)\vvec(g)\notag\\
=&2 d\oneform (\vvec, \wvec)
\end{align}

\end{proofs}
\subsection*{Interior contraction}
\begin{definition} The \begin{bf}interior contraction\end{bf} of a $\dega$-form with respect to a vector field is defined by
\begin{align}
i: \Gamma \tan \man \times \Gamma\Lambda^\dega \man \rightarrow \Gamma \Lambda^{\dega-1}\man, \qquad \vvec, \aform \mapsto i_\vvec \aform
\end{align}
where
\begin{align}
i_\vvec \oneform&=\oneform(\vvec)\notag\\
i_\vvec (\aform\wedge\bform)&=i_\vvec \aform \wedge\bform +\eta(\aform)\wedge i_\vvec \bform\notag\\
i_\vvec i_\vvec \aform&=0
\label{internal_defs}
\end{align}
for all $\vvec \in \Gamma \tan \man$, $\oneform \in\Gamma\Lambda^1 \man$, $\aform\in \Gamma \Lambda^\dega \man$ and $\bform \in \Gamma \Lambda^\degb \man$.
\end{definition}

\subsection*{Hodge Dual}
\begin{definition}
\label{def_starone}
Introduce a $\g$-orthonormal frame $e^a$ such that
\begin{align}
\g=\g_{ab}d\co^a\otimes d\co^b=\eta_{ab}e^a\otimes e^b
\end{align}
where $\eta_{ab}=\pm 1$ for $a=b$ and $\eta_{ab}=0$ for $a\neq b$. Then
\begin{align}
\star 1\in \Lambda^\dimM \man, \qquad \star1 = e^0\wedge e^1\wedge..\wedge e^{\dimM-1},
\end{align}
is the  \begin{bf} volume form \end{bf} on $\man$.
\end{definition}
\begin{definition}
\label{hodge_def}
The \begin{bf} Hodge dual \end{bf} is the map
\begin{align}
\star : \Lambda^\dega \man \rightarrow \Lambda^{\dimM-\dega}\man, \qquad \aform \mapsto \star \aform
\end{align}
where
\begin{align}
\star (1) &= \star 1\notag\\
\star \oneform &= i_{\widetilde{\oneform}}\star 1,\notag\\
\star(\aform\wedge\oneform)&=i_{\widetilde{\oneform}}\star\aform,\notag\\
\star (\funcf \aform)&=\funcf\star \aform.
\end{align}
for all $\oneform \in\Gamma\Lambda^1 \man$, $\aform\in \Gamma \Lambda^\dega \man$ and $\bform \in \Gamma \Lambda^\degb \man$.
Properties (\ref{form_comp}) and (\ref{hodge_def}) define the action of the Hodge dual on any arbitrary $\dega$-form.
\end{definition}
\begin{lemma}
\label{one_one}
Given two vector fields $\vvec, \wvec \in \Gamma \tan \man$ the following relation is true
\begin{align}
\vvecd \wedge \star \wvecd = \g(\vvec, \wvec)\star 1
\end{align}
\end{lemma}
\begin{proofs}
\begin{align}
\vvecd \wedge \star \wvecd=& \vvecd \wedge i_{\wvec} \star 1\notag\\
=& \vvec_a \wvec^b dz^a \wedge i_{\frac{\partial}{\partial z^b}} \star 1\label{starshow}
\end{align}
Using the alternating Leibniz rule (\ref{internal_defs}) yields
\begin{align}
0=i_{\partial z^b}(dz^a\wedge \star 1)&=i_{\partial z^b}dz^a\wedge \star 1-d z^a \wedge i_{\partial z^b}\star 1\notag\\
\textup{and thus}\qquad dz^a \wedge i_{\partial z^b}\star 1 &= \delta^a_b \star 1.
\label{star_two}
\end{align}
Substituting (\ref{star_two}) into (\ref{starshow}) yields
\begin{align}
\vvecd \wedge \star \wvecd=& \vvec_b \wvec^b \star 1\notag\\
=&\g(\vvec, \wvec) \star 1
\end{align}
\end{proofs}

\begin{lemma}
\label{hodge_lemma_def}
Given $\oneform \in \Gamma \Lambda^1 \man$ and $\alpha \in \Gamma \Lambda^p \man$ the following is true
\begin{align}
\oneform \wedge \star \alpha =& -\star(i_{\widetilde{\oneform}}\alpha^{\eta})
\end{align}
where $\alpha^\eta=\eta(\alpha)$.
\end{lemma}
\begin{proofs}
By lemma \ref{one_one} it is clearly true for $\deg(\alpha)=1$. Assume true for $\deg(\alpha)=p$, then for $\wonform \in \Gamma \Lambda^1 \man$
\begin{align}
\star i_{\widetilde{\oneform}}(\alpha\wedge \wonform) =&\star (i_{\widetilde{\oneform}}\alpha\wedge \wonform) + \star (\alpha^{\eta}\wedge  i_{\widetilde{\oneform}} \wonform)\notag\\
=&\star (i_{\widetilde{\oneform}}\alpha\wedge \wonform) + \g(\widetilde{\wonform},  \widetilde{\oneform})\star \alpha^{\eta} \notag\\
=&\star (i_{\widetilde{\oneform}}\alpha\wedge \wonform) + i_{\widetilde{\wonform}}(\oneform \wedge \star \alpha^\eta) + \oneform \wedge i_{\widetilde{\wonform}}\star \alpha^\eta\notag\\
\end{align}
Evaluating the first two terms yields
\begin{align}
\star (i_{\widetilde{\oneform}}\alpha\wedge \wonform) + i_{\widetilde{\wonform}}(\oneform \wedge \star \alpha^\eta)=&  \star (\wonform \wedge i_{\widetilde{\oneform}}\alpha^\eta)- i_{\widetilde{\wonform}}(\star(i_{\widetilde{\oneform}}\alpha))\notag\\
=&\star (\wonform \wedge i_{\widetilde{\oneform}}\alpha^\eta)- \star(i_{\widetilde{\oneform}}\alpha\wedge \wonform)\notag\\
=&\star (\wonform \wedge i_{\widetilde{\oneform}}\alpha^\eta)- \star(\wonform \wedge i_{\widetilde{\oneform}}\alpha^\eta)\notag\\
=&0
\end{align}
Thus
\begin{align}
-\star i_{\widetilde{\oneform}}(\alpha\wedge \wonform) =&-\oneform \wedge i_{\widetilde{\wonform}}\star \alpha^\eta= \oneform \wedge \star( \alpha\wedge \wonform)^\eta\notag\\
\end{align}
\end{proofs}
\begin{lemma}
\label{one_two}
Given 1-forms $\oneform, \wonform, \aform \in \Gamma \Lambda^1 \man$ the following is true
\begin{align}
\oneform \wedge \star (\wonform\wedge \aform) =& \g(\oneformd, \widetilde{\aform})\star \wonform - \g(\oneformd, \widetilde{\wonform})\star \aform
\end{align}
\end{lemma}
\begin{proofs}
By lemma \ref{hodge_lemma_def} we have
\begin{align}
\oneform \wedge \star (\wonform \wedge \aform) =& -\star(i_{\oneformd} (\wonform\wedge \aform)\notag\\
=& -\star(i_{\oneformd}\wonform \wedge \aform - \wonform \wedge i_{\oneformd}\aform)\notag\\
=&  \star(\g(\oneformd, \widetilde{\aform})\wonform)- \star(\g(\oneformd, \wonformd)\aform)\notag\\
=& \g(\oneformd, \widetilde{\aform})\star \wonform - \g(\oneformd, \widetilde{\wonform})\star \aform\notag
\end{align}
\end{proofs}
\begin{lemma}
\label{fiveoneforms}
Given one forms $\alpha, \beta, \gamma, \nu, \omega \in \Lambda^1 \man$ the following is true
\begin{align}
i_{\widetilde{\omega}}\star(\alpha \wedge \beta)\wedge \gamma \wedge \nu=& \big(\beta (\widetilde{\gamma})\omega (\widetilde{\nu})-\omega(\widetilde{\gamma})\beta(\widetilde{\nu})\big)\star \alpha\notag\\
&+ \big(\omega (\widetilde{\gamma})\alpha (\widetilde{\nu})-\alpha(\widetilde{\gamma})\omega(\widetilde{\nu})\big)\star \beta\notag\\
&+\big(\alpha(\widetilde{\gamma})\beta (\widetilde{\nu})-\beta(\widetilde{\gamma})\alpha(\widetilde{\nu})\big)\star \omega
\end{align}
\begin{proofs}
\begin{align}
i_{\widetilde{\omega}}\star(\alpha \wedge \beta)\wedge \gamma \wedge \nu =& \star(\alpha\wedge \beta \wedge \omega) \wedge \gamma \wedge \nu\notag\\
=& -\gamma \wedge \star(\alpha\wedge \beta \wedge \omega)\wedge \nu
\label{firstone}
\end{align}
Using (\ref{hodge_lemma_def}) yields
\begin{align}
  \gamma \wedge \star(\alpha\wedge \beta \wedge \omega)=&-\star\big(i_{\widetilde{\gamma}}(-\alpha\wedge \beta \wedge \omega)\big)\notag\\
  =&\star \big(i_{\widetilde{\gamma}}(\alpha\wedge \beta \wedge \omega)\big)\notag\\
  =& \star \big(\alpha(\widetilde{\gamma})\beta \wedge \omega - \beta(\widetilde{\gamma})\alpha \wedge \omega + \omega (\widetilde{\gamma}) \alpha \wedge \beta\big)\label{secondone}
\end{align}
Thus substituting (\ref{secondone}) into (\ref{firstone}) yields
\begin{align}
i_{\widetilde{\omega}}\star(\alpha \wedge \beta)\wedge \gamma \wedge \nu=-\nu \wedge \star \big(\alpha(\widetilde{\gamma})\beta \wedge \omega - \beta(\widetilde{\gamma})\alpha \wedge \omega + \omega (\widetilde{\gamma}) \alpha \wedge \beta\big)
\end{align}
Now using (\ref{hodge_lemma_def}) again yields
\begin{align}
i_{\widetilde{\omega}}\star(\alpha \wedge \beta)\wedge \gamma \wedge \nu=& \star\Big(i_{\widetilde{\omega}}\big(\alpha(\widetilde{\gamma})\beta \wedge \omega \big)- i_{\widetilde{\omega}}\big(\beta(\widetilde{\gamma})\alpha \wedge \omega \big)+i_{\widetilde{\omega}}\big(\omega(\widetilde{\gamma})\alpha \wedge \beta \big)\Big)\notag\\
=& \big(\beta (\widetilde{\gamma})\omega (\widetilde{\nu})-\omega(\widetilde{\gamma})\beta(\widetilde{\nu})\big)\star \alpha\notag\\
&+ \big(\omega (\widetilde{\gamma})\alpha (\widetilde{\nu})-\alpha(\widetilde{\gamma})\omega(\widetilde{\nu})\big)\star \beta\notag\\
&+\big(\alpha(\widetilde{\gamma})\beta (\widetilde{\nu})-\beta(\widetilde{\gamma})\alpha(\widetilde{\nu})\big)\star \omega\notag
\end{align}
 \end{proofs}
\end{lemma}
\begin{lemma}
\label{hodgeab}
For two forms $\aform, \bform\in \Gamma \Lambda^\dega \man$ of the same degree the following is true
\begin{align}
\aform\wedge\star\bform=\bform\wedge\star\aform.
\end{align}
\end{lemma}
\begin{proofs}
By  (\ref{one_one}) it is clearly true for $\deg(\alpha)=\deg(\beta)=1$. Assume true for $\deg(\alpha)=\deg(\beta)=p$, then for for $\alpha\in \Gamma \Lambda^{p+1} \man$, $\beta \in \Gamma \Lambda^p \man$ and $\oneform \in \Gamma \Lambda^1 \man$ we have
\begin{align}
\aform\wedge\star(\bform\wedge \oneform)=& \aform \wedge i_{\widetilde{\oneform}}\star \bform\notag\\
=&i_{\widetilde{\oneform}}(\aform^\eta \wedge \star \bform)- i_{\widetilde{\oneform}}\aform^\eta \wedge \star \bform\notag\\
=&-\bform\wedge \star i_{\widetilde{\oneform}}\aform^\eta \notag\\
=& \bform \wedge \oneform \wedge \star \aform
\end{align}
\end{proofs}

\subsection*{Lie Derivative}
\begin{definition}
The Lie derivative is the the map
\begin{align}
\l: \Gamma \tan \man \times \Gamma \totimes^{[s]} \man \rightarrow \Gamma \totimes^{[s]} \man, \qquad \vvec, \tenst \mapsto \l_\vvec \tenst,
\end{align}
which is additive linear in both arguments
\begin{align}
\l_\vvec (\tenst+\tenss)= \l_\vvec \tenst + \l_{\vvec} \tenss,\notag\\
\l_{\vvec+\wvec} \tenst= \l_\vvec \tenst + \l_{\wvec} \tenst,
\end{align}
has the properties
\begin{align}
& \l_\vvec f = \vvec(f),\notag\\
& \l_\vvec \wvec = [\vvec, \wvec],
\end{align}
and obeys the Leibniz rule for tensor products, wedge products and contractions
\begin{align}
& \l_\vvec (\tenst \otimes \tenss)=\l_\vvec\tenst\otimes\tenss+\tenst\otimes\l_\vvec\tenss,\\
& \l_\vvec (\aform \wedge \bform)=\l_\vvec\aform\wedge\bform+\aform\wedge\l_\vvec\bform,\\
& \l_\vvec(\aform(\wvec))=\l_\vvec\aform(\wvec)+\aform(\l_\vvec\wvec).\label{lie_contract}
\end{align}
\end{definition}
\begin{lemma}
Cartan's formula
\begin{align}
\l_\vvec=di_\vvec+i_\vvec d
\label{cartan}
\end{align}
\end{lemma}
\begin{proofs}
Trivial for $0$-forms. First prove for $1$-form $\oneform \in \Lambda^1\man$. From (\ref{lie_contract}) and (\ref{twoform})
\begin{align}
\l_\vvec\oneform(\wvec)=& \l_\vvec(\oneform ( \wvec)) - \oneform (\l_\vvec \wvec)\notag\\
=&\vvec(\oneform(\wvec))-\oneform ([\vvec, \wvec])\notag\\
=&2d\oneform (\vvec, \wvec)+\wvec(\oneform( \vvec))\notag\\
=&i_\vvec d\oneform ( \wvec) + d(\oneform (\vvec))(\wvec)\notag\\
=&i_\vvec d\oneform ( \wvec) + d i_\vvec \oneform(\wvec)\notag\\
=&(i_\vvec d\oneform+di_\vvec \oneform)(\wvec)
\end{align}
Now assume true for $p$-form $\alpha \in \Gamma \Lambda^p \man$, show true for $(p+1)$-form  $\alpha \wedge \oneform$.
\begin{align}
(i_\vvec d+ d i_\vvec) (\alpha \wedge \oneform)=& i_\vvec(d\alpha\wedge \oneform+(-1)^p\alpha \wedge d \oneform) + d(i_\vvec \alpha \wedge \oneform + (-1)^p \alpha \wedge i_\vvec \oneform)\notag\\
=& i_\vvec (d\alpha \wedge \oneform) +(-1)^p i_\vvec (\alpha \wedge d\oneform) +d(i_\vvec\alpha \wedge \oneform) +(-1)^p d(\alpha \wedge i_\vvec \oneform)\notag\\
=&i_\vvec d \alpha\wedge\oneform +((-1)^{p+1}+(-1)^p)d\alpha\wedge i_\vvec \oneform +di_\vvec \alpha \wedge \oneform \notag\\
&+((-1)^{p-1}+(-1)^p)i_\vvec \alpha \wedge d\oneform + (-1)^{2p} (\alpha \wedge i_\vvec d \oneform + \alpha \wedge d i_\vvec \oneform)\notag\\
=&(i_\vvec d+ d i_\vvec) \alpha \wedge \oneform + \alpha \wedge (i_\vvec d+ d i_\vvec)\oneform\notag\\
=&\l_\vvec \alpha \wedge \oneform + \alpha \wedge \l_\vvec \oneform\notag\\
=& \l_\vvec (\alpha \wedge \oneform)\notag
\end{align}
Thus by induction true for all $p$-forms.
\end{proofs}

\begin{lemma}
The Lie derivative commutes with the exterior derivative
\begin{align}
\l_\vvec d = d \l_\vvec
\end{align}
\end{lemma}
\begin{proofs}
follows trivially from \ref{cartan}
\end{proofs}

\subsection*{Levi-Civita Connection}
\label{connection}
An affine connection is a map
\begin{align}
\nabla: \Gamma \tan \man \times \Gamma \totimes^{[s]} \man\quad \rightarrow \quad \Gamma \totimes^{[s]} \man, \qquad (\vvec,\tenst) \mapsto \nabla_\vvec \tenst
\end{align}
which is additive linear in both arguments,
\begin{align}
\nabla_\vvec (\tenst+\tenss)= \nabla_\vvec \tenst + \nabla_{\vvec} \tenss,\notag\\
\nabla_{\vvec+\wvec} \tenst= \nabla_\vvec \tenst + \nabla_{\wvec} \tenst,
\end{align}
and satisfies
\begin{align}
& \nabla_\vvec f = \vvec(f),\notag\\
& \nabla_{f\vvec} \tenst = f \nabla_\vvec \tenst,\notag\\
& \nabla_\vvec (\tenst \otimes \tenss)= \nabla_\vvec\tenst \otimes \tenss+ \tenst \otimes \nabla_\vvec\tenss,\notag\\
& \nabla_\vvec (\aform \wedge \bform)=\nabla_\vvec\aform\wedge\bform+\aform\wedge\nabla_\vvec\bform.
\end{align}
The Levi-Civita connection on $\man$ is the unique torsion free metric compatible affine connection.
\section{Pushforwards, pullbacks and curves}
\section*{Pushforward map}
\begin{definition}
Given differentiable manifolds $\man$ and $\manN$ and the smooth map $\mapMN : \man \longrightarrow \manN , \quad \point \longmapsto \mapMN (\point)$,  the \begin{bf}pushforward\end{bf}  of a vector at a point  $ \vvec\atp \in \tan_\point \man$ with respect to $\mapMN$ is the map
\begin{align}
\mapMN_{\ast} : \tan_x\man \longrightarrow \tan_{\phi (\point)}\manN,  \qquad  \vvec\atp \longmapsto \mapMN_{\ast} \vvec\atp,\label{pushstart}
\end{align}
where
\begin{align}
\mapMN_{\ast }\vvec\atp(\funcf) &= \vvec\atp(\funcf \circ \mapMN),\notag\\
\mapMN_{\ast}(\vvec\atp + \wvec \atp) &= \mapMN_{\ast}\vvec\atp + \phi_{\ast} \wvec\atp\nonumber,\\
\text{and}\quad \mapMN_{\ast}(\vvec\atp \wvec\atp) &= \mapMN_{\ast}\vvec\atp\mapMN_{\ast}\wvec\atp
\label{pushforward}
\end{align}
 for $\vvec\atp, \wvec\atp \in \tan_x \man$ and $\funcf\in \Gamma \Lambda^{0} \manN$.
 The pushforward of a vector at a point is naturally extended using (\ref{vecfield_def}) to obtain the pushforward of a vector field
 \begin{align}
\mapMN_{\ast} : \Gamma \tan \man \longrightarrow \Gamma \tan \manN,  \qquad  \vvec \longmapsto \mapMN_{\ast} \vvec.
\end{align}
Let $\man$, $\manN$ and $\manO$ be differential manifolds and let $\mapMN_{\ast} : \tan_x \man \longrightarrow \tan_{\phi (\point)}\manN$ and $\mapNO_{\ast} : \tan_x \manN \longrightarrow \tan_{\phi (\point)}\manO$, then
\end{definition}
 \begin{lemma}The composition of the pushforwards is the pushforward of the composition
 \begin{align}
 \mapNO _{\ast} \circ \mapMN_{\ast} = (\mapNO \circ \mapMN)_{\ast}
  \end{align}
  \end{lemma}
  \begin{proofs}
   \begin{align*}
 (\mapNO \circ \mapMN)_{\ast}\vvec(\funcf) &= \vvec(\funcf \circ \mapNO \circ \mapMN)\\
 &=\mapMN_{\ast}\vvec(\funcf \circ \mapNO)\\
 &=\mapNO_{\ast} \mapMN_{\ast}\vvec(\funcf)
  \end{align*}
\end{proofs}

\begin{lemma}
 Let $(x^{1}, x^{2},...,x^{m})$ be a coordinate basis of $\real^{M}$ and $(y^{1}, y^{2},...,y^{n})$ a coordinate basis of $\real^{N}$. Given a diffeomorphism $\phi$ where
\begin{align}
&\maprmn : \real^{M} \longrightarrow \real^{N},\notag\\
&\point=x^{a}(\point)= (x^{1}(\point), ..,x^{m}(\point))\quad \longmapsto\quad \maprmn(\point)= y^{b}(\maprmn(x^{1}(\point), .., x^{m}(\point))),\notag
\end{align}
then
\begin{align}
\maprmn_{\ast}\frac{\partial}{\partial x^{a}}\Big|_x = \frac{\partial \maprmn^{b}}{\partial x^{a}} \frac{\partial}{\partial y^{b}}, \qquad \text{where}\qquad \maprmn^{b} = y^{b} \circ \maprmn.\label{pushend}
\end{align}
\end{lemma}
\begin{proofs}
\begin{align*}
\maprmn_{\ast}\frac{\partial}{\partial x^{a}}\Big|_x (\funcf) &= \frac{\partial}{\partial x^{a}}\Big|_x (\funcf \circ \maprmn)\\
&= \frac{\partial}{\partial x^{a}}\Big|_x (\funcf \circ y^{b} \circ \maprmn)\\
&= \frac{\partial}{\partial x^{a}} \funcf(y^{b}(\maprmn(\point)))
\end{align*}
where $y^{b}(\maprmn(\point)) = \maprmn^{b}$ by definition. Evaluating using the chain rule yields
\begin{align*}
\frac{\partial}{\partial x^{a}}(\funcf(\maprmn^{b})) &= \frac{\partial \maprmn^{b}}{\partial x^{a}} \frac{\partial \funcf}{\partial y^{b}}\\
&= \frac{\partial \maprmn^{b}}{\partial x^{a}} \frac{\partial}{\partial y^{b}}(\funcf)
\end{align*}
\end{proofs}
\section*{Pullback map}
\begin{definition}
For the smooth maps $\mapMN : \man \rightarrow \manN$ and $\funcf : \manN \rightarrow \real$ the \begin{bf}pullback\end{bf} of $\funcf$ with respect to $\mapMN$ is given by the composition:
\begin{align}
\mapMN^{\ast}\funcf = \mapNO \circ \funcf
\end{align}
\end{definition}
\begin{lemma}
 For the smooth maps $\mapMN : \man \rightarrow \manN$, $\mapNO : \manN \rightarrow \manO$, and $\funcf : \manN \rightarrow \real$  the \begin{bf}pullback\end{bf} of the composition is the composition of the pullbacks
\begin{align}
(\mapMN \circ \mapNO)^{\ast} = \mapMN^{\ast} \circ \mapNO^{\ast}
\end{align}
\end{lemma}

\begin{proofs}
\begin{align*}
 \mapNO^{\ast} (\mapMN^{\ast} \funcf) &= (\funcf \circ \mapMN) \circ \mapNO\\
&= \funcf \circ \mapMN \circ \mapNO\\
&= (\mapMN \circ \mapNO )^{\ast} \funcf
\end{align*}
\end{proofs}
\begin{definition}
 The pullback  of a differential $\dega$-form with respect to the smooth map $\mapMN : \man \longrightarrow \manN ,\quad \point \longmapsto \mapMN (\point)$ is given by:
\begin{align}
\mapMN^{\ast} : \Gamma \Lambda^{\dega} \manN \longrightarrow \Gamma \Lambda^{\dega} \man , \qquad \aform \longmapsto \mapMN^{\ast}\aform
\end{align}
where
\begin{align}
\mapMN^{\ast}(\aform + \bform) &= \mapMN^{\ast}\aform + \mapMN^{\ast}\bform\notag\\
\mapMN^{\ast}(\aform \wedge \bform) &= \mapMN^{\ast} \aform \wedge \mapMN^{\ast} \bform
\end{align}
for all $\aform  \in \Gamma \Lambda^\dega \manN$ and $\bform \in \Gamma \Lambda^\degb \manN$.
\end{definition}
\begin{definition}
The pullback of a 1-form acting on a vector field is the 1-form acting on the pushforward of the vector field
\begin{align}
\mapMN^{\ast}d\funcf(\vvec) = d\funcf (\mapMN_{\ast} \vvec),
\end{align}
for all $\funcf\in \Gamma \Lambda^0 \manN$ and $\vvec \in \Gamma \tan \man$.
\end{definition}
\begin{lemma}
 the pullback commutes with the exterior derivative
\begin{align}
d\mapMN^{\ast}\alpha = \mapMN^{\ast}d \alpha
 \end{align}
for all $\alpha\in\Gamma\Lambda^p\man$.
\end{lemma}
\begin{proofs}
First show for $\funcf \in \Gamma \Lambda^{0} \man$
\begin{align*}
\mapMN^{\ast}d\funcf(\vvec)&=d\funcf(\mapMN_{\ast}\vvec)\\
&=\mapMN_{\ast}\vvec(\funcf)\\
&=\vvec(\funcf\circ \mapMN)\\
&=\vvec(\mapMN^{\ast} \funcf)\\
&=d\mapMN^{\ast}\funcf(\vvec)
\end{align*}
Now show for $\wonform =\funcg d\funcf \in \Gamma \Lambda^{1}\man$ where $\funcg,\funcf  \in \Gamma \Lambda^{0}\man$
\begin{align*}
\mapMN^{\ast}d(\funcg d\funcf) &=\mapMN^{\ast}(d\funcg\wedge d\funcf)\\
&=(\mapMN^{\ast}d\funcg)\wedge (\mapMN^{\ast}d\funcf)\\
&=d\mapMN^{\ast}\funcg\wedge \mapMN^{\ast}d\funcf\\
&=d(\mapMN^{\ast}\funcg \mapMN^{\ast}d\funcf)\\
&=d\mapMN^{\ast}(\funcg d\funcf)
\end{align*}
proof for a general 1-form $\oneform =\oneform_{i}d x^{i}$ follows by linearity. \\
Now, assuming true for $\aform \in \Gamma \Lambda^{\dega}\man$
\begin{align*}
\mapMN^{\ast}d(\oneform \wedge \aform) &= \mapMN^{\ast}[d\oneform \wedge \aform - \oneform \wedge d\aform]\\
&=\mapMN^{\ast}(d\oneform) \wedge \mapMN^{\ast}\aform - \mapMN^{\ast}\oneform \wedge \mapMN^{\ast}d \aform\\
&=d\mapMN^{\ast}\oneform \wedge \mapMN^{\ast}\aform - \mapMN^{\ast}\oneform \wedge d\mapMN^{\ast} \aform\\
&=d[\mapMN^{\ast}\oneform \wedge \mapMN^{\ast} \aform]\\
&=d\mapMN^{\ast}(\oneform \wedge \aform)
\end{align*}
Hence by induction the relation holds for all $\dega$-forms.
\end{proofs}
\begin{lemma}
 The internal contraction of a pullback with respect to a vector field $\vvec$, is equal to the pullback of the internal contraction with respect to the pushforward of $\vvec$
\begin{align}
 i_{\vvec} \mapMN^{\ast} \aform = \mapMN^{\ast} (i_{\mapMN_{\ast \vvec}} \aform)
 \end{align}
 \end{lemma}
\begin{proofs} (by induction)\\
Trivial for 0-forms. First prove for 1-form $\wonform = \funcg d\funcf \quad \in \Gamma \Lambda^{1} \manN$
\begin{align*}
i_{\vvec}(\mapMN^{\ast}\funcg d \funcf) &=i_{\vvec}(\mapMN^{\ast}\funcg\mapMN^{\ast}d\funcf)\\
&= \mapMN^{\ast}\funcg i_{\vvec} \mapMN^{\ast}d\funcf\\
&= \mapMN^{\ast}\funcg(\mapMN^{\ast}d\funcf) . \vvec
\end{align*}
Noticing that $\mapMN^{\ast}d\funcf(\vvec) = d\funcf(\mapMN_{\ast}\vvec)$, and remembering that the pullback of a number doesn't change the number, ie
\begin{align*}
d\funcf(\mapMN_{\ast}\vvec) &= \mapMN^{\ast}(d\funcf.\mapMN_{\ast} \vvec)\\
&= \mapMN^{\ast}(i_{\mapMN_{\ast} \vvec}d\funcf)
\end{align*}
We can therefore write
\begin{align*}
i_{\vvec}(\mapMN^{\ast}\funcg d\funcf) &= \mapMN^{\ast}\funcg \mapMN^{\ast}(i_{\mapMN_{\ast} \vvec}d\funcf)\\
&= \mapMN^{\ast}(\funcg i_{\mapMN_{\ast} \vvec} d\funcf)\\
&=\mapMN^{\ast}[i_{\mapMN_{\ast} \vvec}(\funcg d\funcf)]\\
&= \mapMN^{\ast}(i_{\mapMN_{\ast} \vvec}\oneform)
\end{align*}
proof for a general 1-form $\oneform =\oneform_{i}d x^{i}$ follows by linearity. \\
Now, assuming the relation holds for $\aform \epsilon \Gamma \Lambda^{\dega}\manN$, we have
\begin{align*}
i_{\vvec}\mapMN^{\ast}(\oneform \wedge \aform)& =i_{\vvec}(\mapMN^{\ast}\oneform \wedge \mapMN^{\ast} \aform)\\
&= i_{\vvec}\mapMN^{\ast} \oneform \wedge \mapMN^{\ast} \aform - \mapMN^{\ast} \oneform \wedge \mapMN^{\ast} \aform\\
&=\mapMN^{\ast}(i_{\mapMN_{\ast} \vvec} \oneform) \wedge \mapMN^{\ast}\aform - \mapMN^{\ast}\oneform \wedge \mapMN^{\ast}(i_{\mapMN_{\ast}\vvec}\aform)\\
&=\mapMN^{\ast}(i_{\mapMN_{\ast} \vvec}(\oneform \wedge \aform))
\end{align*}
Thus we have proved by induction that the relation must hold for all $(\dega+1)$ forms, and therefore for any general form  $ \bform\in \Gamma \Lambda^{\degb} \manN$.
\end{proofs}
\section*{Curves}
\begin{definition}
A smooth parameterized curve $\c(s)$ on a manifold $\man$ is a smooth map from an open interval $I\subset \real$ to $\man$,
\begin{align}
\c: I \rightarrow \man, \qquad s \mapsto \c(s).
\end{align}
If $\co^a$ are local coordinates on $\man$ then we use the notation
\begin{align}
\co^a \circ \c(s)= \c^a(s),
\end{align}
thus if $\c(s_0)=\point$ is any point on the image of $\c$ then
\begin{align}
\co^a(\point) =\c^a(s_0).
\end{align}
The tangent vector to $\c$ at $\point$ is
\begin{align}
\cdot\atp\in T_\point \man, \qquad \cdot\atp = \c_\ast\Big(\frac{\partial}{\partial s}\Big) \Big|_{s_0}
\end{align}
For any $f\in \fsmooth(\man)$
\begin{align}
\c_\ast\Big(\frac{\partial}{\partial s}\Big) \Big|_{s_0}(f)= \frac{\partial}{\partial s}(\c^\ast f)\Big|_{s_0}= \frac{\partial}{\partial s}(f \circ \c(s_0))= \frac{\partial f}{\partial \co^a}\frac{\partial \c^a}{\partial s}\Big|_{s_0},
\end{align}
hence
\begin{align}
\cdot\atp=\cdot^a \frac{\partial}{\partial \co^a}\Big|_{s_0}= \frac{\partial \c^a}{\partial s}\Big|_{s_0}\frac{\partial}{\partial \co^a}.
\end{align}
\end{definition}
There is an induced vector field $\cdot \in \Gamma \tan \man$ where $\cdot\atp$ is the tangent vector at $\point$ for all $\point\in \c(s)$.

\section{Integration of $p$-forms}
\begin{definition}Let $\sigma$ be a diffeomorphism from the submanifold $\Sigma\subset \man$ of dimension $\dimN$ to the differentiable manifold $\man$ of dimension $\dimM$.
\begin{align}
\sigma:\Sigma \hookrightarrow \man
\end{align}
If $\co^a$ are local coordinates on $\man$ at $\sigma(\point)$ then $\sigma^\ast$ acting on the local basis of $1$-forms $d\co^a$ is given by
\begin{align}
\sigma^\ast d\co^a=d(\co^a \circ \sigma)
\end{align}
and for any $f\in \fsmooth(\man)$
\begin{align}
\sigma^\ast (f d\co^a)=(f\circ \sigma) d(\co^a \circ \sigma)
\end{align}
If we define a local coordinate system for $\Sigma$ at $\point\in\Sigma$ by
\begin{align}
z^a=\co^a\circ \sigma
\end{align}
then
\begin{align}
\sigma^\ast (f d\co^0\wedge d\co^1\wedge .. \wedge d\co^m)=(f\circ\sigma)dz^0\wedge dz^1 \wedge .. \wedge dz^n
\end{align}

 \end{definition}
\begin{definition}
If $\dimM$-form $\aform\in \Gamma \Lambda^{m} \man$ has compact support then so does the $\dimN$-form $\sigma^\ast \aform \in \Gamma \Lambda^{\dimN} \Sigma$, and
\begin{align}
\int_\man \aform = \int_\Sigma \sigma^\ast \alpha.
\end{align}
\end{definition}
\begin{theorem}
If $\Sigma$ is an oriented differential manifold of dimension $\dimN$, with boundary $\partial \Sigma$ of dimension $(\dimN-1)$ then
\begin{align}
&\int_\Sigma d \aform =\int_{\partial \Sigma} \aform,
\label{stokes_def}
\end{align}
for all $\aform\in \Gamma \Lambda^{\dimN-1}\Sigma$ with compact support. This theorem is often called the generalized \emph{Stokes'} theorem.
\end{theorem}

\begin{theorem}
\label{lie_int}
Let $t$ be a choice of coordinate on a manifold $\man$ such that $\frac{\partial}{\partial t}$ is Killing and let $t$ foliate $\man$ into surfaces $\Sigma_t$ . Then for $\alpha \in \Gamma \Lambda^p \man$
\begin{align}
\frac{d}{d t} \int_{\Sigma_t} \alpha=\int_{\Sigma_t} \l_{\frac{\partial}{\partial t}} \alpha,
\end{align}
and thus
\begin{align}
\int_{\man(t_1, t_0)} \alpha = \int_{t=t_0}^{t_1} dt \int_{\Sigma_t} i_{\partial_t}\alpha
\end{align}
where $\man(t_1, t_0)$ is a submanifold of $\man$ with range of $t$ between $t_0$ and $t_1$.
\end{theorem}

\chapter{Distributional  $p$-forms}
\label{dist_apend}
\section{Definitions}
The space of $C^{\infty}$ functions with compact support is called the space of test functions. We extend this notion to the space of test $\dega$-forms.
\begin{definition}
\label{def_test_form}
 Let $\man$ be a differential manifold of dimension $\dimM$. The space of test $\dega$-forms on $\man$ is  denoted $\Gamma_{0}\Lambda^{\dega} \man$,
\begin{align}
\Gamma_{0}\Lambda^{\dega}\man=\{\varphi \in \Gamma \Lambda^{\dega}\man|\enspace\varphi\enspace \text{has}\enspace \text{compact}\enspace \text{support}\}.
\end{align}
\end{definition}
\begin{definition}
\label{dist_form_def}
The space of $\dega$-form distributions $\Gamma_{D} \Lambda^{\dega}\man$  is the vector space dual to the space of test ($\dimM-\dega$)-forms $\Gamma_{0}\Lambda^{\dimM-\dega}\man$,
\begin{align}
\Gamma_{D}\Lambda^p\man \times\Gamma_{0}\Lambda^{\dimM-\dega}\man \rightarrow \real, \qquad (\Psi, \varphi)  \mapsto \Psi[\varphi] \in \real,
\end{align}
which satisfies
\begin{align}
\Psi[\lambda \varphi + \psi]=\lambda\Psi[\varphi]+\Psi[\psi],
\end{align}
for $\lambda \in \real$, $\varphi, \psi \in \Gamma_{0}\Lambda^{\dimM-\dega}\man$ and $\Psi \in \Gamma_{D} \Lambda^{\dega}\man$.
\end{definition}
\begin{definition}
The subspace of $\Gamma_{D} \Lambda^{\dega}\man$ comprising  \textbf{piecewise continuous} $\dega$-forms is the space of \textbf{regular distributions}. The action of a regular $\dega$-form distribution $\psi^D$ on an ($\dimM-\dega$)-test form $\varphi$ is given by the integral
\begin{align}
\psi^D[\varphi]=\int_{\man} \varphi \wedge &\psi
\end{align}
for any $\varphi \in \Gamma_{0}\Lambda^{\dimM-\dega}\man$ and where $\psi\in\Gamma\Lambda^\dega\man$ is piecewise continuous. We say that $\psi^D$ is the $p$-form distribution associated with the $p$-form $\psi$.\\[3pt]
\end{definition}

\begin{definition}
\label{diff_forms}
 The \textbf{exterior derivative} of a $\dega$-form distribution is defined as:
\begin{align}
d: \Gamma_{D}\Lambda^{\dega}\man\rightarrow\Gamma_{D}\Lambda^{\dega+1}\man, \quad \Psi \mapsto d\Psi
\end{align}
and satisfies
\begin{align}
  d\Psi[\varphi] = -\Psi[d\varphi^{\eta}]
\end{align}
For any $\varphi \in \Gamma_{0} \Lambda^{\dimM-(\dega+1)}\man$
\end{definition}

\begin{lemma}
If $\man$ has no boundary then for any regular distribution $\psi^D \in\Gamma_{D}\Lambda^{\dega}\man$
\begin{align}
d \psi^D[\varphi]&=(d\psi)^D[\varphi]
\end{align}
\end{lemma}
\begin{proofs}
\begin{align*}
d \psi^D[\varphi] &= -\int_{\man}d\varphi^{\eta}\wedge \psi\\
&=\int_{\man}\varphi\wedge d\psi- \int_{\man}d(\varphi^{\eta}\wedge\psi)\\
&=\int_{\man}\varphi\wedge d\psi - \int_{\partial \man}(\varphi^{\eta}\wedge\psi)\\
&=\int_{\man}\varphi \wedge d \psi\\
&= (d\psi)^D[\varphi]
\end{align*}
\end{proofs}
\section{Criteria for regular distributions in N-U coordinates}

\begin{theorem}
\label{dist_oneform}
Let the 1-form $\alpha\in \Gamma \Lambda^1 \mnotc$ be represented in Newman-Unti coordinates by
\begin{align}
\alpha=\alpha_idz^i, \qquad \text{where} \quad z^0=\tau, \quad z^1=R, \quad z^2=\theta, \quad z^3=\phi,
\end{align}
and where the functions $\alpha_i=\alpha_i(\tau, R, \theta, \phi)$ are polynomials in $R$ and are singular on the worldline. Let the most divergent terms in the polynomial functions $\alpha_i$ be denoted
 \begin{align}
 \hat{\alpha}_i=\frac{\alpha_i'(\tau, \theta, \phi)}{R^{\beta_i}}.\label{alpha_beta}
 \end{align}
 where $\displaystyle{\alpha_i'(\tau, \theta, \phi)}$ are bounded and $\beta_i$ are positive constants.
The distribution $\alpha^D\in \Gamma_D \Lambda^1 \m$, where
\begin{align}
\alpha^D[\varphi]=\int_\m  \varphi \wedge \alpha \qquad \text{is finite for all} \quad \varphi \in \Gamma_0 \Lambda^3 \m,
\end{align}
is well defined providing the four constants $\beta_i$ satisfy
\begin{align}
  \beta_0< 3, \quad \beta_1<2, \quad \beta_2<2, \quad \beta_3< 3 .\label{asymp_cond}
\end{align}
\end{theorem}
\begin{proofs}
An arbitrary test 3-form $\varphi\in \Gamma_0 \Lambda^3 \m$ is given in Minkowski coordinates by
$\displaystyle{\varphi=\varphi_{i j k} dy^i \wedge dy^j \wedge dy^k}$. Applying a coordinate transformation such that $\varphi=\hat{\varphi}_{i j k} dz^i \wedge dz^j \wedge dz^k$ where $\{z^i\}$ are NU coordinates  yields the following form for the coefficients $\hat{\varphi}_{i j k}$,
\begin{align}
\hat{\varphi}_{1 2 3}=& R^2 \textup{Y}^2_{1 2 3},\notag\\
\hat{\varphi}_{0 1 2}=& R \textup{Y}^1_{0 1 2},\notag\\
\hat{\varphi}_{0 1 3}=&R \textup{Y}^1_{0 1 3},\notag\\
\textup{and}\quad\hat{\varphi}_{0 2 3}=&R^2 \textup{Y}^2_{0 2 3}+R^3 \textup{Y}^3_{0 2 3}.\label{phi_three}
\end{align}
Here the functions $\textup{Y}^l_{i j k}$ depend on the test functions $\varphi_{i j k}$, sines and cosines of $\theta$ and $\phi$, and the functions $\cdot_i$ and $\cddot_i$. The key result is that they are bounded functions of $\tau$, $\theta$ and $\phi$.

 We are interested in the boundedness of $\alpha^D[\varphi]$ therefore it is sufficient to show that $\displaystyle{\int_\m \varphi \wedge \hat{\alpha}_i dz^i}$ is bounded for all $\varphi \in \Gamma_0 \Lambda^3\m$.
In component form we have
\begin{align}
\Bigg| \int_\m \varphi \wedge \hat{\alpha}_i dz^i \notag\Bigg| =&\Bigg|\int_\m (-\hat{\varphi}_{1 2 3} \hat{\alpha}_0 +\hat{\varphi}_{0 2 3} \hat{\alpha}_1 -\hat{\varphi}_{0 1 3}\hat{\alpha}_2 + \hat{\varphi}_{0 1 2 }\hat{\alpha}_3) d z^{0123}\Bigg|,
 \end{align}
 where $\displaystyle{d z^{0123}=d z^0 \wedge d z^1 \wedge d z^2 \wedge d z^3}$.

 Substituting the relations \ref{phi_three} yields
 \begin{align}
 \Bigg| \int_\m \varphi \wedge \hat{\alpha}_i d z^i \Bigg|\notag
 =&\Bigg|\int_\m  R^2 \textup{Y}^2_{1 2 3}\hat{\alpha}_0d z^{0123}\Bigg| +\Bigg|\int_\m (R^2\textup{Y}^2_{0 2 3}+R^3 \textup{Y}^3_{0 2 3})\hat{\alpha}_1 d z^{0123}\Bigg|\notag\\
 &+\Bigg|\int_\m R \textup{Y}^1_{0 1 3}\hat{\alpha}_2 d z^{0123}\Bigg|+\Bigg|\int_\m  R \textup{Y}^1_{0 1 2} \hat{\alpha}_3 d z^{0123}\Bigg| \label{r_deps}
 \end{align}
Substituting \ref{alpha_beta} and separating with respect to $R$-dependence yields
\begin{align}
 \Bigg| \int_\m \varphi \wedge \hat{\alpha}_i dz^i \Bigg|\leq& \Bigg|\max\Bigg(\int_{\tau=-\infty}^{\tau=\infty}\int_{\theta=0}^{\theta=\pi}\int_{\phi=0}^{\phi=2\pi}\alpha_0'\textup{Y}^2_{1 2 3} d z^{023}\Bigg)\Bigg|\int_0^\epsilon \frac{R^2}{R^{\beta_0}}  d z^1\notag\\
 &+ \Bigg|\max\Bigg(\int_{\tau=-\infty}^{\tau=\infty}\int_{\theta=0}^{\theta=\pi}\int_{\phi=0}^{\phi=2\pi}\alpha_1'(\textup{Y}^2_{0 2 3}+\textup{Y}^3_{0 2 3}) dz^{023}\Bigg)\Bigg|\int_0^\epsilon \frac{R^2}{R^{\beta_1}}   d z^1\notag\\
 &+ \Bigg|\max\Bigg(\int_{\tau=-\infty}^{\tau=\infty}\int_{\theta=0}^{\theta=\pi}\int_{\phi=0}^{\phi=2\pi}\alpha_2'\textup{Y}^1_{0 1 3} dz^{023}\Bigg)\Bigg|\int_0^\epsilon \frac{R}{R^{\beta_2}}   d z^1\notag\\
 &+ \Bigg|\max\Bigg(\int_{\tau=-\infty}^{\tau=\infty}\int_{\theta=0}^{\theta=\pi}\int_{\phi=0}^{\phi=2\pi}\alpha_3'\textup{Y}^1_{0 1 2} dz^{023}\Bigg)\Bigg|\int_0^\epsilon \frac{R}{R^{\beta_3}}   d z^1\label{ints_R}
\end{align}
We now consider the integrals w.r.t $  z^1=R$. The standard integral
\begin{align}
\int_0^\epsilon \frac{R^\gamma}{R^{\beta}}  d R= \Bigg[\frac{R^{1+\gamma-\beta}}{1+\gamma-\beta}\Bigg]^\epsilon_0\label{standard_res}
\end{align}
where $\epsilon\in \real^{+}$, is bounded in the limit $\epsilon\rightarrow 0$ providing $\beta < 1+\gamma$.  Comparison with the powers in \ref{ints_R} yields the conditions \ref{asymp_cond}.
\end{proofs}
\begin{theorem}
\label{dist_twoform}
Let the 2-form $\alpha\in \Gamma \Lambda^2 \mnotc$ be represented in Newman-Unti coordinates by
\begin{align}
\alpha=\alpha_{ij}dz^i\wedge dz^j, \qquad \text{where} \quad z^0=\tau, \quad z^1=R, \quad z^2=\theta, \quad z^3=\phi,
\end{align}
and where the functions $\alpha_{ij}=\alpha_{ij}(\tau, R, \theta, \phi)$ are polynomials in $R$ and are singular on the worldline.
Let the most divergent terms in the functions $\alpha_{ij}$ be denoted
 \begin{align}
 \hat{\alpha}_{ij}=\frac{\alpha_{ij}'(\tau, \theta, \phi)}{R^{\beta_{ij}}}.\label{alpha_beta2}
 \end{align}
 where $\displaystyle{\alpha_{ij}'(\tau, \theta, \phi)}$ are bounded. and $\beta_{ij}$ are positive constants.
 The distribution $\alpha^D\in \Gamma_D \Lambda^2 \m$, where
\begin{align}
\alpha^D[\varphi]=\int_\m \varphi \wedge \alpha  \qquad \text{is finite for all} \quad \varphi \in \Gamma_0 \Lambda^2 \m,
\end{align}
is well defined providing the six constants $\beta_{ij}$ satisfy
\begin{align}
\beta_{01}< 1, \quad \beta_{12}< 2 , \quad  \beta_{13}<2, \notag\\  \beta_{02}<2 , \quad  \beta_{03} <2 , \quad  \beta_{23}<3 .\label{asymp_cond2}
\end{align}

\end{theorem}
\begin{proofs}
An arbitrary test 2-form $\phi\in \Gamma_0 \Lambda^2 \m$ is given by
\begin{align}
\varphi=\varphi_{i j} dy^i \wedge dy^j
\end{align}
Applying a coordinate transformation such that $\varphi=\hat{\varphi}_{i j} dz^i \wedge dz^j$ where $\{z^i\}$ are NU coordinates  yields the following form for the coefficients $\hat{\varphi}_{i j}$,
\begin{align}
\hat{\varphi}_{1 2}=& R \textup{Y}^1_{1 2},\notag\\
\hat{\varphi}_{1 3}=& R \textup{Y}^1_{1 3},\notag\\
\hat{\varphi}_{0 2}=&R \textup{Y}^1_{0 2}+R^2 \textup{Y}^2_{0 2},\notag\\
\hat{\varphi}_{0 3}=&R \textup{Y}^1_{0 3}+R^2 \textup{Y}^2_{0 3},\notag\\
\hat{\varphi}_{0 1}=& \textup{Y}^0_{0 1},\notag\\
\textup{and}\quad\hat{\varphi}_{2 3}=&R^2 \textup{Y}^2_{2 3}.\label{phi_two}
\end{align}
Here as before the functions $\textup{Y}^l_{i j}$  are bounded functions of $\tau$, $\theta$ and $\phi$.
  We are interested in the boundedness of $\alpha^D[\varphi]$ therefore it is sufficient to show that $\displaystyle{\int_\m \varphi \wedge \hat{\alpha}_{ij} dz^i\wedge dz^j }$ is bounded for all $\varphi \in \Gamma_0 \Lambda^2\m$.
Hence
\begin{align}
 \Bigg|\int_\m\varphi \wedge \hat{\alpha}_{ij} dz^i\wedge &dz^j\Bigg|\\ =&\Bigg|\int_\m (-\hat{\varphi}_{1 3}\hat{\alpha}_{02} +\hat{\varphi}_{1 2 }\hat{\alpha}_{03} +\hat{\varphi}_{0  3}\hat{\alpha}_{12} -\hat{\varphi}_{0 2 }\hat{\alpha}_{13} +\hat{\varphi}_{  2 3 }\hat{\alpha}_{01} +\hat{\varphi}_{0  1 }\hat{\alpha}_{23} ) dz^{0123}\Bigg|\notag
 \end{align}
 Substituting the relations \ref{phi_two} yields
\begin{align}
 \Bigg|\int_\m \varphi \wedge \hat{\alpha}_{ij} dz^i\wedge dz^j \notag \Bigg|=&\Bigg|\int_\m   R \textup{Y}^1_{1 3}\hat{\alpha}_{02}dz^{0123}\Bigg|+\Bigg|\int_\m  R \textup{Y}^1_{1 2}\hat{\alpha}_{03}dz^{0123}\Bigg|\notag\\
 &+\Bigg|\int_\m  (R \textup{Y}^1_{0 3}+R^2 \textup{Y}^2_{0 3})\hat{\alpha}_{12}dz^{0123}\Bigg|\notag\\
 &+\Bigg|\int_\m  (R \textup{Y}^1_{0 2}+R^2 \textup{Y}^2_{0 2})\hat{\alpha}_{13}dz^{0123}\Bigg|\notag\\
 &+\Bigg|\int_\m R^2 \textup{Y}^2_{2 3} \hat{\alpha}_{01}dz^{0123}\Bigg|+\Bigg|\int_\m\textup{Y}^0_{0 1} \hat{\alpha}_{23}  dz^{0123}\Bigg|
 \end{align}
Substituting \ref{alpha_beta2} and separating with respect to $R$-dependence yields
\begin{align}
 \Bigg|\int_\m \hat{\alpha}_{ij} dz^i\wedge dz^j \wedge \varphi\notag \Bigg|\leq& \Bigg|\max\Bigg(\int_{\tau=-\infty}^{\tau=\infty}\int_{\theta=0}^{\theta=\pi}\int_{\phi=0}^{\phi=2\pi}\alpha_{02}'\textup{Y}^1_{1 3} dz^{023}\Bigg)\Bigg|\int_0^\epsilon \frac{R}{R^{\beta_{02}}}   d z^1\notag\\
 &+ \Bigg|\max\Bigg(\int_{\tau=-\infty}^{\tau=\infty}\int_{\theta=0}^{\theta=\pi}\int_{\phi=0}^{\phi=2\pi}\alpha_{03}'\textup{Y}^1_{1 2} dz^{023}\Bigg)\Bigg|\int_0^\epsilon \frac{R}{R^{\beta_{03}}}   d z^1\notag\\
 &+\Bigg|\max\Bigg(\int_{\tau=-\infty}^{\tau=\infty}\int_{\theta=0}^{\theta=\pi}\int_{\phi=0}^{\phi=2\pi}\alpha_{12}'(\textup{Y}^1_{0 3}+\textup{Y}^2_{0 3}) dz^{023}\Bigg)\Bigg|\int_0^\epsilon \frac{R}{R^{\beta_{12}}}   d z^1\notag\\
 &+\Bigg|\max\Bigg(\int_{\tau=-\infty}^{\tau=\infty}\int_{\theta=0}^{\theta=\pi}\int_{\phi=0}^{\phi=2\pi}\alpha_{13}'(\textup{Y}^1_{0 2}+\textup{Y}^2_{0 2}) dz^{023}\Bigg)\Bigg|\int_0^\epsilon \frac{R}{R^{\beta_{13}}}   d z^1\notag\\
 &+\Bigg|\max\Bigg(\int_{\tau=-\infty}^{\tau=\infty}\int_{\theta=0}^{\theta=\pi}\int_{\phi=0}^{\phi=2\pi}\alpha_{01}'\textup{Y}^2_{2 3}dz^{023}\Bigg)\Bigg|\int_0^\epsilon \frac{R^2}{R^{\beta_{01}}}   d z^1\notag\\
 &+\Bigg|\max\Bigg(\int_{\tau=-\infty}^{\tau=\infty}\int_{\theta=0}^{\theta=\pi}\int_{\phi=0}^{\phi=2\pi}\alpha_{23}'\textup{Y}^0_{0 1}dz^{023}\Bigg)\Bigg|\int_0^\epsilon \frac{1}{R^{\beta_{23}}}   d z^1
 \end{align}
Once again comparing the integrals with respect to $z^1=R$ with the standard result \ref{standard_res} yields the relations \ref{asymp_cond2}.
\end{proofs} 
\chapter{Dirac Geometry}
\label{dirac_apend}

\section{Definitions}
\label{dirac_coords}
\begin{definition}
\label{dirac_map}

Consider the region $N=\widetilde{N}\backslash \c$ where $\widetilde{N} \subset \m$ is a local neighborhood of the worldline.
For every field point $\point\in N$ there is a unique point $\tau_D(x)$ at which the worldline crosses the plane of simultaneity
according to an observer comoving with the charge at $\point$.
\begin{align}
\c &: \real \rightarrow \m, \quad \tau \mapsto \c(\tau)\\
\tau_D &: \m \rightarrow \real,\quad \point \mapsto \tau_D(\point)
\label{tau_def}
\end{align}
\end{definition}
\begin{definition}
\label{dirac_y}
 Dirac geometry uses a spacelike displacement vector $\y=\point-\c(\tau_D(\point))$, which satisfies
\begin{align}
\g(\y, \cdot(\tau_D))=0, \quad\quad\quad R_D^2=g(\y, \y),
\label{Dirac_geom_def}
\end{align}
to associate a spacetime point with a point on the worldline. We observe that $R_D>0$ is the magnitude of $\y$.
\end{definition}


\begin{definition}
\label{def_vad}
We use the notation $\c_D=\c(\tau_D(x))$. The vector fields $\v_D, \a_D, \adot_D \in \Gamma \tan N $ are defined as
\begin{align}
\v_D\atp&=\cdot^j(\tau_D(\point)) \frac{\partial}{\partial \co^j},\quad\a_D\atp=\cddot^j(\tau_D(\point))\frac{\partial}{\partial \co^j}\quad \text{and} \quad\adot_D\atp=\cdddot^j(\tau_D(\point)) \frac{\partial}{\partial \co^j},
\label{defdirac_V_A}
\end{align}
\end{definition}
\begin{lemma}
The exterior derivative of the Dirac time $\tau_D$ is given by
\begin{align}
d\tau_D=-\frac{\vdual_D}{\g(\y, \a_D)+1}.
\label{tauD_def}
\end{align}
\end{lemma}
\begin{proofs}\\
It follows from definition \ref{dirac_y} that
\begin{align}
0=&d \g(\y, \v_D),\notag\\
=& d \g(\point, \v_D)-d\g(\c_D, \v_D),\notag\\
=&\vdual_D+\big(\g(\point, \a_D)+1-\g(\a_D, \c_D)\big)d\tau_D,\notag\\
=&\vdual_D+(\g(\y, \a_D)+1)d\tau_D.
\end{align}
\end{proofs}
\begin{lemma}
\begin{align}
d R_D= \frac{\ydual}{R_D}
\end{align}
\end{lemma}
\begin{proofs}\\
Let
\begin{align}
\xvec\in \Gamma \tan \m,\qquad \xvec\atp=x^a \frac{\partial}{\partial \co^a} \qquad \text{and} \qquad \cvec_D \in \Gamma \tan \m,\qquad \cvec_D\atp=\c_D^a\frac{\partial}{\partial \co^a},
\end{align}
It follows from definition \ref{dirac_y} that
\begin{align}
dR_D=&d \sqrt{\g(\y, \y)}\notag\\
=&\frac{1}{2 \sqrt{\g(\y, \y)}} d \g(\y, \y)\\[3pt]
d \g(\y, \y)=& d \g(x-\c_D, x-\c_D)\notag\\
=&d\g(\xvec, \xvec)+d\g(\cvec_D, \cvec_D)-2d\g(\xvec, \cvec_D)\\[3pt]
d\g(\xvec, \xvec)&= d(\g_{a b}\point^a \point^b),\nonumber\\
&= \g_{a b} (d \point^a) \point^b + \g_{a b} \point^a (d \point^b),\nonumber\\
&= \point_a d\point^a + \point_a d\point^a,
\end{align}
Now $d\point^a=d\co^a$ since $\point=(\xt, \xx, \xy, \xz)$, therefore
\begin{align}
d\g(\xvec, \xvec)&= 2\point_a d\co^a,\nonumber\\
&= 2 \widetilde{\xvec}.
\end{align}
Also
\begin{align}
d\g(\cvec_D, \cvec_D) &= d(\g_{a b} \c_D^a \c_D^b),\nonumber\\
&=(d \c_D^a)\g_{a b} \c_D^b + (d \c_D^a) \g_{a b} \c_D^b,\nonumber\\
&=2 \c_{D a} \v_D^a d\tau_D,\nonumber\\
&=2 \g( \cvec_D, \v_D) d\tau_D,
\end{align}
and
\begin{align}
d\g(\cvec_D, \xvec) &= d(\g_{a b}\point^a \c_D^b),\nonumber\\
&=\g_{a b} (d\point^a)\c_D^b + \g_{a b} \point^a d(\c_D^b),\nonumber\\
&=\c_{D a} d\point^a + \point_a d(\c_D^a),
\end{align}
where $\displaystyle{d(\c_D^a)=\frac{\partial(\c_D^a)}{\partial \tau}d\tau=\v_D^a d\tau}$, therefore
\begin{align}
d\g(\cvec_D, \xvec) &= \widetilde{\cvec_D} + \point_a  \v_D^a d \tau_D,\nonumber\\
& = \widetilde{\cvec_D} + \g(\xvec, \v_D) d\tau_D.
\end{align}
Thus
\begin{align}
dR_D=&\frac{1}{2R_D} (d\g(\xvec, \xvec)+d\g(\cvec_D, \cvec_D)-2d\g(\xvec, \cvec_D))\notag\\
=&\frac{1}{R_D}\big(\ydual+ \g( \y, \v_D) d\tau_D)
\end{align}
The definition  \ref{Dirac_geom_def} yields
\begin{align}
dR_D=&\frac{\ydual}{R_D}
\end{align}
\end{proofs}
\begin{lemma}
\begin{align}
d\g(\y, \a_D)=& \adual_D-\frac{\vdual_D \g(\y, \adot_D)}{\g(\y, \a_D)+1}\notag\\[2pt]
d\g(\y, \adot_D)=& \adotdual_D-\frac{\vdual_D (\g(\y, \addot_D)+\g(\a_D, \a_D))}{\g(\y, \a_D)+1}\notag\\[2pt]
d\g(\a_D,\a_D)=& \frac{-2\g(\a_D, \adot_D)\vdual_D}{\g(\y, \a_D)+1}
\label{dirac_relations}
\end{align}
\end{lemma}
\begin{definition}
\label{def_normd}
We define the normalized vector field
\begin{align}
\n_D=\frac{\y}{R_D},\qquad \textup{where} \qquad \g(\n_D, \n_D)=1 \qquad \textup{and} \qquad \g(\n_D, \v_D)=0.
\end{align}
\label{nd_def}
\end{definition}

\section{The \LW potential expressed in Dirac Geometry}
\label{dirac_lw}
Dirac geometry is not a natural choice to use to describe electromagnetic phenomena because all retarded (and advanced) quantities are given only as Taylor expansions around the Dirac time $\tau_D$. The retarded stress form must be calculated as such an expansion. Below we give the advanced and retarded \LW potentials.
\begin{lemma}
\label{deltas_def}
The difference $\delta_r=\tau_D-\tau_r$ is given in terms of $R_D$ by
\begin{align}
\delta_r=&R_D-\frac{1}{2}\g(\n, \cddot)R_D^2 +\big(\frac{3}{8}g(\n_D, \a_D)^2 +\frac{1}{6}\g(\n_D, \adot_D)-\frac{1}{24}\g(\a_D, \a_D)\big)R_D^3+\ord(R_D^4).
\label{dirac_taur}
\end{align}
and the difference $\delta_a=\tau_a-\tau_D$ is given by
\begin{align}
\delta_a=&R_D-\frac{1}{2}\g(\n, \cddot)R_D^2 +\big(\frac{3}{8}g(\n_D, \a_D)^2 -\frac{1}{6}\g(\n_D, \adot_D)-\frac{1}{24}\g(\a_D, \a_D)\big)R_D^3+\ord(R_D^4).
\label{dirac_taua}
\end{align}
\end{lemma}
\begin{proofs}
\begin{align}
\c(\tau_r)=\c_D-\v_D\delta_r+\a_D\frac{\delta_r^2}{2}-\addot_D\frac{\delta_r^3}{6}+\adddot_D\frac{\delta_r^4}{24}+\ord(\delta_r^5)
\end{align}
and thus the null vector $\x$ is given by
\begin{align}
\x=\point-\c(\tau_r) &= \point-\c_D+\v_D\delta_r-\a_D\frac{\delta^2}{2}+\adot_D\frac{\delta_r^3}{6}-\addot_D\frac{\delta_r^4}{24}+\ord(\delta_r^5),\notag\\
&=\y+\v_D\delta_r-\a_D\frac{\delta_r^2}{2}+\adot_D\frac{\delta_r^3}{6}-\addot_D\frac{\delta_r^4}{24}+\ord(\delta_r^5).
\label{X_taud}
\end{align}
Substituting (\ref{X_taud}) into the lightcone condition (\ref{null_cond}) gives
\begin{align}
\g(\x, \x)=& \g(\y, \y)+2\g(\y, \v_D)\delta_r-(1+\g(\y, \a_D))\delta_r^2+\frac{1}{3}\g(\y, \adot_D)\delta_r^3,\notag\\
&-\frac{1}{12}(\g(\y, \addot_D)+\g(\a_D, \a_D))\delta_r^4+\ord(\delta_r^5).
\end{align}
Definition (\ref{dirac_y}) and (\ref{nd_def}) yield
\begin{align}
\g(\x, \x)=& R_D^2-(1+R_D\g(\n_D, \a_D))\delta_r^2+\frac{R_D}{3}\g(\n_D, \adot_D)\delta_r^3\notag\\
&-\frac{1}{12}(R_D\g(\n_D, \addot_D)+\g(\a_D, \a_D))\delta_r^4+\ord(\delta_r^5).
\label{xx_d}
\end{align}
We may solve this equation to obtain $\delta_r$ and $\delta_a$ in terms of $R_D$.

Let $\delta_r=a_1 R_D$, then equating coefficients of order $R_D^2$ yields
\begin{align}
a_1^2=1.
\end{align}
We choose $\delta_r > 0$. Knowing that $R_D>0$ it follows that $a_1=+1$.  Now let $\delta_r=R_D+a_2R_D^2$ then then equating coefficients of order $R_D^3$ yields
\begin{align}
0=&2a_2+\g(\n_D, \cddot)\notag\\
\Rightarrow \quad& a_2=-\frac{\g(\n_D, \cddot)}{2}
\end{align}
Let $\delta_r=R_D-\frac{\g(\n_D, \cddot)}{2}R_D^2+a_3R_D^3$, then then equating coefficients of order $R_D^4$ yields
\begin{align}
a_3=\frac{3}{8}g(\n_D, \a_D)^2 +\frac{1}{6}\g(\n_D, \adot_D)-\frac{1}{24}\g(\a_D, \a_D).
\end{align}
Thus to third order $\delta_r$ is given by (\ref{dirac_taur}).

A similar calculation may be performed in to obtain an expression for $\delta_a =\tau_a-\tau_D$ in terms of $R_D$. In this case all quantities on the left hand side are evaluated at the advanced time $\tau_a$, so that instead of solving the retarded null condition $\g(\x, \x)=0$ we must solve the advanced null condition $\g(\w, \w)=0$. Since $\tau_a-\tau_D$ is positive this means that all terms with odd powers of $\delta$ will have opposite sign to those in the retarded calculations. The resulting expression for $\delta_a$ is given by (\ref{dirac_taua}).

\end{proofs}
\begin{lemma}
In terms of the Dirac time $\tau_D$ and the Dirac radius $R_D$ the retarded \LW potential is given by
\begin{align}
\alw_r=&-\frac{\v_D}{R_D} +\big(\a_D +\frac{1}{2}\g(\n_D, \a_D)\v_D\big)\notag\\
&+\Big(\v_D\big(\frac{1}{8}\g(\a_D, \a_D)-\frac{1}{8}\g(\n_D, \a_D)^2-\frac{1}{3}\g(\n_D, \adot_D)\big) -\frac{1}{2}\adot_D-\frac{1}{2}\g(\n_D, \a_D)\a_D\Big)R_D\notag\\
& +\ord(R_D^2),
\label{lw_dirac}
\end{align}
and the advanced \LW potential is given by
\begin{align}
\alw_a=&\frac{\v_D}{R_D} +\big(\a_D -\frac{1}{2}\g(\n_D, \a_D)\v_D\big)\notag\\
&+\Big(-\v_D\big(\frac{1}{8}\g(\a_D, \a_D)-\frac{1}{8}\g(\n_D, \a_D)^2+\frac{1}{3}\g(\n_D, \adot_D)\big) +\frac{1}{2}\adot_D-\frac{1}{2}\g(\n_D, \a_D)\a_D\Big)R_D\notag\\
& +\ord(R_D^2)
\label{alw_dirac}
\end{align}
\end{lemma}
\begin{proofs}

We evaluate the retarded \LW potential as a series in $R_D$.
\begin{align}
\v=\v_D-\a_D\delta_r+\adot_D\frac{\delta_r^2}{2}-\addot_D(\tau_D)\frac{\delta_r^3}{6}+\ord(\tau^4)
\end{align}
Substituting (\ref{dirac_taur}) yields
\begin{align}
\v=& \v_D-\a_dR_D +\frac{1}{2}\big(\adot_D+\a_D\g(\n_D, \a_D)\big)R_D^2\notag\\
&+\Big(\a_D\big(\frac{3}{8}g(\n_D, \a_D)^2 +\frac{1}{6}\g(\n_D, \adot_D)-\frac{1}{24}\g(\a_D, \a_D)\big)-\frac{1}{6}\addot_D-\frac{1}{2}\g(\n_D, \a_D)\adot_D\Big)R_D^3\notag\\
&+\ord(R_D^4)
\label{v_dirac}
\end{align}
Also
\begin{align}
\g(\x, \v)=& -(\g(\y, \a_D)+1)\delta_r+\g(\y, \adot_D)\frac{\delta_r^2}{2}\notag\\
&+\big(\g(\a_D, \a_D)-\g(\y, \addot_D)\big)\frac{\delta_r^3}{6} +\ord(\delta_r^4)
\end{align}
Again substituting (\ref{dirac_taur}) yields
\begin{align}
\g(\x, \v)=&-R_D -\frac{1}{2}\g(\n_D, \a_D)R_D^2\notag\\
&+\Big(\frac{1}{8}\g(\n_D, \a_D)^2 +\frac{1}{2}\g(\n_D, \a_D) -\frac{1}{6}\g(\n_D, \adot_D)+\frac{1}{24}\g(\a_D, \a_D)\Big)R_D^3 +\ord(R_D^4)
\label{gxv_dirac}
\end{align}
Dividing (\ref{v_dirac}) by (\ref{gxv_dirac}) gives (\ref{lw_dirac}). Evaluating the advanced potential
  \begin{align}
  \alw_{\textup{adv}}\atp=\frac{\cdot(\tau_a)}{\g(\w, \cdot(\tau_a))}
  \end{align}
   using the same procedure leads to (\ref{alw_dirac}).
\end{proofs}

 The retarded and advanced \LW fields $\fret$ and $\fadv$ are obtained by taking the exterior derivative of $\alw_{\textup{ret}}$ and $\alw_{\textup{adv}}$ respectively. In 1938 Dirac \cite{Dirac38} showed that the difference between the retarded and advanced fields  is finite on the worldline and given by
 \begin{align}
 \frac{1}{2} (\fret-\fadv)= \frac{2}{3}(g(\cddot, \cddot)\widetilde{\cdot}-\widetilde{\cdddot})
 \label{fminus}
 \end{align}


It is easily seen that taking the sum of expansions
\begin{align}
\fret=\frac{1}{2}(\fadv+\fret)+\frac{1}{2}(\fadv-\fret).
\end{align}
is equivalent to expanding $\fret$ only. This point was  emphasized by Infeld and Wallace \cite{Infeld39}, and later by Havas \cite{Havas48}.


\chapter[Adapted N-U coordinates ]{Adapted N-U coordinates $(\tau, \rnew , \theta, \phi)$}
\label{app_coords}

For the numerical investigation presented in chapter \ref{bendingbeams} we use a coordinate system $(\tau, \rnew , \theta, \phi)$ adapted from the Newman-Unti coordinates. This change in coordinates was initially motivated by our interest in the ultra-relativistic \LW fields. The N-U coordinate system breaks down in the ultra-relativistic limit since $R=-\g(\x, \v)=0$ when $\v$ is null. In the new coordinate system the radial parameter is given by
\begin{align}
\rnew=-\frac{R}{\alpha} = -\g(\x, \partial_{y^0})
\end{align}
which remains non-zero in the ultra-relativistic limit. If $(y^0, y^1, y^2, y^3)$ is the global Lorentzian coordinate chart then the coordinate transformation is given by
\begin{align}
&y^0=\c^0 (\tau) + \rnew \notag\\
&y^1=\c^1(\tau) + \rnew \sin(\theta)\cos(\phi)\notag\\
&y^2=\c^2 (\tau) + \rnew  \sin(\theta)\sin(\phi)\notag\\
&y^3=\c^3(\tau) +\rnew  \cos(\theta).
\label{coords_new}
\end{align}
\begin{lemma}
In terms of the new coordinates the vector fields $\x, \v \in \Gamma \tan \mnotc$ are given by
\begin{align}
&X=\rnew \frac{\partial}{\partial \rnew }\\\notag
&\v=\frac{\partial}{\partial \tau}
\end{align}
\end{lemma}
\begin{proofs}
\begin{align*}
X&= \underline{x} - \c(\tau)\\
&=\rnew \frac{\partial}{\partial y^0} + \rnew  \sin(\theta)\cos(\phi)\frac{\partial}{\partial y^1}+\rnew \sin(\theta) \sin(\phi)\frac{\partial}{\partial y^2} +\rnew \cos\theta \frac{\partial}{\partial y^3}\\
&=\rnew \frac{\partial}{\partial \rnew }\\
\frac{\partial}{\partial \tau}&= \frac{\partial y^0}{\partial \tau} \frac{\partial}{\partial y^0} + \frac{\partial y^1}{\partial \tau}\frac{\partial}{\partial y^1}+\frac{\partial y^2}{\partial \tau}\frac{\partial}{\partial y^2} +\frac{\partial y^3}{\partial \tau }\frac{\partial}{\partial y^3}\\
&=\cdot^0(\tau)\frac{\partial}{\partial y^0} +\cdot^1(\tau)\frac{\partial}{\partial y^1} +\cdot^2(\tau)\frac{\partial}{\partial y^2}+ \cdot^3(\tau)\frac{\partial}{\partial y^3}\\
&=\cdot^a(\tau)\frac{\partial}{\partial y^a}\\
&=\cdot(\tau)\\
&=\v
\end{align*}
\end{proofs}
\begin{lemma}
The Minkowski metric $\g\in \bigotimes^{[ \mathds{F}, \mathds{F}]} \man$ is given by
\begin{align}
g=&-\ccon^2 d\tau \otimes d\tau +\rnew ^2 d\theta \otimes d\theta + \rnew ^2 \sin^2\theta d\phi \otimes d\phi\nonumber\\
&+ \alpha[d\tau \otimes d\rnew  + d\rnew  \otimes d\tau]+\rnew \atheta [d\tau \otimes d\theta + d\theta \otimes d\tau] +\rnew \aphi[d\tau \otimes d\phi \nonumber\\
&+ d\phi \otimes d\tau]\label{g_newcoords}
\end{align}
and the inverse metric $\gdual \in \bigotimes^{[ \mathds{V}, \mathds{V}]} \man$ is given by
\begin{align}
 \gdual=&\frac{\ccon^2 \sin^2(\theta)+\atheta^2 \sin^2(\theta)+\aphi^2}{\sin^2(\theta) \alpha^2} \Big(\frac{\partial}{\partial \rnew }\otimes\frac{\partial}{\partial \rnew }\Big)+\frac{1}{\rnew ^2}\Big(\frac{\partial}{\partial \theta}\otimes\frac{\partial}{\partial \theta}\Big)+\frac{1}{\rnew ^2 \sin^2\theta}\Big(\frac{\partial}{\partial \phi}\otimes\frac{\partial}{\partial \phi}\Big)\nonumber\\
&+\frac{1}{\alpha} \Big(\frac{\partial}{\partial \tau} \otimes \frac{\partial}{\partial \rnew } + \frac{\partial}{\partial \rnew } \otimes \frac{\partial}{\partial \tau}\Big)
-\frac{\atheta}{\alpha \rnew}\Big( \frac{\partial}{\partial \rnew}\otimes  \frac{\partial}{\partial \theta}+  \frac{\partial}{\partial \theta}\otimes  \frac{\partial}{\partial \rnew}\Big)\nonumber \\
&-\frac{\aphi}{\alpha \rnew \sin^2(\theta)}\Big(\frac{\partial}{\partial \rnew}\otimes\frac{\partial}{\partial \phi}+\frac{\partial}{\partial \phi}\otimes\frac{\partial}{\partial \rnew}\Big)
\end{align}
Where $\alpha$ is defined by (\ref{def_alpha}) and $\atheta$ and  $\aphi$ are the derivatives of $\alpha$ with respect to $\theta$ and $\phi$ respectively.  Let $z^0=\tau,\quad z^1=\rnew,\quad z^2=\theta,\quad z^3=\phi$, then the matrices $G'=G'_{ab}=\g(\partial_{z^a},\partial_{z^b})$ and $G'^{-1}=G'^{-1}_{ab}=\gdual(d z^a, d z^b)$ are given by

\[G'=g(\partial_{z^a},\partial_{z^b})= \left( \begin{array}{cccc}\displaystyle
-\ccon^2 &\displaystyle \alpha &\displaystyle \rnew\atheta &\displaystyle \rnew\aphi \\
\displaystyle \alpha & 0 & 0 & 0 \\
\displaystyle \rnew\atheta & 0 &\displaystyle \rnew^2 &0\\
\displaystyle \rnew\aphi & 0 & 0 &\displaystyle\rnew^2 \sin^2\theta \end{array} \right)\]

\[G'^{-1}= \left( \begin{array}{cccc}
0 &\displaystyle \frac{1}{\alpha} & 0 & 0 \\
\displaystyle\frac{1}{\alpha} &\displaystyle \frac{ \ccon^2\sin^2(\theta)+\atheta^2 \sin^2(\theta)+\aphi^2}{ \sin^2(\theta) \alpha^2} &\displaystyle -\frac{\atheta}{\alpha \rnew} &\displaystyle -\frac{\aphi}{\alpha \rnew sin^2(\theta)} \\
0 & \displaystyle-\frac{\atheta}{\alpha \rnew} &\displaystyle \frac{1}{\rnew^2} &0\\
0 &\displaystyle -\frac{\aphi}{\alpha \rnew \sin^2(\theta)} & 0 &\displaystyle \frac{1}{\rnew^2 \sin^2\theta} \end{array} \right)\label{gii}\]

\end{lemma}
\begin{proofs}
\begin{align}
g=-d y^0 \otimes d y^0 +d y^1\otimes d y^1 +d y^2 \otimes d y^2 + d y^3 \otimes d y^3
\end{align}
\begin{flushleft}
$d y^0=\cdot^0(\tau)d\tau +d\rnew$\\
$d y^1=\cdot^1(\tau)d\tau  + \sin(\theta)\cos(\phi) d\rnew + \rnew \cos(\theta)\cos(\phi)d \theta - \rnew \sin(\theta)\sin(\phi)d\phi$\\
$d y^2=\cdot^2{\tau} d\tau + \sin(\theta)\sin(\phi) d\rnew + \rnew \cos(\theta)\sin(\phi)d \theta + \rnew \sin(\theta)\cos(\phi)d\phi$\\
$d y^3=\cdot^3{\tau} d\tau + \cos(\theta) d\rnew - \rnew\sin(\theta) d\theta$\\[1cm]
\end{flushleft}
Thus
\begin{align*}
g &= -\ccon^2 d\tau \otimes d\tau +\rnew^2 d\theta \otimes d\theta + \rnew^2 \sin^2\theta d\phi \otimes d\phi\\
&+(-\cdot^0 +\cdot^1 \sin\theta \cos\phi + \cdot^2 \sin\theta \sin\phi +\cdot^3 \cos\theta)[d\tau \otimes d\rnew + d\rnew \otimes d\tau]\\
&+ (\cdot^1 \rnew \cos\theta \cos\phi + \cdot^2 \rnew \cos\theta \sin\phi -\cdot^3 \rnew \sin \theta)[d\tau \otimes d\theta + d\theta \otimes d\tau]\\
&+(\cdot^2 \rnew \sin\theta \cos\phi - \cdot^1 \rnew \sin\theta \sin\phi) [d\tau \otimes d\phi + d\phi \otimes d\tau]\\[10pt]
&=-\ccon^2 d\tau \otimes d\tau +\rnew^2 d\theta \otimes d\theta + \rnew^2 \sin^2\theta) d\phi \otimes d\phi\\
&+ \alpha[d\tau \otimes d\rnew + d\rnew \otimes d\tau]+\rnew\atheta [d\tau \otimes d\theta + d\theta \otimes d\tau] +\rnew\aphi[d\tau \otimes d\phi + d\phi \otimes d\tau]
\end{align*}\\[3pt]

$\gdual$ follows from the matrix $(\ref{gii})$
\end{proofs}
\begin{corrol}
\begin{align}
\widetilde{d\tau}&=\frac{1}{\alpha}\partial_{\rnew}\notag\\
\widetilde{d \rnew}&=\frac{\ccon^2 \sin^2(\theta)+\atheta^2 \sin^2(\theta)+\aphi^2}{\sin^2(\theta) \alpha^2}\partial_\rnew +\frac{1}{\alpha}\partial_\tau -\frac{\atheta}{\alpha \rnew}\partial_\theta -\frac{\aphi}{\alpha \rnew sin^2(\theta)}\partial\phi\notag\\
\widetilde{d\theta} &=\frac{1}{\rnew^2}\partial_\theta -\frac{\atheta}{\alpha\rnew}\partial_\rnew\notag\\
\widetilde{d\phi} &=\frac{1}{\rnew^2 \sin^2\theta}\partial_\phi -\frac{\aphi}{\alpha \rnew \sin^2(\theta)}\partial_\rnew
\end{align}
\end{corrol}
\begin{proofs}
follows from definition of $\gdual$.
\end{proofs}
\begin{lemma}
The 1-forms $\xdual, \vdual \in \Gamma \Lambda^1 \man$ are given by
\begin{align}
&\widetilde{X}= \rnew \alpha d\tau\\
&\vdual = -\epsilon^2d\tau + \alpha d\rnew + \rnew\atheta d\theta + \rnew\aphi d\phi
\end{align}
\end{lemma}
\begin{proofs}
\begin{align*}
\widetilde{X} &= \rnew g (\frac{\partial}{\partial \rnew}, -)\\
&=\rnew (-\cdot^0 +\cdot^1 \sin\theta \cos\phi + \cdot^2 \sin\theta \sin\phi +\cdot^3 \cos\theta) d\tau\\
&= \rnew \alpha d\tau\\
\vdual &=  g (\frac{\partial}{\partial \tau}, -)\\
&= -c^2 d\tau+ \alpha d\rnew + \rnew\atheta d\theta + \rnew\aphi d\phi
\end{align*}
\end{proofs}

\begin{lemma}
\begin{align}
\star 1 = -\alpha \rnew^2 \sin \theta d\tau \wedge d\rnew \wedge d\theta \wedge d\phi
\end{align}
\end{lemma}
\begin{proofs}
\begin{align*}
\star 1 &= \sqrt{|\det(g)|} d\tau\wedge d \rnew \wedge d\theta \wedge d\phi\\
&=\sqrt{|-\alpha^2 \rnew^4 \sin^2{\theta}|}d\tau\wedge d \rnew \wedge d\theta \wedge d\phi\\
&=-\alpha \rnew^2 \sin \theta d\tau \wedge d\rnew \wedge d\theta \wedge d\phi
\end{align*}
\end{proofs}
\begin{lemma}
\begin{align}
\adual = \dot{\alpha} d\rnew + \rnew\dot{\atheta}d\theta +\rnew\dot{\aphi} d\phi
\end{align}
\end{lemma}
\begin{proofs}
\begin{align*}
\adual &= \frac{d \v_a}{d\tau}d y^a\\
&= -\cddot^0(\tau)  d y^0 + \cddot^1(\tau) d y^1 + \cddot^2(\tau)d y^2 +\cddot^3(\tau)d y^3\\
&=-\cddot^0(\tau)\Big[\cdot^0(\tau) d\tau + d\rnew\Big] \\
  &\quad+\cddot^1(\tau)\Big[\cdot^1(\tau) d\tau + \sin(\theta)\cos(\phi) d\rnew + \rnew \cos(\theta)\cos(\phi)d \theta - \rnew \sin(\theta)\sin(\phi)d\phi\Big]\\
 &\quad+\cddot^2(\tau)\Big[\cdot^2(\tau) d\tau + \sin(\theta)\sin(\phi) d\rnew + \rnew \cos(\theta)\sin(\phi)d \theta + \rnew \sin(\theta)\cos(\phi)d\phi\Big]\\
 &\quad+\cddot^3(\tau)\Big[\cdot^3(\tau) d\tau + \cos(\theta) d\rnew  - \rnew \sin(\theta) d\theta\Big] \\
 &=g(\a, \v) d\tau + \dot{\alpha} d\rnew  + \rnew \dot{\atheta}d\theta +\rnew \dot{\aphi} d\phi\\
 &= \dot{\alpha} d\rnew  + \rnew \dot{\atheta}d\theta +\rnew \dot{\aphi} d\phi
\end{align*}
\end{proofs}
\begin{lemma}
\begin{align}
\a&=\frac{\dot{\alpha}}{\alpha}\partial_\tau +\Bigg[\Big(\frac{\ccon^2 \dot{\alpha}}{\alpha^2}\Big) +\Big( \frac{\dot{\alpha} \atheta^2}{\alpha^2} -\frac{ \atheta \atd}{ \alpha} \Big)+ \frac{1}{\sin^2(\theta)}\Big(\frac{\aphi^2\dot{\alpha}}{\alpha^2}-\frac{\aphi\apd}{\alpha}\Big)\Bigg]\partial_\rnew \nonumber\\
&+\frac{1}{\rnew }\Big(\atd-\frac{\dot{\alpha}\atheta}{ \alpha}\Big)\partial_\theta + \frac{1}{\rnew  \sin^2(\theta)}\Big(\apd-\frac{\dot{\alpha} \aphi}{\alpha }\Big)\partial_\phi
\end{align}
\end{lemma}
\begin{proofs}
\begin{align*}
\a &= \gdual(\adual, -)\\
&=g(\a, \v)\gdual(d\tau, -) +\dot{\alpha} \gdual (d\rnew , -) +\rnew \atd \gdual(d\theta, -) + \rnew \apd \gdual(d\phi, -)\\
&=\frac{g(\a, \v)}{\alpha}\partial_\rnew  +\dot{\alpha}\Bigg(\frac{\ccon^2 \sin^2(\theta)+\atheta^2 \sin^2(\theta)+\aphi^2}{ \sin^2(\theta) \alpha^2}\partial_\rnew  +\frac{1}{\alpha}\partial_\tau -\frac{\atheta}{\alpha \rnew }\partial_\theta - \frac{\aphi}{\alpha \rnew  \sin^2(\theta)}\partial_\phi\Bigg)\\
&+\rnew \atd\Big(\frac{1}{\rnew ^2}\partial_\theta - \frac{\atheta}{\alpha \rnew }\partial_\rnew \Big) +\rnew \apd\Big(\frac{1}{\rnew ^2 \sin^2(\theta)}\partial_\phi - \frac{\aphi}{\alpha \rnew  \sin^2(\theta)}\partial_\rnew \Big)\\
&=\frac{\dot{\alpha}}{\alpha}\partial_\tau +\Bigg[\Big(\frac{g(\a, \v)}{\alpha} + \frac{\ccon^2 \dot{\alpha}}{\alpha^2}\Big) +\Big( \frac{\dot{\alpha} \atheta^2}{\alpha^2} -\frac{ \atheta \atd}{ \alpha} \Big)+ \frac{1}{\sin^2(\theta)}\Big(\frac{\aphi^2\dot{\alpha}}{\alpha^2}-\frac{\aphi\apd}{\alpha}\Big)\Bigg]\partial_\rnew \\
&+\frac{1}{\rnew }\Big(\atd-\frac{\dot{\alpha}\atheta}{ \alpha}\Big)\partial_\theta + \frac{1}{\rnew  \sin^2(\theta)}\Big(\apd-\frac{\dot{\alpha} \aphi}{\alpha }\Big)\partial_\phi
\end{align*}
\end{proofs}
\begin{lemma}
\begin{align}
g(X, \v)&=\rnew  \alpha\\
g(X, \a)&=\rnew \dot{\alpha}
\end{align}
\end{lemma}
\begin{proofs}
\begin{align*}
g(X, \v) &= g(\rnew  \frac{\partial}{\partial \rnew },  \frac{\partial}{\partial \tau})\\
&=\rnew  g(\frac{\partial}{\partial \rnew },  \frac{\partial}{\partial \tau})\\
&=\rnew (-\cdot^0 +\cdot^1 \sin\theta \cos\phi + \cdot^2 \sin\theta \sin\phi +\cdot^3 \cos\theta)\\
&=\rnew  \alpha\\
\end{align*}
For $g(X, \a)$ the only relevant term in the metric is $\alpha d\rnew  \otimes d\tau$, thus
\begin{align*}
g(X, A) &= \rnew  \frac{\dot{\alpha}}{\alpha}g(\partial_\rnew , \partial_\tau)\\
&= \rnew  \dot{\alpha}
\end{align*}
\end{proofs}
\begin{lemma}
\begin{align}
&\alw=-\qe\Big(\frac{\ccon^2}{\alpha \rnew }d\tau + \frac{1}{\rnew }d\rnew  + \frac{\atheta}{\alpha}d\theta + \frac{\aphi}{\alpha}d\phi\Big)\\
&\flw_{\textup{R}}=\qe\frac{(\alpha\atd-\dot{\alpha}\atheta)d\tau \wedge d\theta + (\alpha\apd-\dot{\alpha}\aphi)d\tau \wedge d\phi}{\alpha^2}\\
&\flw_{\textup{C}}=-\qe\ccon^2\Big( \frac{\alpha}{\rnew ^2}d\tau \wedge d\rnew  + \frac{\atheta}{\rnew \alpha^2}d\tau \wedge d\theta + \frac{\aphi}{\rnew \alpha^2}d\tau \wedge d\phi\Big)
\end{align}
\end{lemma}
\begin{proofs}
(53) follows directly from (20) (47) (51)
\[\flw_{\textup{R}} = \qe \frac{g(X, \v)\tilde{X}\wedge \adual-g(X, \a)\tilde{X} \wedge \vdual}{g(X, \v)^3}\]
Using the relations:
\begin{align}
\tilde{X} \wedge \vdual&=\Big(\rnew  \alpha d\tau\Big)\wedge\Big(-\ccon^2d\tau + \alpha d\rnew  + \rnew \atheta d\theta + \rnew \aphi d\phi\Big)\nonumber\\
&= \rnew \alpha^2 d\tau \wedge d\rnew  + \rnew ^2 \alpha \atheta d\tau \wedge d\theta + \rnew ^2 \alpha \aphi d\tau \wedge d\phi\\
\tilde{X} \wedge \adual&=\Big(\rnew  \alpha d\tau\Big)\wedge\Big(g(\a, \v) d\tau + \dot{\alpha} d\rnew  + \rnew \dot{\atheta}d\theta +\rnew \dot{\aphi} d\phi\Big)\nonumber\\
&=\rnew \alpha \dot{\alpha} d\tau \wedge d\rnew  + \rnew ^2 \alpha \atd d\tau \wedge d\theta + \rnew ^2 \alpha \apd d\tau \wedge d\phi
\end{align}
along with (51)(52)gives
\begin{align*}
\flw_{\textup{R}} &=\qe\frac{\alpha \rnew  (\rnew \alpha \dot{\alpha} d\tau \wedge d\rnew  + \rnew ^2 \alpha \atd d\tau \wedge d\theta + \rnew ^2 \alpha \apd d\tau \wedge d\phi)}{(\alpha \rnew )^3}\\
&-\qe\frac{\rnew \dot{\alpha}(\rnew \alpha^2 d\tau \wedge d\rnew  + \rnew ^2 \alpha \atheta d\tau \wedge d\theta + \rnew ^2 \alpha \aphi d\tau \wedge d\phi)}{(\alpha \rnew )^3}\\
&=\qe\frac{1}{\alpha^2}\Big((\alpha\atd-\dot{\alpha}\atheta)d\tau \wedge d\theta + (\alpha\apd-\dot{\alpha}\aphi)d\tau \wedge d\phi\Big)
\end{align*}
\begin{align*}
\flw_{\textup{C}} &=-\qe\frac{\ccon^2 \tilde{X}\wedge \vdual}{g(X, \v)^3} \\
&=-\qe\frac{\ccon^2 (\rnew \alpha^2 d\tau \wedge d\rnew  + \rnew ^2 \alpha \atheta d\tau \wedge d\theta + \rnew ^2 \alpha \aphi d\tau \wedge d\phi)}{(\alpha \rnew )^3}\\
&=-\qe\ccon^2\Big( \frac{1}{\alpha \rnew ^2}d\tau \wedge d\rnew  + \frac{\atheta}{\rnew \alpha^2}d\tau \wedge d\theta + \frac{\aphi}{\rnew \alpha^2}d\tau \wedge d\phi\Big)
\end{align*}
\end{proofs}
\begin{lemma}
\begin{align}
&\star \flw_{\textup{R}}=\qe\Big(\frac{\aphi \dot{\alpha}-\alpha \apd}{\alpha^2 \sin(\theta)}d\tau \wedge d\theta - \frac{\sin(\theta)(\atheta \dot{\alpha} - \alpha \atd)}{\alpha^2}d\tau \wedge d\phi\Big)\\
&\star \flw_{\textup{C}}=\qe\frac{\ccon^2 \sin(\theta)}{\alpha^2}d\theta \wedge d\phi
\end{align}
\end{lemma}
\begin{proofs}
\begin{align}
\star (\tilde{X} \wedge \vdual) &=i_\v i_{X} \star 1\nonumber\\
&=\rnew  i_{\partial_\tau} i_{\partial_\rnew } (\alpha \rnew ^2 \sin \theta d\tau \wedge d\rnew  \wedge d\theta \wedge d\phi )\nonumber\\
&=\rnew ^3 \alpha \sin(\theta)d\theta \wedge d\tau\\
\star (\tilde{X} \wedge \adual) &=i_\a i_{X} \star 1\nonumber\\
&=\rnew ^2 \sin(\theta) (\alpha \atd - \atheta \dot{\alpha})d\tau \wedge d\phi - \frac{\rnew ^2 (\alpha \apd - \aphi \dot{\alpha})}{\sin(\theta)}d\tau \wedge d\theta\nonumber\\
&-\rnew ^3 \dot{\alpha} \sin(\theta) d\theta \wedge d\phi
\end{align}
therefore
\begin{align*}
\star\flw_{\textup{C}} &= \qe \frac{-\ccon^2\star(\tilde{X}\wedge \vdual)}{g(X, \v)^3}\\
&=\qe\frac{\ccon^2 \sin{\theta}}{\alpha^2}d\theta\wedge d\phi \\
\star\flw_{\textup{R}}&=\qe \frac{g(X, \v)\star(\tilde{X}\wedge \adual)-g(X, \a)\star(\tilde{X} \wedge \vdual)}{g(X, \v)^3}\\
&=\qe\Big(\frac{\aphi \dot{\alpha}-\alpha \apd}{\alpha^2 \sin(\theta)}d\tau \wedge d\theta - \frac{\sin(\theta)(\atheta \dot{\alpha} - \alpha \atd)}{\alpha^2}d\tau \wedge d\phi\Big)
\end{align*}
\end{proofs}
\begin{lemma}
The couloumbic and radiative terms of the 1-forms $\elwd$ and $\blwd$ take the form
\begin{align}
\elwd_{\textup{C}}=&\qe \ccon^2 \Big(\frac{\aphi }{\rnew  \alpha^3}d\phi+\frac{1}{\alpha^2 \rnew ^2}d\rnew +\frac{\atheta }{\rnew \alpha^3}d\theta+\frac{\cdot^0+\alpha }{\alpha^2 \rnew ^2}d\tau\notag\\
&+\frac{\atheta^2  }{\rnew ^2 \alpha^3}d\tau+\frac{\aphi^2  }{\rnew ^2 \alpha^3 \sin^2(\theta)}d\tau\Big)\notag\\
\elwd_{\textup{R}}=& \qe\Bigg(\frac{\dot{\alpha}\aphi-\alpha \apd}{\alpha^3}d\phi- \frac{\alpha \atd-\dot{\alpha} \atheta }{\alpha^3}d\theta\notag\\
&-  \Big(\frac{\atheta (\alpha \atd-\dot{\alpha} \atheta)}{\rnew \alpha^3}-\frac{\aphi (\dot{\alpha}\aphi-\alpha\apd)}{\rnew \alpha^3 \sin^2(\theta)}\Big) d\tau\Bigg)
\label{e_newcoords}
\end{align}
and
\begin{align}
\blwd_{\textup{C}}&= \qe \ccon \Big(\frac{\atheta \sin(\theta)}{\rnew \alpha^3}d\phi+\frac{\cdot^1 \sin(\phi)-\cdot^2 \cos(\phi)}{\rnew  \alpha^3}d\theta\Big)\notag\\
&=\qe \ccon \sin(\theta)\Big(\frac{\atheta}{\rnew \alpha^3}d\phi-\frac{\aphi}{\rnew  \alpha^3 \sin^2(\theta)}d\theta\Big)\notag\\
\blwd_{\textup{R}}=&\frac{1}{\ccon}\qe\sin(\theta)\Bigg(\frac{\dot{\alpha} \atheta-\alpha \atd  }{\alpha^3}d\phi+\frac{\alpha \apd-\dot{\alpha} \aphi}{\alpha^3 \sin^2(\theta)}d\theta\notag\\
 &+\Big(\frac{\atheta (\alpha \apd-\dot{\alpha} \aphi) }{\rnew \alpha^3\sin^2(\theta)}
 +\frac{\aphi(\dot{\alpha} \atheta-\alpha \atd)}{\rnew \alpha^3 \sin^2(\theta)}\Big)d\tau\Bigg)\notag\\
 &=\frac{1}{\ccon}\qe\sin(\theta)\Bigg(\frac{\dot{\alpha} \atheta-\alpha \atd  }{\alpha^3}d\phi+\frac{\alpha \apd-\dot{\alpha} \aphi}{\alpha^3 \sin^2(\theta)}d\theta\notag\\
 &+\alpha \frac{\atheta  \apd- \aphi  \atd}{\rnew \alpha^3\sin^2(\theta)}
 d\tau\Bigg)
\end{align}
\end{lemma}
\begin{proofs}
\begin{align}
\frac{\partial}{\partial y^0}&=\frac{\partial \tau}{\partial y^0}\frac{\partial}{\partial \tau}+\frac{\partial \rnew }{\partial y^0}\frac{\partial}{\partial \rnew }+\frac{\partial\theta}{\partial y^0}\frac{\partial}{\partial \theta}+\frac{\partial \phi}{\partial y^0}\frac{\partial}{\partial \phi}\notag\\
&=\frac{c}{\alpha} \Big(-\frac{\partial}{\partial \tau}+(\cdot^0 +\alpha)\frac{\partial}{\partial \rnew}+\frac{\atheta}{r}\frac{\partial}{\partial \theta}+\frac{\aphi}{\rnew\sin^2(\theta)}\frac{\partial}{\partial \phi}\Big)
\end{align}
Results follow from definitions (\ref{eb_def}) and (\ref{bvecdef}).
\end{proofs}

\chapter{MAPLE Input for Part I}
\label{maple_ptwo}
In this thesis the we use the mathematical software MAPLE to implement the computations which support the results presented in parts I and II. In principle there are other programming tools which could have be used, such as MATHEMATICA and MATLAB, each of which has its own advantages and disadvantages. In general, MATHEMATICA and MAPLE are more suited to symbolic computation, whereas MATLAB is more suited to numerical computation.

 In part I of the thesis we require heavy use of symbolic computation. In particular we utilize the tools of differential geometry to manipulate tensors and differential forms. These  tools were readily available to us in MANIFOLDS package \cite{Manifolds} written by Robin Tucker and Charles Wang for use with MAPLE. There are similar packages available for use with other software, such as RICCI for use with MATHEMATICA, and Tensor Toolbox for use with MATLAB, however the availability of the MANIFOLDS package and supporting documentation was an important factor in deciding to use MAPLE instead of other possible programming tools. In addition, the procedural language of MAPLE was appealing to the author based on his experience with $\textup{C}^{++}$ and FORTRAN programming languages.
 
 The calculations carried out for part II of the thesis are more numerical by nature, however rather than adopting a programming tool more suited to numerical calculations we decided it would be more economical to build on the code already written in MAPLE.

The following script was written in MAPLE 15 and can be run with the packages \emph{Plots}, \emph{LinearAlgebra} and the additional package \emph{Manifolds}\cite{Manifolds} with tools for differential geometry.
\begin{linenumbers*}
{\color{red}
\begin{singlespacing}
\begin{flushleft}
\texttt{\# set up coordinate system}\\
\texttt{Manifoldsetup(M,[tau,R,theta,phi],[e,E,0],}\\
\texttt{map(x->simplify(x,symbolic),}\\
\texttt{[e[0]=d(tau),}\\
\texttt{e[1]=d(R),}\\
\texttt{e[2]=d(theta),}\\
\texttt{e[3]=d(phi)])):}
\end{flushleft}
\end{singlespacing}}
\begin{mapcode}
Constants([epsilon, q, ep, b0, b1, b2, b3, a3, R0]);\\
Manfdomain(M, [a, ad, ath, aph, athd, aphd, adphph, adthth], [tau, theta, phi]):\\
Manfdomain(M,[C0,C1,C2,C3,Cd0,Cd1,Cd2,Cd3,Cdd0,Cdd1,Cdd2,Cdd3],[tau]) :\\[5pt]
  g := (-1+2\ask R\ask ad/a)\ask d(tau) \&X d(tau)\\
- (d(tau) \&X d(R)+ d(R) \&X d(tau))\\
+ (R\ind2/a\ind2)\ask(d(theta) \&X d(theta))\\
+ (R\ind2/a\ind2)\ask sin(theta)\ind2 \ask d(phi) \&X d(phi) :\\[5pt]
Mancovmetric(M,g):\\
Manvol(M) := -(R\ind2/a\ind2)\ask sin(theta)\ask`\&\ind`(e[0], e[1], e[2], e[3]) :\\[3pt]
  Basis1 := \{d(tau),d(R),d(theta),d(phi)\} :\\
Basis2 := \{d(tau)\&\ind d(R), d(tau)\&\ind d(theta), d(tau)\&\ind d(phi),\\
           d(R)\&\ind d(theta), d(R)\&\ind d(phi), d(theta)\&\ind d(phi)\} :\\
Basis3 := \{d(tau)\&\ind d(R)\&\ind d(theta), d(tau)\&\ind d(R)\&\ind d(phi),\\
           d(tau)\&\ind d(theta)\&\ind d(phi), d(R)\&\ind d(theta)\&\ind d(phi)\} :\\
Basis4 := \{e(0) \&\ind e(1) \&\ind e(2), e(1) \&\ind e(2) \&\ind e(3),e(2) \&\ind e(3) \&\ind e(0),e(3) \&\ind e(0) \&\ind e(1)\}:  \\[3pt]
a\_sublist:=\{diff(a,tau)=ad,diff(a,theta)=ath,diff(a,phi)=aph,\\
           diff(ath,tau)=athd,diff(aph,tau)=aphd,\\
           diff(ath,phi)=athph,diff(aph,theta)=aphth, diff(ad, theta)=athd, diff(ad, phi)=aphd, diff(aph, phi)=aphph,\\
           diff(ath, theta)=athth, diff(adph, phi)=adphph, diff(adth, theta)=adthth\}:\\[5pt]
Cd\_sublist := \{diff(C0,tau)=Cd0,diff(C1,tau)=Cd1,\\
	       diff(C2,tau)=Cd2,diff(C3,tau)=Cd3\} :\\
Cd\_inv\_sublist := \{Cd0=diff(C0,tau),Cd1=diff(C1,tau),\\
	       	  Cd2=diff(C2,tau),Cd3=diff(C3,tau)\} :\\
Cdd\_sublist := \{diff(C0,tau,tau)=Cdd0,diff(C1,tau,tau)=Cdd1,\\
	       diff(C2,tau,tau)=Cdd2,diff(C3,tau,tau)=Cdd3, diff(Cd0,tau)=Cdd0,diff(Cd1,tau)=Cdd1,\\
               diff(Cd2,tau)=Cdd2,diff(Cd3,tau)=Cdd3\} :\\
Cddd\_sublist := \{ diff(Cdd0,tau)=Cddd0,diff(Cdd1,tau)=Cddd1,\\
diff(Cdd2,tau)=Cddd2,diff(Cdd3,tau)=Cddd3\} :\\[5pt]
aa := -Cd0+Cd1\ask cos(phi)\ask sin(theta)+Cd2\ask sin(phi)\ask sin(theta)+Cd3\ask cos(theta) :\\
aath := diff(aa,theta) :\\
aaph := diff(aa,phi) :\\
aad := subs(Cdd\_sublist,diff(subs(Cd\_inv\_sublist,aa),tau)) :\\
aathd := subs(Cdd\_sublist,diff(subs(Cd\_inv\_sublist,aath),tau)) :\\
aaphd := subs(Cdd\_sublist,diff(subs(Cd\_inv\_sublist,aaph),tau)) :\\
aathth:=subs(Cdd\_sublist,diff(subs(Cd\_inv\_sublist,aath),theta)) :\\
aaphph:=subs(Cdd\_sublist,diff(subs(Cd\_inv\_sublist,aaph),phi)) :\\
aadphph:=subs(Cdd\_sublist,diff(diff(subs(Cd\_inv\_sublist,aad),phi), phi)) :\\
aadthth:=subs(Cdd\_sublist,diff(diff(subs(Cd\_inv\_sublist,aad),theta), theta)) :\\[5pt]
aa\_sublist:=\{a=aa, ath=aath, ad=aad, aph=aaph, athd=aathd, aphd=aaphd, athth=aathth,\\
            aphph=aaphph, adphph=aadphph, adthth=aadthth\}:\\[5pt]
x0 := C0-(R/a):\\
x1 := C1-(R/a)\ask sin(theta)\ask cos(phi):\\
x2 := C2-(R/a)\ask sin(theta)\ask sin(phi):\\
x3 := C3-(R/a)\ask cos(theta):\\[5pt]
 J := Matrix(4, 4):\\
for i from 0 to 3 do J[i+1, 1] := diff(x || i, tau):\\
J[i+1, 2] := diff(x || i, R):\\
J[i+1, 3] := diff(x || i, theta):\\
J[i+1, 4] := diff(x || i, phi) end do:\\
subs(a\_sublist, J):\\[5pt]
DetJ := simplify(Determinant(J)):\\
detJ := (R\ind2/a\ind2)\ask sin(theta) :\\[5pt]
  AdJ := simplify(eval(subs( a\_sublist, Adjoint(J)))):\\[5pt]
df\_tau\_0 :=AdJ[1, 1]/detJ :\\
df\_tau\_1 := AdJ[1, 2]/detJ :\\
df\_tau\_2 := AdJ[1, 3]/detJ :\\
df\_tau\_3 := AdJ[1, 4]/detJ :\\
df\_R\_0   := AdJ[2, 1]/detJ :\\
df\_R\_1   := AdJ[2, 2]/detJ :\\
df\_R\_2   := AdJ[2, 3]/detJ :\\
df\_R\_3   := AdJ[2, 4]/detJ :\\
df\_theta\_0  := AdJ[3, 1]/detJ :\\
df\_theta\_1 := AdJ[3, 2]/detJ :\\
df\_theta\_2 := AdJ[3, 3]/detJ :\\
df\_theta\_3 := AdJ[3, 4]/detJ :\\
df\_phi\_0 := AdJ[4, 1]/detJ :\\
df\_phi\_1 := AdJ[4, 2]/detJ :\\
df\_phi\_2 := AdJ[4, 3]/detJ :\\
df\_phi\_3 := AdJ[4, 4]/detJ :\\[3pt]
PD\_0:=df\_tau\_0\ask PD(tau)+df\_R\_0\ask PD(R) +df\_theta\_0\ask PD(theta) +df\_phi\_0\ask PD(phi):\\
PD\_1:=df\_tau\_1\ask PD(tau)+df\_R\_1\ask PD(R) +df\_theta\_1\ask PD(theta) +df\_phi\_1\ask PD(phi):\\
PD\_2:=df\_tau\_2\ask PD(tau)+df\_R\_2\ask PD(R) +df\_theta\_2\ask PD(theta) +df\_phi\_2\ask PD(phi):\\
PD\_3:=df\_tau\_3\ask PD(tau)+df\_R\_3\ask PD(R) +df\_theta\_3\ask PD(theta) +df\_phi\_3\ask PD(phi):\\[5pt]
 VX := R\ask PD(R) ;\\
 dualX := F2C(\&~(VX)) ;\\
 VV := PD(tau)+VX\ask (ad/a) ;\\
 dualV := collect(F2C(\&~(VV)), Basis1);\\
dualA  :=collect(expand(R\ask (ad\ind2/a\ind2)\ask d(tau) + (-ad/a)\ask d(R) +R\ask ((ad\ask ath)/a\ind2-athd/a)\ask d(theta)+ R\ask ((ad\ask aph)/a\ind2-aphd/a)\ask d(phi)),    Basis1);\\
VA :=collect(expand(F2C( \&~(dualA))), Basis6) ;\\[5pt]
ALW := collect(expand(F2C(dualV/(-R))), Basis1) ;\\
FLW := collect(expand(subs(a\_sublist, d(ALW))), Basis2) ;\\
starFLW:=collect(F2C(\&i (\&star(FLW))), Basis2);\\[5pt]
  stress:=proc(kill);\\
 collect(subs(Cd\_sublist,F2C(((ep/2)\ask ((PD\_||kill \&i FLW) \&\ind(\&star FLW)-(PD\_||kill \&i(\&star FLW))\&\ind FLW)))), Basis3);\\
end proc:\\[5pt]
stress\_0:=stress(0):\\
stress\_1:=stress(1):\\
stress\_2:=stress(2):\\
stress\_3:=stress(3):\\[5pt]
  expansion\_cdot\_sublist:=\{epsilon=1, Cd0=1+(b0\ask tau\ind2/2)+O(tau\ind3), Cd1=(b1\ask tau\ind2/2)+O(tau\ind3), Cd2=(b2\ask tau\ind2/2)+O(tau\ind3), \\ Cd3=a3\ask tau+(b3\ask tau\ind2/2)+O(tau\ind3), Cdd0=b0\ask tau+O(tau\ind2), Cdd1=b1\ask tau+O(tau\ind2), Cdd2=b2\ask tau+O(tau\ind2), Cdd3=a3+b3\ask tau+O(tau\ind2)\};\\[3pt]
  S\_k\_cdot:=proc(sublist, kill)\\
          local spl;\\
spl:=stress\_||kill;\\
subs(sublist, subs(aa\_sublist, collect(expand(subs(Cd1\ask cos(phi)\ask sin(theta)\\
+Cd2\ask sin(phi)\ask sin(theta)+Cd3\ask cos(theta)=a+Cd0,\\
-Cd1\ask cos(phi)\ask sin(theta)-Cd2\ask sin(phi)\ask sin(theta)\\
-Cd3\ask cos(theta)=-a-Cd0,-Cd1\ask cos(phi)\ask cos(theta)\\
-Cd2\ask sin(phi)\ask cos(theta)+Cd3\ask sin(theta)=-ath,Cd1\ask cos(phi)\ask cos(theta)\\
+Cd2\ask sin(phi)\ask cos(theta)-Cd3\ask sin(theta)=ath,\\
-Cd1\ask sin(phi)+Cd2\ask cos(phi)=aph/sin(theta),Cd1\ask sin(phi)\\
-Cd2\ask cos(phi)=-aph/sin(theta), Cd\_sublist,spl)), Basis3)));\\
end proc:\\[5pt]
  intgrd\_0:= series(coeff(S\_k\_cdot(expansion\_cdot\_sublist, 0), `\&\ind`(d(tau), d(theta), d(phi))), tau=0):\\
intgrd\_1:= series(coeff(S\_k\_cdot(expansion\_cdot\_sublist, 1), `\&\ind`(d(tau), d(theta), d(phi))), tau=0):\\
intgrd\_2:= series(coeff(S\_k\_cdot(expansion\_cdot\_sublist, 2), `\&\ind`(d(tau), d(theta), d(phi))), tau=0):\\
intgrd\_3:= series(coeff(S\_k\_cdot(expansion\_cdot\_sublist, 3), `\&\ind`(d(tau), d(theta), d(phi))), tau=0):\\[5pt]
 get\_integrands:= proc();\\
 print(t, intgrd\_0);\\
 print(x, intgrd\_1);\\
 print(y, intgrd\_2);\\
 print(z, intgrd\_3);\\
 end proc:\\[5pt]
  get\_integrals:= proc();
 print(t, factor(simplify(int(int(int(intgrd\_0, phi=0..2\ask   Pi), theta=0..Pi), tau))));\\
 print(x, factor(simplify(int(int(int(intgrd\_1, phi=0..2\ask   Pi), theta=0..Pi), tau))));\\
 print(y, factor(simplify(int(int(int(intgrd\_2, phi=0..2\ask   Pi), theta=0..Pi), tau))));\\
 print(z, factor(simplify(int(int(int(intgrd\_3, phi=0..2\ask   Pi), theta=0..Pi), tau))));\\
 end proc:\\
 get\_integrals();
\end{mapcode}
\newpage
\end{linenumbers*}
\section*{Comments}
{\footnotesize \color{red} \tt 1-7} Set up the Newman-Unti coordinate system $(\tau, R, \theta, \phi)=(\mathtt{tau, R, theta, phi})$\\
{\footnotesize \color{red} \tt 8-12} The global variables are defined. For $i=0..3$ we use notation $\c^i=\mathtt{Ci}$, $\cdot^i=\mathtt{Cdi}$, $\cddot^i=\mathtt{Cddi}$. Also $\alpha=\mathtt{a}$, $\ad=\mathtt{ad}$, $\atheta=\mathtt{ath}$,$\aphi=\mathtt{aph}$, $\apd=\mathtt{aphd}$ etc. The constants $a, b^i$ defining the comoving frame are given by $\mathtt{a}$ and $\mathtt{bi}$ respectively. Also $\mathtt{q}$ and $\mathtt{ep}$ are constants.\\
{\footnotesize \color{red} \tt 13-17} The metric (\ref{g_nu_def}) is input. This associates the manifold $\mathtt{M}$ with Minkowski space $\m$. The function \texttt{Mancovmetric(M, g)} identifies \texttt{g} as the metric on \texttt{M}. The \emph{Manifolds} package will automatically give the inverse metric and the vector and covector bases on $\tan\m$ and $\tan^*\m$. Note that there is no factor of $\ccon^2$ in the metric because 
we use dimensions such that $\g(\cddot, \cddot)=-1$.  \\
{\footnotesize \color{red} \tt 18} \texttt{Manvol(M)} defines the volume $4$-form. Notice the negative orientation.\\
{\footnotesize \color{red} \tt 19-25} Define coordinate bases to simplify output\\
{\footnotesize \color{red} \tt 26-31}  These lines define the relationships between $\alpha$ and its derivatives.\\
{\footnotesize \color{red} \tt 32-41}  These lines define the relationships between the components of $\c$, $\cdot$,$\cddot$ and $\cdddot$ .\\
{\footnotesize \color{red} \tt 42-58} The here we define the parameters \texttt{aa, aad, aath, aaph, aaphd...} which assign the coordinate representations to the variables \texttt{a, ad, ath, aph, aphd...} \\
{\footnotesize \color{red} \tt 59-62} The coordinate transformation from Newman-Unti (\texttt{tau, R,theta, phi})  to Minkowski coordinates (\texttt{x0, x1, x2, x3}).\\
{\footnotesize \color{red} \tt 63-70} We determine the Jacobian \texttt{J} and its determinant. \\
{\footnotesize \color{red} \tt 71-87} We calculate the partial derivatives of the Newman-Unti coordinates with respect to the Minkowski coordinates.\\
{\footnotesize \color{red} \tt 88-95} These lines define the Minkowski basis vectors \texttt{PD\_t}=$\frac{\partial}{\partial x0}$,\texttt{PD\_x}=$\frac{\partial}{\partial x1}$, \texttt{PD\_y}=$\frac{\partial}{\partial x2}$, \texttt{PD\_z}=$\frac{\partial}{\partial x3}$  in terms of Newman-Unti coordinates.\\
{\footnotesize \color{red} \tt 96-103} Defines the vectors $\x=\mathtt{VX}$, $\v=\mathtt{VV}$, and $\a=\mathtt{VA}$ and their duals using (\ref{X_NU}) and (\ref{V_NU}) and (\ref{nu_dual_vecs}). \\
 {\footnotesize \color{red} \tt 104-106} The \LW potential $\alw=\mathtt{ALW}$ is defined using (\ref{LW_def}). The $2$-form field $\flw=\mathtt{FLW}$ may is calculated by taking the exterior derivative. This is included in the \emph{Manifolds} package. The Hodge dual is also used to define $\star \flw=\mathtt{starFLW}$ \\
 {\footnotesize \color{red} \tt 107-114} These lines define the four  stress 3-forms $\stresslw_K=$\texttt{stress\_i} for $\mathtt{i=0, 1, 2, 3}$. \\
 {\footnotesize \color{red} \tt 115-119} Defines the expansion around the momentarily comoving frame\\
 {\footnotesize \color{red} \tt 120-32} A procedure for substituting the expansion into either of the stress 3-forms and simplifying the resulting expression.\\
{\footnotesize \color{red} \tt 133-145} These lines provide the procedure \texttt{get\_integrands} for obtaining the integrands.\\
{\footnotesize \color{red} \tt 147-156} These lines provide the procedure \texttt{get\_integrals} for carry out the integration. \\
{\footnotesize \color{red} \tt 157} This calls the procedure \texttt{get\_integrals}. The result is stated in (\ref{int_Sk}).
\chapter{MAPLE Input for Part II}
\label{mapletwo}
 For the numerical investigation in Part II we use MAPLE to perform many different calculations, integrals and plots for a wide range of input parameters. As a result I have many different files with variations on the code. With hindsight I would have liked to have kept all the code in one file, beautifully annotated and ready to reproduce any calculation. However coding in MAPLE is a skill which I have learnt throughout my PhD and the code I have written hasn't always been the most simple or the most elegant. In this section I present some of the most important code which has been used to obtain the results stated in chapter \ref{bendingbeams}. Once again we use the packages \emph{Plots}, \emph{LinearAlgebra} and  \emph{Manifolds}\cite{Manifolds}.
\section{Part 1 - Setup}
\resetlinenumber[1]
\begin{linenumbers}
{\color{red}
\begin{singlespacing}
\begin{flushleft}
\texttt{\# set up coordinate system}\\
\texttt{Manifoldsetup(M,[tau,r,theta,phi],[e,E,0],}\\
\texttt{map(x->simplify(x,symbolic),}\\
\texttt{[e[0]=d(tau),}\\
\texttt{e[1]=d(r),}\\
\texttt{e[2]=d(theta),}\\
\texttt{e[3]=d(phi)])):}
\end{flushleft}
\end{singlespacing}}
\begin{mapcode}
Constants(epsilon, Lp, Rp, thetap, v, X0, Y0, Z0, q\_e, ep, mu, c):\\
Manfdomain(M,gAV) :\\
Manfdomain(M,[a,ath,aph],[tau,theta,phi]) :\\
Manfdomain(M,[ad,athd,aphd, athph, aphth],[tau,theta,phi]) :\\
Manfdomain(M,[C0,C1,C2,C3,Cd0,Cd1,Cd2,Cd3,Cdd0,Cdd1,Cdd2,Cdd3],[tau]) :\\
Manfdomain(M,[rhat, cthhat, sthhat, cphhat, sphhat, T0], [tau]):\\[5pt]
g := -\ccon\ind 2\textasteriskcentered d(tau) \& X d(tau)\\
+ a\textasteriskcentered (d(tau) \& X d(r)+ d(r) \& X d(tau))\\
+ r\textasteriskcentered  ath \textasteriskcentered  (d(tau) \& X d(theta)+ d(theta) \& X d(tau))\\
+ r\textasteriskcentered  aph\textasteriskcentered (d(tau) \& X d(phi)+ d(phi) \& X d(tau))\\
+ r\textasciicircum 2\textasteriskcentered (d(theta) \& X d(theta))\\
+ r\textasciicircum 2\textasteriskcentered  sin(theta)\textasciicircum 2 \textasteriskcentered  d(phi) \& X d(phi) :\\[5pt]
\texttt{Mancovmetric(M,g):}\\
\texttt{G:=Manconmetric(M):}\\[5pt]
Cd\_sublist := \{diff(C0,tau)=Cd0,diff(C1,tau)=Cd1,\\
diff(C2,tau)=Cd2,diff(C3,tau)=Cd3\} :\\
Cd\_inv\_sublist := \{Cd0=diff(C0,tau),Cd1=diff(C1,tau),\\
Cd2=diff(C2,tau),Cd3=diff(C3,tau)\} :\\
Cdd\_sublist := \{diff(C0,tau,tau)=Cdd0,diff(C1,tau,tau)=Cdd1,\\
diff(C2,tau,tau)=Cdd2,diff(C3,tau,tau)=Cdd3\} :\\[5pt]
aa := -Cd0+Cd1\textasteriskcentered cos(phi)\textasteriskcentered sin(theta)+Cd2\textasteriskcentered sin(phi)\textasteriskcentered sin(theta)+Cd3\textasteriskcentered cos(theta) :\\
aath := diff(aa,theta) :\\
aaph := diff(aa,phi) :\\
aad := subs(Cdd\_sublist,diff(subs(Cd\_inv\_sublist,aa),tau)) :\\
aathd := subs(Cdd\_sublist,diff(subs(Cd\_inv\_sublist,aath),tau)) :\\
aaphd := subs(Cdd\_sublist,diff(subs(Cd\_inv\_sublist,aaph),tau)) :\\[5pt]
Basis1 := \{d(tau),d(r),d(theta),d(phi)\} :\\
Basis2 := \{d(tau)\& \textasciicircum d(r), d(tau)\& \textasciicircum d(theta), d(tau)\& \textasciicircum d(phi),
           d(r)\& \textasciicircum d(theta), d(r)\& \textasciicircum d(phi), d(theta)\& \textasciicircum d(phi)\} :\\
Basis3 := \{d(tau)\& \textasciicircum d(r)\& \textasciicircum d(theta), d(tau)\& \textasciicircum d(r)\& \textasciicircum d(phi),
           d(tau)\& \textasciicircum d(theta)\& \textasteriskcentered d(phi), d(r)\& \textasciicircum d(theta)\& \textasciicircum d(phi)\} :\\[5pt]
VX := r\ask PD(r) :\\
VV := PD(tau) :\\
dualX := \& ~(VX) :\\
dualV := \& ~(VV):\\
dualA  := subs(gAV \ask d(tau) + ad\ask d(r) + r\ask athd\ask d(theta) + r\ask aphd\ask d(phi) ):\\
VA := \& ~(dualA) :\\[5pt]

ALW := q\_e\ask(dualV/(r\ask a)) :\\
FLW := collect(subs(\{diff(a,tau)=ad,diff(a,theta)=ath,diff(a,phi)=aph,
      diff(ath,tau)=athd,diff(aph,tau)=aphd,
      diff(ath,phi)=athph,diff(aph,theta)=athph\},
     simplify(F2C(d(ALW)))),Basis2) :\\
FLWc := subs(ad=0,athd=0,aphd=0,FLW) :\\
FLWr := collect(simplify(FLW - FLWc),Basis2) :\\[5pt]

  x0 := (C0+r)/c:\\
x1 := C1+r\ask sin(theta)\ask cos(phi):\\
x2 := C2+r\ask sin(theta)\ask sin(phi):\\
x3 := C3+r\ask cos(theta):\\[5pt]

J := Matrix(4, 4):\\
J[1, 1]:=Cd0/c:\\
for i from 1 to 3 do J[i+1, 1] := Cd || i:\\
J[i+1, 2] := diff(x || i, r):\\
J[i+1, 3] := diff(x || i, theta):\\
J[i+1, 4] := diff(x || i, phi) end do:\\
J[1,2]:=diff(x0, r):J[1,3]:=diff(x0, theta):J[1,4]:=diff(x0, phi):\\
J:\\[5pt]

DetJ := simplify(Determinant(J)):\\
detJ := -(1/c)\ask a\ask r\ind 2\ask sin(theta) :\\
Manvol(M) :=-(1/c) a\ask r\ind 2\ask sin(theta)\ask`\&\ind`(e[0], e[1], e[2], e[3]) :\\[5pt]

AdJ := simplify(Adjoint(J)):\\[3pt]
df\_tau\_t := AdJ[1, 1]/detJ :\\
df\_tau\_x := AdJ[1, 2]/detJ :\\
df\_tau\_y := AdJ[1, 3]/detJ :\\
df\_tau\_z := AdJ[1, 4]/detJ :\\
\#df\_r\_t   := AdJ[2, 1]/detJ :\\
df\_r\_t   := ((Cd0+a)\ask c)/a :\\
df\_r\_x   := AdJ[2, 2]/detJ :\\
df\_r\_y   := AdJ[2, 3]/detJ :\\
df\_r\_z   := AdJ[2, 4]/detJ :\\
\#df\_theta\_t  := AdJ[3, 1]/detJ :\\
df\_theta\_t := (ath\ask c)/(r\ask a) :\\
df\_theta\_x := AdJ[3, 2]/detJ :\\
df\_theta\_y := AdJ[3, 3]/detJ :\\
df\_theta\_z := AdJ[3, 4]/detJ :\\
df\_phi\_t := AdJ[4, 1]/detJ :\\
df\_phi\_x := AdJ[4, 2]/detJ :\\
df\_phi\_y := AdJ[4, 3]/detJ :\\
df\_phi\_z := AdJ[4, 4]/detJ :\\[5pt]

  PD\_t:=df\_tau\_t\ask PD(tau)+df\_r\_t\ask PD(r) +df\_theta\_t\ask PD(theta) +df\_phi\_t\ask PD(phi):\\
PD\_x:=df\_tau\_x\ask PD(tau)+df\_r\_x\ask PD(r) +df\_theta\_x\ask PD(theta) +df\_phi\_x\ask PD(phi):\\
PD\_y:=df\_tau\_y\ask PD(tau)+df\_r\_y\ask PD(r) +df\_theta\_y\ask PD(theta) +df\_phi\_y\ask PD(phi):\\
PD\_z:=df\_tau\_z\ask PD(tau)+df\_r\_z\ask PD(r) +df\_theta\_z\ask PD(theta) +df\_phi\_z\ask PD(phi):\\[5pt]

PDt\_Fc:=PD\_t \&i FLWc:\\
PDt\_starFc:=collect(subs(aph=aaph,Cd\_sublist,F2C(PD\_t \&i (\&star(FLWc)))), Basis1,simplify):\\
Elec\_c :=(1/c)\ask PD\_t \&i FLWc :\\
Mag\_c :=(1/(c\ask c))\ask collect(subs(aph=aaph,Cd\_sublist,F2C(PD\_t \&i (\&star(FLWc)))),
     Basis1,simplify) :\\
Elec\_r := (1/c)\ask collect(PD\_t \&i FLWr,Basis1) :\\
Mag\_r := (1/(c\ask c))\ask collect(PD\_t \&i F2C(\&star(FLWr)),Basis1) :\\[5pt]

 Elec\_cx := simplify(PD\_x \&i Elec\_c) :\\
Elec\_cy := simplify(PD\_y \&i Elec\_c) :\\
Elec\_cz := simplify(PD\_z \&i Elec\_c) :\\
Elec\_rx := simplify(PD\_x \&i Elec\_r) :\\
Elec\_ry := simplify(PD\_y \&i Elec\_r) :\\
Elec\_rz := simplify(PD\_z \&i Elec\_r) :\\
Mag\_cx := simplify(PD\_x \&i Mag\_c) :\\
Mag\_cy := simplify(PD\_y \&i Mag\_c) :\\
Mag\_cz := simplify(PD\_z \&i Mag\_c) :\\
Mag\_rx := simplify(PD\_x \&i Mag\_r) :\\
Mag\_ry := simplify(PD\_y \&i Mag\_r) :\\
Mag\_rz := simplify(PD\_z \&i Mag\_r) :\\[5pt]
 Energy\_res :=(1/2)\ask (ep\ask ((Elec\_cx+Elec\_rx)\ind 2+(Elec\_cy+Elec\_ry)\ind 2\\
 +(Elec\_cz+Elec\_rz)\ind 2)+(1/mu)\ask ((Mag\_cx+Mag\_rx)\ind 2+(Mag\_cy+Mag\_ry)\ind 2\\
 +(Mag\_cz+Mag\_rz)\ind 2)):\\[5pt]

hat\_ sublist :=
\{T0=(sqrt((X0-C1)\ind 2 + (Y0-C2)\ind 2 + (Z0-C3)\ind 2) +C0)/c,\\
rhat=sqrt((X0-C1)\ind 2 + (Y0-C2)\ind 2 + (Z0-C3)\ind 2),\\
cthhat=(Z0-C3)/(sqrt((X0-C1)\ind 2 + (Y0-C2)\ind 2 + (Z0-C3)\ind 2)),\\
sthhat=sqrt((X0-C1)\ind 2+(Y0-C2)\ind 2)/(sqrt((X0-C1)\ind 2 + (Y0-C2)\ind 2 + (Z0-C3)\ind 2)),\\
cphhat=(X0-C1)/(sqrt((X0-C1)\ind 2+(Y0-C2)\ind 2)),\\
sphhat=(Y0-C2)/(sqrt((X0-C1)\ind 2+(Y0-C2)\ind 2))\}:\\[5pt]

prehat\_ subslist :={cos(theta)=cthhat,sin(theta)=sthhat,\\
cos(phi)=cphhat,sin(phi)=sphhat,r=rhat} :\\[5pt]

Curve3\_ def := \{\\
C0a=epsilon\ask gamma\ask tau,\\
C1a=epsilon\ask gamma\ask v\ask tau,\\
C2a=0,\\
C3a=0 \} :\\[2pt]
Curve2\_ def := \{\\
C0a=epsilon\ask gamma\ask tau,\\
C1a=Lp-epsilon\ask (Rp\ask sin((Lp/(epsilon\ask Rp))-(gamma\ask v\ask tau)/Rp)),\\
C2a=epsilon\ask Rp\ask (1-cos((Lp/(epsilon\ask Rp))-(gamma\ask v\ask tau)/Rp)),\\
C3a=0\} :\\[2pt]
Curve1\_ def :=\{\\
C0a=epsilon\ask gamma\ask tau,\\
C1a=epsilon\ask (gamma\ask v\ask cos(thetap)\ask tau+Lp -Rp\ask sin(thetap)+cos(thetap)\ask (thetap\ask Rp-Lp/epsilon)),\\
C2a=epsilon\ask (-gamma\ask v\ask sin(thetap)\ask tau + Rp\ask (1-cos(thetap))-sin(thetap)\ask (thetap\ask Rp-Lp/epsilon)),\\
C3a=0\} :\\[5pt]

Curve3\_ sublist := eval(subs(Diff=diff,eval(subs(Curve3\_ def,\\
\{\\
C0=C0a,Cd0=Diff(C0a,tau),Cdd0=Diff(C0a,tau,tau),\\
C1=C1a,Cd1=Diff(C1a,tau),Cdd1=Diff(C1a,tau,tau),\\
C2=C2a,Cd2=Diff(C2a,tau),Cdd2=Diff(C2a,tau,tau),\\
C3=C3a,Cd3=Diff(C3a,tau),Cdd3=Diff(C3a,tau,tau)\}\\
)))) :\\[2pt]

Curve2\_ sublist := eval(subs(Diff=diff,eval(subs(Curve2\_ def,\\
\{\\
C0=C0a,Cd0=Diff(C0a,tau),Cdd0=Diff(C0a,tau,tau),\\
C1=C1a,Cd1=Diff(C1a,tau),Cdd1=Diff(C1a,tau,tau),\\
C2=C2a,Cd2=Diff(C2a,tau),Cdd2=Diff(C2a,tau,tau),\\
C3=C3a,Cd3=Diff(C3a,tau),Cdd3=Diff(C3a,tau,tau)\}\\
)))) :\\[2pt]

Curve1\_ sublist := eval(subs(Diff=diff,eval(subs(Curve1\_ def,\\
\{\\
C0=C0a,Cd0=Diff(C0a,tau),Cdd0=Diff(C0a,tau,tau),\\
C1=C1a,Cd1=Diff(C1a,tau),Cdd1=Diff(C1a,tau,tau),\\
C2=C2a,Cd2=Diff(C2a,tau),Cdd2=Diff(C2a,tau,tau),\\
C3=C3a,Cd3=Diff(C3a,tau),Cdd3=Diff(C3a,tau,tau)\}\\
)))) :\\[5pt]

get\_range3 := proc(Values\_sublist)\\
  local Taub ;\\
  Taub := subs(Values\_sublist,X0/(epsilon\ask gamma\ask v)) ;\\
  0..Taub ;\\
end proc :\\
get\_range2 := proc(Values\_sublist)\\
  local Taua ;\\
  Taua := subs(Values\_sublist,-Rp\ask thetap/(gamma\ask v)) ;\\
  Taua..0 ;\\
end proc :\\
get\_range1 := proc(Values\_sublist)\\
  local Taua ;\\
  Taua := subs(Values\_sublist,-Rp\ask thetap/(gamma\ask v)) ;\\
  subs(Values\_sublist,StartTau)..Taua ;\\
end proc :\\[5pt]
\end{mapcode}
\end{linenumbers}
\label{get_fields}
\begin{linenumbers}
\begin{mapcode}
  Get\_Fields := proc(Cnum,Values\_sublist)\\
  local Curve\_sublist;\\
  Curve\_sublist := Curve||Cnum||\_sublist ;\\
  \{\\
    Elec\_cx\_res =\\
    subs(Values\_sublist,subs(Curve\_sublist,\\
    subs(hat\_sublist,subs(prehat\_subslist,\\
    subs(a=aa, ad=aad, ath=aath, athd=aathd, aph=aaph, aphd=aaphd,\\
    Elec\_cx))))) ,
    Elec\_cy\_res =\\
    subs(Values\_sublist,subs(Curve\_sublist,\\
    subs(hat\_sublist,subs(prehat\_subslist,\\
    subs(a=aa, ad=aad, ath=aath, athd=aathd, aph=aaph, aphd=aaphd,\\
    Elec\_cy))))) ,\\
    Elec\_cz\_res =\\
    subs(Values\_sublist,subs(Curve\_sublist,\\
    subs(hat\_sublist,subs(prehat\_subslist,\\
    subs(a=aa, ad=aad, ath=aath, athd=aathd, aph=aaph, aphd=aaphd,\\
    Elec\_cz))))) ,\\
    Elec\_rx\_res =\\
    subs(Values\_sublist,subs(Curve\_sublist,\\
    subs(hat\_sublist,subs(prehat\_subslist,\\
    subs(a=aa, ad=aad, ath=aath, athd=aathd, aph=aaph, aphd=aaphd,\\
    Elec\_rx))))) ,\\
    Elec\_ry\_res =\\
    subs(Values\_sublist,subs(Curve\_sublist,\\
    subs(hat\_sublist,subs(prehat\_subslist,\\
    subs(a=aa, ad=aad, ath=aath, athd=aathd, aph=aaph, aphd=aaphd,\\
    Elec\_ry))))) ,\\
    Elec\_rz\_res =\\
    subs(Values\_sublist,subs(Curve\_sublist,\\
    subs(hat\_sublist,subs(prehat\_subslist,\\
    subs(a=aa, ad=aad, ath=aath, athd=aathd, aph=aaph, aphd=aaphd,\\
    Elec\_rz))))) ,\\

    Mag\_cx\_res =\\
    subs(Values\_sublist,subs(Curve\_sublist,\\
    subs(hat\_sublist,subs(prehat\_subslist,\\
    subs(a=aa, ad=aad, ath=aath, athd=aathd, aph=aaph, aphd=aaphd,\\
    Mag\_cx))))) ,\\
    Mag\_cy\_res =\\
    subs(Values\_sublist,subs(Curve\_sublist,\\
    subs(hat\_sublist,subs(prehat\_subslist,\\
    subs(a=aa, ad=aad, ath=aath, athd=aathd, aph=aaph, aphd=aaphd,\\
    Mag\_cy))))) ,\\
    Mag\_cz\_res =\\
    subs(Values\_sublist,subs(Curve\_sublist,\\
    subs(hat\_sublist,subs(prehat\_subslist,\\
    subs(a=aa, ad=aad, ath=aath, athd=aathd, aph=aaph, aphd=aaphd,\\
    Mag\_cz))))) ,\\
    Mag\_rx\_res =\\
    subs(Values\_sublist,subs(Curve\_sublist,\\
    subs(hat\_sublist,subs(prehat\_subslist,\\
    subs(a=aa, ad=aad, ath=aath, athd=aathd, aph=aaph, aphd=aaphd,\\
    Mag\_rx))))) ,\\
    Mag\_ry\_res =\\
    subs(Values\_sublist,subs(Curve\_sublist,\\
    subs(hat\_sublist,subs(prehat\_subslist,\\
    subs(a=aa, ad=aad, ath=aath, athd=aathd, aph=aaph, aphd=aaphd,\\
    Mag\_ry))))) ,\\
    Mag\_rz\_res =\\
    subs(Values\_sublist,subs(Curve\_sublist,\\
    subs(hat\_sublist,subs(prehat\_subslist,\\
    subs(a=aa, ad=aad, ath=aath, athd=aathd, aph=aaph, aphd=aaphd,\\
    Mag\_rz))))) ,\\

    T0\_res =\\
    subs(Values\_sublist,subs(Curve\_sublist,\\
    subs(hat\_sublist,subs(prehat\_subslist,\\
    subs(a=aa, ad=aad, ath=aath, athd=aathd, aph=aaph, aphd=aaphd,\\
    T0))))) ,\\

    C1\_res =\\
    subs(Values\_sublist,subs(Curve\_sublist,\\
    subs(hat\_sublist,subs(prehat\_subslist,\\
    subs(a=aa, ad=aad, ath=aath, athd=aathd, aph=aaph, aphd=aaphd,\\
    C1))))) ,\\

    aa\_res =\\
    subs(Values\_sublist,subs(Curve\_sublist,\\
    subs(hat\_sublist,subs(prehat\_subslist,\\
    subs(a=aa, ad=aad, ath=aath, athd=aathd, aph=aaph, aphd=aaphd,\\
    a))))) ,

    Energy\_res =\\
    subs(Values\_sublist,subs(Curve\_sublist,\\
    subs(hat\_sublist,subs(prehat\_subslist,\\
    subs(a=aa, ad=aad, ath=aath, athd=aathd, aph=aaph, aphd=aaphd,\\
    sqrt((Elec\_cx+Elec\_rx)\ind 2+(Elec\_cy+Elec\_ry)\ind 2+(Elec\_cz+Elec\_rz)\ind 2))))))\\
  \} ;\\
end proc:
\end{mapcode}
\end{linenumbers}
\section*{Comments}
 {\footnotesize \color{red} \tt 1-98} This section of the code is almost identical to that of Part I with a few notable exceptions:
  \begin{itemize}
  \item We use the coordinate system given by (\ref{coords_new}) where $(\tau, \rnew, \theta, \phi)$=(\texttt{tau, r, theta, phi}).
  \item The metric is now given by (\ref{g_newcoords}).
  \item We We define $\flw_{\textup{C}}=$\texttt{FLWc} by setting all components of acceleration to zero in \texttt{FLW}. We define $\flw_{\textup{R}}=$\texttt{FLWr} as the difference \texttt{FLW}-\texttt{FLWc}.
\end{itemize}
 {\footnotesize \color{red} \tt 99-106} Calculate $\elw_{\textup{C}}=$\texttt{Elec\_c}, $\elw_{\textup{R}}=$\texttt{Elec\_r}, $\blw_{\textup{C}}=$\texttt{Mag\_c} and $\blw_{\textup{R}}=$\texttt{Mag\_r}. This is easily done using (\ref{ebvec_def}).\\
{\footnotesize \color{red} \tt 107-118} Calculate the components of these vectors in the \texttt{x1, x2} and \texttt{x3} directions by taking the internal contractions with respect to \texttt{PD\_x}, \texttt{PD\_y} and \texttt{PD\_z} respectively.\\
{\footnotesize \color{red} \tt 119-121} Calculate the total energy of the electric field $||\elw(\tau, \rnew, \theta, \phi)||^2=$\texttt{Energy\_res}.
{\footnotesize \color{red} \tt 122-130} Input the  substitutions given by (\ref{path_ref}).  \\
{\footnotesize \color{red} \tt 131-147} Define the three sections of the pre-bent path by inputting the components (\ref{prebent_path}) according to (\ref{worldline_comps}). The labels for the axis in the code are different to those given in figure \ref{setup} due to the way I initially set up the trajectory. The axes $x, y, z$ in the figure correspond to \texttt{y}, \texttt{z}, \texttt{x} in the code, and correspondingly the point $\VX=(X_0, Y_0, Z_0)$ is given by \texttt{(Y0, Z0, X0)}. The coordinate system is aligned so that instead of being located at the terminus of the bend as in the figure (\ref{setup}), the origin is located at the end of the small straight line section. As a result the parameter $\texttt{Lp}$ is defined as the negative of the distance $Z$. We use notation $\rrad=$\texttt{Rp} and $\Theta$=\texttt{thetap}.\\
{\footnotesize \color{red} \tt 148-168} Define the three corresponding list of substitutions which will associate a field with a particular trajectory. \\
{\footnotesize \color{red} \tt 169-183} Calculate the ranges of $\tau=$\texttt{tau} for each of the three sections of the path. The values \texttt{Taua} and \texttt{Taub} are the tau values at the start and the end of the bend respectively. The value \texttt{StartTau} is the initial value for \texttt{tau}.\\
{\footnotesize \color{red} \tt 184-186} The procedure \texttt{get\_fields} uses the substitutions in the previous section to output the listed fields as functions of $\tau=$\texttt{tau} for a given section of path and a given set of input parameters. The inputs are the number \texttt{Cnum}$=$1, 2 \textup{or} 3 which tells maple which of the three sections of the path {\footnotesize \color{red} \tt 131-147} we are considering, and a list of numerical inputs of the following format
\begin{mapcode}
Values\_sublist0 :=
subs(gam=1000,\{X0=0,Y0=0,Z0=1,epsilon=1,v=sqrt(1-1/gam\ind2),Lp=0, Rp=1000, thetap=0.1, gamma=gam, StartTau=-20\});
\end{mapcode}
{\footnotesize \color{red} \tt 247-251 }Notice the lab time $T_0(\tau)=$\texttt{T0\_res} and the total energy of the electric field $||\e(\tau, \rnew, \theta, \phi)||^2=$\texttt{Energy\_res} are also obtained as functions of \texttt{tau}.

\section{Part 3 - minimize peak field }
\label{minimize_app}
\begin{linenumbers}
\begin{mapcode}
get\_list3 := proc(Values\_sublist)\\
  local Taub ;\\
  Taub := subs(Values\_sublist,Lp/(epsilon\ask gamma\ask v)) ;\\
  \$(round(Taub)..0) ;\\
end proc :\\
get\_list2 := proc(Values\_sublist)\\
  local Taua, Taub ;\\
  Taub := subs(Values\_sublist,Lp/(epsilon\ask gamma\ask v)) ;\\
  Taua := subs(Values\_sublist,Lp/(epsilon\ask gamma\ask v)-Rp\ask thetap/(gamma\ask v)) ;\\
  \$(round(Taua)..round(Taub) );\\
end proc :\\
get\_list1 := proc(Values\_sublist)\\
  local Taua, Taub ;\\
  Taub := subs(Values\_sublist,Lp/(epsilon\ask gamma\ask v)) ;\\
  Taua := subs(Values\_sublist,Lp/(epsilon\ask gamma\ask v)-Rp\ask thetap/(gamma\ask v)) ;\\
  \$(round(subs(Values\_sublist,StartTau))..round(Taua) )\\
end proc :\\[5pt]

  Max\_ field := proc(Values\_ sublist)\\
  local Field1,Field2,Field3,taurng1,taurng2,taurng3, FUNCT1, FUNCT2, FUNCT3,Taua, Taub,VAL1, VAL2, VAL3, i1, i2,i3, F1, FF1, F2, FF2;\\
Taub := evalf(subs(Values\_ sublist,Lp/(epsilon\ask gamma\ask v))) ;\\
Taua := evalf(subs(Values\_ sublist,Lp/(epsilon\ask gamma\ask v)-Rp\ask thetap/(gamma\ask v))) ;\\
  Field1:=evalf(subs(Get\_ Fields(1,Values\_ sublist),Energy\_ res));\\
  Field2:=evalf(subs(Get\_ Fields(2,Values\_ sublist),Energy\_ res));\\
  Field3:=evalf(subs(Get\_ Fields(3,Values\_ sublist),Energy\_ res));\\
  taurng1 := get\_ list1(Values\_ sublist) ;\\
  taurng2 := get\_ list2(Values\_ sublist) ;\\
  taurng3 := get\_ list3(Values\_ sublist) ;\\
VAL1:=(abs(subs(Values\_ sublist, StartTau))-(abs(round(Taua)))):\\[3pt]
for i1 from 1 to VAL1 do:\\
for i2 from 1 to 100 do:\\
 FUNCT1:= max(subs(tau=taurng1[i1], Field1), subs(tau=taurng1[i1]-i2/100, Field1));\\
end do:\\
end do:\\
ARR:=Array(1..19):\\
\#for i1 from 1 to (abs(round(Taua))) do:\\
for i3 from 1 to 19 do:\\
ARR[i3]:=(evalf(subs(tau=-i3\ask(0.05), Field2)));\\
FUNCT2:=max(ARR);\\
end do:\\
\#end do:\\
max(FUNCT1, FUNCT2);\\
end proc :\\

Values\_sublist1 :=\\
subs(gam=1000,{X0=0,Y0=0,Z0=1,epsilon=1,v=sqrt(1-1/gam\ind 2),\\
gamma=gam,Lp=0,Rp=Rpp, thetap=thetapp, StartTau=-20});\\

thetap\_range:= 1/95, 1/90, 1/85, 1/80, 1/75, 1/70, 1/65, 1/60, 1/55, 1/50, 1/45, 1/40, 1/35, 1/30, 1/25, 1/20, 1/15, 1/10, 1/5, 1;\\

Rp\_range:=500, 1000, 1500, 2000, 2500, 3000, 3500, 4000, 4500, 5000, 5500, 6000, 6500, 7000, 7500, 8000, 8500, 9000, 9500, 10000;\\

B:=Matrix(20, 20);\\

 for i1 from 1 to 20 do:\\
   for i2 from 1 to 20 do:\\
   subs\_LthetaR :=subs(thetapp=thetap\_range[i1],Rpp=Rp\_range[i2], Values\_sublist1);\\
  B[i1, i2]:= Max\_field(subs\_LthetaR);\\
end do;\\
  end do;\\

\end{mapcode}
\end{linenumbers}

\section*{Comments}
 {\footnotesize \color{red} \tt 269-287} The procedures \texttt{get\_list||Cnum} will round the values of \texttt{Taua} and \texttt{Taub} to the nearest integer and output the range of \texttt{tau} as a sequence of integers.\\
 {\footnotesize \color{red} \tt 288-315} The procedure \texttt{Max\_ field} will compare the peak value of \texttt{Energy\_ res} for the initial straight line and the bend for a number of values of \texttt{tau}. The field given by the second straight line is negligible. The peak field for the straight segment is given by the local variable \texttt{FUNCT1} and the peak field for the bend is given by  \texttt{FUNCT2}.\\
  {\footnotesize \color{red} \tt 316-330} These lines of code will create a $20\times20$ matrix \texttt{B} whose elements are the peak fields corresponding to the given values of  \texttt{thetap} and \texttt{Rp}. These values correspond to those given in table \ref{minimize_table} and the resulting matrix was used to plot figure \ref{minimize} using the MAPLE function \emph{matrixplot}.
\section{Part 4 - Plots}
\begin{linenumbers}
\begin{mapcode}
Plot\_Field\_tau := proc(Values\_sublist,Field)\\
  local Field1,Field2,Field3,taurng1,taurng2,taurng3;\\
  Field1:=subs(Get\_Fields(1,Values\_sublist),Field) ;\\
  Field2:=subs(Get\_Fields(2,Values\_sublist),Field) ;\\
  Field3:=subs(Get\_Fields(3,Values\_sublist),Field) ;\\
  taurng1 := get\_range1(Values\_sublist) ;\\
  taurng2 := get\_range2(Values\_sublist) ;\\
  taurng3 := get\_range3(Values\_sublist) ;\\
  display(\\
    plot(Field1,tau=taurng1,color=BLACK,\_rest),\\
    plot(Field2,tau=taurng2,color=RED,\_rest, numpoints=1000),\\
    plot(Field3,tau=taurng3,color=BLUE,\_rest)\\
  ):\\
end proc:\\[5pt]

Plot\_Field\_T0 := proc(Values\_sublist,Field)\\
  local Field1,Field2,Field3,taurng1,taurng2,taurng3;\\
  taurng1 := get\_range1(Values\_sublist) ;\\
  taurng2 := get\_range2(Values\_sublist) ;\\
  taurng3 := get\_range3(Values\_sublist) ;\\
  Field1:=subs(Get\_Fields(1,Values\_sublist),[T0\_res,Field,tau=taurng1]) ;\\
  Field2:=subs(Get\_Fields(2,Values\_sublist),[T0\_res,Field,tau=taurng2]) ;\\
  Field3:=subs(Get\_Fields(3,Values\_sublist),[T0\_res,Field,tau=taurng3]) ;\\
  display(\\
    plot(Field1,color=BLACK,\_rest),\\
    plot(Field2,color=RED,\_rest),\\
    plot(Field3,color=BLUE,\_rest)\\
  ):\\
end proc :\\[5pt]

Values\_sublist1 :=\\
subs(gam=1000, \{X0=0.005,Y0=0,Z0=0.0005,epsilon=1,v=sqrt(1-1/gam\ind 2),gamma=gam,Lp=0,thetap=0.13,Rp=0.5,  StartTau=-100\\
, c=3\ask 10\ind (8), q\_e=-1.80951262\ask 10\ind (-8)\});\\

Values\_sublist2 :=\\
subs(gam=1000, \{X0=0.005,Y0=0, Z0=0.0005, epsilon=1,v=(sqrt(1-1/gam\ind 2)),gamma=gam,Lp=0,thetap=0,Rp=0.5,  StartTau=-100, c=3\ask 10\ind (8),\\ q\_e=(-1.80951262\ask 10\ind (-8))\});\\[5pt]

EEx:=subs(Get\_Fields(1,Values\_sublist2),  Elec\_cx\_res+Elec\_rx\_res):\\
EEy:=subs(Get\_Fields(1,Values\_sublist2),  Elec\_cy\_res+Elec\_ry\_res):\\
EEz:=subs(Get\_Fields(1,Values\_sublist2),  Elec\_cz\_res+Elec\_rz\_res):\\
TT:=subs(Get\_Fields(1,Values\_sublist2),(T0\_res)):\\[5pt]

EEx1:=subs(Get\_Fields(1,Values\_sublist1),Elec\_cx\_res+Elec\_rx\_res):\\
EEy1:=subs(Get\_Fields(1,Values\_sublist1),Elec\_cy\_res+Elec\_ry\_res):\\
EEz1:=subs(Get\_Fields(1,Values\_sublist1),Elec\_cz\_res+Elec\_rz\_res):\\
TT1:=subs(Get\_Fields(1,Values\_sublist1),(T0\_res)):\\[5pt]

EEx2:=subs(Get\_Fields(2,Values\_sublist1),Elec\_cx\_res+Elec\_rx\_res):\\
EEy2:=subs(Get\_Fields(2,Values\_sublist1),Elec\_cy\_res+Elec\_ry\_res):\\
EEz2:=subs(Get\_Fields(2,Values\_sublist1),Elec\_cz\_res+Elec\_rz\_res):\\
TT2:=subs(Get\_Fields(2,Values\_sublist1),(T0\_res)):\\[5pt]

EEx3:=subs(Get\_Fields(3,Values\_sublist1),Elec\_cx\_res+Elec\_rx\_res):\\
EEy3:=subs(Get\_Fields(3,Values\_sublist1),Elec\_cy\_res+Elec\_ry\_res):\\
EEz3:=subs(Get\_Fields(3,Values\_sublist1),Elec\_cz\_res+Elec\_rz\_res):\\
TT3:=subs(Get\_Fields(3,Values\_sublist1),(T0\_res)):\\[5pt]

part1x:=plot([10\ind (12)\ask TT1, abs(EEx1), tau=get\_range1(Values\_sublist1)], color=black, numpoints=10000):\\
part2x:=plot([10\ind (12)\ask TT2, abs(EEx2), tau=get\_range2(Values\_sublist1)], color=red,resolution=600, numpoints=50000):\\
part3x:=plot([10\ind (12)\ask TT3, abs(EEx3), tau=get\_range3(Values\_sublist1)], color=blue, numpoints=10000):\\

part1y:=plot([10\ind (12)\ask TT1, abs(EEy1), tau=get\_range1(Values\_sublist1)], color=black, numpoints=10000):\\
part2y:=plot([10\ind (12)\ask TT2, abs(EEy2), tau=get\_range2(Values\_sublist1)], color=red,resolution=600, numpoints=50000):\\
part3y:=plot([10\ind (12)\ask TT3, abs(EEy3), tau=get\_range3(Values\_sublist1)], color=blue, numpoints=10000):\\

part1z:=plot([10\ind (12)\ask TT1, abs(EEz1), tau=get\_range1(Values\_sublist1)], color=black, numpoints=10000):\\
part2z:=plot([10\ind (12)\ask TT2, abs(EEz2), tau=get\_range2(Values\_sublist1)], color=red,resolution=600, numpoints=50000):\\
part3z:=plot([10\ind (12)\ask TT3, abs(EEz3), tau=get\_range3(Values\_sublist1)], color=blue, numpoints=10000):\\

Resize(display(part1x, part2x, part3x, axes=boxed, view=[15.8..16.8, 0..8], axesfont=[TIMES, ROMAN, 20], thickness=3));\\
Resize(display(part1y, part2y, part3y, axes=boxed, view=[15.8..16.8, 0..8], axesfont=[TIMES, ROMAN, 20], thickness=3));\\
Resize(display(part1z, part2z, part3z, axes=boxed, view=[15.8..16.8, 0..8], axesfont=[TIMES, ROMAN, 20], thickness=3));\\

\end{mapcode}
\end{linenumbers}
\section*{Comments}
{\footnotesize \color{red} \tt 331-334} This procedure will plot any field in the list \texttt{Get\_fields} (or combination thereof) against \texttt{tau} for a given set of inputs. We can plot the field due to the straight trajectory by setting \texttt{thetap}$=0$.\\
{\footnotesize \color{red} \tt 345-361} This procedure will plot any field  in the list as a function of \texttt{T0}.\\
{\footnotesize \color{red} \tt 362-417} This will make the plots given in figure \ref{fig_fields_comp}.

\section{Part 5 - Convolution}
\label{convolution_app}

 \begin{linenumbers}
 \begin{mapcode}
rho\_ box:= (t,a, b) -> 1/a\ask(Heaviside(t+a/2+b)-Heaviside(t-a/2+b));\\
plot(rho\_ box(t,0.0005, 0),t=-0.01..0.01,title="box distribution",colour=brown,axes=boxed);\\
rho\_ Gauss:= (t, a, b) -> 1/(a\ask sqrt(2\ask Pi))\ask exp((-(t-b)\ind2)/(2\ask a\ind2));\\
plot(rho\_ Gauss(t, 0.5, 0),t=-1..1,title="Gaussian distribution",colour=brown,axes=boxed,numpoints=10000);\\[5pt]

conv:=proc(PEAK, a, b, N,comp )\\
      local t, i, E\_seq, tau\_seq,rho\_seq,E0\_seq,  conv,sum1 ;\\
      global convx1, convy1, convz1, convx2, convy2, convz2 ;\\
if PEAK=1 then\\
for i from 0 to (N-1) do\\
t:=16.6667;\\
\#solve(a+((b-a)/N)\ask (i+1/2)=TT,tau);\\
\#print("-----",
tau\_1\_||i :=solve(t-a-((b-a)/N)\ask (i+1/2)=10\ind(12)\ask TT,tau);\\
\#print(tau\_1\_||i) ;\\
EEE\_0\_||i :=evalf(subs(tau=tau\_1\_||i, EE||comp));\\
EEE\_1\_||i :=EEE\_0\_||i\ask evalf(rho\_Gauss(t-a-((b-a)/N)\ask (i+1/2), b-a, t));\\
rho\_1\_||i:=rho\_Gauss(t-a-((b-a)/N)\ask (i+1/2), b-a, t);\\
end do:\\
sum1:=add(EEE\_1\_||i, i=0..N-1);\\
conv||comp||PEAK:=sum1/add(evalf(rho\_Gauss(t-a-((b-a)/N)\ask (i+1/2), b-a, t)), i=0..N-1);\\
print(conv||comp||PEAK);\\[3pt]
elif PEAK=2 then\\
for i from 0 to (N-1) do\\
t:=16.685;\\
tau\_1\_||i :=fsolve(t-a-((b-a)/N)\ask (i+1/2)=10\ind(12)\ask TT||PEAK,tau);\\
\#print(tau\_1\_||i) ;\\
EEE\_0\_||i :=evalf(subs(tau=tau\_1\_||i, EE||comp||PEAK));\\
EEE\_1\_||i :=EEE\_0\_||i\ask evalf(rho\_Gauss(t-a-((b-a)/N)\ask (i+1/2), b-a, t));\\
rho\_1\_||i:=rho\_Gauss(t-a-((b-a)/N)\ask (i+1/2), b-a,t);\\
end do:\\
sum1:=add(EEE\_1\_||i, i=0..N-1);\\
conv||comp||PEAK:=sum1/add(evalf(rho\_Gauss(t-a-((b-a)/N)\ask (i+1/2), b-a, t)), i=0..N-1);\\
print(conv||comp||PEAK);\\
end if:\\
end proc:
 \end{mapcode}
 \end{linenumbers}
\section*{Comments}
{\footnotesize \color{red} \tt 347-352} Defines the charge profile $\rho(\nu)$. We can use either a box profile or a Gaussian profile.\\
{\footnotesize \color{red} \tt 424-458} Procedure for calculating the convolution $(\ref{E_Tot})$. The convolution has to be evaluated for the pre-bent path and for the straight path for a selection of different bunch lengths. We adopt a Gaussian form for $\rhoLab$ and define the bunch length as the full width at half maximum (FWHM). The results are given in table \ref{table3}.

\end{appendices}

\end{mainmatter}

\begin{backmatter}

\fancyhfoffset[L,R]{\marginparsep+\marginparwidth}
\fancyhead[LO,RE]{}
\fancyhead[LE,RO]{\slshape{Bibliography}}
\fancyfoot[C]{\thepage}
\renewcommand{\chaptermark}[1]{\markboth{#1}{}}
\renewcommand{\sectionmark}[1]{\markright{#1}{}}
\addcontentsline{toc}{chapter}{Bibliography}
\bibliographystyle{unsrt}
\bibliography{all_references}

\begin{thebibliography}{10}

\bibitem{FerrisGratus11}
M.~R. {Ferris} and J.~{Gratus}.
\newblock {The origin of the Schott term in the electromagnetic self force of a
  classical point charge}.
\newblock {\em Journal of Mathematical Physics}, 52(9):092902, September 2011.

\bibitem{GratusFerris11}
J.~{Gratus} and M.~R. {Ferris}.
\newblock {Bending a Beam to Significantly Reduce Wakefields of Short Bunches}.
\newblock {\em ArXiv e-prints 1108.4625}, August 2011.

\bibitem{deGroot72}
S.~R de~Groot and L~G Suttorp.
\newblock {\em Foundations of Electrodynamics}.
\newblock North-Holland, Amsterdam, 1972.

\bibitem{Rohrlich65}
Fritz Rohrlich.
\newblock {\em Classical Charged Particles}.
\newblock Addison Wesley, Reading, Mass, 1965.

\bibitem{Maxwell65}
J.~Clerk Maxwell.
\newblock A dynamical theory of the electromagnetic field.
\newblock {\em Phil. Trans. R. Soc. Lond.}, 155:pp. 459--512, 1865.

\bibitem{Lorentz16}
H.~A. Lorentz.
\newblock {\em The Theory of Electrons and Its Application to the Phenomena of
  Light and Radiant Heat}.
\newblock B. G. Teubner, Leipzig, 1916.

\bibitem{Obukhov}
Friedrich~W. Hehl and Yuri~N. Obukhuv.
\newblock {\em Foundations of Classical Electrodynamics}.
\newblock Birkh$\ddot{\textup{a}}$user, 2003.

\bibitem{Temple38}
G.~Temple.
\newblock New systems of normal co-ordinates for relativistic optics.
\newblock {\em Proceedings of the Royal Society of London. Series A,
  Mathematical and Physical Sciences}, 168(932):pp. 122--148, 1938.

\bibitem{Newman}
Ezra~T. Newman and T.~W.~J. Unti.
\newblock A class of null flat-space coordinate systems.
\newblock {\em {Journal of Mathematical Physics}}, {4}({12}):{1467--\&},
  {1963}.

\bibitem{Galtsov02}
Dmitri~V. Gal'tsov and Pavel Spirin.
\newblock Radiation reaction reexamined: bound momentum and the schott term.
\newblock {\em Gravitation and Cosmology}, 12:1--10, 2006.

\bibitem{Trautman62}
I~Robinson and A~Trautman.
\newblock {Some Spherical Gravitational Waves in General Relativity}.
\newblock {\em Proceedings of the Royal Society of London. Series A,
  Mathematical and Physical Sciences}, {265}({1323}):{463--\&}, {1962}.

\bibitem{Ellis80}
G.~F.~R. Ellis.
\newblock Limits to verification in cosmology.
\newblock {\em Annals of the New York Academy of Sciences}, 336(1):130--160,
  1980.

\bibitem{Howie}
John~M. Howie.
\newblock {\em Complex Analysis}.
\newblock Springer Verlag, 2007.

\bibitem{Landau80}
L~D Landau and E~M Lifshitz.
\newblock {\em Classical Theory of Fields (4th Edition)}.
\newblock Butterworth-Heinemann, 1980.

\bibitem{Rohrlich97}
F.~Rohrlich.
\newblock The dynamics of a charged sphere and the electron.
\newblock {\em American Journal of Physics}, 65(11):1051, 1997.

\bibitem{erber61}
Thomas Erber.
\newblock The classical theories of radiation reaction.
\newblock {\em Protein Science}, 9:343--392, 1961.

\bibitem{Dirac38}
P.~A.~M. Dirac.
\newblock Classical theory of radiating electrons.
\newblock {\em Proceedings of the Royal Society of London. Series A,
  Mathematical and Physical Sciences}, 167(929):pp. 148--169, 1938.

\bibitem{Eliezer}
C.~Jayaratnam Eliezer.
\newblock The classical equations of motion of an electron.
\newblock {\em Mathematical Proceedings of the Cambridge Philosophical
  Society}, 42(03):278--286, 1946.

\bibitem{Bonnor}
W.B. Bonnor.
\newblock A new equation of motion for a radiating charged particle.
\newblock {\em Proc. R. Soc. Lond. A}, 337:591--598, 1974.

\bibitem{Parrott86}
Stephen Parrott.
\newblock {\em Relativistic Electrodynamics and Differential Geometry}.
\newblock Springer, 1986.

\bibitem{Poisson99}
Eric Poisson.
\newblock An introduction to the lorentz-dirac equation.
\newblock {\em Preprint arXiv gr-qc/9912045}, 1999.

\bibitem{Lorentz}
H.~A. Lorentz.
\newblock {\em The Theory of Electrons and its Applications to the Phenomena of
  Light and Radiant Heat}.
\newblock leipzig and Berlin: B. G. Teubner, 1916.

\bibitem{Schott}
G.~A. Schott.
\newblock {\em Electromagnetic Radiation and the Mechanical Reactions arising
  from it}.
\newblock Cambridge University Press, 1912.

\bibitem{Norton}
Andrew~H Norton.
\newblock The alternative to classical mass renormalization for tube-based
  self-force calculations.
\newblock {\em {Classical and Quantum Gravity}}, {26}({10}), {2009}.

\bibitem{Mehra73}
Jagdish Mehra.
\newblock {\em The Physicists Conception of Nature}.
\newblock D. Reidel Publishing Company, 1973.

\bibitem{Teitelboim70}
Claudio Teitelboim.
\newblock Splitting of the maxwell tensor: Radiation reaction without advanced
  fields.
\newblock {\em Phys. Rev. D}, 1(6):1572--1582, Mar 1970.

\bibitem{Bhabha39}
H.~J. Bhabha.
\newblock Classical theory of mesons.
\newblock {\em Proceedings of the Royal Society of London. Series A,
  Mathematical and Physical Sciences}, 172(950):pp. 384--409, 1939.

\bibitem{Tucker06}
D.A. Burton, J.~Gratus, and R.W. Tucker.
\newblock Asymptotic analysis of ultra-relativistic charge.
\newblock {\em Annals of Physics}, 322(3):599 -- 630, 2007.

\bibitem{Rowe75}
E.~G.~Peter Rowe.
\newblock Resolution of an ambiguity in the derivation of the lorentz-dirac
  equation.
\newblock {\em Phys. Rev. D}, 12(6):1576--1587, Sep 1975.

\bibitem{Yokoya90}
K.~Yokoya.
\newblock Impedence of slowly tapered structures.
\newblock {\em CERN Report}, (SL-90-88-AP), 1990.

\bibitem{Warnock93}
R.~L. Warnock.
\newblock An intergo-algebraic equation for high frequency wake fields in a
  tube with smoothly varying radius.
\newblock {\em SLAC Report}, (SLAC-PUB-6038), 1993.

\bibitem{Stupakov96}
G.~V. Stupakov.
\newblock Geometrical wake of a smooth flat collimator.
\newblock {\em SLAC Report}, (SLAC-PUB-7167), 1996.

\bibitem{Stupakov01}
G.~V. Stupakov.
\newblock Impedance of small-angle collimators in high-frequency limit.
\newblock {\em SLAC Report}, (SLAC-PUB-8857), 2001.

\bibitem{Stupakov07}
G.~Stupakov.
\newblock Low frequency impedance of tapered transitions with arbitrary cross
  sections.
\newblock {\em Phys. Rev. ST Accel. Beams}, 10(9):094401, Sep 2007.

\bibitem{Bane07}
K.~L.~F. Bane, G.~Stupakov, and I.~Zagorodnov.
\newblock Impedance calculations of nonaxisymmetric transitions using the
  optical approximation.
\newblock {\em Phys. Rev. ST Accel. Beams}, 10(7):074401, Jul 2007.

\bibitem{Bane10}
G.~Stupakov, K.~L.~F. Bane, and I.~Zagorodnov.
\newblock Impedance scaling for small angle transitions.
\newblock {\em Phys. Rev. ST Accel. Beams}, 14(1):014402, Jan 2011.

\bibitem{Podobedov06}
B.~Podobedov and S.~Krinsky.
\newblock Transverse impedance of axially symmetric tapered structures.
\newblock {\em Phys. Rev. ST Accel. Beams}, 9(5):054401, May 2006.

\bibitem{Podobedov07}
B.~Podobedov and S.~Krinsky.
\newblock Transverse impedance of tapered transitions with elliptical cross
  section.
\newblock {\em Phys. Rev. ST Accel. Beams}, 10(7):074402, Jul 2007.

\bibitem{Smith11}
J.~D.~A. Smith.
\newblock {\em Calculations of Collimator Wakefield}.
\newblock PhD thesis, Lancaster University, UK, 2011.

\bibitem{Jackson99}
J~D Jackson.
\newblock {\em Classical Electrodynamics (3rd Edition)}.
\newblock Wiley, 1999.

\bibitem{GotoTucker}
Shin itiro Goto and Robin~W Tucker.
\newblock Electromagnetic fields produced by moving sources in a curved beam
  pipe.
\newblock {\em Journal of Mathematical Physics}, 50(6):063510, 2009.

\bibitem{BennTucker}
I.M. Benn and R.~W. Tucker.
\newblock {\em An Introduction to Spinors and Geometry with Applications in
  Physics}.
\newblock Adam Hilger, Bristol and New York, 1987.

\bibitem{Frankel03}
Theodore Frankel.
\newblock {\em The Geometry of Physics}.
\newblock Cambridge University Press, 2003.

\bibitem{Infeld39}
L.~Infeld and P.~R. Wallace.
\newblock The equations of motion in electrodynamics.
\newblock {\em Phys. Rev.}, 57:797--806, May 1940.

\bibitem{Havas48}
Peter Havas.
\newblock On the classical equations of motion of point charges.
\newblock {\em Phys. Rev.}, 74:456--463, Aug 1948.

\bibitem{Manifolds}
Robin~W Tucker and Charles Wang.
\newblock Manifolds: A maple package for differential geometry, 1996.

\end{thebibliography}

\end{backmatter}
\end{document}